%% file: cryst-arxiv-v2.tex
\documentclass[sigplan,authorversion,screen]{acmart}
                                                                          
\AtBeginDocument{%
  \providecommand\BibTeX{{%
    \normalfont B\kern-0.5em{\scshape i\kern-0.25em b}\kern-0.8em\TeX}}}
\acmConference%[Submitted version (slightly adapted)]
              {Report version (submitted version, corrections performed)}
\acmYear{2022.}
\acmISBN{} % \acmISBN{978-x-xxxx-xxxx-x/YY/MM}
\acmDOI{} % \acmDOI{10.1145/nnnnnnn.nnnnnnn}
\startPage{1}

\setcopyright{rightsretained}
\input{macros.tex}
\usepackage{afterpage}
\usepackage{float}
\usepackage{bussproofs}
\usepackage{calc}
\usepackage{crossreftools} % https://tex.stackexchange.com/questions/405645/referencing-custom-label-with-enumitem
    \makeatletter
    \newcommand{\optionaldesc}[2]{%
      \phantomsection
      #1\protected@edef\@currentlabel{#1}\label{#2}%
    }
    \makeatother
\usepackage{placeins}
\usepackage{mathtools}
\usepackage{moreenum}
\usepackage{soul}   
  \let\emptyset\relax
  \let\savebigtimes\bigtimes
  \let\bigtimes\relax
  \usepackage{mathabx}
  \let\bigtimes\savebigtimes
  
  \let\emptyset\saveemptyset
\usepackage{graphicx}

\usepackage{enumitem}
  
  \let\degree\relax
  \let\savebigtimes\bigtimes
  \let\bigtimes\relax
  \usepackage{mathabx}
  \let\bigtimes\savedegree
\usepackage{tikz,fp,tikz-cd}
\usetikzlibrary{arrows.meta,arrows,calc,automata,angles,quotes,matrix,%backgrounds,calc,chains,fadings,fit,patterns,snakes,
                positioning,through,shapes.geometric,shapes.symbols,shapes.multipart,through,trees,
                decorations.markings,decorations.text,
                }                

\pgfdeclaredecoration{funbisim}{final}{
  \state{final}[width=\pgfdecoratedpathlength]{
    \draw[->]
      (0,\pgfdecorationsegmentamplitude+0.35mm) -- +(\pgfdecoratedpathlength,0);
    \draw[-]
      (0,-\pgfdecorationsegmentamplitude-0.35mm) -- +(\pgfdecoratedpathlength,0);}}
\tikzset{
  funbisim/.style={
    decoration={funbisim, amplitude=0.25ex},
    decorate,
    funbisim options/.style={#1}    
  }}

\pgfdeclaredecoration{funbisimright}{final}{
  \state{final}[width=\pgfdecoratedpathlength]{
    \draw[->]
      (0,\pgfdecorationsegmentamplitude+0.35mm) -- +(\pgfdecoratedpathlength,0);
    \draw[-]
      (0,-\pgfdecorationsegmentamplitude-0.35mm) -- +(\pgfdecoratedpathlength,0);}}
\tikzset{
  funbisimright/.style={
    decoration={funbisimright, amplitude=0.25ex},
    decorate,
    funbisimright options/.style={#1}    
  }}

\pgfdeclaredecoration{funbisimleft}{final}{
  \state{final}[width=\pgfdecoratedpathlength]{
    \draw[->]
      (0,-\pgfdecorationsegmentamplitude-0.35mm) -- +(\pgfdecoratedpathlength,0);
    \draw[-]
      (0,\pgfdecorationsegmentamplitude+0.35mm) -- +(\pgfdecoratedpathlength,0);}}
\tikzset{
  funbisimleft/.style={
    decoration={funbisimleft, amplitude=0.25ex},
    decorate,
    funbisimleft options/.style={#1}    
  }}

\pgfdeclaredecoration{bisim}{final}{
  \state{final}[width=\pgfdecoratedpathlength]{
    \draw[<->]
      (0,\pgfdecorationsegmentamplitude+0.35mm) -- +(\pgfdecoratedpathlength,0);
    \draw[-]
      (0,-\pgfdecorationsegmentamplitude-0.35mm) -- +(\pgfdecoratedpathlength,0);}}
\tikzset{
  bisim/.style={
    decoration={bisim, amplitude=0.25ex},
    decorate,
    bisim options/.style={#1}    
  }}

\pgfdeclaredecoration{funonebisim}{final}{
  \state{final}[width=\pgfdecoratedpathlength]{
    \draw[->,magenta]
      (0,\pgfdecorationsegmentamplitude+0.35mm) -- +(\pgfdecoratedpathlength,0);
    \draw[-,magenta]
      (0,-\pgfdecorationsegmentamplitude-0.35mm) -- +(\pgfdecoratedpathlength,0);}}
\tikzset{
  funonebisim/.style={
    decoration={funonebisim, amplitude=0.25ex},
    decorate,
    funonebisim options/.style={#1}    
  }}

\pgfdeclaredecoration{funonebisimdashed}{final}{
  \state{final}[width=\pgfdecoratedpathlength]{
    \draw[->,densely dashed,magenta]
      (0,\pgfdecorationsegmentamplitude+0.35mm) -- +(\pgfdecoratedpathlength,0);
    \draw[-,densely dashed,magenta]
      (0,-\pgfdecorationsegmentamplitude-0.35mm) -- +(\pgfdecoratedpathlength,0);}}
\tikzset{
  funonebisimdashed/.style={
    decoration={funonebisimdashed, amplitude=0.25ex},
    decorate,
    funonebisimdashed options/.style={#1}    
  }}

\pgfdeclaredecoration{funonebisimright}{final}{
  \state{final}[width=\pgfdecoratedpathlength]{
    \draw[->,magenta]
      (0,\pgfdecorationsegmentamplitude+0.35mm) -- +(\pgfdecoratedpathlength,0);
    \draw[-,magenta]
      (0,-\pgfdecorationsegmentamplitude-0.35mm) -- +(\pgfdecoratedpathlength,0);}}
\tikzset{
  funonebisimright/.style={
    decoration={funonebisimright, amplitude=0.25ex},
    decorate,
    funonebisimright options/.style={#1}    
  }}

\pgfdeclaredecoration{funonebisimleft}{final}{
  \state{final}[width=\pgfdecoratedpathlength]{
    \draw[->,magenta]
      (0,-\pgfdecorationsegmentamplitude-0.35mm) -- +(\pgfdecoratedpathlength,0);
    \draw[-,magenta]
      (0,\pgfdecorationsegmentamplitude+0.35mm) -- +(\pgfdecoratedpathlength,0);}}
\tikzset{
  funonebisimleft/.style={
    decoration={funonebisimleft, amplitude=0.25ex},
    decorate,
    funonebisimleft options/.style={#1}    
  }}

\pgfdeclaredecoration{onebisim}{final}{
  \state{final}[width=\pgfdecoratedpathlength]{
    \draw[<->,magenta]
      (0,\pgfdecorationsegmentamplitude+0.35mm) -- +(\pgfdecoratedpathlength,0);
    \draw[-,magenta]
      (0,-\pgfdecorationsegmentamplitude-0.35mm) -- +(\pgfdecoratedpathlength,0);}}
\tikzset{
  onebisim/.style={
    decoration={onebisim, amplitude=0.25ex},
    decorate,
    onebisim options/.style={#1}    
  }}

\makeatletter
\def\calcLength(#1,#2)#3{%
\pgfpointdiff{\pgfpointanchor{#1}{center}}%
             {\pgfpointanchor{#2}{center}}%
\pgf@xa=\pgf@x%
\pgf@ya=\pgf@y%
\FPeval\@temp@a{\pgfmath@tonumber{\pgf@xa}}%
\FPeval\@temp@b{\pgfmath@tonumber{\pgf@ya}}%
\FPeval\@temp@sum{(\@temp@a*\@temp@a+\@temp@b*\@temp@b)}%
\FProot{\FPMathLen}{\@temp@sum}{2}%
\FPround\FPMathLen\FPMathLen5\relax
\global\expandafter\edef\csname #3\endcsname{\FPMathLen}
}
\makeatother

\DeclareRobustCommand*{\mycirc}[1]{%
  \tikz[baseline=(C.base)]
    \node[draw,circle,inner sep=1pt](C){#1};}

\tikzset{
  my dash/.style={dash pattern=on 5pt off 2pt}
         }% end of tikzset

\newcommand\ellipsebyfoci[4]{% options, focus pt1, focus pt2, cste
  \path[#1] let \p1=(#2), \p2=(#3), \p3=($(\p1)!.5!(\p2)$)
  in \pgfextra{
    \pgfmathsetmacro{\angle}{atan2(\y2-\y1,\x2-\x1)}
    \pgfmathsetmacro{\focal}{veclen(\x2-\x1,\y2-\y1)/2/1cm}
    \pgfmathsetmacro{\lentotcm}{\focal*2*#4}
    \pgfmathsetmacro{\axeone}{(\lentotcm - 2 * \focal)/2+\focal}
    \pgfmathsetmacro{\axetwo}{sqrt((\lentotcm/2)*(\lentotcm/2)-\focal*\focal}
  }
  (\p3) ellipse[x radius=\axeone cm,y radius=\axetwo cm, rotate=\angle];
} 

\setcopyright{none}
\begin{document}

%\title{Cutting Twin-Crystals as Near-Collapsed Process Interpretations of Regular Expressions}
\title[Milner's Proof System for Regular Expressions Mod.\ Bisimilarity is Complete]
      {Milner's Proof System for Regular Expressions Modulo Bisimilarity is Complete}
\subtitle{Crystallization: Near-Collapsing Process Graph Interpretations of Regular Expressions}

\author{Clemens Grabmayer}
\email{clemens.grabmayer@gssi.it}
\orcid{0000-0002-2414-1073}
%\email{webmaster@marysville-ohio.com}
\affiliation{%
  \institution{Gran Sasso Science Institute}
  \streetaddress{P.O. Box 1212}
  \city{L'Aquila}
  \state{Abruzzo}
  \country{Italy}
  \postcode{67100 AQ}
}

% \affiliation{%

% \affiliation{%

% \affiliation{%

% \affiliation{%

% \affiliation{%

% \affiliation{%

% \affiliation{%

\renewcommand{\shortauthors}{C. Grabmayer}

\begin{abstract}
  \input{abstract-short.tex}

\end{abstract}


% \keywords{regular expressions, process algebra, bisimilarity, process graphs, complete proof system} 

%   \includegraphics[width=\textwidth]{acmart/samples/sampleteaser}

\maketitle

\raggedbottom
%\flushbottom

%--------------------------
% article
%--------------------------
\input{cryst-arxiv-art.tex}

%--------------------------

%------------------
%% Acknowledgments
%------------------ 
\begin{acks}
   {Wan Fokkink} introduced me to Milner's questions from \cite{miln:1984} in 2005.
    This facilitated my work on them, and led me to the decision result \cite{baet:corr:grab:2007} with Jos Baeten and Flavio Corradini.
  In 2015, Wan suggested to tackle the problem together
    by looking for minimization strategies for star expressions that he had pioneered in \cite{fokk:zant:1994,fokk:1996:term:cycle:LGPS,fokk:1997:pl:ICALP}.
  For that it led me from the structure constraints of LEE and layered LEE for process graphs
    (generalizing `well-behaved specifications' of processes by Flavio Corradini)
      to an idea for tackling Milner's axiomatization question in full generality,
        weekly meetings with Wan in 2015--2018 were crucial.
  Together with Wan's research visit to GSSI in 2019, these meetings
    led us to the completeness result \cite{grab:fokk:2020:lics} for the tailored restriction \BBP\ of Milner's system to `\onefree\ star expressions'.
  While I continued further on this path by myself for the past two years, % (caused by the physical distance, and the Covid-19 pandemic),
    my work with Wan has formed the firm basis of~this~work.

  I am very thankful to {Luca Aceto} for comments on parts of this submission and previous ones,
    and for giving me the chance to complete this work within the PRIN project
    \emph{{\nf IT MATTERS} - Methods and Tools for Trustworthy Smart Systems}
    (project ID: \text{2017FTXR7S\_005}).

  I want to thank {Emilio Tuosto} very much for discussing aspects of my work with him,
    for the chance to speak about it in the FM/SE meetings he organized at GSSI,
      and also for quick help, repeatedly.

  For their help during the TAPS-typesetting process of the LICS conference article,
    I want to thank Vincent van Oostrom for crucial advice about how to proceed in order to overcome an impasse, 
    and Alessandro Aloisio for an idea that helped me locate an unspecific, wrongly located typesetting error.

  Finally I want to thank the {anonymous reviewers} of the LICS submission for their interest and thorough reading,
    and a conceptualizing summary. %, and the note of caution (see on pages~\pageref{note:of:caution:start}--\pageref{note:of:caution:end}).
    Their comments signaled the need to make details in Sect.~\ref{completeness:proof} better accessible,
      and provided me with an idea to visually improve~Fig.~\ref{fig:proof:structure} by using a `pyramid' structure.
\end{acks}

%-------------
% bibliography
%-------------
%\clearpage
\bibliographystyle{ACM-Reference-Format}
\bibliography{cryst}

%-------------------------------
% appendix
%-------------------------------
\input{cryst-arxiv-app-v1v2.tex}

\end{document}

%% file: macros.tex
\newtheorem{thm}{Theorem}[section]
  \counterwithin{thm}{section}

\newtheorem*{prop*}{Proposition}
\newtheorem{prop}[thm]{Proposition}
\newtheorem{lem}[thm]{Lemma}
\newtheorem*{lem*}{Lemma}
\newtheorem*{thm*}{Theorem}
\newtheorem{cor}[thm]{Corollary}
\newtheorem{cor*}{Corollary}
\theoremstyle{remark}

\newtheorem*{rem*}{Remark}
\theoremstyle{definition}
\newtheorem{defi}[thm]{Definition}
\newtheorem*{defi*}{Definition}
\newtheorem{exa}[thm]{Example}

\definecolor{azure}{rgb}{0.94,1.00,1.00}
\definecolor{brown}{rgb}{.75,.25,.25}
\definecolor{cyan}{rgb}{0.25,0.88,0.82}
\definecolor{chocolate}{rgb}{0.82,0.41,0.12}
\definecolor{darkcyan}{rgb}{0.5,0,1}
\definecolor{darkgreen}{rgb}{0,0.39,0}
\definecolor{darkmagenta}{rgb}{0.5,0,0.5}
\definecolor{darkgoldenrod}{RGB}{184,134,11}
\definecolor{firebrick}{RGB}{175,25,25}
\definecolor{forestgreen}{rgb}{0.13,0.55,0.13}
\definecolor{goldenrod}{RGB}{218,165,32}
\definecolor{grey}{RGB}{128,128,128}
\definecolor{lightcyan}{rgb}{0.88,1.00,1.00}
\definecolor{lightpink}{rgb}{1.00,0.71,0.76}
\definecolor{myyellow}{RGB}{235,235,0}
\definecolor{lightyellow}{rgb}{1.00,1.00,0.88}
\definecolor{lightgoldenrod}{rgb}{0.83,0.97,0.51}
\definecolor{lightgoldenrodyellow}{rgb}{0.98,0.98,0.82}
\definecolor{lightskyblue}{rgb}{0.53,0.81,0.98}
\definecolor{moccasin}{rgb}{1.00,0.89,0.71}
\definecolor{magenta}{rgb}{1,0,1}
\definecolor{navyblue}{rgb}{0,0,0.5}
\definecolor{orange}{rgb}{1.0,0.65,0.0}
\definecolor{orangered}{rgb}{1.0,0.27,0.0}
\definecolor{palegreen}{rgb}{0.60,0.98,0.60}
\definecolor{powderblue}{rgb}{0.69,0.88,0.90}
\definecolor{purple}{rgb}{1,0.5,1}
\definecolor{royalblue}{RGB}{65,105,225}
\definecolor{mediumblue}{RGB}{0,0,205}
\definecolor{cornflowerblue}{RGB}{100,149,237}
\definecolor{springgreen}{rgb}{0.0,1.0,0.5}
\definecolor{turquoise}{rgb}{0.25,0.88,0.82}
\definecolor{snow}{rgb}{1.00,0.98,0.98}
\definecolor{tan}{rgb}{0.82,0.71,0.55}
\definecolor{red}{rgb}{1,0,0}
\definecolor{violetred}{RGB}{208,32,144}

\newcommand{\colorin}[1]{\textcolor{#1}}
\newcommand{\alert}{\colorin{red}}
\newcommand{\black}{\colorin{black}}

\newcommand{\chocolate}{\colorin{chocolate}}

\newcommand{\darkcyan}{\colorin{darkcyan}}

\newcommand{\forestgreen}{\colorin{forestgreen}}
\newcommand{\firebrick}{\colorin{firebrick}}

\newcommand{\magenta}{\colorin{magenta}}
\newcommand{\mediumblue}{\colorin{mediumblue}}

\newcommand{\royalblue}{\colorin{royalblue}}

\newcommand{\nb}{\nobreakdash}
\newcommand{\punc}[1]{\ensuremath{\hspace*{1.5pt}{#1}}}

\newcommand{\nf}{\normalfont}
\newcommand{\textnf}[1]{\text{\nf{#1}}}

\newcommand{\bs}[1]{\boldsymbol{#1}}

\newcommand{\funin}{\mathrel{:}}
\newcommand{\fap}[2]{{#1}(\hspace*{-0.5pt}{#2}\hspace*{-0.5pt})}

\newcommand{\bfap}[3]{{#1}({#2},\hspace*{0.5pt}{#3})}

\newcommand{\iap}[2]{#1_{#2}}
\newcommand{\bap}[2]{#1_{#2}}
\newcommand{\pap}[2]{#1^{#2}}
\newcommand{\bpap}[3]{#1_{#2}^{#3}}
\newcommand{\pbap}[3]{#1_{#3}^{#2}}

\newcommand{\sidrelon}{\iap{\textrm{\nf id}}}

\newcommand{\sproj}{\pi}
\newcommand{\sproji}{\iap{\sproj}}
\newcommand{\proji}[1]{\fap{\sproji{#1}}}

\newcommand{\sdefdby}{{:=}}
\newcommand{\defdby}{\mathrel{\sdefdby}}

\newcommand\tuple[1]{\langle #1 \rangle}

\newcommand\tuplespace{\hspace*{0.5pt}}
\newcommand\pair[2]{\tuple{#1, \tuplespace #2}}
\newcommand\triple[3]{\tuple{#1, \tuplespace #2, \tuplespace #3}}

\newcommand{\nat}{\mathbb{N}} %{{\iap{\natplus}{0}}}
\newcommand{\natplus}{\pap{\nat}{+}} %{{>}0}} {\mathbb{N}}

\newcommand{\sgraphof}{\textit{graph}}
\newcommand{\graphof}{\fap{\sgraphof}}

\newcommand{\BNFor}{\:\mid\:}
\newcommand{\BNFdefdby}{\:{::=}\:}

\newcommand{\family}[2]{\setexp{#1}_{#2}}

\newcommand{\sectionto}[2]{({#2}\,{#1})}

\newcommand{\eqcl}[1]{\iap{\left[{#1}\right]}}

\newcommand{\ssynteq}{{\mediumblue{=}}} %{{\equiv}} %{{=}}%{{\isnewest{\iap{=}{\eqlogic}}}}%{{\equiv}}
\newcommand{\synteq}{\mathrel{\ssynteq}} %{\mathrel{\ssynteq}}              

\newcommand{\sformeq}{{=}}%{{\approx}}
\newcommand{\formeq}{\mathrel{\sformeq}}

\newcommand{\sred}{{\to}}

\newcommand{\sredi}[1]{{\iap{\sred}{#1}}}

\newcommand{\seqrel}{{\pmb{=}}}

\newcommand{\sconvred}{{\leftarrow}}

\newcommand{\sredrsci}[1]{{\pbap{\leftrightarrow}{=}{#1}}}

\newcommand{\sredtc}{\sred^{+}}

\newcommand{\sredtci}[1]{{\iap{\sredtc}{#1}}}
\newcommand{\redtci}[1]{\mathrel{\sredtci{#1}}}
\newcommand{\sredrtc}{\sred^{*}}

\newcommand{\sredrtci}[1]{{\iap{\sredrtc}{#1}}}

\newcommand{\sredb}[1]{{\bap{\sred}{#1}}}

\newcommand{\sredbp}[2]{{\bpap{\sred}{#1}{#2}}}
\newcommand{\sconvredb}[1]{{\bap{\sconvred}{#1}}}

\newcommand{\scomprewrels}[2]{{#1}\cdot{#2}}
\newcommand{\comprewrels}[2]{\mathrel{\scomprewrels{#1}{#2}}}

\newcommand{\descsetexpmid}{\mathrel{\vert}}
\newcommand{\descsetexpbigmid}{\mathrel{\big\vert}}

\newcommand{\descsetexp}[2]{\left\{{#1}\descsetexpmid{#2}\right\}}

\newcommand{\descsetexpbig}[2]{\bigl\{{#1}\descsetexpbigmid{#2}\bigr\}}

\newcommand{\sdom}{\text{\nf dom}}
\newcommand{\domof}{\fap{\sdom}}
\newcommand{\sactdom}{\text{\nf dom}_{\text{\nf act}}}
\newcommand{\actdomof}{\fap{\sactdom}}

\newcommand{\sactcodom}{\text{\nf cod}_{\text{\nf act}}}
\newcommand{\actcodomof}{\fap{\sactcodom}}

\newcommand{\sran}{\text{\nf ran}}
\newcommand{\ranof}{\fap{\sran}}

\newcommand{\dashedmapsto}{\begin{tikzpicture}[scale=0.8]\draw[|->,densely dashed,thick] (0,0) to (0.62,0);\end{tikzpicture}}

\newcommand{\picarrowstart}{\raisebox{2pt}{\begin{tikzpicture}%
                                             \draw[<-,very thick,>=latex,chocolate,shorten <=2pt](0,0) -- ++ (180:{12pt});%
                                           \end{tikzpicture}}}
  
\newcommand{\pictermvert}{\begin{tikzpicture}%
                           \node[draw,chocolate,very thick,circle,minimum width=2.5pt,fill,inner sep=0pt,outer sep=2pt](v){};%
                           \draw[thick,chocolate] (v) circle (0.12cm);%
                         \end{tikzpicture}}

\newcommand{\sphifun}{\phi}
\newcommand{\phifun}{\fap{\sphifun}}
\newcommand{\sphifuni}{\iap{\sphifun}}

\newcommand{\scpfun}{{\textit{cp}}}
\newcommand{\scpfunon}[1]{\bap{\scpfun}{\hspace*{-0.25pt}#1}}
\newcommand{\cpfunon}[1]{\fap{\scpfunon{#1}}}

\newcommand{\sliftltfun}{\pap{\overline{\smash{\sphifun}\rule{0pt}{0.4\baselineskip}}}{{\scriptscriptstyle (\firstfloor)}}\hspace*{-2.5pt}}

\newcommand{\sliftltfunn}[1]{\pap{\overline{\smash{#1}\rule{0pt}{0.4\baselineskip}}}{{\scriptscriptstyle (\firstfloor)}}\hspace*{-2.5pt}}

\newcommand{\scompfuns}[2]{{#1}\circ{#2}}
\newcommand{\compfuns}[2]{\fap{\scompfuns{#1}{#2}}}

\newcommand{\sabinrel}{{R}}

\newcommand{\setexp}[1]{\left\{{#1}\right\}}

\renewcommand{\emptyset}{\varnothing}

\newcommand{\slogand}{\wedge}

\newcommand{\logand}{\mathrel{\slogand}}
 
\newcommand{\slognot}{\neg}
\newcommand{\lognot}[1]{\slognot{#1}}

\newcommand{\emptyword}{\epsilon}

\let\depth\relax  % depth is used in raisebox -- to use that in captions this helps https://tex.stackexchange.com/questions/253431/use-raisebox-with-depth-in-caption
\newcommand{\tightfbox}[1]{{\fboxsep=1.5pt\fbox{#1}}}

\newcommand{\backlink}{backlink}

\newcommand{\emptystep}{emp\-ty-step}

\newcommand{\maximalwrt}[1]{${#1}$\nb-max\-i\-mal}

\newcommand{\entrytransition}{en\-try tran\-si\-tion}
\newcommand{\entrytransitions}{\entrytransition{s}}

\newcommand{\reducedwrt}[1]{\mbox{${#1\hspace*{0.5pt}}$}\nb-re\-duced}

\newcommand{\loopentry}{loop-en\-try}
\newcommand{\Loopentry}{Loop-en\-try}

\newcommand{\txtloopsbackto}{loops-back-to}

\newcommand{\loopentrytransition}{loop-entry tran\-si\-tion}
\newcommand{\loopentrytransitions}{\loopentrytransition{s}}

\newcommand{\entrybodylabeling}{en\-try\discretionary{/}{}{/}body-la\-be\-ling}

\newcommand{\LLEEpreserving}{\LLEE\nb-pre\-serving}
\newcommand{\LLEEpreservingly}{\LLEE\nb-pre\-servingly}

\newcommand{\LLEEPreserving}{\LLEE\nb-Pre\-serving}

\newcommand{\LLEEwitness}{{\protect\nf LLEE}-wit\-ness}
\newcommand{\LLEEwitnesses}{{\protect\nf LLEE}-wit\-nesses}

\newcommand{\LLEEchart}{{\nf LLEE}\nb-chart}

\newcommand{\LLEEcharts}{{\nf LLEE}\nb-chart{s}}

\newcommand{\LLEEonechart}{{\nf LLEE}\hspace*{1pt}-\onechart}

\newcommand{\LLEEonecharts}{{\nf LLEE}\hspace*{1pt}-\onechart{s}}

\newcommand{\LLEEoneLTS}{{\nf LLEE}\hspace*{1pt}\nb-\oneLTS}

\newcommand{\LLEEoneLTSs}{{\nf LLEE}\hspace*{1pt}-\oneLTS{s}}

\newcommand{\LLEEonelim}{\text{\nf LLEE--$\sone$-lim}}

\newcommand{\localtransfer}{lo\-cal-trans\-fer}

\newcommand{\nonempty}{non-emp\-ty}

\newcommand{\nontrivial}{non-triv\-i\-al}

\newcommand{\onetransition}{$\sone$\nb-tran\-si\-tion}
\newcommand{\onetransitions}{\oneTransition{s}}
\newcommand{\oneTransition}{$\sone$\nb-tran\-si\-tion}

\newcommand{\onefree}{$\sone$\nb-free}

\newcommand{\stexponefree}{$\stexpone$\nb-free}

\newcommand{\onechart}{$\sone$\nb-chart}

\newcommand{\onecharts}{\onechart{s}}

\newcommand{\onecollapsible}{$\sone$\nb-col\-laps\-ible}

\newcommand{\onebisimulation}{$\sone$\nb-bi\-si\-mu\-la\-tion}
\newcommand{\onebisimulations}{\onebisimulation{s}}

\newcommand{\onebisimilar}{$\sone$\nb-bi\-si\-mi\-lar}
\newcommand{\oneBisimilar}{$\sone$\nb-Bi\-si\-mi\-lar}

\newcommand{\onebisimilarity}{$\sone$\nb-bi\-si\-mi\-la\-ri\-ty}

\newcommand{\onebisimulating}{$\sone$\nb-bi\-si\-mu\-la\-ting}

\newcommand{\onecollapsed}{$\sone$\nb-col\-lapsed}

\newcommand{\perpetual}{per\-pet\-u\-al}
\newcommand{\perpetualloop}{\perpetual-loop}

\newcommand{\precrystalline}{pre\-crys\-tal\-line}

\newcommand{\properaction}{prop\-er-ac\-tion}

\newcommand{\provablein}[1]{{$#1$}\nb-pro\-vable}
\newcommand{\Provablein}[1]{{$#1$}\nb-Pro\-vable}
\newcommand{\provablyin}[1]{{$#1$}\nb-pro\-vably}

\newcommand{\nearcollapse}{near-col\-lapse}
\newcommand{\Nearcollapse}{Near-Col\-lapse}
\newcommand{\nearcollapsed}{\nearcollapse{d}}
\newcommand{\Nearcollapsed}{\Nearcollapse{d}}

\newcommand{\reflexivesymmetric}{re\-flex\-ive-sym\-met\-ric}

\newcommand{\selfinverse}{self-in\-verse}

\newcommand{\subonechart}{sub-$\sone$\nb-chart}

\newcommand{\subonecharts}{\subonechart{s}}

\newcommand{\completewrt}[1]{${#1}$\nb-com\-plete}
\newcommand{\Completewrt}[1]{${#1}$\nb-Com\-plete}

\newcommand{\transitionclosed}{tran\-si\-tion-closed}

\newcommand{\txtconnectthrough}{con\-nect-through}

\newcommand{\transitionact}[1]{{${#1}$}\nb-tran\-si\-tion}

\newcommand{\twincrystal}{twin-crys\-tal}
\newcommand{\twincrystals}{\twincrystal{s}}
\newcommand{\Twincrystal}{Twin-Crys\-tal}
\newcommand{\Twincrystals}{\Twincrystal{s}}

\newcommand{\LTS}{LTS}

\newcommand{\oneLTS}{$\sone$\nb-\LTS}
\newcommand{\oneLTSs}{\oneLTS{s}}

\newcommand{\aLTS}{\mathcal{L}}
\newcommand{\aoneLTS}{\underline{\smash{\aLTS}}\rule{0pt}{0.575\baselineskip}}

\newcommand{\aoneLTSi}{\iap{\aoneLTS}}

\newcommand{\oneLTSof}{\fap{\aoneLTS}}

  \newcommand{\subscriptindtrans}{\mediumblue{\scriptscriptstyle\pmb{\otind{\cdot}}}}

\newcommand{\subLTS}{sub\nb-\LTS}

\newcommand{\suboneLTS}{sub\nb-\oneLTS}

\newcommand{\scc}{scc}
\newcommand{\sccs}{\scc's}

\newcommand{\thickbar}[1]{\mathbf{\bar{\text{$#1$}}}}
\newcommand{\overlinebar}[1]{\mathbf{\overline{\text{$#1$}}}}

\newcommand{\aspec}{\mathcal{S}}
\newcommand{\aspeci}{\iap{\aspec}}

\newcommand{\astexp}{e}
\newcommand{\bstexp}{f}

\newcommand{\astexpi}{\iap{\astexp}}

\newcommand{\astexpacc}{\astexp'}

\newcommand{\astexpacci}{\iap{\astexpacc}}

\newcommand{\StExp}{\textit{StExp}}
\newcommand{\StExpover}{\fap{\StExp}}

\newcommand{\stexpzero}{0}
\newcommand{\stexpone}{1}

\newcommand{\sstexpit}{\sstar}
\newcommand{\stexpit}[1]{{#1^{\sstexpit}}}

\newcommand{\sstexpprod}{{\cdot}}%{.}
\newcommand{\stexpprod}[2]{{#1}\mathrel{\sstexpprod}{#2}}
\newcommand{\sstexpsum}{+}
\newcommand{\stexpsum}[2]{{#1}\sstexpsum{#2}}
\newcommand{\sone}{\firebrick{1}}

\newcommand{\sstar}{*}

\renewcommand{\prod}{\mathrel{\cdot}}

\newcommand{\sstexpbit}{\circledast} %{\iap{\sstar}{2}}
\newcommand{\stexpbit}[2]{{#1}\hspace*{0.35pt}\pap{}{\sstexpbit}\hspace*{-0.6pt}{#2}}

\newcommand{\looplab}[1]{{\darkcyan{[#1]}}}
\newcommand{\loopsteplab}{\looplab}

\newcommand{\loopnsteplab}[1]{\darkcyan{[{#1}]}}%{\looplab[{#1}]}

\newcommand{\bodylabcol}{\text{\nf\darkcyan{bo}}}
\newcommand{\bodylab}{\bodylabcol}%{\text{\nf bo}}

\newcommand{\bodytransition}{body tran\-si\-tion}

\newcommand{\sprocsem}{P}
\newcommand{\procsem}[1]{\llbracket{#1}\rrbracket_{\sprocsem}}

\newcommand{\slt}[1]{{\xrightarrow{#1}}}
\newcommand{\lt}[1]{\mathrel{\slt{#1}}}

\newcommand{\sterminates}{{\downarrow}}
\newcommand{\terminates}[1]{{#1}{\sterminates}}

\newcommand{\snotterminates}{\ndownarrow}  %{\!\not\,\downarrow} 
\newcommand{\notterminates}[1]{{#1}{\snotterminates}}

\newcommand{\soneterminates}{{\pap{\downarrow}{\hspace*{-1.5pt}\onebrackscript}}}
\newcommand{\oneterminates}[1]{{#1}{\soneterminates}}
\newcommand{\soneterminatesi}[1]{{\bpap{\downarrow}{#1}{\hspace*{-1.5pt}\onebrackscript}}}
\newcommand{\oneterminatesi}[2]{{#2}{\soneterminatesi{#1}}}

\newcommand{\achart}{\mathcal{C}}
\newcommand{\acharti}{\iap{\achart}}

\newcommand{\aonechart}{\underline{\smash{\mathcal{C}}}\rule{0pt}{0.575\baselineskip}}%   %{{\heightof{\text{$\mathcal{C}$}}}}
\newcommand{\aonecharti}{\iap{\aonechart}}
\newcommand{\bonechart}{\underline{\smash{\mathcal{D}}}\rule{0pt}{0.575\baselineskip}}%
\newcommand{\bonecharti}{\iap{\bonechart}}
\newcommand{\conechart}{\underline{\smash{\mathcal{E}}}\rule{0pt}{0.575\baselineskip}}%

\newcommand{\aonechartacc}{\underline{\smash{\mathcal{C}}\hspace*{-0.4pt}'}}
\newcommand{\bonechartacc}{\underline{\mathcal{D}}\hspace*{0pt}'}
\newcommand{\aonechartacci}{\iap{\aonechartacc}}
\newcommand{\bonechartacci}{\iap{\bonechartacc}}
\newcommand{\aonechartdacc}{\underline{\mathcal{C}}''}
\newcommand{\aonechartdacci}{\iap{\aonechartdacc}}
\newcommand{\aonecharthat}{\hspace*{0.2pt}\Hat{\hspace*{-0.75pt}\aonechart}\hspace*{-0.75pt}} %{\hspace*{0.75pt}}
\newcommand{\aonecharthati}[1]{\hspace*{0.2pt}\iap{\Hat{\hspace*{-0.75pt}\aonechart}}{#1}\hspace*{-0.75pt}}
\newcommand{\aonecharthatip}[2]{\hspace*{0.2pt}\bpap{\Hat{\hspace*{-0.75pt}\aonechart}}{#1}{#2}\hspace*{-0.75pt}}
\newcommand{\bonecharthati}[1]{\hspace*{0.2pt}\iap{\Hat{\hspace*{-0.75pt}\bonechart}}{#1}\hspace*{-0.75pt}} 
\newcommand{\conecharthati}[1]{\hspace*{0.2pt}\iap{\Hat{\hspace*{-0.75pt}\conechart}}{#1}\hspace*{-0.75pt}} 
\newcommand{\aonecharthatacc}{\hspace*{0.2pt}\aonecharthat\hspace*{0.4pt}'\hspace*{-0.75pt}}
\newcommand{\aonecharthatacci}[1]{\hspace*{0.2pt}\aonecharthat\hspace*{0.4pt}'_{#1}\hspace*{-0.75pt}}
\newcommand{\aonecharthatdacci}[1]{\hspace*{0.2pt}\aonecharthat\hspace*{0.4pt}''_{#1}\hspace*{-0.75pt}}

\newcommand{\aonechartsubscript}{\underline{\smash{\mathcal{C}}}\rule{0pt}{5pt}}%   %{{\heightof{\text{$\mathcal{C}$}}}}
\newcommand{\bonechartsubscript}{\underline{\smash{\mathcal{D}}}\rule{0pt}{5pt}}%   %{{\heightof{\text{$\mathcal{C}$}}}}
\newcommand{\bonechartaccsubscript}{\underline{\smash{\mathcal{D}'}}\rule{0pt}{5pt}}%   %{{\heightof{\text{$\mathcal{C}$}}}}
\newcommand{\aonecharthatsubscript}{\hat{\aonechartsubscript}} %{\widehat{\aLTS}}%
\newcommand{\bonecharthatsubscript}{\widehat{\bonechartsubscript}} %{\widehat{\aLTS}}%
\newcommand{\bonecharthataccsubscript}{\widehat{\bonechartaccsubscript}} %{\widehat{\aLTS}}%
\newcommand{\bonecharthatsubscripti}[1]{\widehat{\bonechartsubscript}_{#1}} %{\widehat{\aLTS}}%
\newcommand{\bonecharthataccsubscripti}[1]{\widehat{\bonechartsubscript}{}'_{#1}} %{\widehat{\aLTS}}%

\newcommand{\chartof}{\fap{\achart}}

\newcommand{\indchart}[1]{\iap{{#1}}{\subscriptindtrans}}
\newcommand{\indchartof}{\indchart}

\newcommand{\onechartof}{\fap{\aonechart}}

\newcommand{\selevation}[1]{\iap{\text{\nf\sf E}}{\hspace*{-0.25pt}#1\hspace*{-0.5pt}}}
\newcommand{\elevationofabove}[1]{\fap{\selevation{#1}}}
\newcommand{\selevatechart}[1]{\iap{\text{\nf\sf E}}{\hspace*{-0.25pt}#1\hspace*{-0.5pt}}}
\newcommand{\elevatechart}[1]{\fap{\selevatechart{#1}}}
\newcommand{\groundfloor}{\bs{0}}
\newcommand{\firstfloor}{\bs{1}}

\newcommand{\aloop}{\mathcal{LC}}

\newcommand{\aoneloop}{\underline{\smash{\aloop}}}

\newcommand{\indsubonechartinat}[1]{\fap{\aonecharti{#1}}}

\newcommand{\otind}[1]{({#1}]}

\newcommand{\iact}[1]{\mediumblue{(}\hspace*{-0.5pt}{#1}\hspace*{-0.25pt}\mediumblue{]}}
\newcommand{\silt}[1]{\slt{\iact{#1}}}

\newcommand{\silti}[2]{{\xrightarrow{\iact{#1}}}{_{#2}}}
\newcommand{\ilt}[1]{\mathrel{\silt{#1}}}

\newcommand{\ilti}[2]{\mathrel{\silti{#1}{#2}}}

\newcommand{\indtranss}{{\silt{\cdot}}}

\newcommand{\sasol}{s}
\newcommand{\sasoli}{\iap{\sasol}}
\newcommand{\asol}{\fap{\sasol}}
\newcommand{\asoli}[1]{\fap{\iap{\sasol}{#1}}}

\newcommand{\sasoltildei}[1]{\tilde{\sasol}_{#1}}

\newcommand{\sextrsol}{\sasol}
\newcommand{\sextrsolof}{\iap{\sextrsol}}

\newcommand{\sterminatesconst}{\tau}

\newcommand{\sterminatesconstof}{\iap{\sterminatesconst}}

\newcommand{\terminatesconstof}[1]{\fap{\sterminatesconstof{#1}}}

\newcommand{\actions}{A}

\newcommand{\oneactions}{\underline{\actions}}

\newcommand{\aact}{a}
\newcommand{\bact}{b}
\newcommand{\cact}{c}

\newcommand{\aacti}{\iap{\aact}}

\newcommand{\cacti}{\iap{\cact}}

\newcommand{\aoneact}{\firebrick{\underline{a}}}

\newcommand{\aoneacti}[1]{\firebrick{\iap{\aoneact}{#1}}}

\newcommand{\verts}{V}

\newcommand{\start}{\averti{\hspace*{-0.5pt}\text{\nf s}}}
\newcommand{\transs}{{\sred}}%{T}
\newcommand{\exts}{{\sterminates}}%{E}
\newcommand{\termexts}{{\sterminates}}%{F}
\newcommand{\states}{V}%{S}
\newcommand{\alab}{\darkcyan{l}}
\newcommand{\alabi}[1]{\iap{\alab}{\darkcyan{#1}}}

\newcommand{\vertsof}{\fap{\verts\hspace*{-1pt}}}

\newcommand{\vertsi}[1]{\iap{\verts}{\hspace*{-0.25pt}{#1}}}
\newcommand{\statesi}{\iap{\states}}

\newcommand{\starti}[1]{\averti{\text{\nf s},#1}}
\newcommand{\transsi}[1]{\iap{\transs}{\hspace*{-0.25pt}{#1}}}

\newcommand{\termextsi}[1]{\iap{\termexts}{{#1}}}

\newcommand{\asetverts}{W\hspace*{-1pt}}

\newcommand{\asetvertsi}[1]{\iap{\asetverts}{\hspace*{-0.75pt}#1}}

\newcommand{\asettranss}{U}

\newcommand{\csetverts}{P}

\newcommand{\csetvertsi}{\iap{\csetverts}}

\newcommand{\astate}{v}
\newcommand{\bstate}{w}

\newcommand{\astatei}{\iap{\astate}}
\newcommand{\bstatei}{\iap{\bstate}}

\newcommand{\astateacc}{\astate'}
\newcommand{\bstateacc}{\bstate'}

\newcommand{\astateacci}{\iap{\astateacc}}
\newcommand{\bstateacci}{\iap{\bstateacc}}

\newcommand{\avert}{v}
\newcommand{\bvert}{w}
\newcommand{\cvert}{u}
\newcommand{\dvert}{x}
\newcommand{\averti}{\iap{\avert}}
\newcommand{\bverti}{\iap{\bvert}}
\newcommand{\cverti}{\iap{\cvert}}
\newcommand{\dverti}{\iap{\dvert}}

\newcommand{\bvertacc}{\bvert'}
\newcommand{\cvertacc}{\cvert'}

\newcommand{\bvertacci}{\iap{\bvertacc}}
\newcommand{\cvertacci}{\iap{\cvertacc}}

\newcommand{\averttilde}{\tilde{\avert}}

\newcommand{\avertbar}{\overlinebar{\avert}}
\newcommand{\bvertbar}{\overlinebar{\bvert}}
\newcommand{\cvertbar}{\overlinebar{\cvert}}
\newcommand{\dvertbar}{\overlinebar{\dvert}}

\newcommand{\cvertbaracc}{\thickbar{\cvert}'}

\newcommand{\avertbari}{\iap{\avertbar}}
\newcommand{\bvertbari}{\iap{\bvertbar}}
\newcommand{\cvertbari}{\iap{\cvertbar}}
\newcommand{\dvertbari}{\iap{\dvertbar}}

\newcommand{\cvertbaracci}{\iap{\cvertbaracc}}

\newcommand{\atrans}{\tau}

\newcommand{\sstransitions}{T\hspace*{-2.25pt}r}
\newcommand{\transitionsinfrom}[1]{\fap{\iap{\sstransitions}{#1}}}

\newcommand{\gensubchartofby}[2]{{#1}{\downarrow}^{\hspace*{-0.75pt}#2}_{\hspace*{-0.5pt}*}}
\newcommand{\connectthrough}[2]{con\-nect-{$#1$}-through-to-{$#2$}}

\newcommand{\connectthroughonechart}[2]{con\-nect-{$#1$}-through-to-{$#2$} \onechart}

\newcommand{\onetermexts}{\soneterminates}

\newcommand{\sfunbisim}{%
    \setbox0=\hbox{\kern-.1ex{$\rightarrow$}\kern-.1ex}
    \setbox1=\vbox{\hbox{\raise .1ex \box0}\hrule}%
    {\hbox{\kern.05ex\box1\kern.1ex}}
  }
\newcommand{\funbisim}{\hspace*{-0pt}\mathrel{\sfunbisim}}

\newcommand{\sshaftfunbisim}{
    \setbox0=\hbox{\kern-.1ex{{---}}\kern-.1ex}
    \setbox1=\vbox{\hbox{\raise .1ex \box0}\hrule}%
    {\hbox{\kern.05ex\box1\kern.1ex}}
  }

\newcommand{\sconvfunbisim}{%
    \setbox0=\hbox{\kern-.1ex{$\leftarrow$}\kern-.1ex}
    \setbox1=\vbox{\hbox{\raise .1ex \box0}\hrule}%
    {\hbox{\kern.05ex\box1\kern.1ex}}
  }

\newcommand{\sbisim}{%
    \setbox0=\hbox{\kern-.1ex{$\leftrightarrow$}\kern-.1ex}
    \setbox1=\vbox{\hbox{\raise .1ex \box0}\hrule}%
    \hbox{\kern.1ex\box1\kern.1ex}
  }
\newcommand{\bisim}{\mathrel{\sbisim\hspace*{1pt}}}
\newcommand{\sbisimsubscript}{\scalebox{0.75}{$\sbisim$}}

\newcommand{\sfunonebisim}{%
    \setbox0=\hbox{\kern-.1ex{\magenta{$\rightarrow$}}\kern-.1ex}
    \setbox1=\vbox{\hbox{\raise .1ex \box0}\hrule}%
    \ensuremath{\magenta{\hbox{\kern.05ex\box1\kern.1ex}}}
  }
\newcommand{\funonebisim}{\mathrel{\sfunonebisim}}

\newcommand{\sconvfunonebisim}{%
    \setbox0=\hbox{\kern-.1ex{$\leftarrow$}\kern-.1ex}
    \setbox1=\vbox{\hbox{\raise .1ex \box0}\hrule}%
    \ensuremath{\magenta{\hbox{\kern.05ex\box1\kern.1ex}}}
  }
\newcommand{\convfunonebisim}{\mathrel{\sconvfunonebisim}}

\newcommand{\sonebisim}{%
    \setbox0=\hbox{\kern-.1ex{$\leftrightarrow$}\kern-.1ex}
    \setbox1=\vbox{\hbox{\raise .1ex \box0}\hrule}%
    \ensuremath{\magenta{\hbox{\kern.1ex\box1\kern.1ex}}}
  }
\newcommand{\onebisim}{\mathrel{\sonebisim}}

    \newcommand{\sfunbisimos}{%
        \setbox0=\hbox{\kern-.1ex{$\rightarrow$}\kern-.1ex}
        \setbox1=\vbox{\hbox{\raise .1ex \box0}\hrule}%
        {\pap{\hbox{\kern.05ex\box1\kern.1ex}}{\hspace*{0.5pt}\subotr}}
      }

    \newcommand{\sconvfunbisimos}{%
        \setbox0=\hbox{\kern-.1ex{$\leftarrow$}\kern-.1ex}
        \setbox1=\vbox{\hbox{\raise .1ex \box0}\hrule}%
        {\pap{\hbox{\kern.05ex\box1\kern.1ex}}{\hspace*{0.5pt}\subotr}}
      }

    \newcommand{\sbisimos}{%
        \setbox0=\hbox{\kern-.1ex{$\leftrightarrow$}\kern-.1ex}
        \setbox1=\vbox{\hbox{\raise .1ex \box0}\hrule}%
        \ensuremath{\pap{\mathrel{\hbox{\kern.1ex\box1\kern.1ex}}}{\hspace*{0.5pt}\subotr}}
      }

\newcommand{\sfunonebisimvia}[1]{%
    \setbox0=\hbox{\kern-.1ex{$\rightarrow$}\kern-.1ex}
    \setbox1=\vbox{\hbox{\raise .1ex \box0}\hrule}%
    {\bap{\magenta{\hbox{\kern.05ex\box1\kern.1ex}}}{#1}}
  }

\newcommand{\sconvfunonebisimvia}[1]{%
    \setbox0=\hbox{\kern-.1ex{$\leftarrow$}\kern-.1ex}
    \setbox1=\vbox{\hbox{\raise .1ex \box0}\hrule}%
    {\bap{\magenta{\hbox{\kern.05ex\box1\kern.1ex}}}{#1}}
  }

\newcommand{\sonebisimvia}[1]{%
    \setbox0=\hbox{\kern-.1ex{$\leftrightarrow$}\kern-.1ex}
    \setbox1=\vbox{\hbox{\raise .1ex \box0}\hrule}%
    {\bap{\magenta{\hbox{\kern.05ex\box1\kern.1ex}}}{\hspace*{-1.5pt}#1}}
  }
\newcommand{\onebisimvia}[1]{\mathrel{\sonebisimvia{#1}}}  

\newcommand{\sonebisimi}{\sonebisimvia}

\newcommand{\sonebisimon}{\sonebisimi}
\newcommand{\onebisimon}[1]{\mathrel{\sonebisimon{#1}}}
\newcommand{\sisonesubstateof}{{\magenta{\sqsubseteq}}} %{\pap{\sqsubseteq}{\subotr}}
\newcommand{\isonesubstateof}{\mathrel{\sisonesubstateof}}
\newcommand{\sisonesubstateofin}[1]{{\bap{\magenta{\sqsubseteq}}{#1}}}
\newcommand{\isonesubstateofin}[1]{\mathrel{\sisonesubstateofin{#1}}}
\newcommand{\abisim}{B}

\newcommand{\abisimdoublebar}{\abisim^{=}} %{{\overline{\overline{\abisim}}}}

\newcommand{\abisimslice}{B}

\newcommand{\abisimchart}{\mathcal{B}}

\newcommand{\sbehinc}{{\sqsubseteq}}

\newcommand{\sbehinca}[1]{{\prescript{#1}{}{\sbehinc}}}
\newcommand{\behinca}[1]{\mathrel{\sbehinca}}

\newcommand{\sonebehinc}{{\pap{\sbehinc}{\subotr}}}

\newcommand{\sonebehinca}[1]{{{}_{#1}\sonebehinc}}
\newcommand{\onebehinca}[1]{\mathrel{\sonebehinca}}
\newcommand{\sRSP}{\textrm{\nf RSP}}
\newcommand{\RSPstar}{\text{$\sRSP^{*}\hspace*{-1pt}$}}

\newcommand{\sUSP}{\textrm{\nf USP}}

\newcommand{\USP}{\sUSP}

\newcommand{\assocstexpsum}{\fap{\text{\nf assoc}}{\sstexpsum}}
\newcommand{\assocstexpprod}{\fap{\text{\nf assoc}}{\sstexpprod}}
\newcommand{\commstexpsum}{\fap{\text{\nf comm}}{\sstexpsum}}
\newcommand{\neutralstexpsum}{\fap{\text{\nf neutr}}{\sstexpsum}}  
\newcommand{\idempotstexpsum}{\fap{\text{\nf idempot}}{\sstexpsum}}

\newcommand{\rightdistr}{\fap{\text{\nf r-distr}}{\sstexpsum,\sstexpprod}}

\newcommand{\rightidstexpprod}{\fap{\text{\nf id}_{\text{\nf r}}}{\sstexpprod}}
\newcommand{\leftidstexpprod}{\fap{\text{\nf id}_{\text{\nf l}}}{\sstexpprod}}
\newcommand{\recdefstexpit}{\fap{\text{\nf rec}}{{}^{*}}}

\newcommand{\termbodystexpit}{\fap{\text{\nf trm-body}}{{}^{*}}}

\newcommand{\deadlockax}{\text{$\stexpzero$}}

\renewcommand{\prod}{\mathrel{\cdot}}
\renewcommand{\assocstexpsum}{\textnf{(A1)}}
\renewcommand{\neutralstexpsum}{\textnf{(A2)}}
\renewcommand{\commstexpsum}{\textnf{(A3)}}
\renewcommand{\idempotstexpsum}{\textnf{(A4)}}
\renewcommand{\assocstexpprod}{\textnf{(A5)}}
\renewcommand{\rightdistr}{\textnf{(A6)}}
\renewcommand{\leftidstexpprod}{\textnf{(A7)}}
\renewcommand{\rightidstexpprod}{\textnf{(A8)}}
\renewcommand{\deadlockax}{\textnf{(A9)}}
\renewcommand{\recdefstexpit}{\textnf{(A10)}}
\renewcommand{\termbodystexpit}{\textnf{(A11)}}

\newcommand{\wg}{w.g.}
\newcommand{\weaklyguarded}{weak\-ly guard\-ed}

\newcommand{\sloopsbackinloopltortc}[1]{{\pap{}{#1}{\lefttorightarrow^{*}}}}

\newcommand{\sloopsbackto}{{\raisebox{0.8\depth}{\scalebox{1}[-1]{$\lefttorightarrow$}}}}  %{{\lefttorightarrow}} %{\raisebox{\depth}{\scalebox{1}[-1]{$\lefttorightarrow$}}} {{\lefttorightarrow}} %{{\lefttorightarrow}} %{{\downtouparrow}} %{{\rcurvearrowright}} %{{\nnearrow}}
\newcommand{\loopsbackto}{\mathrel{\sloopsbackto}}

\newcommand{\sloopsbacktortc}{{\sloopsbackto^{*}}}
\newcommand{\loopsbacktortc}{\mathrel{\sloopsbacktortc{\!}}}

\newcommand{\sconvloopsbackto}{{\righttoleftarrow}} %{\reflectbox{$\downtouparrow$}} %{{\rcurvearrowright}} %{{\nnearrow}}

\newcommand{\sconvloopsbacktortc}{{\pap{\sconvloopsbackto}{{\!}*}}}
\newcommand{\convloopsbacktortc}{\mathrel{\sconvloopsbacktortc{\!}}}
\newcommand{\onescriptbs}{\scalebox{0.75}{$\scriptstyle\bs{1}$}}
\newcommand{\onebrackscript}{\scalebox{0.75}{$\firebrick{\scriptstyle (1)}$}}

\newcommand{\sdloopsbackto}{{\!}\prescript{}{\textit{d}}{\hspace*{-0.75pt}\sloopsbackto}}  % {{\iap{\downtouparrow}{{\!}\textit{d}}}}  %{{\circlearrowleft}}  %{{\rcirclearrowup}} %{{\curvearrowbotright}}

\newcommand{\dloopsbackto}{\mathrel{\sdloopsbackto}}

\newcommand{\sconvdloopsbackto}{{\!}\prescript{}{\textit{d}}{\hspace*{-0.75pt}{\sconvloopsbackto}}} %{\reflectbox{$\downtouparrow$}} %{{\rcurvearrowright}} %{{\nnearrow}}
\newcommand{\convdloopsbackto}{\mathrel{\sconvdloopsbackto}}

  \newcommand{\subotr}{\hspace*{-1pt}{\scriptscriptstyle (\sone)}}

\newcommand{\sotelimred}{\sredb{\subotr}}
\newcommand{\otelimred}{\mathrel{\sotelimred}}

\newcommand{\sotelimredtc}{\sredbp{\subotr}{+}}

\newcommand{\sotelimconvred}{\sconvredb{\subotr}}

\newcommand{\Fone}{\text{$\bap{\text{\nf\bf F}}{\hspace*{-1pt}\bs{1}}$}}

\newcommand{\milnersys}{\text{$\text{\sf Mil}$}}

\newcommand{\milnersysmin}{\text{\smash{$\text{\sf Mil}^{\boldsymbol{-}}$}}}

\newcommand{\BBP}{\text{$\text{\sf BBP}$}}

\newcommand{\asys}{\mathcal{S}}

\newcommand{\seqin}[1]{{\iap{=}{#1}\hspace*{1pt}}}
\newcommand{\eqin}[1]{\mathrel{\seqin{#1}}}

\newcommand{\milnersyseq}{\eqin{\milnersys}}

\newcommand{\BBPeq}{\eqin{\BBP}}

\newcommand{\saTSS}{{\mathcal{T}}}
\newcommand{\StExpTSS}{\text{$\saTSS$}}
\newcommand{\StExpTSSover}[1]{\text{$\fap{\StExpTSS}{#1}$}}

\newcommand{\sLEE}{\text{\nf LEE}}
\newcommand{\sLLEE}{\text{\nf LLEE}}

\newcommand{\LEE}{\sLEE}

\newcommand{\LLEE}{\sLLEE}

\newcommand{\pivot}{\textit{piv}} %{\averti{\textit{p}}}
\newcommand{\pivottwinc}{\textit{piv}} %{\averti{\textit{p}}}
\newcommand{\pivottwincpb}[1]{\pap{\pivottwinc}{({#1})}}

\newcommand{\toptwinc}{\textit{top}} %{\averti{\textit{p}}}
\newcommand{\toptwincpb}[1]{\pap{\toptwinc}{({#1})}}

\newcommand{\twinc}{\csetverts}
\newcommand{\twincpivot}{\csetvertsi{1}}
\newcommand{\twinctop}{\csetvertsi{2}}

\newcommand{\twincambig}{\textit{A\hspace*{-1pt}m}}
\newcommand{\twincambigtop}{\iap{\twincambig}{2}}
\newcommand{\twincambigpivot}{\iap{\twincambig}{1}}

\newcommand{\twincunambig}{\textit{U\hspace*{-0.5pt}n}}
\newcommand{\twincunambigtop}{\iap{\twincunambig}{2}}
\newcommand{\twincunambigpivot}{\iap{\twincunambig}{1}}  

\newcommand{\sboundary}{\partial}
\newcommand{\boundaryof}[1]{{\sboundary}\hspace*{-1pt}{#1}}

\newcommand{\sprop}{\scalebox{0.7}{$\textit{prop}$}}
\newcommand{\spropred}{{\iap{\xrightarrow{\hspace*{2pt}{\phantom{\aact}}\hspace*{1pt}}}{\sprop}}}

\newcommand{\spropredrtc}{{\bpap{\xrightarrow{\hspace*{2pt}{}\hspace*{1pt}}}{\sprop}{*}}}
\newcommand{\propredrtc}{\mathrel{\spropredrtc}}

\newcommand*{\StrikeThruDistance}{0.3cm}%  

%% file: abstract-short.tex
Milner (1984) defined a process semantics for regular expressions.
  He formulated a sound proof system for bisimilarity of process interpretations of regular expressions,
    and asked whether this system is complete.
    
We report conceptually on a proof that shows that Milner's system is complete,
  by motivating and %illustrating 
                    describing all of its main steps.
We substantially refine the completeness proof by Grabmayer and Fokkink (2020) for the restriction of Milner's system to `\stexponefree' regular expressions.
As a crucial complication we recognize that 
  process graphs with emp\-ty-step transitions that satisfy the layered loop-existence and elimination property \LLEE\
  are not closed under bisimulation collapse (unlike process graphs with \LLEE\ that only have proper-step transitions). 
We circumnavigate this obstacle by defining a \LLEEpreserving\ `crystallization procedure' for such process graphs.
By that we obtain `near-col\-lapsed' process graphs with \LLEE\
  whose strongly connected components are either collapsed or of `\twincrystal' shape.
Such \nearcollapsed\ process graphs guarantee provable solutions 
 for bisimulation collapses of process interpretations of regular expressions.

%% file: cryst-arxiv-art.tex
%----------------------
\section{Introduction}%
  \label{intro}
%----------------------

Kleene \cite{klee:1951} (1951) introduced regular expressions, which are widely studied in formal language theory.
In a typical formulation, they are constructed from constants 0, 1, letters $a$ from some alphabet
  (interpreted as the formal languages $\emptyset$, $\setexp{\emptyword}$, and $\setexp{a}$,  
   where $\emptyword$ is the empty word)
and binary operators $+$ and $\cdot$, and the unary Kleene star ${}^{\sstexpit}$
  (which are interpreted as language union, concatentation, and iteration).  

Milner \cite{miln:1984} (1984) introduced a process semantics for regular expressions.
He defined an interpretation $\chartof{\astexp}$ of regular expressions $\astexp$ as charts (finite process graphs):
the interpretation %$\chartof{\cdot}$
                   of $\stexpzero$ is deadlock, of $\stexpone$ is successful termination, letters $a$ are atomic actions,
the operators $\sstexpsum$ and $\sstexpprod$ stand for choice and concatenation of processes,
and (unary) Kleene star $\stexpit{(\cdot)}$ represents iteration with the option to terminate successfully 
  before each execution of the iteration body.
He then defined the process semantics of `star expressions' (regular expressions in this context) $\astexp$
  as `star behaviors' $\procsem{\astexp} \defdby \eqcl{\chartof{\astexp}}{\sbisimsubscript}$,
  that is, as equivalence classes of chart interpretations with respect to bisimilarity $\sbisim$.   
Milner was interested in an axiomatization of equality of `star behaviors'.
  For this purpose he adapted Salomaa's complete proof system \cite{salo:1966} for language equivalence on regular expressions
  to a system %(we denote by) 
              $\milnersys$ (see Def.~\ref{def:milnersys}) that is sound for equality of denoted star behaviors.
Recognizing that Salomaa's proof strategy cannot be followed directly, he left completeness as an open question.  

Over the past 38 years, completeness results have been obtained for restrictions of Milner's system to the following subclasses of star expressions: 
  (a)~without $\stexpzero$ and $\stexpone$, but with binary star iteration $\stexpbit{\astexpi{1}}{\astexpi{2}}$ %$\stexpbit{(\cdot)}{(\cdot)}$
      % with iteration-part $\astexpi{1}$ and exit-part $\astexpi{2}$ instead 
      instead of unary star~\cite{fokk:zant:1994},
  (b)~with~$\stexpzero$, with iterations restricted to exit-less ones $\stexpprod{\stexpit{(\cdot)}}{\stexpzero}$,
    without~$\stexpone$ \cite{fokk:1997:pl:ICALP} and with $\stexpone$~\cite{fokk:1996:term:cycle:LGPS},
  (c)~without $\stexpzero$, and with only restricted occurrences of~$\stexpone$~\cite{corr:nico:labe:2002}, and
  (d)~`\onefree' expressions formed with~$\stexpzero$, without~$\stexpone$, but with binary instead of unary iteration~\cite{grab:fokk:2020:lics}.
By refining concepts developed in \cite{grab:fokk:2020:lics} for the proof of (d) % (which refined \cite{fokk:1997:pl:ICALP})
  % Owing  idea from \cite{fokk:1997:pl:ICALP}, but crucially drawing from the result for (d) in \cite{grab:fokk:2020:lics}, 
   % Standing on the shoulders of specifically the recent result for (d) in \cite{grab:fokk:2020:lics},
  we can finally establish completeness~of~$\milnersys$.

\smallskip
\noindent{\bf The aim of this article.} 
  We provide an outline of the completeness proof for $\milnersys$.
    Hereby our focus is on the main new concepts and results.
    % For scrutiny of most of the crucial arguments of the proof
    %   we refer to the appendix, which is the kernel of a monograph %\cite{grab:2022:kernel:monograph} 
    %                                                                on the proof
    %   that we are writing.
  While details are sometimes only hinted at in this article,
    we think that the crystallization technique we present
      opens up a wide space for other applications (we suggest one in Sect.~\ref{conclusion}).
  We want to communicate this technique in summarized form to the community
    in order to stimulate its further development. 
\section{Motivation for the chosen proof strategy}   
  \label{motivation:proof:strategy}  
%---------------------------------------------------

We explain the main obstacle we encountered for developing our proof strategy
  through explaining shortcomings of existing approaches. 
Finally we describe crucial new concepts that we use for adapting
  the collapse strategy from \cite{grab:fokk:2020:lics,grab:fokk:2020:lics:arxiv}.
  
\smallskip
\noindent{\bf Obstacle for the `bisimulation chart' proof strategy.}\label{bisim:chart:strategy:start}
  Milner \cite{miln:1984} recognized that completeness of the proof system $\milnersys$
    cannot be established along the lines of Salomaa's completeness proof
      for his   proof system $\Fone$ of language equivalence of regular expressions \cite{salo:1966}.
  The reason is as follows.          
  Adopting Salomaa's proof strategy would mean
    (i)~to link given bisimilar chart interpretations $\chartof{\astexpi{1}}$ and $\chartof{\astexpi{2}}$ 
      of star expressions $\astexpi{1}$ and $\astexpi{2}$ 
        via a chart $\abisimchart$ that represents a bisimulation between $\chartof{\astexpi{1}}$ and $\chartof{\astexpi{2}}$,
   (ii)~to use this link via functional bisimulations from $\abisimchart$ to $\chartof{\astexpi{1}}$ and $\chartof{\astexpi{2}}$ 
          to prove equal in $\milnersys$ 
             the provable solutions $\astexpi{1}$ of $\chartof{\astexpi{1}}$, and $\astexpi{2}$ of $\chartof{\astexpi{2}}$,
   (iii)~to extract from $\abisimchart$ a star expression $\astexp$ that provably solves $\abisimchart$,
         and then is provably equal to $\astexpi{1}$ and $\astexpi{2}$.          
  Here a `provable solution' of a chart~$\achart$ is a function $\sasol$ from the set of vertices of $\achart$ to star expressions
    such that the value $\asol{\avert}$ at a vertex $\avert$ can be reconstructed, provably in $\milnersys$,
      from the transitions to, and the values of $\sasol$ at, the immediate successor vertices of $\avert$ in $\achart$,
      and (non-)termination at $\avert$. 
  By the `principal value' of a provable solution we mean its value at the start vertex.
  In %the pictures below
     pictures we write `$\astexp$ is solution' for `$\astexp$ is the principal value of~a~solution'.
  \input{figs/fig-bisimchart-strategy.tex}
  First by (i) star expressions $\astexpi{1}$ and $\astexpi{2}$ 
    can be shown to be the principal values of provable solutions of their chart interpretations $\chartof{\astexpi{1}}$ and $\chartof{\astexpi{2}}$, respectively.
  These solutions can be transferred backwards by (ii) 
    over the functional bisimulations from the bisimulation chart $\abisimchart$ to $\chartof{\astexpi{1}}$ and $\chartof{\astexpi{2}}$, respectively.
  It follows that $\astexpi{1}$ and $\astexpi{2}$ are the principal values of two provable solutions of $\abisimchart$. 
  However, now the obstacle appears, because the extraction procedure in (iii) of a proof of $\astexpi{1} \formeq \astexpi{2}$ in $\milnersys$
    cannot work, like Salomaa's, for all charts $\abisimchart$ irrespective of the actions of its transitions.
  An example that demonstrates that is the chart $\acharti{12}$ in Ex.~4.1 in \cite{grab:fokk:2020:lics,grab:fokk:2020:lics:arxiv}.  
  This is because some charts are unsolvable
    (``[In] contrast with the case for languages---an arbitrary system of guarded equations 
      in [star]-behaviours cannot in general be solved in star expressions'' \cite{miln:1984}\label{quotation:unsolvable:specs}),
  but turn into a solvable one if all actions in it are replaced by a single one. 
    % not any more if all actions in it are replaced by a single one. 
  The reason for the failure of Salomaa's extraction procedure is then that 
    the absence in~$\milnersys$ of 
    the \emph{left}-distributivity law  $x \cdot (y + z) = x \cdot y + x \cdot z$ (it is not sound under bisimilarity) %in Salomaa's system 
    frequently prevents applications of the fixed-point rule $\RSPstar$ in $\milnersys$ unlike for the system $\Fone$ that Salomaa proved complete.
  We conclude that such a bisimulation-chart proof strategy, inspired by Salomaa % Salomaa's proof in 
                                                                         \cite{salo:1966},
    is not expedient for showing completeness of $\milnersys$.\label{bisim:chart:strategy:end}
    
  However, if the fixed-point rule $\RSPstar$ in Milner's system is replaced in $\milnersys$
    by a general unique-solvability rule scheme \USP\ for guarded systems of equations,
  then a proof system arises to which the bisimulation-chart proof strategy is applicable.
    That system can therefore be shown to be complete comparatively easily (as noted in \cite{grab:2021:calco}).

\smallskip
\noindent{\bf Loop existence and elimination.}
  A sufficient structural condition for solvability of a chart,
    and correspondingly of a linear system of recursion equations, 
    by a regular expression modulo bisimilarity
    was given by Grabmayer and Fokkink in \cite{grab:fokk:2020:lics}: 
  the `loop existence and elimination' condition \LEE,
      and its `layered' specialization \LLEE,
    which is independent of the specific actions in a chart.    
  These properties are refinements for graphs 
    of `well-behaved specifications' due to Baeten and Corradini in \cite{baet:corr:2005}
      that single out a class of `palm trees' (trees with back-links) that specify star expressions under the process interpretation. 
  For showing that the tailored restriction \BBP\ of Milner's system $\milnersys$
    to `\stexponefree' star expressions (without $\stexpone$, but with binary instead of unary star iteration)      
      is complete, the following properties were established in~\cite{grab:fokk:2020:lics}:  
  \crtcrossreflabel{{\bf (I$_{\hspace*{0.5pt}\text{\nf \st{1}}}$)}}[Ionefree]
    Chart \textbf{\textit{i}}nterpretations of \stexponefree\ star expressions are \LLEEcharts.
  \crtcrossreflabel{{\bf (S$_{\hspace*{0.5pt}\text{\nf \st{1}}}$)}}[Sonefree]
    Every \stexponefree\ star expression $\astexp$ is the principal value of a provable \textbf{\textit{s}}olution 
      of its chart interpretation $\chartof{\astexp}$. 
  \crtcrossreflabel{{\bf (E$_{\hspace*{0.5pt}\text{\nf \st{1}}}$)}}[Eonefree]
    From every \LLEEchart\ $\achart$ a provable solution of $\achart$ (by \stexponefree\ star expressions) can be \textbf{\textit{e}}xtracted.
  \crtcrossreflabel{{\bf (SE$_{\hspace*{0.5pt}\text{\nf \st{1}}}$)}}[SEonefree]
    All provable \textbf{\textit{s}}olutions of a \LLEEchart\ are provably \textbf{\textit{e}}qual.
  \crtcrossreflabel{{\bf (T$_{\hspace*{-1pt}\text{\nf \st{1}}}$)}}[Tonefree]
    Every provable solution can be \textbf{\textit{t}}ransferred 
                                   from the target to the source chart of a functional bisimulation
    to obtain a provable solution of the source chart.
  \crtcrossreflabel{{\bf (C$_{\hspace*{-1pt}\text{\nf \st{1}}}$)}}[Conefree]
    The bisimulation \textbf{\textit{c}}ollapse of a \LLEEchart\ is again a \LLEEchart.   
    
  As a consequence of these properties, a finite chart $\achart$ is expressible by a \stexponefree\ star expression modulo bisimilarity
    if and only if
  the bisimulation collapse of $\achart$ satisfies \LLEE.

\smallskip
\noindent{\bf The `bisimulation collapse' proof strategy for \BBP\ (\cite{grab:fokk:2020:lics}).}
  For the completeness proof of the tailored restriction $\BBP$ of Milner's system $\milnersys$ 
    to \stexponefree\ star expressions,
  Grabmayer and Fokkink in \cite{grab:fokk:2020:lics} linked bisimilar chart interpretations $\chartof{\astexpi{1}}$ and $\chartof{\astexpi{2}}$ 
    of \stexponefree\ star expressions $\astexpi{1}$ and $\astexpi{2}$
    via the joint bisimulation collapse $\acharti{0}$. 
  That argument, which we recapitulate below, can be illustrated as follows:
  \input{figs/fig-bisimcoll-strategy.tex}    
  By \ref{Sonefree}, 
  the star expressions $\astexpi{1}$ and $\astexpi{2}$ are the principal values of provable solutions $\sasoli{1}$ and $\sasoli{2}$ 
    of their chart interpretations $\chartof{\astexpi{1}}$ and $\chartof{\astexpi{2}}$, respectively.
  Furthermore, by \ref{Ionefree} the chart interpretations of the \stexponefree\ star expressions  $\chartof{\astexpi{1}}$ and $\chartof{\astexpi{2}}$ have the property \LLEE.
  Since \LLEE\ is preserved under the operation of bisimulation collapse due to \ref{Conefree}, 
    the joint bisimulation collapse $\acharti{0}$ of $\chartof{\astexpi{1}}$ and $\chartof{\astexpi{2}}$ is again a \LLEEchart.
  Therefore a provable solution $\sasoli{0}$ can be extracted from $\acharti{0}$ due to \ref{Eonefree}. Let $\astexpi{0}$ be its principal value. 
  The solution $\sasoli{0}$ can be transferred from $\acharti{0}$ backwards over the functional bisimulations 
  from $\chartof{\astexpi{1}}$ and $\chartof{\astexpi{2}}$ to $\acharti{0}$ due to \ref{Tonefree},
  and thereby defines provable solutions $\sasoltildei{1}$ of $\chartof{\astexpi{1}}$ and $\sasoltildei{2}$ of $\chartof{\astexpi{2}}$,
  both with $\astexpi{0}$ as the principal value. 
  Now having provable solutions $\sasoli{1}$ and $\sasoltildei{1}$ of the \LLEEchart~$\chartof{\astexpi{1}}$,
    these solutions are \provablyin{\BBP} equal by \ref{SEonefree}, and hence also their principal values $\astexpi{1}$ and $\astexpi{0}$,
    that is $\astexpi{1} \BBPeq \astexpi{0}$.
  Analogously $\astexpi{1} \BBPeq \astexpi{0}$ can be established.
  Then $\astexpi{1} \BBPeq \astexpi{2}$ follows by applying symmetry and transitivity proof rules of equational logic.

\smallskip  
\noindent{\bf Obstacles for a `bisimulation collapse' %proof
                                                      strategy for $\milnersys$.}  
  A generalization of this argument for arbitrary star expressions runs into two problems
    that can be illustrated as:
  \input{figs/fig-bisimcoll-strategy-fail.tex}    
  First, see \ref{Ionefail}, there are star expressions $\astexp$ whose chart interpretation~$\chartof{\astexp}$ satisfies neither \LEE\ nor \LLEE,
    as was noted in \cite{grab:2021:TERMGRAPH2020-postproceedings}.
  In order to still be able to utilize \LLEE, in \cite{grab:2021:TERMGRAPH2020-postproceedings} a variant chart interpretation $\onechartof{\astexp}$ 
    was defined for star expressions $\astexp$
    such that $\onechartof{\astexp}$ is a `\LLEEonechart', that is, a chart with `\onetransitions' (explicit \emptystep\ transitions)
        that satisfies \LLEE, and $\onechartof{\astexp}$ is `\onebisimilar' to the chart interpretation $\chartof{\astexp}$. 
  Hereby `\onebisimulations' and `\onebisimilarity'
    are adaptations of bisimulations and bisimilarity to \onecharts.   
        
  However, use of the variant chart interpretation encounters 
  the second obstacle \ref{Conefail} as illustrated above.
  A part of it that was also observed in \cite{grab:2021:TERMGRAPH2020-postproceedings}
    is that \LLEEonecharts\ are not closed under bisimulation collapse, unlike \LLEEcharts.
  While this renders the bisimulation collapse proof strategy unusable,
    we here show that an adaptation to a `\onebisimulation\ collapse' strategy is not possible, either,
    if it is based on `\onebisimulation\ collapsed' \onecharts\ 
    in which none of its vertices are \onebisimilar.
  In particular we show \ref{Conefail} as the first of our key observations and concepts as listed below:
  \begin{enumerate}[label=\protect\mycirc{\arabic{*}},align=left,leftmargin=*,itemsep=0.5ex]%
                   %[label=\protect\encircle{\protect\arabic{*}},align=right,itemsep=0.5ex]%
    \item{}\label{contrib:concept:1}
      \emph{\LLEEonecharts\ are not in general collapsible to (collapsed) \LLEEonecharts.}
      Nor do \onebisimilar\ \LLEEonecharts\ always have a joint (\onebisimilarity) minimization.   
      We demonstrate this by an example (see Fig.~\ref{fig:countex:collapse}). %, and Prop.~\ref{prop:not:LLEE:pres:collapse:minimization}).
  \end{enumerate}        
  The second part of \ref{contrib:concept:1} prevents 
    a change from the `\onebisimilarity\ collapse strategy'
      to a `joint minimization strategy' that would use the weaker statement that any two \onebisimilar\ \LLEEonecharts\
        were always jointly minimizable under \onebisimilarity\ (which is also wrong due to \ref{contrib:concept:1}).
  We want to credit 
  Schmid, Rot, and Silva's careful coinductive analysis in \cite{schm:rot:silv:2021}
  of the proof in~\cite{grab:fokk:2020:lics},
  which helped us to realize that this possibility needs to be excluded as well.
  %
  % The second part of \ref{contrib:concept:1} also prevents the generalization of a variation
  %   of the proof in \cite{grab:fokk:2020:lics} sketched above
  %     that uses that any two bisimilar \LLEEcharts\ (thus without \onetransitions)
  %       are jointly minimizable under bisimilarity (which was shown by Schmid, Rot, and Silva in \cite{schm:rot:silv:2021}). 

\smallskip
\noindent{\bf How we recover the collapse proof strategy for $\milnersys$.}
  We define `crystallized' approximations with \LLEE\ of collapsed \LLEEonecharts\
    in order to show that bisimulation collapses of \LLEEonecharts\ have provable solutions.
  For this purpose we combine the following concepts and their properties:

\begin{enumerate}[label=\protect\mycirc{\arabic{*}},leftmargin=*,align=left,itemsep=0.5ex,start=2]%
                 %[label=\protect\encircle{\protect\arabic{*}},align=right,itemsep=0.5ex]%
    % \item{}\label{contrib:concept:1}
    %   \LLEEonecharts\ are \emph{not closed under \onebisimulation\ collapse}
    %     nor under joint (\onebisimilarity) minimization.   
    %   We demonstrate this by a natural example (see Fig.~\ref{fig:countex:collapse}).
    %
  \item{}\label{contrib:concept:2}
    \emph{\Twincrystals{}}: These are \onecharts\ with a single strongly connected component (\scc) 
                     that exhibit a \selfinverse\ symmetry function that links \onebisimilar\ vertices.
      \Twincrystals\ abstract our example that demonstrates~\ref{contrib:concept:1}.
  \item{}\label{contrib:concept:3}  
    \emph{\Nearcollapsed} \onecharts: These are \onecharts\ in which \onebisimilar\ vertices
      appear as pairs that are linked by a \selfinverse\ function that induces a `grounded \onebisimulation\ slice'. 
      \Twincrystals\ are \nearcollapsed\ \LLEEonecharts.
    
  \item{}\label{contrib:concept:4}    
    \emph{Crystallization:}
      By this we understand a process of step-wise minimization of \LLEEonecharts\ under \onebisimilarity\
        %by crucially using the \txtconnectthrough\ operation from \cite{grab:fokk:2020:lics},
      that produces \onebisimilar\ `crystallized' \LLEEonecharts\ 
        in which all strongly connected components are collapsed or of \twincrystal\ shape.
      This process uses the \txtconnectthrough\ operation from \cite{grab:fokk:2020:lics}
        for \onebisimilar\ vertices. 
      We show that crystallized \onecharts\ are \nearcollapsed.  
    
  \item{}\label{contrib:concept:5}
    \emph{Complete \provablein{\milnersys} solution} of a \onechart~$\aonechart$:
      %so we call any 
      This is a \provablein{\milnersys} solution of $\aonechart$
        with the property that its values for %any two 
                                              \onebisimilar\ vertices of $\aonechart$
          are \provablyin{\milnersys} equal.   
      Any complete \provablein{\milnersys} solution of a \onechart~$\aonechart$
        yields a \provablein{\milnersys} solution of the bisimulation collapse~of~$\aonechart$.

  \item{}\label{contrib:concept:6}  
    \emph{Elevation} of vertex sets above \onecharts:
      This is a concept of partially unfolding \onecharts\ that facilitates us to show that 
    \nearcollapsed\ \weaklyguarded\ \LLEEonecharts\ 
    have complete \provablein{\milnersys} solutions. 
    %are uniquely solvable in Milner's system \milnersys. 
   
\end{enumerate}

With these conceptual tools we will be able to recover
  the collapse proof strategy for \LLEEonecharts. 
The idea is to establish, for given \onebisimilar\ \LLEEonecharts~$\aonecharti{1}$ and $\aonecharti{2}$,
  a link via which solutions of $\aonecharti{1}$ and $\aonecharti{2}$ with the same principal value can be obtained.
  We create such a link via the crystallized \LLEEonechart\ $\aonecharti{10}$ of (one of them, say) $\aonecharti{1}$
  and the joint bisimulation collapse $\acharti{0}$ of $\aonecharti{1}$ and $\aonecharti{2}$. 
  Then a solution of $\aonecharti{10}$ can obtained,
     transferred first to $\acharti{0}$, and then to $\aonecharti{1}$ and to $\aonecharti{2}$.  
% This argument will be illustrated in Fig.~\ref{fig:proof:structure} in Section~\ref{completeness:proof},
%   where we give the completeness proof for~$\milnersys$~based~on~lemmas.
%  
% %-------------
% \subsubsection{Recovered bisimulation collapse strategy (explained directly)} \mbox{}
% %-------------
% \medskip
% \noindent
%
% With the conceptual tools \ref{contrib:concept:2}--\ref{contrib:concept:5} on page~\pageref{contrib:concept:2} in place,
%   the collapse proof strategy can be recovered for \LLEEonecharts.
More precisely, the central part of our completeness proof for $\milnersys$
  can be illustrated as follows:
  \begin{center}
    \input{figs/fig-part-proof-structure.tex}
  \end{center}
Here we start from the assumption of \onebisimilar\ \onecharts\ $\aonecharti{1}$ and $\aonecharti{2}$
  of which at least $\aonecharti{1}$ satisfies \LLEE. 
Then the \LLEEonechart~$\aonecharti{1}$ can be crystallized with as result a \onebisimilar\ \LLEEonechart~$\aonecharti{10}$.
Due to \LLEE, a provable solution $\sasoli{10}$ can be extracted from $\aonecharti{10}$. Let $\astexpi{10}$ be the principal value of $\sasoli{10}$.
Since crystallized \onecharts\ are \nearcollapsed\ \LLEEonecharts,
  the solution $\sasoli{10}$ of $\aonecharti{10}$ is complete. 
As such it defines a provable solution with principal value $\astexpi{10}$
  on the joint bisimulation collapse $\aonecharti{0}$ of $\aonecharti{10}$, $\aonecharti{1}$, and $\aonecharti{2}$.       
This solution can be pulled back conversely over the functional \onebisimulations\ from $\aonecharti{0}$ to $\aonecharti{1}$ and to $\aonecharti{2}$. 
In this way we obtain provable solutions $\sasoli{1}$ of $\aonecharti{1}$ and $\sasoli{2}$ of $\aonecharti{2}$,
  both with the same principal value $\astexpi{10}$.  

This is the central argument of our completeness proof for $\milnersys$
  that we give in Sect.~\ref{completeness:proof} based on lemmas. 
Its illustration appears in the illustration of the completeness proof in Fig.~\ref{fig:proof:structure}.

\section{Preliminaries}%
  \label{prelims}
%-----------------------

%\begin{defi}\label{def:StExp}
  Let $\actions$ be a set whose members we call \emph{actions}.
  The set $\StExpover{\actions}$ of \emph{star expressions over actions in $\actions$} 
  is defined by the following grammar, where $\aact\in\actions\,$:
  \begin{center}
    $
    \astexp, \astexpi{1}, \astexpi{2}
      \:\BNFdefdby\:
    \stexpzero
      \BNFor
    \stexpone
      \BNFor  
    \aact
      \BNFor
    \stexpsum{\astexpi{1}}{\astexpi{2}}
      \BNFor
    \stexpprod{\astexpi{1}}{\astexpi{2}}
      \BNFor
    \stexpit{\astexp} % \quad (\aact\in\actions)
    $
  \end{center} 

\begin{defi}[\onecharts\ and \oneLTSs]\label{def:onechart:oneLTS}
  A \emph{\onechart}
  is a 6\nb-tuple $\tuple{\verts,\actions,\sone,\start,\transs,\sterminates}$ 
    with $\verts$ a \underline{\smash{finite}} 
                        set of \emph{vertices},
    $\actions$ a \underline{\smash{finite}} set of \emph{(proper)} $\emph{action labels}$,
    $\sone\notin\actions$ the specified \emph{empty step label},
    $\start\in\verts$ the \emph{start vertex} (hence $\verts \neq \emptyset$),
    $\transs \subseteq \verts\times\oneactions\times\verts$ the \emph{labeled transition relation},
    where $\oneactions \defdby \actions \cup \setexp{\sone}$ is the set of action labels including $\sone$, 
    and $\sterminates \subseteq \verts$ a set of \emph{vertices with immediate termination} (or \emph{terminating} vertices). %,
    % for short, the \emph{terminating states}.
  In such a \onechart, 
  we call a transition in $\transs\cap(\verts\times\actions\times\verts)$ (labeled by a \emph{proper action} in $\actions$) 
          a \emph{proper transition},
  and a transition in $\transs\cap(\verts\times\setexp{\sone}\times\verts)$ (labeled by $\sone$)
          a \emph{\onetransition}.  
  Reserving non-underlined action labels like $\aact,\bact,\ldots$ for proper actions,
  we use highlighted underlined action label symbols like $\aoneact$ for actions labels in the set $\oneactions$ 
  that includes~the~label~$\sone$. 
  % We highlight in red transition labels that may involve $\sone$.
  
  % By a \emph{\oneLTS} (short for a \emph{labeled transition system with \onetransitions\ and immediate termination})
  %   we understand a 5\nb-tuple $\tuple{\verts,\actions,\sone,\transs,\sterminates}$ 
  %     with concepts as explained above for `\onecharts' but without a specified start vertex. \
  A \emph{\oneLTS} (%short for 
                    a \emph{labeled transition system with \onetransitions\ and immediate termination})
    is a 5\nb-tuple $\tuple{\verts,\actions,\sone,\transs,\sterminates}$ 
      with concepts as explained above %for `\onecharts' 
                                       but without a %specified 
                                                     start vertex.     
  For every \onechart~$\aonechart = \tuple{\verts,\actions,\sone,\start,\transs,\termexts}$ 
    we denote by $\oneLTSof{\aonechart} = \tuple{\verts,\actions,\sone,\transs,\termexts}$ 
      the \oneLTS\ \emph{underlying $\aonechart$} (or \emph{that underlies $\aonechart$}).
  
  We say that a \onechart\ $\aonechart$ (a \oneLTS\ $\aoneLTS$) is \emph{weakly guarded}~{(\emph{\wg})} 
    if $\aonechart$ (resp. $\aoneLTS$) does not have an infinite path of \onetransitions.
    
  By an \emph{induced \transitionact{\aact}} $\avert \ilt{\aact} \bvert$, for a proper action $\aact\in\actions$,\vspace*{-1.5pt} in a \onechart~$\aonechart$
    (in a \oneLTS~$\aoneLTS$) we mean 
    a path of the form $\avert \lt{\sone} \cdots \lt{\sone} \cdot \lt{\aact} \bvert$ in $\aonechart$ (in $\aoneLTS$)
    that consists of a finite number of \onetransitions\ that ends with a proper \transitionact{\aact}.
  By\vspace*{-1.5pt} \emph{induced termination} $\oneterminates{\avert}$, for $\avert\in\verts$ we mean that there is a path
    $\avert \lt{\sone} \cdots \lt{\sone} \averttilde$ with $\terminates{\averttilde}$ in~$\aonechart$ (in $\aoneLTS$).
  
  By a \emph{chart} (a \emph{LTS}) we mean a \onetransition\ free \onechart\ (a \onetransition\ free LTS).   
  Let $\aonechart = \tuple{\verts,\actions,\sone,\start,\transs,\termexts}$ be a \onechart.\vspace*{-1mm} 
  We define by $\indchartof{\aonechart} \defdby \tuple{\verts,\actions,\sone,\start,\indtranss,\onetermexts}$
    \emph{the induced chart of $\aonechart$} 
      % of mean the \onechart\ $\indchartof{\aonechart} \defdby \tuple{\verts,\actions,\sone,\start,\indtranss,\onetermexts}$
        whose transitions are the induced transitions of $\aonechart$,
          and whose terminating vertices are the vertices of $\aonechart$ with induced termination.
  Note that $\indchartof{\aonechart}$ is \onetransition\ free.
  Also, for every vertex $\bvert\in\verts$
  we denote by   
    $\gensubchartofby{\aonechart}{\bvert} \defdby \tuple{\verts,\actions,\sone,\bvert,\transs,\termexts}$        
  the \emph{generated sub(-$\sone$)-chart of $\aonechart$ at $\bvert$}.
\end{defi}

\begin{defi}[\onebisimulating\ slices, \onebisimulations\ for \oneLTSs]\label{def:onebisim:slices:betw:oneLTSs}
  Let $\aoneLTSi{i} = \tuple{\statesi{i},\actions,\sone,\transsi{i},\termextsi{i}}$ for $i\in\setexp{1,2}$
    be \oneLTSs. 
   
  A \emph{\onebisimulating\ slice between $\aoneLTSi{1}$ and $\aoneLTSi{2}$}
  is a binary relation $\abisim \subseteq \statesi{1}\times\statesi{2}$, 
  with active domain $\asetvertsi{1} \defdby \actdomof{\abisim} = \proji{1}{\abisim}$,
    and 
  active codomain $\asetvertsi{2} \defdby \actcodomof{\abisim} 
                                        = \proji{2}{\abisim}$, 
    where $\sproji{i} \funin \statesi{1}\times\statesi{2} \to \statesi{i}, \proji{i}{\pair{\astatei{1}}{\astatei{2}}} = \astatei{i}$, 
          for %$\bstatei{1}\in\statesi{1}$, $\bstatei{2}\in\statesi{2}$, 
              $i\in\setexp{1,2}$, 
  such that $\abisim\neq\emptyset$, 
  and for all $\pair{\astatei{1}}{\astatei{2}}\in\abisim$ the three conditions hold:
  \begin{itemize}[labelindent=0em,leftmargin=*,align=left,itemsep=0.25ex,labelsep=0.25em]
    \item[(forth)$_\text{\nf s}$]
      $ \begin{aligned}[t]
          &
          \forall \aact\in\actions\,
            \forall \astateacci{1}\in\statesi{1} 
              \bigl(\, 
                \astatei{1} \ilti{\aact}{1} \astateacci{1}
                  \: 
                \alert{\underline{\black{\logand\:\astateacci{1}\in\asetvertsi{1}}}}  
                \\[-0.75ex]
                & \hspace*{6ex}
                  \;\Longrightarrow\; 
                    \exists \astateacci{2}\in\statesi{2}
                      \bigl(\, \astatei{2} \ilti{\aact}{2} \astateacci{2} 
                                 \logand
                               \pair{\astateacci{1}}{\astateacci{2}}\in\abisim \,)
            \,\bigr) \punc{,}
        \end{aligned} $
    \item[(back)$_\text{\nf s}$]
      $ \begin{aligned}[t]
          &
          \forall \aact\in\actions\,
            \forall \astateacci{2}\in\statesi{2} 
              \bigl(\, 
                \astatei{2} \ilti{\aact}{2} \astateacci{2}
                  \: 
                \alert{\underline{\black{\logand\:\astateacci{2}\in\asetvertsi{2}}}}  
                \\[-0.75ex]
                & \hspace*{6ex}
                  \;\Longrightarrow\; 
                    \exists \astateacci{1}\in\statesi{1}
                      \bigl(\, \astatei{1} \ilti{\aact}{1} \astateacci{1} 
                                 \logand
                               \pair{\astateacci{2}}{\astateacci{1}}\in\abisim \,)
            \,\bigr) \punc{,}
        \end{aligned} $
    \item[(termination)]
      $ \oneterminatesi{1}{\astatei{1}}
          \;\;\Longleftrightarrow\;\;
            \oneterminatesi{2}{\astatei{2}} \punc{.}$
  \end{itemize}
  Here (forth)$_\text{\nf s}$ entails $\astateacci{2}\in\asetvertsi{2}$,
    and (back)$_\text{\nf s}$ entails $\astateacci{1}\in\asetvertsi{1}$.
  % The condition (forth)$_\text{\nf s}$ entails $\astateacci{2}\in\asetvertsi{2}$,
  %  and the condition (back)$_\text{\nf s}$ entails $\astateacci{1}\in\asetvertsi{1}$.
    
  A \emph{\onebisimulation\ between $\aoneLTSi{1}$ and $\aoneLTSi{2}$}
    is a \onebisimulation\ slice $\abisim$ between $\aoneLTSi{1}$ and $\aoneLTSi{2}$
      such that the active domain of $\abisim$, and the active codomain of $\abisim$, are \transitionclosed\
      (that is, closed under $\transsi{1}$ and $\transsi{2}$, respectively),
      or equivalently, a \nonempty\ relation $\abisim \subseteq \statesi{1}\times\statesi{2}$
        such that for every $\pair{\bverti{1}}{\bverti{2}}\in\abisim$ 
          the conditions (forth), (back), (termination) hold,
          where (forth), and (back) result from (forth)$_\text{\nf s}$ and (back)$_\text{\nf s}$
            by dropping the \alert{\underline{\black{underlined}}} conjuncts. 
        
  By a \emph{\onebisimulation\ slice between $\aoneLTSi{1}$ and $\aoneLTSi{2}$}
    we mean a \onebisimulating\ slice between $\aoneLTSi{1}$ and $\aoneLTSi{2}$
    that is contained in a \onebisimulation\ between $\aoneLTSi{1}$ and $\aoneLTSi{2}$.
  
  A \onebisimulating\ slice (a \onebisimulation\ slice, a \onebisimulation) \emph{on} a \oneLTS~$\aoneLTS$
  is a \onebisimulating\ slice (and resp., a \onebisimulation\ slice, a \onebisimulation) between $\aoneLTS$ and $\aoneLTS$. % itself.  
\end{defi}

\begin{defi}[(funct.) \onebisimulation\ between \onecharts]\label{def:onebisim:betw:onecharts}
  We consider \onecharts~$\aonecharti{i} = \tuple{\vertsi{i},\actions,\sone,\starti{i},\transsi{i},\termextsi{i}}$
    for~\mbox{$i\in\setexp{1,2}$}. 
    
  A \emph{\onebisimulation} between \onecharts~$\aonecharti{1}$ and $\aonecharti{2}$
    is a \onebisimulation\ $\abisim \subseteq \vertsi{1}\times\vertsi{2}$ 
           between the \oneLTSs~$\oneLTSof{\aonecharti{1}}$ and $\oneLTSof{\aonecharti{2}}$ 
                          underlying $\aonecharti{1}$ and $\aonecharti{2}$, respectively,
    such that additionally: 
  \begin{itemize}[labelindent=0em,leftmargin=*,align=left,itemsep=0.25ex,labelsep=0.4em]
    \item[(start)]
      $\pair{\starti{1}}{\starti{2}}\in\abisim$ \hspace*{0.15em} ($\abisim$ relates start vertices of $\aonecharti{1}$ and $\aonecharti{2}$)
  \end{itemize}    
  holds; 
  thus $\abisim$ must satisfy (start), and, for all $\pair{\bverti{1}}{\bverti{2}}\in\abisim$, 
    the conditions (forth), (back), (termination) from Def.~\ref{def:onebisim:slices:betw:oneLTSs}.  
    
  By a \emph{functional} \onebisimulation\ \emph{from $\aonecharti{1}$ to $\aonecharti{2}$}
    we mean a \onebisimulation\ between $\aonecharti{1}$ and $\aonecharti{2}$ that is the graph of a partial function from $\vertsi{1}$ to $\vertsi{2}$. 
  By $\aonecharti{1} \onebisim \aonecharti{2}$
    (by $\aonecharti{1} \funonebisim \aonecharti{2}$)
    we denote
    that there is a \onebisimulation\ between $\aonecharti{1}$ and $\aonecharti{2}$
      (respectively, a functional \onebisimulation\ from $\aonecharti{1}$ to $\aonecharti{2}$).
\end{defi}

\begin{defi}\label{def:substate}\label{def:onebisim-collapse}\nf
  Let $\aonechart = \tuple{\verts,\actions,\sone,\start,\transs,\termexts}$ be a \onechart.
  
  By $\sonebisimon{\aonechart}$ we denote \emph{\onebisimilarity\ on $\aonechart$},
  the largest \onebisimulation\ (which is the union of all \onebisimulations) between $\aonechart$ and $\aonechart$ itself.
   % (which is the union of all \onebisimulations\ between $\aonechart$ and $\aonechart$).
  If $\bverti{1} \onebisimon{\aonechart} \bverti{2}$ holds for vertices $\bverti{1},\bverti{2}\in\verts$,
  then we say that \emph{$\bverti{1}$ and $\bverti{2}$ are \onebisimilar~in~$\aonechart$}.
  \\\indent
  We call $\aonechart$ %\onechart~$\aonechart = \tuple{\verts,\actions,\sone,\start,\transs,\termexts}$ 
  \emph{\onecollapsed}, and a \emph{\onebisimulation\ collapse}, 
  if $\sonebisimon{\aonechart} = \sidrelon{\verts}$ holds, that is, if \onebisimilar\ vertices of $\aonechart$ are identical.
  % if the only \onebisimulation\ between $\achart$ and $\achart$
  % is the identity relation $\sidrelon{\verts} = \descsetexp{ \pair{\cvert}{\cvert} }{ \cvert\in\verts }$ on $\verts$.
  If, additionally, $\aonechart$ does not contain any \onetransitions,
  then we call $\aonechart$  \emph{collapsed}, and a \emph{bisimulation collapse}. 
  
  Let $\bverti{1},\bverti{2}\in\verts$.
  We say that \emph{$\bverti{1}$ is a substate of $\bverti{2}$}, denoted by $\bverti{1} \isonesubstateofin{\aonechart} \bverti{2}$,
  if the pair $\pair{\bverti{1}}{\bverti{2}}$ forth-progresses to \onebisimilarity\ on $\aonechart$ in the sense of the following conditions:
  \begin{itemize}[labelindent=0.5em,leftmargin=*,align=left,itemsep=0.25ex,labelsep=0.75em]
      
    \item[(prog-forth)]
      $ 
      \begin{aligned}[t]
        &
        \forall \bvertacci{1}\in\vertsi{1}
            \forall \aact\in\actions
                \bigl(\,
                  \bverti{1} \ilt{\aact} \bvertacci{1} 
                  \\[-0.65ex]
                  & \hspace*{-3ex}%\quad
                    \;\;\Longrightarrow\;\;
                      \exists \bvertacci{2}\in\vertsi{2}
                        \bigl(\, \bverti{2} \ilt{\aact} \bvertacci{2} 
                                   \logand
                                 \bvertacci{1} \onebisimon{\aonechart} \bvertacci{2}  \,)
              \,\bigr) \punc{,}
       \end{aligned}    
       $
      
    \item[(prog-termination)]
      $ \oneterminates{\bverti{1}}
          \;\;\Longrightarrow\;\;
            \oneterminates{\bverti{2}} \punc{.}$
  \end{itemize}
\end{defi}

% %----------
% \subsection{Process semantics of regular expressions,
%              and Milner's proof system}%
%       \label{procsem:milnersys::prelims}
% %----------

\begin{defi}\label{def:chartof}
  The \emph{chart interpretation of} 
  a star expression~$\astexp\in\StExpover{\actions}$
  is the ($\sone$-tr.\ free) chart 
  $\chartof{\astexp} = \tuple{\vertsof{\astexp},\actions,\sone,\astexp,\\\transs\cap(\vertsof{\astexp}{\times}\actions{\times}\vertsof{\astexp}),\termexts{\cap}\vertsof{\astexp}}$
  %a (\onetransition\ free) chart 
  where $\vertsof{\astexp}$ consists of all star expressions that are reachable from $\astexp$
  via transitions of the labeled transition relation $ \transs \subseteq \StExpover{\actions}\times\actions\times\StExpover{\actions}$,
  which is defined, together with the imm.-termination relation $\sterminates \subseteq \StExpover{\actions}$,
  by derivability in the transition system specification (TSS)~$\StExpTSSover{\actions}$,
  where $\aact\in\actions$, $\astexp,\astexpi{1},\astexpi{2},\astexpacc\in\StExpover{\actions}$:      
  \begin{gather*}
   \begin{aligned}
     &
     \AxiomC{\phantom{$\terminates{\stexpone}$}}
     \UnaryInfC{$\terminates{\stexpone}$}
     \DisplayProof
     & \hspace*{-1.5ex} &
     \AxiomC{$ \terminates{\astexpi{1}} $}
     \UnaryInfC{$ \terminates{(\stexpsum{\astexpi{1}}{\astexpi{2}})} $}
     & \hspace*{2ex} &
     \AxiomC{$ \terminates{\astexpi{i}} $}
     \UnaryInfC{$ \terminates{(\stexpsum{\astexpi{1}}{\astexpi{2}})} $}
     \DisplayProof
     & \hspace*{2ex} &
     \AxiomC{$\terminates{\astexpi{1}}$}
     \AxiomC{$\terminates{\astexpi{2}}$}
     \BinaryInfC{$\terminates{(\stexpprod{\astexpi{1}}{\astexpi{2}})}$}
     \DisplayProof
     & \hspace*{2ex} &
     \AxiomC{$\phantom{\terminates{\stexpit{\astexp}}}$}
     \UnaryInfC{$\terminates{(\stexpit{\astexp})}$}
     \DisplayProof
   \end{aligned} 
   \displaybreak[0]\\%[0.25ex]
   \begin{aligned}
     & 
     \AxiomC{$\phantom{\astexpacci{i} a \:\lt{a}\: \stexpone\astexpacci{i}}$}
     \UnaryInfC{$a \:\lt{a}\: \stexpone$}
     \DisplayProof
     & \hspace*{1ex} &
     \AxiomC{$ \astexpi{i} \:\lt{a}\: \astexpacci{i} $}
     \UnaryInfC{$ \stexpsum{\astexpi{1}}{\astexpi{2}} \:\lt{a}\: \astexpacci{i} $}
     \DisplayProof 
     & \hspace*{1ex} &
     \AxiomC{$ \phantom{\astexpacci{i}}  \astexp \:\lt{a}\: \astexpacc     \phantom{\astexpacci{i}} $}
     \UnaryInfC{$\stexpit{\astexp} \:\lt{a}\: \stexpprod{\astexpacc}{\stexpit{\astexp}}$}
     \DisplayProof
   \end{aligned}
   \displaybreak[0]\\%[0.25ex]
   \begin{aligned}
     &
     \AxiomC{$ \astexpi{1} \:\lt{a}\: \astexpacci{1} $}
     \UnaryInfC{$ \stexpprod{\astexpi{1}}{\astexpi{2}} \:\lt{a}\: \stexpprod{\astexpacci{1}}{\astexpi{2}} $}
     \DisplayProof
     & \hspace*{1ex} &
     \AxiomC{$\terminates{\astexpi{1}}$}
     \AxiomC{$ \astexpi{2} \:\lt{a}\: \astexpacci{2} $}
     \BinaryInfC{$ \stexpprod{\astexpi{1}}{\astexpi{2}} \:\lt{a}\: \astexpacci{2} $}
     \DisplayProof
   \end{aligned}
  \end{gather*} 
\end{defi}

\renewcommand{\prod}{\mathrel{\cdot}}
\renewcommand{\assocstexpsum}{\textnf{(A1)}}
\renewcommand{\neutralstexpsum}{\textnf{(A2)}}
\renewcommand{\commstexpsum}{\textnf{(A3)}}
\renewcommand{\idempotstexpsum}{\textnf{(A4)}}
\renewcommand{\assocstexpprod}{\textnf{(A5)}}
\renewcommand{\rightdistr}{\textnf{(A6)}}
\renewcommand{\leftidstexpprod}{\textnf{(A7)}}
\renewcommand{\rightidstexpprod}{\textnf{(A8)}}
\renewcommand{\deadlockax}{\textnf{(A9)}}
\renewcommand{\recdefstexpit}{\textnf{(A10)}}
\renewcommand{\termbodystexpit}{\textnf{(A11)}}
\begin{defi}\label{def:milnersys}
  Milner's proof system \milnersys\ on star expressions has the following axioms
  (here numbered differently):
\begin{alignat*}{4}
  \assocstexpsum & \;\, & 
    e + (f + g) 
      & {} \formeq 
    (e + f) + g
    & \quad\;\;
  \leftidstexpprod & \;\, & 
    e
      & {} \formeq 
    1 \prod e
    \\[-0.25ex]
  \neutralstexpsum & &
    e + 0 
      & {} \formeq 
    e
    &  
  \rightidstexpprod & &
    e
      & {} \formeq
    e \prod 1
    \displaybreak[0]\\[-0.25ex]
  \commstexpsum & &
    e + f 
      & {} \formeq 
    f + e
    &
  \deadlockax & &
    0
      & {} \formeq
    0 \prod e
    \displaybreak[0]\\[-0.25ex]
  \idempotstexpsum & &
    e + e 
      & {} \formeq 
    e   
    & \quad
  \recdefstexpit & & 
    e^*
      & {} \formeq
    1 + e \prod e^*
    \displaybreak[0]\\[-0.25ex]
  \assocstexpprod & &
    e \prod (f \prod g) 
      & {} \formeq 
    (e \prod f) \prod g 
    &
  \termbodystexpit & \; &
      e^*
        & {} \formeq
      (1 + e)^*
    \displaybreak[0]\\[-0.25ex]
  \rightdistr & &
    (e + f) \prod g 
      & {} \formeq
    e \prod g  +  f \prod g
\end{alignat*}
% \begin{alignat*}{4}
%   \assocstexpsum & \;\; & 
%     e + (f + g) 
%       & {} \formeq 
%     (e + f) + g
%     \\[-0.25ex]
%   \neutralstexpsum & &
%     e + 0 
%       & {} \formeq e
%     \displaybreak[0]\\[-0.25ex]
%   \commstexpsum & &
%     e + f 
%       & {} \formeq 
%     f + e 
%     \displaybreak[0]\\[-0.25ex]
%   \idempotstexpsum & &
%     e + e 
%       & {} \formeq 
%     e 
%     \displaybreak[0]\\[-0.25ex]
%   \assocstexpprod & &
%     e \prod (f \prod g) 
%       & {} \formeq 
%     (e \prod f) \prod g
%     \displaybreak[0]\\[-0.25ex]
%   \rightdistr & &
%     (e + f) \prod g 
%       & {} \formeq
%     e \prod g  +  f \prod g
%     \displaybreak[0]\\[-0.25ex]
%   \leftidstexpprod & \;\; & 
%     1 \prod e 
%       & {} \formeq
%     e
%     \displaybreak[0]\\[-0.25ex]
%   \rightidstexpprod & &
%     e \prod 1
%       & {} \formeq
%     e 
%     & \quad
%   \recdefstexpit & \;\; & 
%     e^*
%       & {} \formeq
%     1 + e \prod e^*
%     \\[-0.25ex]
%   \deadlockax & &
%     0 \prod e
%       & {} \formeq
%     e
%     &
%   \termbodystexpit & \; &
%       e^*
%         & {} \formeq
%       (1 + e)^*
% \end{alignat*}
  The rules of $\milnersys$ are the basic inference rules of equational logic
  (reflexivity, symmetry, transitivity of $\formeq$,
   compatibility of $\formeq$ with $+$, $\prod$, $(\cdot)^*$)
  as well as the fixed-point rule $\RSPstar$: 
  \begin{center}\renewcommand{\fCenter}{\formeq}
   \Axiom$ e \fCenter f \prod e + g $
   \RightLabel{\RSPstar\ {\small (if $\notterminates{f}$)}}
   \UnaryInf$ e \fCenter f^* \prod g $
   \DisplayProof
  \end{center}
  By $\astexpi{1} \milnersyseq \astexpi{2}$ we denote that $\astexpi{1} = \astexpi{2}$ is derivable in \milnersys.
  
  By $\milnersysmin$ we denote the purely equational part of $\milnersys$
    that results by dropping the rule scheme \RSPstar\ from $\milnersys$.
\end{defi}

% %----------
% \subsection{Provable, and complete provable, solutions}
%   \label{solutions::prelims}
% %----------

\begin{defi}\label{def:provable:solution} % be an \eqlogicbased\ proof system for star expressions over $\actions$ 
  While we formulate the stipulations below for \oneLTSs, we will use them also for \onecharts. 
  So we let $\aoneLTS = \tuple{\verts,\actions,\sone,\transs,\exts}$ be a \oneLTS,
  and we let $\asys \in \setexp{\milnersys,\milnersysmin}$. 
  
  By a \emph{star expression function on $\aoneLTS$} we mean a function $\sasol \funin \verts \to \StExpover{\actions}$
    on the vertices of $\aoneLTS$.   
  Now we let $\avert\in\verts$.   
  We say that such a star expression function $\sasol$ on $\aoneLTS$ 
    is an \emph{\provablein{\asys} solution of $\aoneLTS$ at $\avert$} if it holds that:
    \begin{center}   
      $
      \asol{\avert}
        \,\eqin{\asys}\,
          \stexpsum{\terminatesconstof{\aoneLTS}{\avert}}
                   {\displaystyle\sum_{i=1}^{n} \stexpprod{\aoneacti{i}}{\asol{\averti{i}}}} 
      $%\vspace*{-2pt}
    \end{center}  
    % holds 
    given that 
    $\transitionsinfrom{\aoneLTS}{\avert}
       =
     \descsetexpbig{ \avert \lt{\aoneacti{i}} \averti{i} }{ i \in\setexp{1,\ldots,n} }$
    is a (possibly redundant) list representation
    of transitions from $\avert$ in~$\aoneLTS$, 
    % induce the \actiontarget\ set
    % \alert{%
    % $\atsiof{\aoneLTS}{\avert}
    %    =
    %  \setexp{ \pair{\aoneacti{1}}{\averti{1}}, \ldots, \pair{\aoneacti{n}}{\averti{n}} }$}
    % (with possible repetitions in this list representations) 
    and where $\terminatesconstof{\aoneLTS}{\avert}$
    is the \emph{termination constant $\terminatesconstof{\aoneLTS}{\avert}$ of $\aoneLTS$ at $\avert$}
      defined as $\stexpzero$ if $\notterminates{\avert}$,
                      and as $\stexpone$ if $\terminates{\avert}$.
    This definition does not depend on the specifically chosen list representation of $\transitionsinfrom{\aoneLTS}{\avert}$, %of $\atsiof{\aoneLTS}{\avert}$,
    because $\asys$ %extends \ACI, and therefore it 
                    contains the associativity, commutativity, and idempotency axioms~for~$\sstexpsum$.
  
  By an \emph{\provablein{\asys} solution of $\aoneLTS$} 
    (with \emph{principal value $\asol{\start}$} at the start vertex $\start$)
    we mean a star expression function $\sasol$ on $\aoneLTS$ that is an \provablein{\asys} solution of $\aoneLTS$ at every vertex~of~$\aoneLTS$.
    
  We say that an \provablein{\asys} solution $\sasol$ of $\aoneLTS$ is~\emph{\completewrt{\asys}}~if\hspace*{0.25pt}:
  \begin{equation}\tag{concept \ref{contrib:concept:5}}
    \bverti{1} \onebisimvia{\aoneLTS} \bverti{2}
      \;\;\Longrightarrow\;\;
    \asol{\bverti{1}} \eqin{\asys} \asol{\bverti{2}} \punc{,} 
  \end{equation}    
  holds for all $\bverti{1},\bverti{2}\in\verts$, 
    that is, if values of the solution $\sasol$ at \onebisimilar\ vertices of $\aoneLTS$ are \provablyin{\asys} equal. 
\end{defi}

% Schmid, Rot, and Silva \cite{schm:rot:silv:2021} 

The following lemma gathers preservation statements of (complete) provable solutions
  under (functional) \onebisimilarity\ that are crucial for the completeness proof.  

\begin{lem}\label{lem:preservation:sols}
  On \weaklyguarded\ \onecharts, 
    the following statements hold for all star expressions $\astexp\in\StExp\,$: 
    % On \weaklyguarded\ \onecharts, the following preservation statements hold 
    %   for (\completewrt{\milnersys}) \provablein{\milnersys} solvability,
    %     for all star expressions $\astexp\in\StExp\,$ it holds:
    %  
  \begin{enumerate}[label={{\rm (\roman*)}},leftmargin=1.9em,align=right,itemsep=0.25ex]
    \item{}\label{it:conv:funbisims:lem:preservation:sols}
      % On \weaklyguarded\ \onecharts,
      \Provablein{\milnersys} solvability with principal value $\astexp$
        is preserved under converse functional \onebisimilarity.  
    
    \item{}\label{it:collapse:lem:preservation:sols} 
      \Completewrt{\milnersys} \provablein{\milnersys} solvability with principal value $\astexp$ of a \wg\ \onechart~$\aonechart$ 
        implies \provablein{\milnersys} solvability with principal value $\astexp$ of the bisimulation collapse~of~$\aonechart$.~(See~\ref{contrib:concept:5}.)
        % If a \wg\ \onechart~$\aonechart$ has a \completewrt{\milnersys} \provablein{\milnersys} solution with principal value $\astexp$,
        %  then also the bisimulation collapse of $\aonechart$ has a \provablein{\milnersys} solution with principal~value~$\astexp$. 
            
    \item{}\label{it:onebisims:lem:preservation:sols}
      \Completewrt{\milnersys} \provablein{\milnersys} solvability with principal value $\astexp$
        is preserved under \onebisimilarity.    
  \end{enumerate}
\end{lem} 

\section{LLEE-$\protect\sone$-Charts} %{ and LLEE-$\protect\sone$-LTSs} %{\protect\LLEEOnechartsbf\ and \LLEEoneLTSsbf}%
  \label{LLEEonecharts}
%--------------------------------------------------
We use the adaptation of the `loop existence and elimination property' \LEE\ from \cite{grab:fokk:2020:lics} 
  to \onecharts\ as described in \cite{grab:2020:TERMGRAPH2020-postproceedings:arxiv,grab:2021:calco}.
Here we only briefly explain the concept by examples,
  and refer to \cite{grab:2020:TERMGRAPH2020-postproceedings:arxiv,grab:2021:calco} and to the appendix for the definitions. 
Crucially, we gather statements from \cite{grab:2020:TERMGRAPH2020-postproceedings:arxiv,grab:2021:calco} that we need
  for~the~proof.

\LEE\ is defined by a stepwise elimination procedure of `loop \subonecharts' from a given \onechart~$\aonechart$.
A run of this procedure is illustrated in Fig.~\ref{fig:ex:LEE}.
A \onechart~$\aoneloop = \tuple{\verts,\actions,\sone,\start,\transs,\termexts}$ is hereby called a \emph{loop \onechart} if
  it satisfies three 
               conditions:%\enlargethispage{2.5ex}
  \begin{enumerate}[label={(L\arabic{*})},align=right,leftmargin=*,itemsep=0ex]
    \item{}\label{loop:1}
      There is an infinite path from the start vertex $\start$.
    \item{}\label{loop:2}  
      Every infinite path from $\start$ returns to $\start$ after a positive number of transitions.
      % (and so visits $\start$ infinitely often).
    \item{}\label{loop:3}
      Immediate termination is only permitted at the start vertex, that is, $\termexts\subseteq\setexp{\start}$.
  \end{enumerate}
A \emph{loop \subonechart\ of} a \onechart~$\aonechart$
    is a loop \onechart~$\aoneloop$
    that is a \subonechart\ of $\aonechart$ 
      with some vertex $\avert\in\verts$ of $\aonechart$ as start vertex,
    such that $\aoneloop$ is constructed, for a nonempty set $\asettranss$ of transitions of $\aonechart$ from $\avert$,
    by all paths that start with a transition in $\asettranss$ and continue onward until $\avert$ is reached again
  (so the transitions in $\asettranss$ are the \loopentrytransitions~of~$\aoneloop$).
\emph{Eliminating a loop \subonechart\ $\aoneloop$ from a \onechart\ $\aonechart$}
  consists of removing all \loopentrytransitions\ of $\aoneloop$ from $\aonechart$, 
  and then also removing all vertices and transitions that become unreachable. 
Fig.~\ref{fig:ex:LEE} shows a successful three-step run of the loop elimination procedure.   
A \onechart\ $\aonechart$ has the \emph{loop existence and elimination property} (\LEE)
  if the procedure, started on~$\aonechart$, of repeated eliminations of loop \subonecharts\
  results in a \onechart\ without an infinite path.
If, in a successful elimination process from a \onechart~$\aonechart$,
  \loopentrytransitions\ are never removed from the body of a previously eliminated loop \subonechart,
  then we say that $\aonechart$ satisfies \emph{layered \LEE} (\LLEE),
  and is a \emph{\LLEEonechart}.  
\LLEEoneLTSs\ are \oneLTSs\ that are defined analogously.  
While the property \LLEE\ leads to a formally easier concept of `witness', it is equivalent~to~\LEE. 
Since the resulting \onechart~$\aonechart'''$ in Fig.~\ref{fig:ex:LEE} does not have 
an infinite path, and no \loopentry\ transitions have been removed from a previously eliminated loop \subonechart,
we conclude that the initial \onechart~$\aonechart$ satisfies \LLEE\ as well as \LEE.%
  \input{figs/fig-ex-LEE.tex}%
  \input{figs/fig-ex-LLEEw-1-2-3.tex}%
  
A \emph{\LLEEwitness\ $\aonecharthat$ of} a \onechart~$\aonechart$
is the recording of a successful run of the loop elimination procedure
by attaching to a transition $\atrans$ of $\aonechart$ the marking label $n$ for $n\in\natplus$ 
 (in pictures indicated as $\looplab{n}$, in steps as $\sredi{\looplab{n}}$) 
 forming a \emph{\loopentry\ transition}
if $\atrans$ is eliminated in the $n$\nb-th step,
and by attaching marking label $0$ to all other transitions of $\aonechart$
 (in pictures neglected, in steps indicated as $\sredi{\bodylab}$)
 forming a \emph{body transition}. 
Formally, \LLEEwitnesses\ arise as \emph{\entrybodylabeling{s}} from \onecharts,
  and are charts in which the transition labels are pairs of action labels over $\actions$,
  and marking labels~in~$\nat$.
  %
% For vertices $\bvert$ and $\avert$   
%   we say that $\bvert$ \emph{\txtloopsbackto} $\avert$, denoted by $\bvert \loopsbackto \avert$,
%   if $\avert \comprewrels{\sredi{\looplab{n}}}{\sredrtci{\bodylab}} \bvert \redtci{\bodylab} \avert$ holds with $n\in\natplus$ and
%      such that $\avert$ is only encountered again at the end.  
  
The \LLEEwitness~$\aonecharthati{1}$ in Fig.~\ref{fig:ex:LLEEw-1-2-3}\vspace*{-1.5pt}  
  arises from the run of the loop elimination procedure in Fig.~\ref{fig:ex:LEE}.\vspace*{-1.5pt}
The \LLEEwitnesses~$\aonecharthati{2}$ and $\aonecharthati{3}$ of $\aonechart$ in Fig.~\ref{fig:ex:LLEEw-1-2-3} 
  record two other successful runs of the loop elimination procedure of length 4 and 2, respectively, 
  where for $\aonecharthati{3}$ we have permitted to eliminate two loop subcharts at different vertices together in the first step.

\begin{defi}[\protect\onecharts\ with \protect\LLEEonelim]%[\onetransition\ limited \LLEEonecharts]
            \label{def:LLEEonelim}
  Let $\aonechart$ be a \onechart. 
  
  Let $\aonecharthat$ be a \LLEEwitness\ of $\aonechart$.  
  We say that $\aonecharthat$ \emph{is \onetransition\ limited}
    if every \onetransition\ of $\aonechart$ lifts to a \emph{\backlink} (a transition from the body of a loop \subonechart\ back to its start) in $\aonecharthat$. 
  Fixing also a weaker property, we say that
     $\aonecharthat$ is \emph{guarded} if all of its \loopentry\ transitions are proper~transitions. 
    
  We say that $\aonechart$ \emph{satisfies \LLEEonelim}, and % that $\aonechart$
                                                            \emph{is \onetransition\ limited},
    if $\aonechart$ has a \onetransition\ limited \LLEEwitness.
  We say that $\aonechart$ \emph{is guarded} if $\aonechart$ has a guarded \LLEEwitness.
\end{defi}

  % We say that $\aonechart$ \emph{has the property \LLEEonelim} 
  % if $\aonechart$ is weakly guarded, satisfies \LLEE, and is \onetransition\ limited.  

  % We note that \onetransition\ limited \LLEEwitnesses\ are also guarded.
%   because \loopentry\ transitions in \onetransition\ limited \LLEEwitnesses\ cannot be \onetransitions,
%      as they \loopentry\ transitions are not backlinks.

% guarantee that the underlying \onecharts\ are `guarded' 
%   in the sense that all \loopentrytransitions\ are proper transitions, because no \loopentry\ transitions can be a backlink. 

% Note that \onetransition\ limited \LLEEwitnesses\ of weakly guarded \LLEEwitnesses\ are `guarded' 
%   in the sense that all \loopentrytransitions\ must be proper transitions, because
%     weak guardedness prevents \onetransition\ self-loops, and \loopentry\ transitions cannot be backlinks. 

We note that \onetransition\ limited \LLEEonecharts\ are guarded
  (since \loopentry\ transition are not backlinks),
and guarded \LLEE-\onecharts\ are weakly guarded
(since in a guarded \LLEEwitness\ every \onetransition\ path
 is a \bodytransition\ path, which as in every \LLEEwitness\ is guaranteed to be finite).

\begin{lem}\label{lem:LLEEonechart:2:LLEEonelimonechart}
  Every %finite 
        \weaklyguarded\ \LLEEonechart\
  is \onebisimilar\ to a %finite 
                         (guarded)
                         \onechart\ with \LLEEonelim.                        
\end{lem}

% \todo{guarded \LLEEwitness: all \loopentry\ transitions are proper transitions}

% %----------
% \subsection{Extraction and unique solvability}%
%   \label{extraction:unique-solvability::LLEEonecharts}
% %----------

Two crucial properties of \LLEEonecharts\ that motivate
  their use, like for \LLEEcharts\ earlier in \cite{grab:fokk:2020:lics,grab:fokk:2020:lics:arxiv},
  are their provable solvability and unique solvability modulo provability~in~$\milnersys$.
The two lemmas below that express these properties
  are generalizations to \LLEEonecharts\ of Prop.~5.5 and Prop.~5.8 in \cite{grab:fokk:2020:lics,grab:fokk:2020:lics:arxiv}),
  and have been proved in \cite{grab:2021:calco:arxiv,grab:2021:calco}.  
% These properties are expressed by the two lemmas below,
%   which are generalizations to \LLEEonecharts\ of Prop.~5.5 and Prop.~5.8 in \cite{grab:fokk:2020:lics,grab:fokk:2020:lics:arxiv}),
%   and were proved in \cite{grab:2021:calco:arxiv,grab:2021:calco}.

\begin{lem}\label{lem:LLEEonecharts:extraction}
  From every guarded \LLEEwitness~$\aonecharthat$
    of a (\weaklyguarded) \LLEEonechart\ $\aonechart$
      a \provablein{\milnersysmin} solution $\sextrsolof{\aonecharthatsubscript}$ 
                                     of $\aonechart$ can be extracted effectively.
\end{lem}

\begin{lem}\label{lem:LLEEonecharts:uniquely:solvable}
  For every guarded \LLEEonechart\ $\aonechart$ it holds that
    any two \provablein{\milnersys} solutions of $\aonechart$
    are \provablyin{\milnersys} equal.
\end{lem}

%------------------------------------
% Figure proof structure
%------------------------------------
\input{figs/fig-proof-structure-exp.tex}
%------------------------------------
% \afterpage{\FloatBarrier}

% %----------
% \subsection{Solution invariance under transfer functions}
% %----------

In Sect.~\ref{nearcollapsed:LLEEonecharts:complete:solvable} we will need
  a consequence of Lem.~\ref{lem:LLEEonecharts:uniquely:solvable},
    namely provable invariance of provable solutions under `transfer functions',
      which define functional \onebisimulations. 
While this statement holds also for \LLEEonecharts,
  we have to formulate it for \LLEEoneLTSs\ for use later in Sect.~\ref{nearcollapsed:LLEEonecharts:complete:solvable}.        
    % We formulate the definitions in this section for \oneLTSs,
    %   because we will need them for \oneLTSs\ later in Sect.~\ref{nearcollapsed:LLEEonecharts:complete:solvable}. 
    % They apply, however, equally for \onecharts.  

\begin{defi}\label{def:transfer:function}
  A \emph{transfer (partial) function} between \oneLTSs\ $\aoneLTSi{1}$ and $\aoneLTSi{2}$,
      for $\aoneLTSi{i} = \tuple{\statesi{i},\actions,\sone,\transsi{i},\termextsi{i}}$ where $i\in\setexp{1,2}$,
    is a partial function $\sphifun \funin \statesi{1} \rightharpoonup \statesi{2}$ 
      whose graph $\descsetexp{ \pair{\astate}{\phifun{\astate}} }{ \astate \in \statesi{1} }$
        is a \onebisimulation\ between $\aoneLTSi{1}$ and $\aoneLTSi{2}$.
\end{defi}

% \begin{defi}[provable invariance of solutions under partial functions]\label{def:prov:invariance:part:funs}
%   We consider a \oneLTS~$\aoneLTS = \tuple{\states,\actions,\sone,\transs,\termexts}$,
%   and a partial function $\sphifun \funin \states \rightharpoonup \states$ on its set $\states$ of states. 
%   Let $\asys$ be an \eqlogicbased\ proof system for star expressions over $\actions$.
%   
%   Let $\sasol \funin \states \to \StExpover{\actions}$ be a star-expression function on $\states$. 
%   We say that $\sasol$ is \emph{\provablyin{\asys} invariant under $\sphifun$}
%     if for all $\bvert\in\states$ it holds: 
%   \begin{equation}\label{eq:def:prov:invariance:part:funs}
%     \asoli{1}{\bvert}
%       \eqin{\asys}
%     \asoli{2}{\phifun{\bvert}} \punc{,}
%       \qquad\quad
%         \text{for all $\bvert\in\domof{\sphifun}$.} 
%   \end{equation} 
%   
%   Accordingly, we say for a \onechart~$\aonechart = \tuple{\verts,\actions,\sone,\start,\transs,\termexts}$
%   and a star expression function \mbox{$\sasol \funin \verts \to \StExpover{\actions}$} on $\verts$
%   that $\sasol$ is \emph{\provablyin{\asys} invariant under $\sphifun$} 
%     if~\eqref{eq:def:prov:invariance:part:funs}~holds.
% \end{defi}

% \begin{lem}\label{lem:LLEEoneLTSs:transfer:functions}
%   \Provablein{\milnersys} solutions on \wg\ \LLEEoneLTSs\
%     are \provablyin{\milnersys} invariant under transfer functions.
% \end{lem}

\begin{lem}\label{lem:sol:inv:under:tfuns}
  Let $\sphifun \funin \vertsi{1} \rightharpoonup \vertsi{2}$ 
    be a transfer function between \LLEEoneLTS~$\aoneLTSi{i} = \tuple{\vertsi{i},\actions,\sone,\transsi{i},\termextsi{i}}$, for $i\in\setexp{1,2}$. 
  Then for all \provablein{\milnersys} solutions $\sasoli{1}$ of $\aoneLTSi{1}$, and $\sasoli{2}$ of $\aoneLTSi{2}$ it holds
  that $\sasoli{1}$ coincides \provablyin{\milnersys} with the precomposition $\scompfuns{\sasoli{2}}{\sphifun}$~of~$\sasoli{2}$~and~$\sphifun$:
  \begin{equation*}%\label{eq:def:prov:invariance:part:funs}
    \asoli{1}{\bvert}
      \eqin{\milnersys}
    \asoli{2}{\phifun{\bvert}} \punc{,}
      \quad
        \text{for all $\bvert\in\domof{\sphifun}\punc{.}$} 
  \end{equation*} 
\end{lem}

% %----------
% \subsection{\protect\Onechart\ interpretation}%
%   \label{onechart-interpretation:unique-solvability::LLEEonecharts}
% %----------

A substantial obstacle for the use of \LLEEonecharts\ was already recognized in \cite{grab:fokk:2020:lics}:
  the chart interpretation of star expressions does not in general define \LLEEcharts. %^\onecharts\ with \LLEE. 
This obstacle can, however, be navigated successfully by using the result from \cite{grab:2020:TERMGRAPH2020-postproceedings:arxiv,grab:2021:TERMGRAPH2020-postproceedings}
  that a variant chart interpretation can be defined 
    that produces \onebisimilar\ \LLEEonecharts\ instead.

\begin{lem}\label{lem:onechart:int}
  For every star expression $\astexp$, 
    there is a \emph{\onechart\ inter\-pre\-ta\-tion $\onechartof{\astexp}$ of $\astexp$}
    that has the following properties:
    \begin{enumerate}[label=(\roman{*}),align=right,labelsep=0.75ex,itemsep=0.2ex]
     \item{}\label{it:1:lem:onechart:int}\label{IVone.i}
       $\onechartof{\astexp}$ is a \onetransition\ limited (guarded) \LLEEonechart,
     \item{}\label{it:2:lem:onechart:int}\label{IVone.ii} 
       $\onechartof{\astexp} \funonebisim \chartof{\astexp}$, and hence also $\onechartof{\astexp} \onebisim \chartof{\astexp}$. 
         % that is, $\onechartof{\astexp}$ is \onebisimilar\ to the chart interpretation $\chartof{\astexp}$ of $\astexp$,
         
     \item{}\label{it:3:lem:onechart:int}\label{IVone.iii} 
       $\astexp$ is the principal value of a \provablein{\milnersys} solution of~$\onechartof{\astexp}$. 
    \end{enumerate}\smallskip
\end{lem}

%----------
\section{Completeness proof based on lemmas}%
  \label{completeness:proof}
%----------

We anticipate the completeness proof for Milner's system by basing 
  it on the following lemmas, which are faithful abbreviations of statements
    as formulated in other sections.    
The chosen acronyms for these lemmas stem from the letters that are typeset in boldface italics
  in their statements: 
  %
  % The completeness proof for Milner's system~$\milnersys$ uses the following lemmas
  %   % where we link in brackets to the corresponding more precise statements in later sections
  %   % (the acronyms are explained
  %   %  by typesetting the motivating letters in the lemmas in boldface italics):
  %    (the acronyms stem from the letters in the
  %     statements of the lemmas that are typeset in boldface italics):   
  %
\begin{enumerate}[label={{\bf (SE)$_{\protect\onescriptbs}$}},leftmargin=*,align=right,labelsep=0.75ex,itemsep=0.45ex] 
    % \item[{\crtcrossreflabel{{\bf (SI)}}[SIone]}]%[{\crtcrossreflabel{{\bf (SI$_{\hspace*{0pt}\text{\nf 1}}$)}}[SIone]}]%[{\bf (SI$_{\hspace*{0pt}\text{\nf 1}}$)}]%
    %   Every star expression $\astexp$ is the principal value 
    %     of a provable \textbf{\textit{s}}olution of its chart \textbf{\textit{i}}nterpretation~$\chartof{\astexp}$.
    %     % (Refined by~Theorem~\ref{thm:chart-int:solvable:milnersysmin}.)
      %
  \item[{\crtcrossreflabel{{\bf (IV)}}[IVone]}]%[{\crtcrossreflabel{{\bf (IV$_{\hspace*{-1.5pt}\text{\nf 1}}$)}}[IVone]}]%[{\bf (IV$_{\hspace*{-1.5pt}\text{\nf 1}}$)}]%
    For every star expression $\astexp$, 
    there is a \emph{\onechart\ \textbf{\textit{i}}nter\-pre\-ta\-tion (\textbf{\textit{v}}ariant) $\onechartof{\astexp}$ of $\astexp$}
      with the properties \ref{it:1:lem:onechart:int}, \ref{it:2:lem:onechart:int}, and \ref{it:3:lem:onechart:int}
        in Lem.~\ref{lem:onechart:int} (see above).

  \item[{\crtcrossreflabel{\bf (T)}[Tone]}]%[{\crtcrossreflabel{{\bf (P$_{\hspace*{-1pt}\text{\nf 1}}$)}}[Tone]}]%[{\bf (P$_{\hspace*{-1pt}\text{\nf 1}}$)}]%
    Provable solutions can be \textbf{\textit{t}}ransferred backwards over a transfer function between weakly guarded \onecharts.
    % All provable solutions can be \textbf{\textit{p}}ulled back from the target to the source \onechart\ of  a functional \onebisimulation,
    % if both are weakly guarded.
    (See Lem.~\ref{lem:preservation:sols},~\ref{it:conv:funbisims:lem:preservation:sols}). 
    %  A refined version for labeled transitions systems will be shown in Lemma~\ref{lem:transf:sol:via:transf:fun:oneLTSs:tc}.
    %  Its proof uses {\bf (P$_{\hspace*{-1pt}\text{\nf \st{1}}}$)}, see Theorem~\ref{thm:funbisims:reflect:solvability:charts},
    %  and \ref{RSone}.) %{\bf (RS$_{\hspace*{0pt}\text{\nf 1}}$)}.)
     %
  \item[{\crtcrossreflabel{{\bf (E)}}[Eone]}]%[{\crtcrossreflabel{{\bf (E$_{\hspace*{0pt}\text{\nf 1}}$)}}[Eone]}]%[{\bf (E$_{\hspace*{0pt}\text{\nf 1}}$)}]%
    From every %weakly 
               guarded \LLEEonechart\ $\aonechart$ a provable solution of $\aonechart$ can be \textbf{\textit{e}}xtracted.
    (See Lem.~\ref{lem:LLEEonecharts:extraction}.)\enlargethispage{2ex}
 \item[{\crtcrossreflabel{{\bf (CR)}}[CRone]}]%[{\crtcrossreflabel{{\bf (CR$_{\hspace*{-0.25pt}\text{\nf 1}}$)}}[CRone]}]%[{\bf (CR$_{\hspace*{-0.25pt}\text{\nf 1}}$)}]%
   Every %\fb, 
         %\wg\
         guarded \LLEEonechart\ 
   can be transformed into a \onebisimilar\ (guarded) \textbf{\textit{cr}}ystallized (\LLEE-)\onechart.
   (See Thm.~\ref{thm:crystallization:nearcollapse}, based on Def.~\ref{def:crystallized}.)
 \item[{\crtcrossreflabel{{\bf (CN)}}[CNone]}]%[{\crtcrossreflabel{{\bf (CN$_{\hspace*{-1pt}\text{\nf 1}}$)}}[CNone]}]%[{\bf (CN$_{\hspace*{-1pt}\text{\nf 1}}$)}]%
   Every \textbf{\textit{c}}rystallized \onechart\ is \textbf{\textit{n}}ear-col\-lapsed.
   (Lem.~\ref{lem:crystallized:is:nearcollapsed}.)
  \item[{\crtcrossreflabel{{\bf (N\hspace*{-0.75pt}C)}}[NCone]}]%[{\crtcrossreflabel{{\bf (N\hspace*{-0.75pt}C$_{\hspace*{0pt}\text{\nf 1}}$)}}[NCone]}]%[{\bf (N\hspace*{-0.75pt}C$_{\hspace*{0pt}\text{\nf 1}}$)}]%
    Solutions extracted %by \ref{Eone} %{\bf (E$_{\hspace*{0pt}\text{\nf 1}}$)}) 
                        from \textbf{\textit{n}}ear-collapsed %weakly 
                                                              guarded \LLEEonecharts\ 
                          % via \ref{Eone} 
                        are \textbf{\textit{c}}omplete provable solutions.
    (See Lem.~\ref{lem:nearcollapsed:2:compl:solvable}.)                    
    % (Spezialized as Theorem~\ref{thm:pseudo:near:collapsed:LLEEonechart:compl:solv},
    %  whose proof employs \ref{Tone}, %{\bf (P$_{\hspace*{-1pt}\text{\nf 1}}$)}, 
    %                      see Theorem~\ref{thm:funbisims:reflect:solvability:charts}.)
    %
  \item[{\crtcrossreflabel{{\bf (CC)}}[CCone]}]
    If a \weaklyguarded\ \onechart~$\aonechart$ has a \textbf{\textit{c}}omplete provable solution with principal value $\astexp$,
      then also the (\onetransition\ free) bisimulation \textbf{\textit{c}}ollapse of $\aonechart$ has a provable solution with principal value $\astexp$. 
      (See~Lem.~\ref{lem:preservation:sols},~\ref{it:collapse:lem:preservation:sols}.)
  
    % \item[{\crtcrossreflabel{{\bf (TC)}}[TCone]}]%[{\crtcrossreflabel{{\bf (TC$_{\hspace*{-1pt}\text{\nf 1}}$)}}[TCone]}]%[{\bf (TC$_{\hspace*{-1pt}\text{\nf 1}}$)}]%
    %   Between weakly guarded \onecharts\ it is possible to \textbf{\textit{t}}ransfer
    %     \textbf{\textit{c}}omplete provable solutions while preserving their principal values.
    %     % (See Theorem~\ref{thm:transf:sol:via:onebisim:onecharts}.)
    %
  \item[{\crtcrossreflabel{{\bf (SE)}}[SEone]}]%[{\crtcrossreflabel{{\bf (SE$_{\hspace*{0pt}\text{\nf 1}}$)}}[SEone]}]%[{\bf (SE$_{\hspace*{0pt}\text{\nf 1}}$)}]%
    All provable \textbf{\textit{s}}olutions of a %weakly 
                                                guarded \LLEEonechart\ are provably \textbf{\textit{e}}qual. 
    (See Lem.~\ref{lem:LLEEonecharts:uniquely:solvable}.)
\end{enumerate}%

\begin{thm}\label{thm:milnersys:complete}
  Milner's proof system $\milnersys$ is complete with respect to process semantics equality %$\procsemeq$ 
    of regular expressions.
\end{thm}%

\begin{proof}%[The completeness proof for Milner's system $\milnersys$.]
  (\emph{See Fig.~\ref{fig:proof:structure} for an illustration}.)
  Let $\astexpi{1},\astexpi{2}\in\StExpover{\actions}$ be star expressions
  such that $\procsem{\astexpi{1}} = \procsem{\astexpi{2}}$ holds, that is,
  their process interpretations coincide.
  This means that 
    the behaviors $\eqcl{\chartof{\astexpi{1}}}{\sbisimsubscript}$ and $\eqcl{\chartof{\astexpi{2}}}{\sbisimsubscript}$
    of the chart interpretations $\chartof{\astexpi{1}}$ of $\astexpi{1}$ and $\chartof{\astexpi{2}}$ of $\astexpi{2}$ coincide. 
  Therefore $\chartof{\astexpi{1}} \bisim \chartof{\astexpi{2}}$ holds,
    that is, $\chartof{\astexpi{1}}$ and $\chartof{\astexpi{2}}$ are bisimilar.
  We have to show $\astexpi{1} \milnersyseq \astexpi{2}$, that is, that $\astexpi{1} \formeq \astexpi{2}$ can be proved in $\milnersys$. 
  %
  % Suppose that $\astexpi{1}$ and $\astexpi{2}$ are star expressions
  % with the same process semantics $\procsem{\astexpi{1}} = \procsem{\astexpi{2}}$.
  % That is, the behaviors $\procsem{\astexpi{1}} = \eqcl{\chartof{\astexpi{1}}}{\sbisimsubscript}$
  %                    and $\procsem{\astexpi{2}} = \eqcl{\chartof{\astexpi{2}}}{\sbisimsubscript}$
  % that are induced by the chart interpretations $\chartof{\astexpi{1}}$ of $\astexpi{1}$ and $\chartof{\astexpi{2}}$ of $\astexpi{2}$ are equal.   
  % It follows that the chart interpretations $\chartof{\astexpi{1}}$ of $\astexpi{1}$ and $\chartof{\astexpi{2}}$ of $\astexpi{2}$ are bisimilar.
  % We have to show that $\astexpi{1} = \astexpi{2}$ can be proved in $\milnersys$. 
  
  % By \ref{SIone}, %{\bf (SI$_{\hspace*{0pt}\text{\nf 1}}$)}
  %   $\astexpi{1}$ and $\astexpi{2}$ are principal values 
  % of provable solutions %$\sasoli{1}$ and $\sasoli{2}$ 
  %                       of their chart interpretations $\chartof{\astexpi{1}}$ and $\chartof{\astexpi{2}}$, respectively.
  %
  Due to \ref{IVone}, %{\bf (VI$_{\hspace*{0pt}\text{\nf 1}}$)} 
    the \onechart\ interpretations $\onechartof{\astexpi{1}}$ of $\astexpi{1}$ and $\onechartof{\astexpi{2}}$ of $\astexpi{2}$ 
    are %weakly 
        guarded \LLEEonecharts\ 
    that are \onebisimilar\ to $\chartof{\astexpi{1}}$ and to $\chartof{\astexpi{2}}$, respectively. 
  With this, $\chartof{\astexpi{1}} \bisim \chartof{\astexpi{2}}$   
    %As a consequence we get 
  entails $\onechartof{\astexpi{1}} \onebisim \onechartof{\astexpi{2}}$,
    that is, $\onechartof{\astexpi{1}}$ and $\onechartof{\astexpi{2}}$ are \onebisimilar. 
    % Then all of the charts~$\chartof{\astexpi{1}}$, $\chartof{\astexpi{2}}$, and \onecharts~$\onechartof{\astexpi{1}}$, $\onechartof{\astexpi{2}}$ 
    %   are \onebisimilar.
    % Hence $\onechartof{\astexpi{1}} \onebisim \onechartof{\astexpi{2}}$.
  By part~\ref{IVone.iii} of \ref{IVone}, %{\bf (RS$_{\hspace*{0pt}\text{\nf 1}}$)} 
    $\astexpi{1}$ and $\astexpi{2}$ are 
      the principal values of provable solutions $\sasoli{1}$ and $\sasoli{2}$ of $\onechartof{\astexpi{1}}$ and $\onechartof{\astexpi{2}}$, respectively.
      
  We now focus on $\onechartof{\astexpi{1}}$, leaving aside $\onechartof{\astexpi{2}}$ for the moment.
  (Equally we could start from $\onechartof{\astexpi{2}}$, and build up a symmetrical argument).
  % (by symmetry we could start a dual argument also from $\onechartof{\astexpi{2}}$).
  By \ref{CRone} %{\bf (CR$_{\hspace*{-0.25pt}\text{\nf 1}}$)} 
    we find that the %weakly 
        guarded \LLEE -\onechart~$\onechartof{\astexpi{1}}$
  can be transformed into a \onebisimilar\ %weakly 
                                           crystallized \onechart\ $\aonecharti{10}$, which is a guarded \LLEEonechart.  
    Then $\aonecharti{10} \onebisim \onechartof{\astexpi{1}}$ holds.
  Due to \ref{CNone}, %{\bf (CN$_{\hspace*{-1pt}\text{\nf 1}}$)} 
    $\aonecharti{10}$ is also \nearcollapsed. 
  Now we can apply \ref{NCone} to $\aonecharti{10}$ %{\bf (N\hspace*{-0.75pt}C$_{\hspace*{0pt}\text{\nf 1}}$)} 
    in order to conclude that the provable solution $\sasoli{10}$ that is extracted from $\aonecharti{10}$ by \ref{Eone} %{\bf (E$_{\hspace*{0pt}\text{\nf 1}}$)}
    is a complete provable solution of the \nearcollapsed\ \onechart~$\aonecharti{10}$.
  Let $\astexpi{10}$ be the principal value of $\sasoli{10}$.
  
  Now let $\acharti{0}$ be the (\onetransition\ free) bisimulation collapse of $\aonecharti{10}$.
  By applying \ref{CCone} to the complete provable solution $\sasoli{10}$ of $\aonecharti{10}$, which is weakly guarded since it is guarded,
    we obtain a provable solution $\sasoli{0}$ of $\acharti{0}$ that also has the principal value $\astexpi{10}$.
  Due to $\onechartof{\astexpi{1}} \onebisim \onechartof{\astexpi{2}}$,   
    % Since all charts and \onecharts\ above are \onebisimilar, 
     $\acharti{0}$ is the joint bisimulation collapse of $\onechartof{\astexpi{1}}$ and $\onechartof{\astexpi{2}}$.
  It follows that there are functional \onebisimulations\ from $\onechartof{\astexpi{1}}$ and from $\onechartof{\astexpi{2}}$ to $\acharti{0}$,
    that is, $\onechartof{\astexpi{1}} \funonebisim \acharti{0} \convfunonebisim \onechartof{\astexpi{2}}$.    
  Now we can use \ref{Tone} to transfer the provable solution $\sasoli{0}$ 
    from $\acharti{0}$ to $\onechartof{\astexpi{1}}$ and to $\onechartof{\astexpi{2}}$.
  We obtain provable solutions $\sasoltildei{1}$ of $\onechartof{\astexpi{1}}$, and $\sasoltildei{2}$ of $\onechartof{\astexpi{2}}$,
    both of which have $\astexpi{10}$ as their principal value.
  
  Since both $\sasoli{1}$ and $\sasoltildei{1}$ are provable solutions of the %weakly 
                                                                                   guarded LLEE-\mbox{\onechart} $\onechartof{\astexpi{1}}$,
    we can apply %{\bf (SE$_{\hspace*{0pt}\text{\nf 1}}$)} 
                 \ref{SEone} to find that $\sasoli{1}$ and $\sasoltildei{1}$ are provably equal. 
  In particular, the principal values $\astexpi{1}$ of~$\sasoli{1}$ and $\astexpi{10}$ of~$\sasoltildei{1}$ are provably equal.
    That is, $\astexpi{1} \milnersyseq \astexpi{10}$ holds. 
  Analogously, as both $\sasoli{2}$ and $\sasoltildei{2}$ are provable solutions of the %weakly 
                                                                                        guarded \LLEEonechart~$\onechartof{\astexpi{2}}$,
    \ref{SEone} %{\bf (SE$_{\hspace*{0pt}\text{\nf 1}}$)}
    also  entails that $\sasoli{2}$ and $\sasoltildei{2}$ are provably equal. 
  Therefore also the principal values $\astexpi{2}$ of $\sasoli{2}$ and $\astexpi{10}$ of $\sasoltildei{2}$ are provably equal.
    That is, $\astexpi{2} \milnersyseq \astexpi{10}$ holds.
  
  From $\astexpi{1} \milnersyseq \astexpi{10}$ and $\astexpi{2} \milnersyseq \astexpi{10}$
    we obtain $\astexpi{1} \milnersyseq \astexpi{2}$,
      and hence that $\astexpi{1} \formeq \astexpi{2}$ is provable in Milner's system $\milnersys$.
\end{proof}

%\newpage
%-------------------------------------------------------
\section{Failure of LLEE-preserving $\sone$-collapse}
  \label{failure:LLEEpreserving:collapse}
%-------------------------------------------------------

Here we expand on observation \ref{contrib:concept:1},
  due to which we have realized in Sect.~\ref{motivation:proof:strategy} that the bisimulation collapse strategy in \cite{grab:fokk:2020:lics} 
    cannot be extended directly to a \onebisimulation\ collapse strategy for showing completeness of $\milnersys$. 
We define the properties `collapsible', `\onecollapsible', and `jointly minimizable' formally,  
  formulate their failure for \LLEEonecharts, and suggestively explain the reason by means of~an~example. 
       
Let $\aonechart$ be a \LLEEonechart.
  We say that $\aonechart$ is \emph{\LLEEpreservingly\ \onecollapsible}
                              (\emph{\LLEEpreservingly\ collapsible})
    if $\aonechart$ has a \onebisimulation\ collapse that is a \LLEEonechart\
       (and respectively, the bisimulation collapse of $\aonechart$ is a \LLEEonechart). 
       
  We say that two \LLEEonecharts~$\aonecharti{1}$ and $\aonecharti{2}$ 
    that are \onebisimilar\ (that is, with $\aonecharti{1} \onebisim \aonecharti{2}$)
    are \emph{\LLEEpreservingly\ jointly minimizable (under functional \onebisimilarity~$\sfunonebisim$)}
      if there is a \LLEEonechart\ $\aonecharti{0}$
        such that $\aonecharti{1} \funonebisim \aonecharti{0} \convfunonebisim \aonecharti{2}$. % holds.                                            

% \begin{prop}[\ref{contrib:concept:1}]\label{prop:not:LLEE:pres:collapse:minimization}
%   Weakly guarded \LLEEonecharts\ are:
%   \begin{enumerate}[label={(\roman{*})},align=right,leftmargin=*,itemsep=0.25ex]
%     \item{}\label{it:1:prop:not:LLEE:pres:collapse:minimization}
%       not in general \LLEEpreservingly\ collapsible,  
%     \item{}\label{it:2:prop:not:LLEE:pres:collapse:minimization}
%       not in general \LLEEpreservingly\ \onecollapsible,
%     \item\label{it:3:prop:not:LLEE:pres:collapse:minimization}
%       not in general \LLEEpreservingly\ jointly minimizable.  
%   \end{enumerate}
%   Statements~\ref{it:1:prop:not:LLEE:pres:collapse:minimization} and \ref{it:2:prop:not:LLEE:pres:collapse:minimization}
%     are witnessed by the \LLEEonechart~$\aonechart$ in Fig.~\ref{fig:countex:collapse}.
%   Statement~\ref{it:3:prop:not:LLEE:pres:collapse:minimization}  
%     is witnessed by the generated sub-\LLEEonecharts\ $\gensubchartofby{\aonechart}{\bverti{1}}$ and $\gensubchartofby{\aonechart}{\bverti{2}}$ 
%       of the \LLEEonechart~$\aonechart$ in Fig.~\ref{fig:countex:collapse}.
% \end{prop}

\input{figs/fig-countex-collapsible}%
\begin{prop}[\ref{contrib:concept:1}]\label{prop:not:LLEE:pres:collapse:minimization}
  The following two statements hold:\vspace*{-0.5ex}
  \begin{enumerate}[label={(\roman{*})},align=right,leftmargin=*,itemsep=0.25ex]
    \item{}\label{it:1:prop:not:LLEE:pres:collapse:minimization}
      W.g.\ \LLEEonecharts\ are not in general \LLEEpreservingly\ \onecollapsible,
        hence not in general \LLEE-pres.\ collapsible. 
    \item\label{it:2:prop:not:LLEE:pres:collapse:minimization}
      Two \onebisimilar\ \wg\ \LLEEonecharts\ 
        are not in general \LLEEpreservingly\ jointly minimizable under $\sfunonebisim$. 
  \end{enumerate}
  Statement~\ref{it:1:prop:not:LLEE:pres:collapse:minimization} % and \ref{it:2:prop:not:LLEE:pres:collapse:minimization}
    is witnessed by the \LLEEonechart~$\aonechart$ in Fig.~\ref{fig:countex:collapse}.
  Statement~\ref{it:2:prop:not:LLEE:pres:collapse:minimization}  
    is witnessed by the \onebisimilar\ generated sub-\LLEEonecharts\ $\gensubchartofby{\aonechart}{\bverti{1}}$ and $\gensubchartofby{\aonechart}{\bverti{2}}$ 
      of the \LLEEonechart~$\aonechart$ in Fig.~\ref{fig:countex:collapse}.
\end{prop}

That the \LLEEonechart~$\aonechart$ in Fig.~\ref{fig:countex:collapse} is not \onecollapsible\
  is suggested there by exhibiting a natural\vspace*{-1pt} \onebisimulation\ collapse that is not a \LLEEonechart:
    the \onebisimilar\ \onechart\ $\aonecharti{1}$ %$\connthroughin{\aonechart}{\bverti{1}}{\bverti{2}}$ 
      that results by `connecting through' all incoming transitions at $\bverti{1}$ in $\aonechart$ over to $\bverti{2}$.
From %$\connthroughin{\aonechart}{\bverti{2}}{\bverti{1}}$ 
     $\aonecharti{1}$ it is easy to intuit that the bisimulation collapse of $\aonechart$ cannot be a \LLEEonechart, either,
  and hence that $\aonechart$ is not \LLEEpreservingly\ collapsible.
While that is only part of the (remaining) proof 
  of Prop.~\ref{prop:not:LLEE:pres:collapse:minimization},~\ref{it:1:prop:not:LLEE:pres:collapse:minimization}, 
  it is also easy to check\vspace*{-0.5pt} 
    that the \connectthroughonechart{\bverti{2}}{\bverti{1}} $\aonecharti{1}$   %$\connthroughin{\aonechart}{\bverti{2}}{\bverti{1}}$ 
                                                               of $\aonechart$
    cannot be a \LLEEonechart, either. 
For the proof that $\gensubchartofby{\aonechart}{\bverti{1}}$ and $\gensubchartofby{\aonechart}{\bverti{2}}$
  are not \LLEEpreservingly\ jointly minimizable it is crucial to realize
    that a function that maps $\bverti{1}$ to $\bverti{2}$, and $\bverti{2}$ to $\bverti{1}$ 
      cannot be extended into a transfer function on $\aonechart$.
        (That function, however, defines a `grounded' \onebisimulation\ slice on $\aonechart$, see Def.~\ref{def:grounded:onebisim:slice} later,
                                                                                               and Fig.~\ref{fig:ex:local:transfer:function}).

\input{figs/fig-ex-local-transfer-function}%
For motivating concepts in the next sections
  we will use the simplified version $\aonecharti{s}$ in Fig.~\ref{fig:ex:local:transfer:function}
  of the \LLEEonechart~$\aonechart$ in Fig.~\ref{fig:countex:collapse}.
We emphasize, however, that although $\aonecharti{s}$ is not collapsible, it is \onecollapsible.
  Hence $\aonecharti{s}$ does not witness statement \ref{it:1:prop:not:LLEE:pres:collapse:minimization},
  nor can it be used for showing \ref{it:2:prop:not:LLEE:pres:collapse:minimization}~in~Prop.~\ref{prop:not:LLEE:pres:collapse:minimization}.
% Therefore it cannot be used to show \ref{it:2:prop:not:LLEE:pres:collapse:minimization} in Prop.~\ref{prop:not:LLEE:pres:collapse:minimization}.

%-----------------------
\section{Twin-Crystals}%
  \label{twincrystals}
%-----------------------

The \onechart\ $\aonechart$ in Fig.~\ref{fig:countex:collapse} and Prop.~\ref{prop:not:LLEE:pres:collapse:minimization} 
  will turn out to be prototypical for strongly connected components in \LLEEonecharts\
    that are `nearly collapsed' but not \LLEEpreservingly\ collapsible any further. 
By isolating a number of its properties from $\aonechart$, and from its simplified version $\aonecharti{s}$ in Fig.~\ref{fig:ex:local:transfer:function},
  we define the central concept of `\twincrystal'.
% The counterexample \LLEEonechart\ $\aonechart$ in Fig.~\ref{fig:countex:collapse},
%   and its simplified version $\aonecharti{s}$ in Fig.~\ref{fig:ex:local:transfer:function}
%     lead us to the concept of `\twincrystal'.
Each of the \LLEEonecharts\ $\aonechart$ and $\aonecharti{s}$ consists of a single \scc\  
  that is of `\twincrystal\ shape', which exhibits a certain kind of 
                                                     symmetry with respect to \onebisimilarity.
  % Note that each of $\aonechart$ and $\aonecharti{s}$ consists only of a single \scc. 
Our proof utilizes this symmetry for proving uniqueness of provable solutions
  for \LLEEonecharts\ of which all \sccs\ are \onecollapsed\ or of \twincrystal\ shape. 
Before describing \twincrystals, we first define, also motivated by the two examples,
`grounded \onebisimulation\ slices', `local transfer functions', and `\nearcollapsed'~\onecharts.  

On both of the \LLEEonecharts\ $\aonechart$ in Fig.~\ref{fig:countex:collapse}, and $\aonecharti{s}$ in Fig.~\ref{fig:ex:local:transfer:function} % Rem.~\ref{rem:prop:not:LLEE:pres:collapse:minimization}
  there are \nontrivial\ functional \onebisimulation\ slices, 
  which are suggested by the \magenta{magenta links}.
These \onebisimulation\ slices
  cannot be extended to \onebisimulations\ that are defined by transfer functions.
    %for $\aonecharti{s}$ (see Fig.~\ref{fig:ex:local:transfer:function}), % see in Ex.~\ref{ex:local:transfer:function} below,
However, they have an expedient additional property
    that will permit us to work with `local transfer functions' instead.
Namely, that induced transitions from a vertex $\bverti{1}$ of pair $\pair{\bverti{1}}{\bverti{2}}$ of a slice $\abisim$
  to a vertex $\bvertacci{1}$ outside of the domain of $\abisim$ 
    can be joined by an induced transition with the same label from $\bverti{2}$ to $\bverti{1}$,
and vice versa. We call \onebisimulation\ slices with this property `grounded',
and functions that induce them `local transfer functions'.

\begin{defi}\label{def:grounded:onebisim:slice}
  Let $\aoneLTS = \tuple{\states,\actions,\sone,\transs,\termexts}$ be a \oneLTS.\nopagebreak[4]
  
  By a \emph{grounded \onebisimulation\ slice on $\aoneLTS$}
    we mean a \onebisimulating\ slice~$\abisim \subseteq \states\times\states$ on $\aoneLTS$ %\emph{is grounded} 
    such that for all $\pair{\bstatei{1}}{\bstatei{2}}\in\abisim$
  the following additional forth/back conditions hold:
  \begin{itemize}[labelindent=0.2em,leftmargin=*,align=left,itemsep=0.25ex,labelsep=0.25em]
    \item[(forth)$_\text{\nf g}$]
      $ \begin{aligned}[t] 
          \forall \aact\in\actions\,
            \forall \bstateacci{1}\in\statesi{1}
                \bigl(\, & 
                  \bstatei{1} \ilt{\aact} \bstateacci{1}
                    \logand
                  \bstateacci{1}\notin\asetvertsi{1} 
                  \\[-1ex] 
                     & \;\Longrightarrow\;
                      \bstatei{2} \ilt{\aact} \bstateacci{1}
                    \logand
                  \bstateacci{1}\notin\asetvertsi{2}  
              \,\bigr) \punc{,} 
        \end{aligned}$
      
    \item[(back)$_\text{\nf g}$]
      $ \begin{aligned}[t]
          \forall \aact\in\actions\,
            \forall \bstateacci{2}\in\statesi{1}
                \bigl(\, &  
                  \bstatei{1} \ilt{\aact} \bstateacci{2}
                    \logand
                  \bstateacci{2}\notin\asetvertsi{1}
                    \\[-1ex]
                    & \;\Longleftarrow\;
                  \bstatei{2} \ilt{\aact} \bstateacci{2}
                    \logand
                  \bstateacci{2}\notin\asetvertsi{2}  
              \,\bigr) \punc{.} 
       \end{aligned} $
  \end{itemize} 
  where $\asetvertsi{1} \defdby \actdomof{\abisim}$, 
        and $\asetvertsi{2} \defdby \actcodomof{\abisim}$ are the active domain, and the active codomain of $\abisim$, respectively.
\end{defi}

\begin{lem}%{Proposition (\emph{Grounded bisimulation slices induce bisimulations})}
            \label{lem:gr:bisim:slice}
  For every grounded bisimulation slice $\abisimslice \subseteq \states\times\states$ 
      on a \oneLTS~$\aoneLTS = \tuple{\states,\actions,\sone,\transs,\termexts}$,
    the relation
    $
    \abisimdoublebar \defdby \abisimslice \,\cup\, \seqrel 
    $
  is a \onebisimulation\ on $\aoneLTS$.
\end{lem}

\begin{defi}\label{def:local:transfer:function}
  A \emph{\localtransfer\ %(partial) 
                          function} on a \oneLTS~$\aoneLTS = \tuple{\states,\actions,\sone,\transs,\termexts}$
    is a partial function $\sphifun \funin \states \rightharpoonup \states$ 
      whose graph $\descsetexp{ \pair{\astate}{\phifun{\astate}} }{ \astate \in \states }$
        is a grounded \onebisimulation\ slice on ${\aoneLTS}$. 
\end{defi}

\begin{exa}\label{ex:local:transfer:function}
  Both of the functions $\scpfunon{\twinc}$ on the \onecharts~$\aonechart$ in Fig.~\ref{fig:countex:collapse}, 
    and $\scpfunon{\twinc}$ on $\aonecharti{s}$ in Fig.~\ref{fig:ex:local:transfer:function} 
  are local transfer functions: in particular in Fig.~\ref{fig:ex:local:transfer:function} it can be checked easily
    that $\scpfunon{\twinc}$ defines a grounded \onebisimulation\ slice. 
  Neither of these local transfer functions can be extended into a transfer function.
  For $\aonecharti{s}$ this can be checked in Fig.~\ref{fig:countex:collapse}: 
      all pairs of the identity function have to be added 
        in order to extend $\graphof{\scpfunon{\twinc}}$ into a \onebisimulation,
      thereby violating functionality of the relation.  
\end{exa}

  % inthis can be easy from Fig.~\ref{fig:ex:local:transfer:function}.
  %   is a local transfer function, because it defines a grounded \onebisimulation\ slice. 
  %     This can be checked from the following picture (see $\scpfunon{\twinc}$ in magenta):\vspace*{-2ex}
  % % \begin{flushleft}
  % %   \hspace*{-3ex}
  % %   \scalebox{0.7}{\input{figs/ex-local-transfer-function.tex}}\vspace*{-1ex}
  % % \end{flushleft}
  % But $\scpfunon{\twinc}$ does not induce a functional \onebisimulation\ on $\aonecharti{s}$,
  %   and hence not a transfer function on $\aonecharti{s}$:
  %     all pairs of the identity function have to be added 
  %       in order to extend $\graphof{\scpfunon{\twinc}}$ into a \onebisimulation.
  %         For instance, for progress via \transitionsact{\aact} the two brown pairs must be added,
  %           due to one of which the extending \onebisimulation\ cannot be functional.

We will say that a \onechart~$\aonechart$ (and respectively, a \oneLTS~$\aoneLTS$) is `\nearcollapsed'
  if \onebisimilarity\ on $\aonechart$ (on $\aoneLTS$) is the reflexive--symmetrical closure
    of the union of the graphs of finitely many local transfer functions on $\aoneLTS$. 

\begin{defi}[\ref{contrib:concept:3}]%[(locally) near-collapsed]
                  \label{def:local:nearcollapsed}
  Let $\aoneLTS$ be a \oneLTS\ with state set $\verts$, and let $\csetverts\subseteq\verts$ be a subset of the vertices of $\aoneLTS$. 
  We say that $\aoneLTS$ \emph{is locally \nearcollapsed\ for $\csetverts$}
    if there are local transfer functions $\sphifuni{1},\ldots,\sphifuni{n} \funin \verts \rightharpoonup \verts$ on $\aoneLTS$ such that:
  % \begin{equation}%\label{eq:def:local:nearcollapsed}
  \begin{center}
    $ 
    \sonebisimon{\aoneLTS} \cap ( \csetverts \times \csetverts )
      \,\subseteq\,
    \sredrsci{\sabinrel} \punc{,} 
    \quad
      \text{for } \sabinrel = \bigcup_{i=1}^n \graphof{\sphifuni{i}} \punc{,}
    $
  \end{center} 
  % \end{equation}
  where $\sredrsci{\sabinrel}$ means the \reflexivesymmetric\ closure of $\sabinrel$.
  We say that $\aoneLTS$ \emph{is \nearcollapsed} 
    if $\aoneLTS$ is \nearcollapsed\ for $\verts$. %its set $\verts$ of vertices. 
   
  A \onechart~$\aonechart$ is \emph{(locally for $\csetverts$) \nearcollapsed}
    if the \oneLTS\ $\oneLTSof{\aonechart}$ underlying $\aonechart$ is (locally for $\csetverts$) \nearcollapsed. 
\end{defi}

\begin{example}
  Each of the \onecharts\ $\aonechart$ in Fig.~\ref{fig:countex:collapse}, 
    and $\aonecharti{s}$ in Fig.~\ref{fig:ex:local:transfer:function}
  contains only one pair of distinct \onebisimilar\ vertices.
  Since these vertices are related, respectively, by the appertaining transfer function $\scpfunon{\aonechart}$,
    it follows that $\aonechart$ and $\aonecharti{s}$ are locally \nearcollapsed\ for the sets $\csetverts$ of all their vertices, respectively.
  Therefore $\aonechart$ and $\aonecharti{s}$ are \nearcollapsed.
\end{example}

% We say that $\bvert$ \emph{\txtloopsbackto} $\avert$ for vertices $\avert$, $\bvert$ in a \LLEEwitness~$\aonechart$
%   if $\avert \redi{\looplab{n}} \avertacc \redrtci{\bodylab} \bvert \redtci{\bodylab} \avert$ holds
%     with $\avert$ encountered again only at the end.

     %
\input{figs/fig-twincrystal-schema.tex}%
%
% %--------------------------------------------------
% \subsection{Definition and basic properties}%
%   \label{definition:properties::twincrystals}
% %--------------------------------------------------
%
% 
For vertices $\bvert$ and $\avert$ in a \LLEEwitness~$\aonecharthat$
  we write $\bvert \loopsbackto \avert$, $\bvert$ \emph{\txtloopsbackto} $\avert$,
    % we say that $\bvert$ \emph{\txtloopsbackto} $\avert$, denoted by $\bvert \loopsbackto \avert$,
  if $\avert \comprewrels{\sredi{\looplab{n}}}{\sredrtci{\bodylab}} \bvert \redtci{\bodylab} \avert$ holds with $n\in\natplus$ and
     such that $\avert$ is only encountered again at~the~end.
We fix $\sectionto{\avert}{\sloopsbacktortc} \defdby \descsetexp{ \bvert }{ \bvert \loopsbacktortc \avert }$,
  the \emph{\txtloopsbackto\ part~of~$\avert$}. 
  %
% By $\sectionto{\avert}{\sloopsbacktortc} \defdby \descsetexp{ \bvert }{ \bvert \loopsbacktortc \avert }$
%    we denote the \emph{\txtloopsbackto\ part of $\avert$}. 
   
The structure of the \LLEEonecharts~$\aonechart$ and $\aonecharti{s}$ above
  can be described by the illustration in Fig.~\ref{fig:twincrystal:schema},
together with the following eight properties
  that define when a vertex set $\twinc$ in a \LLEEonechart~$\aonechart = \tuple{\verts,\actions,\sone,\start,\transs,\termexts}$
    is the \emph{carrier of a \twincrystal} (and that the \subLTS\ induced by $\twinc$ in $\aonechart$ is a \emph{\twincrystal})
      with respect to vertices $\toptwinc,\pivottwinc\in\twinc\subseteq\verts$, sets $\twinctop,\twincpivot\subseteq\twinc$,
        a \LLEEwitness~$\aonecharthat$ of $\aonechart$
          with binary \txtloopsbackto\ relation $\sloopsbackto$,
        and a \nonempty\ set $E_2$ of transitions~from~$\toptwinc\,$: 
\begin{enumerate}[label={(tc-\arabic*)},leftmargin=*,align=right,itemindent=0ex,itemsep=0.25ex]
  
    \item{}\label{tc-1}
      $\twinc = \sectionto{\toptwinc}{\sloopsbacktortc}$ is a maximal \txtloopsbackto\ part. 
          % (consisting of the cyclic part of all nested loop \subonecharts\ at and below,
          %  that is, from $\toptwinc$ reachable vertices from which $\toptwinc$ is reachable
          %  again by the \txtloopsbackto\ relation $\sloopsbackto$). 
        Then $\toptwinc$ is \maximalwrt{\sloopsbackto}, and $\twinc$ is an \scc.
      
    \item{}\label{tc-2}
      $\twinc = \twincpivot \uplus \twinctop$ for
      $\twincpivot \defdby \sectionto{\pivottwinc}{\sloopsbacktortc}$
        and
      $\twinctop \defdby \sectionto{\toptwinc}{\sloopsbackinloopltortc{E_2}}$,
      the \txtloopsbackto\ part generated by transitions~in~$E_2$. 
      We call $\pivottwinc\,$ \emph{pivot vertex} and $\toptwinc\,$ \emph{top vertex}.
      Then $\pivottwinc\in\twincpivot \subseteq \twinc$, 
            $\toptwinc\in\twinctop \subseteq \twinc$,
       and $\family{P_{i}}{i\in\setexp{1,2}}$ is a partition~of~$\twinc$.
      
    \item{}\label{tc-3}
      $\aonechart$ is not \onecollapsed\ for $\twinc$. Hence $\twinc$ contains \onebisimilarity\ redundancies. 
      
    \item{}\label{tc-4}
      (Using terminology from Def.~\ref{lem:reducing:1brs} later:) % in Sect.~\ref{crystallization}:
      All `reduced' \onebisimilarity\ redundancies in $\twinc$ 
          are of `\precrystalline' form \ref{R3.4},
          with one vertex in $\twincpivot$, and the other in $\twinctop$. 
        %
        % are \precrystalline\ and
        % have position $\tuple{\bverti{1},\pivottwinc,\toptwinc,\bvertbari{2},\bverti{2}}$ for some $\bvertbari{2}\in\verts$.
        % % (Then $\pivottwinc\dloopsbackto\toptwinc \logand \pivottwinc\neq\toptwinc \logand \lognot{(\pivottwinc \ltrtc{\sone} \toptwinc )}$ due to \ref{pos:3}, \ref{pos:3.5}).
    
    \item{}\label{tc-5} 
      Proper transitions from $\pivottwinc\,$ \emph{favor} $\twincpivot$:
        whenever a prop\-er transition from $\pivottwinc$ is \onebisimilar\ to a vertex in $\twincpivot$,
          then its target is in $\twincpivot$. 
      
    \item{}\label{tc-6} 
        Proper transitions from $\toptwinc\,$ \emph{favor} $\twinctop$ \mbox{} (confer \ref{tc-5}).
        
    \item{}\label{tc-7}
      $\twinc$ is \emph{squeezed in $\aonechart$}: 
        no vertex in $\twinc$ is \onebisimilar\ to a vertex outside of $\twinc$. 
    \item{}\label{tc-8}
      $\twinc$ is \emph{grounded in $\aonechart$}:
        any two transitions from $\twinc$ with \onebisimilar\ targets outside of $\twinc$
        have the same target.
  \end{enumerate}
  
  %On the basis of this definition, 
  For carrier sets $\twinc$ of a \twincrystal\ in $\aonechart$ 
    it can be shown that \onebisimilarity\ redundancies in $\twinc$ occur 
      with one vertex in $\twincpivot$, and the other in $\twinctop$.
  Hence a vertex in $\twincpivot$ may have a \onebisimilar\ counterpart, which can only be in $\twinctop$; and vice versa.    
  Then the following derived concepts can be introduced (see Fig.~\ref{fig:twincrystal:schema}).
  The pivot part $\twincpivot$ partitions into the sets $\twincunambigpivot$ and $\twincambigpivot$ 
    of unambiguous, and ambiguous vertices, in $\twinc$
    that are unique, and respectively, are not unique up to $\onebisimon{\aonechart}$. 
  The top part $\twinctop$ partitions into the sets $\twincunambigtop$ and $\twincambigtop$ 
    of unambiguous, and ambiguous vertices, respectively.
  The boundary vertices $\boundaryof{\twincambigpivot}$ of $\twincpivot$ have \onetransition\ paths to $\pivottwinc$,
    and the boundary vertices $\boundaryof{\twincambigtop}$ of $\twinctop$ have \onetransition\ paths to $\toptwinc$.
  There are no transitions directly from $\twincambigpivot\setminus\boundaryof{\twincambigpivot}$ to $\twincunambigpivot$,
    and no transitions directly from $\twincambigtop\setminus\boundaryof{\twincambigtop}$ to $\twincunambigtop$.
    
  \oneBisimilar\ vertices in $\twincambigpivot$ and $\twincambigtop$ are related 
    by the \emph{counterpart (partial) function} $\scpfunon{\twinc} \funin \verts \rightharpoonup \verts$ \emph{on $\aonechart$} 
    with domain and range contained in $\twinc$  that is defined, for all~$\bvert\in\verts$,~by:
    \begin{align*}
      \cpfunon{\twinc}{\bvert} \defdby
                \begin{cases}
                  \bvertbar         & \parbox[t]{\widthof{if $\bvert\in\twincambig$, and $\bvertbar$ the \onebisimilar}}%{\widthof{if \onebisimilar\ counterpart of $\bvert$ in $\twinc$\punc{,}}}
                                             {if $\bvert\in\twincambig$, and $\bvertbar$ the \onebisimilar\
                                              \\
                                              \phantom{if} {counterpart} of $\bvert$ in $\twinc$\punc{,}}
                  \\
                  \text{undefined}  & \text{if $\bvert\notin\twincambig$\punc{.}}
                \end{cases}   
    \end{align*}%
    %\enlargethispage{4ex}
    \vspace*{-1.5ex}

\begin{exa}\label{ex:twincrystal}
  Each of the \onetransition\ limited \LLEEonecharts\ $\aonechart$ in Fig.~\ref{fig:countex:collapse} 
    and $\aonecharti{2}$ in Fig.~\ref{fig:ex:local:transfer:function}
    consists of a single \scc\ that is carrier of a \twincrystal.
    In the appertaining figures, the parts $\twincpivot$ and $\twinctop$ are colored blue and green, respectively.
\end{exa}

The picture in Fig.~\ref{fig:twincrystal:schema}
  suggests the symmetric nature of \twincrystals:\label{symm:nature:twincrystal}
    if in the underlying \LLEEwitness\
      the \loopentry\ transitions from $\twinctop$ to $\twincpivot$ (maximum level $m$) are relabeled into body transitions,
        and the body transitions from $\twincpivot$ to $\twinctop$ into \loopentry\ transitions (level $m$),
          then a \twincrystal\ with permuted roles of $\pivottwinc$ and $\toptwinc$ arises.
But for the completeness proof only the properties of \twincrystal\ shaped \sccs\ 
  that are formulated by the two lemmas below are crucial.           
  
% The picture in Fig.~\ref{fig:twincrystal:schema}
%   suggests the symmetric nature of \twincrystals:\label{symm:nature:twincrystal}
%     if in the underlying \LLEEwitness\
%       the \loopentry\ transitions from $\twinctop$ to $\twincpivot$ are relabeled into body transitions,
%         and the body transitions from $\twincpivot$ to $\twinctop$ into \loopentry\ transitions of appropriate level,
%           then a \twincrystal\ with permuted roles of $\pivottwinc$ and $\toptwinc$ arises.
% But for the completeness proof only the properties of \twincrystal\ shaped \sccs\ 
%   that are formulated by the two lemmas below are crucial.  

\begin{lem}
  The counterpart function %$\scpfunon{\twinc}$ 
                           on the carrier $\twinc$ of a \twincrystal\ in a \LLEEonechart~$\aonechart$
    is a local transfer function. % on $\twinc$.
\end{lem}

\begin{lem}
  If $\csetverts$ is the carrier of a \twincrystal\ in a \LLEEonechart~$\aonechart$,
    then $\aonechart$ is locally \nearcollapsed\ for $\csetverts$. 
\end{lem}

%-----------------------------------------
% Figure reduced bisimulation redundancies
%-----------------------------------------
%\input{figs/fig-lem-red-br-pos-R3.tex}
%--------------------------------------
%
% 
%---------------------------------------------
\section{Crystallization of LLEE-\protect\onecharts}%
  \label{crystallization}
%---------------------------------------------

For this central part of the proof
  we sketch how every weakly guarded \LLEEonechart\ can be minimized under \onebisimilarity\
    far enough to obtain a \LLEEonechart\ in  `crystallized' form.
By that we mean that the resulting \onebisimilar\ \LLEEonechart\ is \onecollapsed\ apart from that some of its \sccs\ may be \twincrystals. 
A `groundedness' clause in the definition will ensure
  that every crystallized \onechart\ is also \nearcollapsed. 
      
The minimization process we describe here is a refinement of the process for \LLEEcharts\ (without \onetransitions) 
  that was defined in \cite{grab:fokk:2020:lics,grab:fokk:2020:lics:arxiv}.
We first find that if a \LLEEonechart~$\aonechart$ is not \onecollapsed,
  then $\aonechart$ contains a `\onebisimilarity\ redundancy' from one of three `reduced kinds' (with subkinds).
  This is stated below by Lem.~\ref{lem:reducing:1brs}, which is a generalization of Prop.~6.4 in \cite{grab:fokk:2020:lics,grab:fokk:2020:lics:arxiv}.
Second, we argue that every reduced \onebisimilarity\ redundancy can be eliminated \LLEEpreservingly\
  except if it belongs to a subkind which can be found in the not \onecollapsible\ \LLEEonechart\ in Fig.~\ref{fig:countex:collapse}.
Third, we define a few further transformations for cutting \sccs\ into \twincrystals.

\begin{lem}[kinds of reduced \protect\onebisimilarity\ redundancies]\label{lem:reducing:1brs}\vspace*{-2pt}
  Let $\aonechart$ %$ = \tuple{\verts,\actions,\sone,\start,\transs,\termexts}$ 
    be a %finite 
         \onechart, %with \LLEEonelim.
  and let $\aonecharthat$ a \onetransition\ limited \LLEEwitness\ of~$\aonechart$.
  Suppose that $\aonechart$ is not a bisimulation collapse.
  
  Then $\aonechart$ contains a \onebisimilarity\ redundancy $\pair{\bverti{1}}{\bverti{2}}$  
  (distinct \onebisimilar\ vertices $\bverti{1}$ and $\bverti{2}$ in $\aonechart$)\vspace*{-1.5pt}
  that satisfies, with respect to\/ $\aonecharthat$, one of the position conditions (kinds)
    \ref{R1} (with subkinds \ref{R1.1}, \ref{R1.2})   %% \crtcrossreflabel{\mbox{\nf (R1)}}[R1] 
    \ref{R2}, or
    \ref{R3} (with subkinds \ref{R3.1}, \ref{R3.2}, \ref{R3.3}, \ref{R3.4}) % \crtcrossreflabel{\mbox{\nf (R3)}}[R3] 
  that are illustrated in Fig.~\ref{fig:ill:pos:R123}.%
  \input{figs/fig-ill-pos-R123.tex}%
  Note that the vertices $\bverti{1}$ and $\bverti{2}$ are in the same \scc\ for position kinds \ref{R2} and \ref{R3},
    but in different \sccs\ for position kind \ref{R1}. 
\end{lem}

\begin{defi}
  We consider a \onebisimilarity\ redundancy\vspace*{-1.5pt} $\pair{\bverti{1}}{\bverti{2}}$ in $\aonechart$
    under the assumptions of Lem.~\ref{lem:reducing:1brs} on $\aonecharthat$~and~$\aonechart$. 
    % Let $\aonechart$ be a \LLEEonechart\ with \onetransition\ limited \LLEEwitness\ $\aonecharthat$.
    % %
    % Let $\pair{\bverti{1}}{\bverti{2}}$ be a \onebisimilarity\ in $\aonechart$.
  
  We say that $\pair{\bverti{1}}{\bverti{2}}$ is \emph{reduced (with respect to $\aonecharthat$)} 
    if it is of one of the kinds  \ref{R1}--\ref{R3} in Fig.~\ref{fig:ill:pos:R123}. % Lem.~\ref{lem:reducing:1brs}.
    %
    % Such a reduced \onebisimilarity\ redundancy we call \emph{simple} (\emph{\precrystalline}, and resp., \emph{crystalline})
    %     if it is of kind or subkind \ref{R1}, \ref{R2}, \ref{R3.1}, \ref{R3.2}, or \ref{R3.3}
    %     (of subkind \ref{R3.4}, 
    %        and respectively, 
    %      of subkind \ref{R3.4} but \underline{\smash{not}} of subkind \ref{R3.4.1} or \ref{R3.4.2} in Fig.~\ref{fig:ill:pos:R3:41:42}).% 
  We say that $\pair{\bverti{1}}{\bverti{2}}$ is \emph{simple}
    if it is of (sub-)kind \ref{R1}, \ref{R2}, \ref{R3.1}, \ref{R3.2}, or \ref{R3.3} in Fig.~\ref{fig:ill:pos:R123}.
  We say that $\pair{\bverti{1}}{\bverti{2}}$ is \emph{\precrystalline} 
    if it is of subkind \ref{R3.4} in Fig.~\ref{fig:ill:pos:R123}.
  We say that $\pair{\bverti{1}}{\bverti{2}}$ is \emph{crystalline} 
    if it is \precrystalline, but neither of subkind \ref{R3.4.1}~nor~\ref{R3.4.2} in Fig.~\ref{fig:ill:pos:R3:41:42}.%
      \input{figs/fig-ill-pos-R3-41-42.tex}%
\end{defi} 

\begin{exa}
  In the \LLEEonecharts~$\aonechart$ in Fig.~\ref{fig:countex:collapse}, and $\aonecharti{s}$ in Fig.~\ref{fig:ex:local:transfer:function},%\pagebreak[4]  
  the pair $\pair{\tightfbox{$abcd_1$}}{\tightfbox{$abcd_2$}}$ of vertices forms a crystalline reduced \onebisimilarity\ redundancy,
  since it is of kind \ref{R3.4}, hence \precrystalline, but not of kind \ref{R3.4.1} nor \ref{R3.4.2}.
\end{exa}
    
From a \LLEEonechart, every reduced \onebisimilarity\ redundancy that is not crystalline
  can be \emph{eliminated \LLEEpreservingly} by which we mean that the result is a \onebisimilar\ \LLEEonechart.  
The transformations needed are adaptations of the \connectthrough{\bverti{1}}{\bverti{2}} operation 
  from \cite{grab:fokk:2020:lics,grab:fokk:2020:lics:arxiv} 
    in which the incoming transitions at vertex~$\bverti{1}$
      are redirected to a \onebisimilar\ vertex~$\bverti{2}$.
Here this operation typically requires an unraveling step in which 
  loop levels above $\bverti{1}$ that are reachable by \onetransitions\ 
    are removed by similar transition redirections. 
In an example
  we illustrate the elimination of a reduced \onebisimilarity\ redundancy of kind \ref{R3.2} in Fig.~\ref{fig:ex:elim:R3.2}.%
\input{figs/fig-ex-elim-R-3-2.tex}%

\begin{lem}\label{lem:elim:all:but:cryst:1brs}
  Every reduced \onebisimilarity\ redundancy $\pair{\bverti{1}}{\bverti{2}}$  
    can be eliminated \LLEEpreservingly\ from a \onetransition\ limited \LLEEonechart\
      provided it is of either of the following kinds:
  \begin{enumerate}[label={(\roman{*})},align=right,leftmargin=*,itemsep=0.25ex]\vspace*{-0.1ex}
    \item
      $\pair{\bverti{1}}{\bverti{2}}$ is simple, or
    \item 
      $\pair{\bverti{1}}{\bverti{2}}$ is \precrystalline, but not crystalline. 
  \end{enumerate}
\end{lem}

\begin{center}%\vspace{-0.25ex}
  \input{figs/fig-ex-pars-insulated-exp.tex}%
\end{center}%\vspace{-1ex}  
In order to cut \twincrystals\ in \scc's we also need to safeguard
  that the joining loop vertices of crystalline reduced \onebisimilarity\ redundancies
    are `parsimoniously insulated' from above. 
The top vertices $\avert$ of \onebisimilarity\ redundancies of kind \ref{R3.3} and \ref{R3.4} 
  are, due to the occurring \onetransitions, substates of $\bverti{1}$ and $\bverti{2}$. 
Therefore any such vertex $\avert$ can be \emph{insulated from above}, that is, turned into a \maximalwrt{\sloopsbackto} vertex,
  by redirecting all induced transitions from $\avert$ into the \txtloopsbackto\ part of $\avert$ or below.   
In the example in Fig.~\ref{fig:ex:pars:insulated:exp}, where all proper transitions have the same action label, in an insulation step (first step)
  the transition from $\avert$ to $\cvert$ is redirected to the \onebisimilar\ target $\bvertbari{2}$.
In the arising \onechart\ $\avert$ is not yet \emph{parsimonious},
  because the \loopentrytransition\ from $\avert$ to $\bvertbari{1}$ 
    can be redirected to \onebisimilar\ target $\bverti{2}$ (second step),
      thereby making less use of the \txtloopsbackto\ part of $\avert$, and eliminating it
        (and permitting further minimization).

We now define `crystallized' (\LLEE-$\sone$-)charts as follows.

\begin{defi}[crystallized \onechart\ (\ref{contrib:concept:4})]\label{def:crystallized}
                 \label{def:crystallized:onechart}
  Let $\aonechart$ % $ = \tuple{\verts,\actions,\sone,\start,\transs,\termexts}$ 
    be a \onechart.
  
  We say that $\aonechart$ \emph{is crystallized}
  if there is a \LLEEwitness~$\aonecharthat$ of $\aonechart$
  such that the following four conditions hold:
  \begin{enumerate}[label={(cr-\arabic*)},leftmargin=*,align=right,itemindent=0ex,itemsep=0ex]%\vspace{-0.5ex}
    \item{}\label{cr-1}
      $\aonechart$ is a (finite) \onechart\ with \LLEEonelim,
      and specifically,
      $\aonechart$ is \onetransition\ limited with respect to $\aonecharthat$.
      % (Then $\aonechart$ is also weakly guarded by Lemma~\ref{lem:LLEEonelim:impl:guarded:wg:finite}.)
    \item{}\label{cr-2}
      Every $\aonecharthat$-reduced $\sone$-bisim. redundancy is crystalline. %\reducedwrt{\aonecharthat} \onebisimilarity\ redundancy in $\aonechart$ is crystalline.
    \item{}\label{cr-3}
      Every crystalline \reducedwrt{\aonecharthat} \onebisimilarity\ redundancy in $\aonechart$
        is parsimoniously insulated from above. 
      % For every crystallized \reducedwrt{\aonecharthat} \onebisimilarity\ redundancy in $\aonechart$
      %   the joining loop vertex is \maximalwrt{\sloopsbackto} and parsimonious.
      % Every loop vertex of $\aonecharthat$ 
      %   that joins a crystalline \reducedwrt{\aonecharthat} \onebisimilarity\ redundancy in $\aonechart$
      %   is \maximalwrt{\sloopsbackto} and parsimonious.
    \item{}\label{cr-4}
      Every carrier of an \scc\ in $\aonechart$ is grounded in $\aonechart$.     
  \end{enumerate}      
  Then we also say that $\aonechart$ is \emph{crystallized with respect to} $\aonecharthat$.  
\end{defi}

% %---------------------------------------------------------------------------
% \subsection{Twin-crystal \protect\sccs\ in crystallized \protect\onecharts}%
%   \label{crystallized:is:nearcollapsed::crystallization}
% %---------------------------------------------------------------------------

The lemma below gathers properties of crystallized \onecharts.
The subsequent proposition justifies the term `crystallized \onechart'
  by explaining the connection with \twincrystals.

\begin{lem}
  Every \onechart~$\aonechart$ that is crystallized with respect to a \LLEEwitness~$\aonecharthat$
    with \txtloopsbackto\ relation $\sloopsbackto$,
  satisfies:
  \begin{enumerate}[label={(\roman{*})},leftmargin=*,align=right,itemindent=0ex,itemsep=0.25ex]
    \item{}\label{it:1:lem:props:crystallized:onecharts}
      $\aonechart$ is \onecollapsed\ apart from within \sccs,
        i.e.\ \onecollapsed\ for \txtloopsbackto\ parts of \maximalwrt{\aonecharthat} loop vertices of $\aonecharthat$.
    \item{}\label{it:2:lem:props:crystallized:onecharts}
      $\aonechart$ is \onecollapsed\ for every \txtloopsbackto\ part 
      of a loop vertex of $\aonecharthat$ that is {not} \maximalwrt{\sloopsbackto}.
    % \item{}\label{it:3:lem:props:crystallized:onecharts}    
    %   Every \reducedwrt{\aonecharthat} \onebisimilarity\ redundancy %$\pair{\bverti{1}}{\bverti{2}}$ 
    %                                                                 in $\aonechart$ 
    %     is crystallized, and its joining loop vertex is \maximalwrt{\sloopsbackto}.
      % Every \reducedwrt{\aonecharthat} \onebisimilarity\ redundancy $\pair{\bverti{1}}{\bverti{2}}$ in $\aonechart$ 
      %   is crystallized, and has as its joining loop vertex 
      %   a \maximalwrt{\sloopsbackto} loop vertex of $\aonecharthat$. 
  \end{enumerate}\label{lem:props:crystallized:onecharts}
\end{lem} 

\begin{prop}[crystallized $\Rightarrow$ $\sone$-coll./\twincrystal\ \scc's]\label{prop:crystallized:2:twincrystal:sccs}
  For every carrier $\csetverts$ of a \scc\ 
    in a \onechart~$\aonechart$ that is crystallized, % with respect to a \LLEEwitness~$\aonecharthat$,
    either $\aonechart$ is \onecollapsed\ for $\csetverts$,
    or $\csetverts$ is the carrier~of~a~\twincrystal. % with respect to $\aonecharthat$.
\end{prop}

\begin{lem}[\ref{contrib:concept:4}]\label{lem:crystallized:is:nearcollapsed}
  Every crystallized \onechart\ is \nearcollapsed.
\end{lem}

By combining \LLEEpreserving\ eliminations of \precrystalline\ \onebisimilarity\ redundancies,
  of parsimonious insulation of crystalline \onebisimilarity\ redundancies, and of grounding of \sccs\
  %
  % By combining the \LLEEpreserving\ transformations as stated by Lem.~\ref{lem:elim:all:but:cryst:1brs}, 
  %                                                                Lem.~\ref{lem:pars:insulate:1brs}, and Lem.~\ref{lem:grounding:sccs},
  % and using Lem.~\ref{lem:crystallized:is:nearcollapsed},
  we are able to prove our main auxiliary statement.

\begin{thm}[crystallization, nearcollapse, (\ref{contrib:concept:3},\ref{contrib:concept:4})]\label{thm:crystallization:nearcollapse}
  Every %finite, 
        weakly guarded \LLEEonechart~$\aonechart$ 
  can be transformed, together with a \LLEEwitness~$\aonecharthat$ of $\aonechart$,
    into a \onebisimilar\ \LLEEonechart~$\aonechartacc$ with \onetransition\ limited \LLEEwitness~$\aonecharthatacc$ 
    such that $\aonechart$ is crystallized with respect to $\aonecharthatacc$,
    and $\aonechartacc$ is \nearcollapsed.
\end{thm}

% %----------------------------------------------------------------
% \subsection{Crystallized \protect\onecharts\ are near-collapsed}%
%   \label{crystallized:is:nearcollapsed::crystallization}
% %----------------------------------------------------------------

% \begin{defi}[near-collapsed]
%   We say that a \onechart~$\aonechart$ \emph{is \nearcollapsed} if \ldots . 
% \end{defi}

\input{figs/fig-counterex-transf-fun-elevation.tex}%
  \afterpage{\FloatBarrier}

%\newpage
%---------------------------------------------------------------------
\section{Near-collapsed LLEE-\protect\onecharts\ 
         have complete solutions}
  \label{nearcollapsed:LLEEonecharts:complete:solvable}
        %{Complete solutions for \protect\nearcollapsed\ \protect\onecharts}
        %{Collapsed \protect\LLEEonechartsbf\ are solvable}
        % {Crystallized \onecharts\ 
        % and near-collapsed \LLEEonechartsbf\ have complete solutions}
        %{Crystallized \protect\onecharts\ have complete solutions}
        % {Solving \LLEEonechartsbf\ and their bisimulation collapses}
%---------------------------------------------------------------------

We show that \nearcollapsed\ \LLEEonecharts\ have complete solutions
  by linking local transfer functions, as in the definition of `\nearcollapsed', 
  to transfer functions in order to be able to use the transfer-property \ref{Tone}
  for provable solutions.

Local-Transfer functions can be linked to transfer functions via the concept of the \oneLTS' $\elevatechart{\asetverts}{\aoneLTS}$
  that is the `elevation of a set $\asetverts$ of vertices above' a \oneLTS~$\aoneLTS = \tuple{\states,\actions,\sone,\transs,\termexts}$,
which is constructed as follows.
The set of vertices of $\elevatechart{\asetverts}{\aoneLTS}$
consists of two copies of its set $\verts$ of vertices,
the `ground floor' $\verts\times\setexp{\groundfloor}$, and the `first floor' $\verts\times\setexp{\firstfloor}$. 
These two copies of the set of vertices of $\aoneLTS$
are linked by copies of the corresponding transitions of $\aoneLTS$ 
with the %\underline{\smash{exception}} 
         \emph{exception} that \emph{proper} %\underline{\smash{proper}} transitions 
$\triple{\averti{1}}{\aact}{\averti{2}}$ of $\aoneLTS$
do \underline{\smash{not}} give rise to a proper transition of the form
$\triple{\pair{\averti{1}}{\firstfloor}}{\aact}{\pair{\averti{2}}{\firstfloor}}$ on the first floor
%\underline{\smash{if}} 
\emph{if} the vertex $\averti{2}$ is not contained in $\asetverts$.
Those transitions get redirected as transitions $\triple{\pair{\averti{1}}{\firstfloor}}{\aact}{\pair{\averti{2}}{\groundfloor}}$
to target the corresponding copy $\pair{\averti{2}}{\groundfloor}$ of $\averti{2}$ on the ground floor. 
Note that such redirections from the first floor to the ground floor do not happen for \onetransitions. 
The \suboneLTS\ of $\elevatechart{\asetverts}{\aoneLTS}$ that consists of all transitions between vertices on the ground floor
is an exact copy of the original \oneLTS~$\aoneLTS$. 
Yet within the elevation $\elevatechart{\asetverts}{\aoneLTS}$ of $\asetverts$ above $\aoneLTS$,
a number of vertices on the ground floor will have additional incoming \properaction\ transitions from vertices on the first floor.

\begin{example}
  Fig.~\ref{fig:ex:lem:transfer:lift:local:transfer} contains two copies of the elevation $\elevationofabove{\alert{\asetverts}}{\aonecharti{s}}$ 
        of the set $\alert{\asetverts} \defdby \setexp{ \tightfbox{$abcd_1$}, \tightfbox{$abcd_2$} }$ 
        of vertices of the \onechart~$\aonecharti{s}$ in Fig.~\ref{fig:ex:local:transfer:function} %from Rem.~\ref{rem:prop:not:LLEE:pres:collapse:minimization}
          above $\aonecharti{s}$. 
\end{example}    
    
Indeed, Fig.~\ref{fig:ex:lem:transfer:lift:local:transfer} also shows how the local transfer function $\protect\magenta{\scpfunon{\black{\twinc}}}$ 
      on \onechart~$\aonecharti{s}$ in Fig.~\protect\ref{fig:ex:local:transfer:function} %from Rem.~\ref{rem:prop:not:LLEE:pres:collapse:minimization}
    can be lifted to a transfer function $\protect\sliftltfunn{\protect\magenta{\scpfunon{\black{\twinc}}}}$
    on the elevation $\protect\elevationofabove{\protect\alert{\protect\asetverts}}{\protect\aonecharti{s}}$ 
      of $\protect\alert{\asetverts}$  % $\defdby \setexp{ \tightfbox{$abcd_1$}, \tightfbox{$abcd_2$}}$ 
          above $\aonecharti{s}$:    
namely by defining $\protect\sliftltfunn{\protect\magenta{\scpfunon{\black{\twinc}}}}$ as $\protect\magenta{\scpfunon{\black{\twinc}}}$
  on the first floor, and as the identity function on the ground floor of $\protect\elevationofabove{\protect\alert{\protect\asetverts}}{\protect\aonecharti{s}}$. 

In general the following statement holds, which is a generalization to \onecharts\ of Prop.~2.4 in \cite{grab:2021:NWPT}.
Every local transfer function 
    $\sphifun \funin \states \rightharpoonup \states$ 
  on a \oneLTS~$\aoneLTS$   
  with $\asetverts \defdby \domof{\sphifun} \cap \ranof{\sphifun}$ %\vspace*{-2pt}
  lifts to a transfer function $\sliftltfun$ on the elevation $\elevationofabove{\asetverts}{\aoneLTS}$ of $\asetverts$ over~$\aoneLTS$,
    via the projection transfer function $\sproji{1}$ from $\elevationofabove{\asetverts}{\aoneLTS}$ to $\aoneLTS$, 
  such that the diagram below commutes for vertices on the first floor 
    %$\elevationofabove{\asetverts}{\aoneLTS}^{(\firstfloor)} $ % $= \elevationofabove{\asetverts}{\aoneLTS} \cap (\verts\times{\firstfloor})$ 
    of $\elevationofabove{\asetverts}{\aoneLTS}$:
  \begin{equation}\label{diag:lem:transfer:lift:local:transfer}
    \hspace*{-3ex}
    \begin{aligned}[c]
      \input{figs/diag-lifting-transfer-fun.tex}
    \end{aligned}
    \hspace*{3ex}
    \begin{aligned}
      & \hspace*{-1.5ex}\text{for all $\bvert\in\domof{\sphifun}$:}
      \\[0.5ex]
      &
      \compfuns{(\sproji{1}}{\sphifun)}{\pair{\bvert}{\firstfloor}}
      \\[-0.5ex]
      & 
      {} \qquad \synteq
      \\[-1ex]
      & \compfuns{(\sliftltfunn{\sphifun}}{\sproji{1})}{\pair{\bvert}{\firstfloor}} 
    \end{aligned}
  \end{equation}

Together with preservation of \LLEE\ for elevations, 
the possibility to lift local transfer functions to transfer functions on elevations
  facilitates us to use invariance of provable solutions under transfer functions between \LLEEoneLTSs\
  for proving invariance of provable solutions under local transfer functions on \LLEEoneLTSs. 
Then we can use this fact to show complete solvability of \nearcollapsed\ \mbox{\LLEEonecharts}.  
 % As an application of this statement we then can show complete solvability of \nearcollapsed\ \LLEEonecharts.

\begin{lem}%
    \label{lem:sol:inv:under:ltfuns}
  Let $\sphifun \funin \verts \rightharpoonup \verts$ be a local transfer function
    on a \wg\ \LLEEoneLTS~$\aoneLTS = \tuple{\verts,\actions,\sone,\transs,\termexts}$.
  Then every \provablein{\milnersys} solution $\sasol$ of $\aoneLTS$ 
    is \provablyin{\milnersys} invariant~under~$\sphifun$: %, that is:
  \begin{equation}\label{eq:lem:sol:inv:under:ltfuns}
    \asol{\bvert}
      \milnersyseq
    \asol{\phifun{\bvert}}
      \quad \text{(for all $\bvert\in\domof{\sphifun}$) \punc{.}} 
  \end{equation}
  % that is, $\sasol$ is \provablyin{\milnersys} invariant under $\sphifun$.
\end{lem}

\begin{proof}[Proof (Sketch)]
  We use the diagram in \eqref{diag:lem:transfer:lift:local:transfer},
    and that $\elevationofabove{\asetverts}{\aoneLTS}$ is also a \LLEEoneLTS.
  Since $\sproji{1}$ and $\scompfuns{\sliftltfunn{\sphifun}}{\sproji{1}}$ are transfer functions,
  by Lem.~\ref{lem:sol:inv:under:tfuns} %entails that
  $\scompfuns{(\sasol}{\sproji{1})}$
    and 
  $\scompfuns{(\sasol}{\scompfuns{\sproji{1}}{\sliftltfun)}}$
  are \provablyin{\milnersys} equal.  
  Then by using we diagram commutativity on the first floor in \eqref{diag:lem:transfer:lift:local:transfer}
  we get:   
  $\asol{\bvert}
     \synteq
   \compfuns{(\sasol}{\sproji{1})}{\pair{\bvert}{\firstfloor}}
     \milnersyseq
   \compfuns{(\sasol}{\scompfuns{\sproji{1}}{\sliftltfun)}}{\pair{\bvert}{\firstfloor}}
   $
   $  \synteq
   \compfuns{(\sasol}{\scompfuns{\sphifun}{\sproji{1}})}{\pair{\bvert}{\firstfloor}}
     \synteq
   \asol{\phifun{\bvert}}$,
  for all $\bvert\in\domof{\sphifun}$. In this way we have obtained \eqref{eq:lem:sol:inv:under:ltfuns}.
\end{proof}

% %----------
% \subsection%{Near-collapsed \protect\LLEEonechartsbf\ have\\ complete solutions}
%            {Crystallized \protect\onecharts\ have complete solutions}
% %---------
% 
% 
% 
% 
\begin{lem}[\ref{contrib:concept:6}]\label{lem:nearcollapsed:2:compl:solvable}
  Every %finite, 
        \wg\ \LLEEonechart\ that is \nearcollapsed\
    has a \completewrt{\milnersys} \provablein{\milnersys} solution.
\end{lem}

\section{Conclusion}%
  \label{conclusion}
%--------------------

% \begin{thm}
%   Every %\wg\
%         guarded \LLEEonechart\ can be transformed into a \onebisimilar, crystallized, \nearcollapsed\ \LLEEonechart.
% \end{thm}  
% 
% \begin{corollary}
%   Every chart interpretation of a star expression
%     can be transformed into a \onebisimilar, crystallized, and \nearcollapsed\ \LLEEonechart. 
% \end{corollary}

As a consequence of the crystallization process for \LLEEonecharts\ and of Thm.~\ref{thm:crystallization:nearcollapse}
  we also obtain a new characterization of expressibility of finite process graphs
    in the process semantics.
For this purpose we say that a \onechart~$\aonechart$ is \emph{expressible by a regular expression modulo \onebisimilarity}
  if $\aonechart$ is \onebisimilar\ to the chart interpretation of a star expression. 

\begin{cor}\label{cor:expressible}
  A \onechart~$\aonechart$ is expressible modulo \onebisimilarity\ % in the process semantics
    if and only if
  $\aonechart$ is \onebisimilar\ to a crystallized, and hence to a \nearcollapsed, \LLEEonechart.
\end{cor}

Since the size of a crystallized \onechart\ is bounded by at most double the size of an \onebisimulation\ collapse
(as every vertex is \onebisimilar\ to at most one other vertex in \twincrystals, and crystallized \onecharts), 
  this characterization raises the hopes for a polynomial algorithm
  for recognizing expressibility of finite process graphs. %, 
    %in view of the fact that crystallized \onecharts\ are at most double the size of a \onebisimulation\ collapse. 
Such a recognition algorithm would substantially improve on
  the superexponential algorithm for deciding expressibility in \cite{baet:corr:grab:2007}.

%% file: figs/fig-bisimchart-strategy.tex
\begin{center}
  \vspace*{-0.5ex}
\begin{tikzpicture}
  \matrix[anchor=center,row sep=0.5cm,column sep=0.9cm] {
      & \node(bisimchartpos){};
    \\
    \node(chartofeonepos){};
      & & \node(chartofetwopos){};
    \\
    };  
  % nodes  
  \path(bisimchartpos) ++ (0cm,0cm) node(bisimchart){$\abisimchart$}; 
  \path(chartofeonepos) ++ (0cm,0cm) node(chartofeone){$\chartof{\astexpi{1}}$};
  \path(chartofetwopos) ++ (0cm,0cm) node(chartofetwo){$\chartof{\astexpi{2}}$};  
  %
    
  %  
  % arrows  
  %  
  \draw[funbisimleft,thick,shorten >=-1pt] (bisimchart) to (chartofeone); 
  \draw[funbisimright,thick,shorten >=-1pt](bisimchart) to (chartofetwo);  
  \draw[bisim,thick] (chartofeone) to node[below]{\small (assm)} (chartofetwo);

  \path (chartofetwo) ++ (0.5cm,0cm) node[right]{$\astexpi{i}$ is solution of $\chartof{\astexpi{i}}$ ($i\in\setexp{1,2}$)};
  \path (chartofetwo) ++ (1.75cm,0.4cm) node{$\Uparrow$};
  \path (bisimchart) ++ (0.9cm,0.25cm) node[right]{$\left.%\{
                                                   \parbox{\widthof{$\astexpi{2}$ is solution of $\abisimchart$}}
                                                          {$\astexpi{1}$ is solution of $\abisimchart$%
                                                           \\[-0.35ex]
                                                           $\astexpi{2}$ is solution of $\abisimchart$}
                                                 \right\}
                                                   \overset{\mbox{\bf\large\alert{?}}}{\Longrightarrow}
                                                     \astexpi{1} \mathrel{\alert{\milnersyseq}} \astexpi{2}$};  
  
\end{tikzpicture}  
  \vspace*{-1ex}
\end{center}

%% file: figs/fig-bisimcoll-strategy.tex
\begin{center}
  \vspace*{-0.5ex}
\begin{tikzpicture}
  \matrix[anchor=center,row sep=0.5cm,column sep=0.85cm] {
    \node(chartofeonepos){};
      & & \node(chartofetwopos){};
    \\  
      & \node(bisimcollpos){};
    \\
    };  
  % nodes  
  \path(chartofeonepos) ++ (0cm,0cm) node(chartofeone){$\chartof{\astexpi{1}}$};  
    \path (chartofeone) ++ (0cm,0.4cm) node{\forestgreen{\LLEE}}; 
  \path(chartofetwopos) ++ (0cm,0cm) node(chartofetwo){$\chartof{\astexpi{2}}$};   
    \path (chartofetwo) ++ (0cm,0.4cm) node{\forestgreen{\LLEE}}; 
  \path(bisimcollpos) ++ (0cm,0cm) node(bisimcoll){$\acharti{0}$}; 
    \path (bisimcoll) ++ (0cm,-0.4cm) node{\forestgreen{\LLEE}};   
  %
    
  %  
  % arrows  
  %  
  \draw[bisim,thick] (chartofeone) to node[above]{\small (assm)} (chartofetwo); 
  \draw[funbisimleft,thick]  (chartofeone) to (bisimcoll); 
  \draw[funbisimright,thick] (chartofetwo) to (bisimcoll);  
  
  \draw[-implies,double equal sign distance] ($(bisimcoll) + (1.25cm,0cm)$) to ($(bisimcoll) + (2.15cm,0.5cm)$);
  \draw[-implies,double equal sign distance,out=225,in=-45,distance=1.1cm] ($(bisimcoll) + (0.6cm,-0.5cm)$) to ($(bisimcoll) + (-1.85cm,0.45cm)$);
  \draw[-implies,double equal sign distance,out=270,in=180,distance=0.6cm] ($(bisimcoll) + (-2.9cm,-0.6cm)$) to ($(bisimcoll) + (2.25cm,-1cm)$);
  
  %
  % solutions
  %
  \path (chartofeone) ++ (-0.75cm,-0.425cm) node[left]{\parbox{\widthof{$\astexpi{1}$ is solution}}
                                                              {$\astexpi{1}$ is solution%
                                                               \\[-0.4ex]
                                                               $\underbrace{\text{$\astexpi{0}$ is solution}}_{{}}$
                                                               \\[-1.5ex]
                                                               \hspace*{\fill}$\Downarrow$\hspace*{\fill}
                                                               \\[-0.4ex]
                                                               \hspace*{\fill}$\astexpi{1} \mathrel{\mediumblue{\BBPeq}} \astexpi{0}$\hspace*{\fill}%
                                                               }};
                                                            
  \path (chartofetwo) ++ (1.1cm,%-0.425cm
                                -0.7cm) node[right]{\parbox{\widthof{$\astexpi{2}$ is solution}}
                                                            {$\astexpi{2}$ is solution%
                                                              \\[-0.4ex]
                                                              $\underbrace{\text{$\astexpi{0}$ is solution}}_{{}}$
                                                              \\[-1.5ex]
                                                              \hspace*{\fill}$\Downarrow$\hspace*{\fill}
                                                              \\[-0.4ex]
                                                              \hspace*{\fill}$\astexpi{0} \mathrel{\mediumblue{\BBPeq}} \astexpi{2}$\hspace*{\fill}%
                                                              \\[-0.4ex]
                                                              \hspace*{\fill}$\Downarrow$\hspace*{\fill}
                                                              \\[-0.4ex]
                                                              \hspace*{\fill}$\astexpi{1} \mathrel{\mediumblue{\BBPeq}} \astexpi{2}$\hspace*{\fill}%
                                                              }};
  \path (bisimcoll) ++ (0.5cm,-0.2cm) node[right]{\parbox{\widthof{solution $\astexpi{0}$}}
                                                          {has%
                                                           \\[-0.6ex]
                                                           solution $\astexpi{0}$}};

%   \path (chartofetwo) ++ (0.5cm,0cm) node[right]{$\astexpi{i}$ is solution of $\chartof{\astexpi{i}}$ ($i\in\setexp{1,2}$)};
%   \path (chartofetwo) ++ (1.75cm,0.4cm) node{$\Uparrow$};
%   \path (bisimcoll) ++ (0.9cm,0.25cm) node[right]{$\left.%\{
%                                                    \parbox{\widthof{$\astexpi{2}$ is solution of $\chartof{\astexpi{2}}$}}
%                                                           {$\astexpi{1}$ is solution of $\acharti{0}$%
%                                                            \\[-0.35ex]
%                                                            $\astexpi{2}$ is solution of $\acharti{0}$}
%                                                  \right\}
%                                                    \overset{\mbox{\bf\large\alert{?}}}{\Rightarrow}
%                                                      \astexpi{1} \milnersyseq \astexpi{2}$};  
  
\end{tikzpicture}  
  \vspace*{-1ex}
\end{center}

%% file: figs/fig-bisimcoll-strategy-fail.tex
\begin{center}
\begin{tikzpicture}
  \matrix[anchor=center,row sep=0.4cm,column sep=1.75cm] {
    \node(chartintnotLLEEpos){}; & & \node(aonechartpos){};
    \\
    \\
    & & \node(aonechartcollpos){};
    \\
    };
  \path (chartintnotLLEEpos) ++ (0cm,0.15cm) node (chartintnotLLEE){$ \chartof{\astexp} $};
    \path (chartintnotLLEE) ++ (0.25cm,0.175cm) node[right]{\alert{\st{\LLEE}}};  
    \path (chartintnotLLEE) ++ (-0.45cm,0cm) node[left]{\text{\crtcrossreflabel{\st{\bf (I)}}[Ionefail]:}};
  \path (aonechartpos) ++ (0cm,0cm) node(aonechart){$ \aonechart $};
    \path (aonechart) ++ (0.1cm,0.3cm) node[right]{\forestgreen{\LLEE}};   
    \path (aonechart) ++ (-0.45cm,0.135cm) node[left]{\text{\crtcrossreflabel{\st{\bf (C)}}[Conefail]:}}; 
  \path (aonechartcollpos) ++ (0cm,0cm) node(aonechartcoll){$ \aonecharti{0} $}; 
    \path (aonechartcoll) ++ (0.1cm,0.3cm) node[right]{\alert{\st{\LLEE}}};     
    \path (aonechartcoll) ++ (-0.25cm,0.1cm) node[left]{any \onebisimulation\ collapse}; 
    \path (aonechartcoll) ++ (1cm,0.1cm) node[right]{of $\aonechart$}; 
   \draw[funonebisimleft,shorten >=4pt] (aonechart) to (aonechartcoll);

\end{tikzpicture}  
\end{center}

%% file: figs/fig-part-proof-structure.tex
\begin{tikzpicture} 
  \matrix[anchor=center,row sep=0.8cm,column sep=1.25cm,
          every node/.style={minimum width=1em}] {
    \node(aonechartone)[text height=0.675em]{$\aonecharti{1}$};
      & &[-0.85cm] \node(aonecharttwo)[text height=0.675em]{$\aonecharti{2}$};
    \\%[0.25cm]
    \node(aonechartcryst)[text height=0.675em]{$\aonecharti{10}$}; 
    \\[-0.9cm]
       & \node(bisimcoll)[text height=0.675em]{$\aonecharti{0}$};
    \\
    };  
 
  \path (aonechartone) ++ (-0.3cm,0.05cm) node[left]{$\astexpi{10}$ is solution}; 
  \path (aonechartone) ++ (-0.25cm,0.4cm) node[left]{\forestgreen{\LLEE}};

  \path (aonechartcryst) ++ (-0.25cm,0.4cm) node[left]{\forestgreen{\LLEE}}; 
  \path (aonechartcryst) ++ (-0.25cm,-0.575cm) node{\parbox{\widthof{has complete}}
                                                     {has complete%
                                                      \\[-0.5ex]
                                                      \hspace*{\fill}solution $\astexpi{10}$}}; 
  \path (aonechartcryst) ++ (-0.275cm,0cm) node[left]{\chocolate{\emph{crystallized}}};
    % \path (aonechartcryst) ++ (-0.2cm,0.15cm) node[left]{\parbox{\widthof{crystallized,}}
    %                                                             {\hspace*{\fill}\emph{crystallized,}\mbox{}%
    %                                                              \\[-0.5ex]
    %                                                              \emph{near-}%
    %                                                              \\[-0.5ex]
    %                                                              collapsed}};

  \path (bisimcoll) ++ (-0.5cm,-0.4cm) node[right]{$\astexpi{10}$ is solution}; 
  \path (bisimcoll) ++ (0.35cm,0.2cm) node[right]{\parbox{\widthof{bisimulation}}
                                                        {\emph{bisimulation}%
                                                         \\[-0.5ex]
                                                         \emph{collapse}}};
    
  \path (aonecharttwo) ++ (0.3cm,0.05cm) node[right]{$\astexpi{10}$ is solution};

  % bisimulations, functional bisimulations  
  %
  \draw[onebisim] (aonechartone) to  (aonecharttwo);
  \draw[onebisim] (aonechartcryst) to %node[left,xshift=0cm,pos=0.5]{\small\ref{CRone}} 
                                      (aonechartone);
  \draw[funonebisim,shorten <=0pt,shorten >=4pt] (aonechartone) to (bisimcoll);
  \draw[funonebisim,shorten <=0pt,shorten >=0pt] (aonechartcryst) to  (bisimcoll);
  \draw[funonebisimleft,shorten <=9pt,shorten >=6pt] (aonecharttwo.center) to (bisimcoll.center);

  %
  % implications
  %  
  \draw[-implies,double equal sign distance,shorten <=3pt,shorten >=6pt] 
    ($(aonechartcryst) + (0.6cm,-0.6cm)$) to ($(bisimcoll) + (-0.2cm,-0.4cm)$);
  % \draw[-implies,double equal sign distance,out=-30,in=210,distance=0.5cm,shorten <=3pt,shorten >=3pt] 
  %   ($(aonechartcryst) + (0cm,-0.9cm)$) to ($(bisimcoll) + (-0.3cm,-0.5cm)$);
    
  \draw[-implies,double equal sign distance,out=0,in=-90,distance=1cm,shorten <=0pt,shorten >=0pt] 
    ($(bisimcoll) + (1.7cm,-0.425cm)$) to ($(aonecharttwo) + (1.6cm,-0.15cm)$);
  
  \draw[-implies,double equal sign distance,out=190,in=-95,distance=2.65cm,shorten <=0pt,shorten >=0pt] 
    ($(bisimcoll) + (-0.4cm,-0.55cm)$) to ($(aonechartone) + (-2.2cm,-0.2cm)$);

\end{tikzpicture}

%% file: figs/fig-ex-LEE.tex
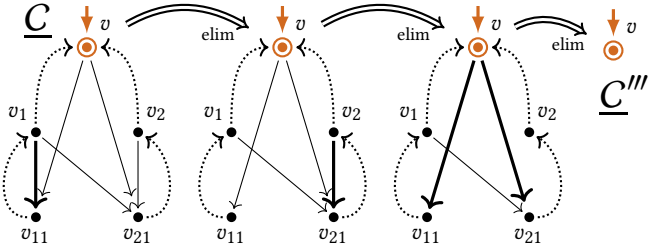
\begin{figure}[t!]
\begin{center}
\begin{tikzpicture}
 
\matrix[anchor=center,row sep=1cm,column sep=0.55cm,
        every node/.style={draw,very thick,circle,minimum width=2.5pt,fill,inner sep=0pt,outer sep=2pt}] at (0,0) {
    & \node[chocolate](v--1){};
  \\          
  \node(v1){};
    & & \node(v2){};
  \\
  \node(v11){};
    & & \node(v21){};
  \\   
  };   
\path (v--1) ++ (-0.65cm,0.35cm) node(label){\LARGE $\aonechart$};

\draw[<-,very thick,>=latex,chocolate,shorten <= 2pt](v--1) -- ++ (90:0.55cm);   
% (v--1) 
\draw[thick,chocolate] (v--1) circle (0.12cm);
\path (v--1) ++ (0.25cm,0.3cm) node{\small $\avert$};
\draw[->,shorten >= 0.175cm,shorten <= 2pt] 
  (v--1) to % node[left,pos=0.36,xshift=0.075cm]{\small $\small\black{\aact}$} node[right,pos=0.4,xshift=-0.075cm,yshift=1pt]{\small $\loopnsteplab{2}$}  
         (v11);
\draw[->,shorten >= 0.175cm,shorten <= 2pt] 
  (v--1) to % node[right,pos=0.36,xshift=-0.05cm]{\small $\small\black{\bact}$} node[left,pos=0.6,xshift=0.075cm,yshift=1pt]{\small $\loopnsteplab{2}$} 
         (v21);

% (v1) 
\path (v1) ++ (-0.225cm,0.25cm) node{\small $\averti{1}$};
\draw[->,very thick,shorten >= 0pt]
  (v1) to % node[left,pos=0.25,xshift=0.075cm]{\small $\small\black{\aact}$} node[left,pos=0.6,xshift=0.075cm]{\small $\loopnsteplab{1}$} 
          (v11);
\draw[->,thick,densely dotted,out=90,in=180,distance=0.5cm,shorten >=2pt](v1) to (v--1);
\draw[->,shorten <= 0pt,shorten >= 0pt] (v1) to %node[below]{$\bact$} 
                                                (v21); 

% (v11)
\path (v11) ++ (0cm,-0.25cm) node{\small $\averti{11}$};
\draw[->,thick,densely dotted,out=180,in=210,distance=0.45cm](v11) to (v1);

% (v2) 
\path (v2) ++ (0.25cm,0.25cm) node{\small $\averti{2}$};
\draw[->,shorten >= 0pt]
  (v2) to % node[right,pos=0.25,xshift=-0.075cm]{\small $\small\black{\bact}$} node[right,pos=0.6,xshift=-0.075cm]{\small $\loopnsteplab{1}$} 
          (v21);
\draw[->,thick,densely dotted,out=90,in=0,distance=0.5cm,shorten >= 2pt](v2) to (v--1);

% (v21) 
\path (v21) ++ (0cm,-0.25cm) node{\small $\averti{21}$};
\draw[->,thick,densely dotted,out=-0,in=-30,distance=0.45cm](v21) to (v2);

\matrix[anchor=center,row sep=1cm,column sep=0.55cm,
        every node/.style={draw,very thick,circle,minimum width=2.5pt,fill,inner sep=0pt,outer sep=2pt}] at (2.6,0) {
    & \node[chocolate](v--2){};
  \\           
  \node(v1){};
    & & \node(v2){};
  \\
  \node(v11){};
    & & \node(v21){};
  \\   
  };   
  %
%\path (v--2) ++ (0cm,-1cm) node(label){\Large $\aonechart$, $\aonecharthat$};  

\draw[<-,very thick,>=latex,chocolate,shorten <= 2pt](v--2) -- ++ (90:0.55cm);   
% (v--2) 
\draw[thick,chocolate] (v--2) circle (0.12cm);
\path (v--2) ++ (0.25cm,0.3cm) node{\small $\avert$};
\draw[->,shorten <= 2pt] 
  (v--2) to % node[left,pos=0.36,xshift=0.075cm]{\small $\small\black{\aact}$} node[right,pos=0.4,xshift=-0.075cm,yshift=1pt]{\small $\loopnsteplab{2}$}  
         (v11);
\draw[->,shorten >= 0.175cm,shorten <= 2pt] 
  (v--2) to % node[right,pos=0.36,xshift=-0.05cm]{\small $\small\black{\bact}$} node[left,pos=0.6,xshift=0.075cm,yshift=1pt]{\small $\loopnsteplab{2}$} 
         (v21);

% (v1) 
\path (v1) ++ (-0.225cm,0.25cm) node{\small $\averti{1}$};
% \draw[->,very thick,shorten >= 0pt]
%   (v1) to % node[left,pos=0.25,xshift=0.075cm]{\small $\small\black{\aact}$} node[left,pos=0.6,xshift=0.075cm]{\small $\loopnsteplab{1}$} 
%           (v11);
\draw[->,thick,densely dotted,out=90,in=180,distance=0.5cm,shorten >=2pt](v1) to (v--2);
\draw[->,shorten <= 0pt,shorten >= 0pt] (v1) to %node[below]{$\bact$} 
                                                (v21); 

% (v11)
\path (v11) ++ (0cm,-0.25cm) node{\small $\averti{11}$};
\draw[->,thick,densely dotted,out=180,in=210,distance=0.45cm](v11) to (v1);

% (v2) 
\path (v2) ++ (0.25cm,0.25cm) node{\small $\averti{2}$};
\draw[->,very thick,shorten >= 0pt]
  (v2) to % node[right,pos=0.25,xshift=-0.075cm]{\small $\small\black{\bact}$} node[right,pos=0.6,xshift=-0.075cm]{\small $\loopnsteplab{1}$} 
          (v21);
\draw[->,thick,densely dotted,out=90,in=0,distance=0.5cm,shorten >= 2pt](v2) to (v--2);

% (v21) 
\path (v21) ++ (0cm,-0.25cm) node{\small $\averti{21}$};
\draw[->,thick,densely dotted,out=-0,in=-30,distance=0.45cm](v21) to (v2);

\draw[-implies,thick,double equal sign distance, bend left,distance=1.1cm,
               shorten <= 0.5cm,shorten >= 0.4cm
               ] (v--1) to node[below,pos=0.7]{\scriptsize elim} (v--2);

\matrix[anchor=center,row sep=1cm,column sep=0.55cm,
        every node/.style={draw,very thick,circle,minimum width=2.5pt,fill,inner sep=0pt,outer sep=2pt}] at (5.2,0) {
    & \node[chocolate](v--3){};
  \\           
  \node(v1){};
    & & \node(v2){};
  \\
  \node(v11){};
    & & \node(v21){};
  \\   
  };   
  %
%\path (v--3) ++ (0cm,-1cm) node(label){\Large $\aonechart$, $\aonecharthat$};  

\draw[<-,very thick,>=latex,chocolate,shorten <= 2pt](v--3) -- ++ (90:0.55cm);   
% (v--3) 
\draw[thick,chocolate] (v--3) circle (0.12cm);
\path (v--3) ++ (0.25cm,0.3cm) node{\small $\avert$};
\draw[->,very thick,shorten <= 2pt] 
  (v--3) to % node[left,pos=0.36,xshift=0.075cm]{\small $\small\black{\aact}$} node[right,pos=0.4,xshift=-0.075cm,yshift=1pt]{\small $\loopnsteplab{2}$}  
         (v11);
\draw[->,very thick,shorten >= 0.1cm,shorten <= 2pt] 
  (v--3) to % node[right,pos=0.36,xshift=-0.05cm]{\small $\small\black{\bact}$} node[left,pos=0.6,xshift=0.075cm,yshift=1pt]{\small $\loopnsteplab{2}$} 
         (v21);

% (v1) 
\path (v1) ++ (-0.225cm,0.25cm) node{\small $\averti{1}$};
% \draw[->,very thick,shorten >= 0pt]
%   (v1) to % node[left,pos=0.25,xshift=0.075cm]{\small $\small\black{\aact}$} node[left,pos=0.6,xshift=0.075cm]{\small $\loopnsteplab{1}$} 
%           (v11);
\draw[->,thick,densely dotted,out=90,in=180,distance=0.5cm,shorten >=2pt](v1) to (v--3);
\draw[->,shorten <= 0pt,shorten >= 0pt] (v1) to %node[below]{$\bact$} 
                                                (v21); 

% (v11)
\path (v11) ++ (0cm,-0.25cm) node{\small $\averti{11}$};
\draw[->,thick,densely dotted,out=180,in=210,distance=0.45cm](v11) to (v1);

% (v2) 
\path (v2) ++ (0.25cm,0.25cm) node{\small $\averti{2}$};
% \draw[->,very thick,shorten >= 0pt]
%  (v2) to % node[right,pos=0.25,xshift=-0.075cm]{\small $\small\black{\bact}$} node[right,pos=0.6,xshift=-0.075cm]{\small $\loopnsteplab{1}$} 
%          (v21);
\draw[->,thick,densely dotted,out=90,in=0,distance=0.5cm,shorten >= 2pt](v2) to (v--3);

% (v21) 
\path (v21) ++ (0cm,-0.25cm) node{\small $\averti{21}$};
\draw[->,thick,densely dotted,out=-0,in=-30,distance=0.45cm](v21) to (v2);

\draw[-implies,thick,double equal sign distance, bend left,distance=1.1cm,
               shorten <= 0.5cm,shorten >= 0.4cm
               ] (v--2) to node[below,pos=0.7]{\scriptsize elim} (v--3);

\matrix[anchor=center,row sep=1cm,column sep=0.55cm,
        every node/.style={draw,very thick,circle,minimum width=2.5pt,fill,inner sep=0pt,outer sep=2pt}] at (7,0) {
    & \node[chocolate](v--4){};
  \\           
  \node[draw=none,fill=none](v1){};
    & & \node[draw=none,fill=none](v2){};
  \\
  \node[draw=none,fill=none](v11){};
    & & \node[draw=none,fill=none](v21){};
  \\   
  };   
\path (v--4) ++ (0.15cm,-0.6cm) node(label){\LARGE $\aonechart'''$};

\draw[<-,very thick,>=latex,chocolate,shorten <= 2pt](v--4) -- ++ (90:0.55cm);   
% (v--4) 
\draw[thick,chocolate] (v--4) circle (0.12cm);
\path (v--4) ++ (0.25cm,0.3cm) node{\small $\avert$};

\draw[-implies,thick,double equal sign distance, bend left,distance=0.75cm,
               shorten <= 0.45cm,shorten >= 0.3cm
               ] (v--3) to node[below,pos=0.7]{\scriptsize elim} (v--4);

\end{tikzpicture}
\end{center}
  \vspace*{-2.5ex}
  \caption{\label{fig:ex:LEE}%
           A successful run of the loop elimination procedure.
           The start vertex is indicated by \protect\picarrowstart,
           immediate termination by a boldface ring.
           \protect\Loopentry\ transitions of loop \protect\subonecharts\ eliminated in the next step are marked in bold. 
           Action labels are neglected, however dotted arrows indicate \protect\onetransitions. 
           }
\end{figure}    

%% file: figs/fig-ex-LLEEw-1-2-3.tex
\begin{figure}
  \vspace*{-2ex}
\begin{center}
\begin{tikzpicture}
\matrix[anchor=center,row sep=1cm,column sep=0.55cm,,
        every node/.style={draw,very thick,circle,minimum width=2.5pt,fill,inner sep=0pt,outer sep=2pt}] at (0,0) {
    & \node[chocolate](v){};
  \\[-0.25ex]              
  \node(v1){};
    & & \node(v2){};
  \\[0.25cm]
  \node(v11){};
    & & \node(v21){};
  \\   
  };   
\path (v) ++ (-0.6cm,0.5cm) node(label){\LARGE $\aonecharthati{1}$};

\draw[<-,very thick,>=latex,chocolate,shorten <= 2pt](v) -- ++ (90:0.55cm);   
% (v) 
\draw[thick,chocolate] (v) circle (0.12cm);
\path (v) ++ (0.25cm,-0.3cm) node{\small $\avert$};
\draw[->,thick,darkcyan,shorten >= 0.175cm,shorten <= 2pt] 
  (v) to %node[left,pos=0.36,xshift=0.075cm]{\small $\small\black{\aact}$} 
         node[right,pos=0.4,xshift=-0.075cm,yshift=1pt]{\small $\loopnsteplab{3}$}  (v11);
\draw[->,thick,darkcyan,shorten >= 0.175cm,shorten <= 2pt] 
  (v) to %node[right,pos=0.36,xshift=-0.05cm]{\small $\small\black{\bact}$} 
         node[left,pos=0.6,xshift=0.075cm,yshift=1pt]{\small $\loopnsteplab{3}$} (v21);

% (v1) 
\path (v1) ++ (-0.225cm,0.25cm) node{\small $\averti{1}$};
\draw[->,thick,darkcyan,shorten >= 0pt]
  (v1) to % node[left,pos=0.25,xshift=0.075cm]{\small $\small\black{\aact}$} 
          node[left,pos=0.6,xshift=0.075cm]{\small $\loopnsteplab{1}$} (v11);
\draw[->,thick,densely dotted,out=90,in=180,distance=0.5cm,shorten >=2pt](v1) to (v);
\draw[->,shorten <= 0pt,shorten >= 0pt] (v1) to % node[below]{$\bact$} 
                                                (v21); 

% (v11)
\path (v11) ++ (0cm,-0.25cm) node{\small $\averti{11}$};
\draw[->,thick,densely dotted,out=180,in=210,distance=0.7cm](v11) to (v1);

% (v2) 
\path (v2) ++ (0.25cm,0.25cm) node{\small $\averti{2}$};
\draw[->,thick,darkcyan,shorten >= 0pt]
  (v2) to % node[right,pos=0.25,xshift=-0.075cm]{\small $\small\black{\bact}$} 
          node[right,pos=0.6,xshift=-0.075cm]{\small $\loopnsteplab{2}$} (v21);
\draw[->,thick,densely dotted,out=90,in=0,distance=0.5cm,shorten >= 2pt](v2) to (v);

% (v21) 
\path (v21) ++ (0cm,-0.25cm) node{\small $\averti{21}$};
\draw[->,thick,densely dotted,out=0,in=-30,distance=0.7cm](v21) to (v2);

\matrix[anchor=center,row sep=1cm,column sep=0.55cm,,
        every node/.style={draw,very thick,circle,minimum width=2.5pt,fill,inner sep=0pt,outer sep=2pt}] at (2.9,0) {
    & \node[chocolate](v){};
  \\[-0.25ex]              
  \node(v1){};
    & & \node(v2){};
  \\[0.25cm]
  \node(v11){};
    & & \node(v21){};
  \\   
  };   
\path (v) ++ (-0.6cm,0.5cm) node(label){\LARGE $\aonecharthati{2}$};

\draw[<-,very thick,>=latex,chocolate,shorten <= 2pt](v) -- ++ (90:0.55cm);   
% (v) 
\draw[thick,chocolate] (v) circle (0.12cm);
\path (v) ++ (0.25cm,-0.3cm) node{\small $\avert$};
\draw[->,thick,darkcyan,shorten >= 0.175cm,shorten <= 2pt] 
  (v) to % node[left,pos=0.36,xshift=0.075cm]{\small $\small\black{\aact}$} 
         node[right,pos=0.4,xshift=-0.075cm,yshift=1pt]{\small $\loopnsteplab{4}$}  (v11);
\draw[->,thick,darkcyan,shorten >= 0.175cm,shorten <= 2pt] 
  (v) to % node[right,pos=0.36,xshift=-0.05cm]{\small $\small\black{\bact}$} 
         node[left,pos=0.6,xshift=0.075cm,yshift=1pt]{\small $\loopnsteplab{3}$} (v21);

% (v1) 
\path (v1) ++ (-0.225cm,0.25cm) node{\small $\averti{1}$};
\draw[->,thick,darkcyan,shorten >= 0pt]
  (v1) to % node[left,pos=0.25,xshift=0.075cm]{\small $\small\black{\aact}$} 
          node[left,pos=0.6,xshift=0.075cm]{\small $\loopnsteplab{2}$} (v11);
\draw[->,thick,densely dotted,out=90,in=180,distance=0.5cm,shorten >=2pt](v1) to (v);
\draw[->,shorten <= 0pt,shorten >= 0pt] (v1) to %node[below]{$\bact$} 
                                                (v21); 

% (v11)
\path (v11) ++ (0cm,-0.25cm) node{\small $\averti{11}$};
\draw[->,thick,densely dotted,out=180,in=210,distance=0.7cm](v11) to (v1);

% (v2) 
\path (v2) ++ (0.25cm,0.25cm) node{\small $\averti{2}$};
\draw[->,thick,darkcyan,shorten >= 0pt]
  (v2) to % node[right,pos=0.25,xshift=-0.075cm]{\small $\small\black{\bact}$} 
          node[right,pos=0.6,xshift=-0.075cm]{\small $\loopnsteplab{1}$} (v21);
\draw[->,thick,densely dotted,out=90,in=0,distance=0.5cm,shorten >= 2pt](v2) to (v);

% (v21) 
\path (v21) ++ (0cm,-0.25cm) node{\small $\averti{21}$};
\draw[->,thick,densely dotted,out=0,in=-30,distance=0.7cm](v21) to (v2);

\matrix[anchor=center,row sep=1cm,column sep=0.55cm,,
        every node/.style={draw,very thick,circle,minimum width=2.5pt,fill,inner sep=0pt,outer sep=2pt}] at (5.8,0) {
    & \node[chocolate](v){};
  \\[-0.25ex]              
  \node(v1){};
    & & \node(v2){};
  \\[0.25cm]
  \node(v11){};
    & & \node(v21){};
  \\   
  };   
\path (v) ++ (-0.6cm,0.5cm) node(label){\LARGE $\aonecharthati{3}$};

\draw[<-,very thick,>=latex,chocolate,shorten <= 2pt](v) -- ++ (90:0.55cm);   
% (v) 
\draw[thick,chocolate] (v) circle (0.12cm);
\path (v) ++ (0.25cm,-0.3cm) node{\small $\avert$};
\draw[->,thick,darkcyan,shorten >= 0.175cm,shorten <= 2pt] 
  (v) to % node[left,pos=0.36,xshift=0.075cm]{\small $\small\black{\aact}$} 
         node[right,pos=0.4,xshift=-0.075cm,yshift=1pt]{\small $\loopnsteplab{2}$}  (v11);
\draw[->,thick,darkcyan,shorten >= 0.175cm,shorten <= 2pt] 
  (v) to % node[right,pos=0.36,xshift=-0.05cm]{\small $\small\black{\bact}$} 
         node[left,pos=0.6,xshift=0.075cm,yshift=1pt]{\small $\loopnsteplab{2}$} (v21);

% (v1) 
\path (v1) ++ (-0.225cm,0.25cm) node{\small $\averti{1}$};
\draw[->,thick,darkcyan,shorten >= 0pt]
  (v1) to % node[left,pos=0.25,xshift=0.075cm]{\small $\small\black{\aact}$} 
          node[left,pos=0.6,xshift=0.075cm]{\small $\loopnsteplab{1}$} (v11);
\draw[->,thick,densely dotted,out=90,in=180,distance=0.5cm,shorten >=2pt](v1) to (v);
\draw[->,shorten <= 0pt,shorten >= 0pt] (v1) to %node[below]{$\bact$} 
                                                (v21); 

% (v11)
\path (v11) ++ (0cm,-0.25cm) node{\small $\averti{11}$};
\draw[->,thick,densely dotted,out=180,in=210,distance=0.7cm](v11) to (v1);

% (v2) 
\path (v2) ++ (0.25cm,0.25cm) node{\small $\averti{2}$};
\draw[->,thick,darkcyan,shorten >= 0pt]
  (v2) to % node[right,pos=0.25,xshift=-0.075cm]{\small $\small\black{\bact}$} 
          node[right,pos=0.6,xshift=-0.075cm]{\small $\loopnsteplab{1}$} (v21);
\draw[->,thick,densely dotted,out=90,in=0,distance=0.5cm,shorten >= 2pt](v2) to (v);

% (v21) 
\path (v21) ++ (0cm,-0.25cm) node{\small $\averti{21}$};
\draw[->,thick,densely dotted,out=0,in=-30,distance=0.7cm](v21) to (v2);

\end{tikzpicture}
\end{center}  
  \vspace*{-2.5ex}
  \caption{\label{fig:ex:LLEEw-1-2-3}%
           Three \protect\LLEEwitnesses\ of the \protect\onechart~$\aonechart$ in Fig.~\protect\ref{fig:ex:LEE}.
           $\aonecharthati{1}$~is a recording of the successful procedure run in Fig.~\ref{fig:ex:LEE} 
           of the order in which \loopentry\ transitions have been removed.
           }
\end{figure}

%% file: figs/fig-proof-structure-exp.tex
%\afterpage{%
\begin{figure*}[t!]
\begin{tikzpicture} 
  \matrix[anchor=center,row sep=1cm,column sep=0.7cm,ampersand replacement=\&] {
    \node(onechartints){};
        \& \node(onechartofeonepos){};
          \& 
            \& 
              \& 
                \& \node(onechartofetwopos){};
    \\[0.35cm]
    \node(chartints){};
      \&
        \&[0cm] \node(chartofeonepos){};
          \& 
            \&[0cm] \node(chartofetwopos){};
    \\[-0.65cm]
    \node(crystonechart){};
      \& \node(onechartofeonecrystpos){};
    \\[0.25cm]
    \node(bisimcollapse){};
     \& 
       \&
         \& \node(bisimcollpos){};
    \\
    };  
  % nodes
  %  
  \path(chartofeonepos) ++ (-0.1cm,0cm) node(chartofeone){\large $\chartof{e_1}$};  
      % \path(chartofeone) ++ (-0.55cm,0cm) node[left]{(by \ref{SIone}) $e_1$ is solution};  
  \path(chartofetwopos) ++ (0.1cm,0cm) node(chartofetwo){\large $\chartof{e_2}$}; 
      % \path(chartofetwo) ++ (0.55cm,0cm) node[right]{$e_2$ is solution (by \ref{SIone})};  
  \path(chartints) ++ (-6.25cm,0cm) node[right]{\em chart interpretations};
  \path(onechartofeonepos) ++ (0cm,0cm) node(onechartofeone){\large $\onechartof{e_1}$};  ;
      \path(onechartofeonepos) ++ (-0.75cm,-0.445cm) node[left]{$ e_1 \milnersyseq e_{10} \underset{\text{\nf\ref{SEone}}}{\Longleftarrow} 
                                                                 \left\{\,
                                                                 \parbox{\widthof{(by \ref{Tone}) $e_{10}$ is solution}}
                                                                        {\hspace*{\fill}
                                                                         guarded, \forestgreen{\LLEE}
                                                                         \\
                                                                         \hspace*{\fill}
                                                                         $e_1$ is solution%
                                                                         \\
                                                                         (by \ref{Tone}) $e_{10}$ is solution}\right.%\}
                                                                                                      $}; 
  \path(onechartofetwopos) ++ (0cm,0cm) node(onechartofetwo){\large $\onechartof{e_2}$};
      \path(onechartofetwopos) ++ (0cm,-0.435cm) node[right]{$\left.%\{
                                                           \parbox{\widthof{(by \ref{Tone}) \forestgreen{\LLEE}, guarded}}
                                                                  {\phantom{(by \ref{Tone})} \forestgreen{\LLEE}, guarded
                                                                   \\
                                                                   \phantom{(by \ref{Tone})} $e_2$ is solution%
                                                                   \\
                                                                   (by \ref{Tone}) $e_{10}$ is solution}\,\right\} \underset{\text{\nf\ref{SEone}}}{\Longrightarrow} \; e_{10} \milnersyseq e_2 $}; 
      \path(onechartofetwopos) ++ (1.25cm,-2cm) node[right]{Finally:};
      \path(onechartofetwopos) ++ (1.25cm,-2.8cm) node[right]{$\left.%\{
                                                           \parbox{\widthof{$e_1 \milnersyseq e_{10}$}}
                                                                  {$e_1 \milnersyseq e_{10}$
                                                                   \\
                                                                   $e_{10} \milnersyseq e_2$} \,\right\} \; \Longrightarrow \;\; e_{1} \milnersyseq e_2 \punc{.}$};                                                              
  \path(onechartints) ++ (-6.25cm,0.45cm) node[right]{\em \onechart\ interpretations};    
  \path(onechartofeonecrystpos) ++ (0cm,-0.5cm) node(onechartofeonecryst)[yshift=0cm]{\large $\aonecharti{10}$};
      \path(onechartofeonecryst) ++ (-0.55cm,-0.075cm) node[left]{\parbox{\widthof{(by \ref{NCone}, \ref{CNone}) $e_{10}$ is complete solution}}
                                                                       {\hspace*{\fill}%
                                                                        guarded, \forestgreen{\LLEE}%
                                                                        \\[-0.2ex]
                                                                        (by \ref{NCone}, \ref{CNone}) $e_{10}$ is complete solution }}; 
  \path(crystonechart) ++ (-6.25cm,-0.2cm) node[right]{\em crystallized \onechart}; 
  \path(bisimcollpos) ++ (0cm,0cm) node(bisimcoll){\large $\acharti{0}$};
      \path(bisimcollpos) ++ (0.2cm,-0.25cm) node[right]{$e_{10}$ is solution (by \ref{CCone})};  
      % \path(bisimcollpos) ++ (-0.2cm,-0.25cm) node[left]{(by \ref{CCone}) $e_{10}$ is solution};  
  \path(bisimcollapse) ++ (-6.25cm,-0.1cm) node[right]{\em bisimulation collapse};    
  %  
  % arrows  
  %  
  \draw[bisim](chartofeone) to node[above]{\small (assm)} (chartofetwo);  
  \draw[funonebisim] (onechartofeone) to node[right,pos=0.3,xshift=0.025cm]{\small\ref{IVone}} (chartofeone);
  \draw[funonebisim] (onechartofetwo) to node[left,pos=0.3,xshift=-0.025cm]{\small\ref{IVone}} (chartofetwo);
  \draw[onebisim] (onechartofeone) to  (onechartofetwo);
  \draw[onebisim] (onechartofeonecryst) to node[left,xshift=-0.025cm,pos=0.4]{\small\ref{CRone}} (onechartofeone);
  %\draw[funonebisim,shorten <=4pt,shorten >=6pt] (onechartofeone) to (bisimcoll);
    \draw[-,magenta,out=275,in=140,shorten <=18pt,shorten >=20pt,distance=1.8cm] ($(onechartofeone) + (-0.5ex,0ex)$) to ($(bisimcoll) + (-0.5ex,0ex)$);
    \draw[->,magenta,out=275,in=140,shorten <=18pt,shorten >=24pt,distance=1.8cm] ($(onechartofeone) + (0.5ex,0ex)$) to ($(bisimcoll) + (1ex,0ex)$);

  \draw[funonebisim,shorten <=-3pt,shorten >=-2pt] (onechartofeonecryst) to  (bisimcoll);
  %
  %
  % \draw[funonebisimleft,shorten >=-2pt] (onechartofetwo) to (bisimcoll);
    \draw[->,magenta,out=275,in=45,shorten <=8pt,shorten >=8pt,distance=1.75cm] ($(onechartofetwo) + (-0.5ex,0ex)$) to ($(bisimcoll) + (-0.5ex,0ex)$);
    \draw[-,magenta,out=275,in=45,shorten <=8pt,shorten >=4pt,distance=1.75cm] ($(onechartofetwo) + (0.5ex,0ex)$) to ($(bisimcoll) + (1ex,0ex)$);

  \draw[funbisimright,shorten >=4pt] (chartofeone) to (bisimcoll);
  \draw[funbisimleft, shorten >=4pt] (chartofetwo) to (bisimcoll);

\end{tikzpicture}
  \vspace*{-1ex}
  \caption{\label{fig:proof:structure}%
           Structure of the completeness proof (see proof of Thm.~\ref{thm:milnersys:complete}):
           The argument starts from the assumption $\protect\chartof{\astexpi{1}} \bisim \protect\chartof{\astexpi{2}}$, that the 
             chart interpretations of $\protect\astexpi{1}$ and $\protect\astexpi{2}$ are bisimilar.
           It uses \protect\onebisimilarity\ with the \protect\onechart\ interpretations~$\protect\onechartof{\astexpi{1}}$ and $\protect\onechartof{\astexpi{2}}$ 
             of $\astexpi{1}$ and $\astexpi{2}$
             (which expand the chart interpretations), 
             the crystallization $\protect\aonecharti{10}$ of $\protect\onechartof{\astexpi{1}}$
             (which arises by \protect\LLEEpreservingly\ minimizing $\protect\onechartof{\astexpi{1}}$ and crystallization operations),
             and the joint bisimulation collapse $\acharti{0}$ of all of these ($\sone$-)charts.
           The conclusion is $\astexpi{1} \milnersyseq \astexpi{2}$, that $\astexpi{1}$ and $\astexpi{2}$ are provably equal in Milner's system $\milnersys$. 
           By ``$\bstexp$ is (complete) solution of $\aonechart$'' we here mean that ``$\bstexp$ is the principal value of a (complete) provable solution of $\aonechart$''.    
           The indicated lemmas are explained in Sect.~\ref{completeness:proof}. 
    }
\end{figure*}%}

%% file: figs/fig-countex-collapsible.tex
\begin{figure}[!t]
  \begin{flushleft}
    \hspace*{-4.8ex}%
    \scalebox{0.76}{\input{figs/ex-limit-collapse.tex}}
    \\[-12.5ex]
    \hspace*{15ex}
    \scalebox{0.35}{\input{figs/ex-limit-collapse-1-small.tex}}
  \end{flushleft}
  \vspace*{-2.25ex}
  \caption{\label{fig:countex:collapse}%
           A \protect\onechart~$\aonechart$ with \protect\LLEEwitness~$\aonecharthat$ 
                                                 (colored \protect\loopentry\ transitions of level~1, green, of level~2, blue)
           that is not \protect\onebisimulation\ collapsed: %the magenta dashed link relates \protect\onebisimilar\ vertices.\vspace*{-0.75mm}
               the correspondences $\protect\magenta{\protect\dashedmapsto}$ via the `counterpart function $\magenta{\scpfun}$'  
                 indicate a (grounded) functional \protect\onebisimulation\ slice on $\aonechart$.
           $\aonechart$ is not \protect\LLEEpreservingly\ \protect\onecollapsible.      
           The result $\aonecharti{1}$ %$ = \protect\connthroughin{\aonechart}{\bverti{1}}{\bverti{2}}$ 
                                       (small)
               of connecting through $\bverti{1}$ to $\bverti{2}$ in $\aonechart$
             is a \protect\onebisimulation\ collapse of~$\aonechart$,~but~not~a~\protect\mbox{\protect\LLEEonechart}.
           % \Twincrystal\ shape of $\aonechart$ is explained by the colored regions,  
           The colored regions explain $\aonechart$ as a \twincrystal, see Sect.~\ref{twincrystals}.
               % that does not satisfy \protect\LEE .
               %, and hence cannot\vspace*{-0.75mm} be a \protect\LLEEonechart\ 
               %(the indicated \protect\entrybodylabeling~$\aonecharthati{1}$ %$ = \protect\connthroughin{\aonecharthat}{\bverti{1}}{\bverti{2}}$ 
               % defines \loopentrytransitions\ but is not a \protect\LLEEwitness).} 
             }
\end{figure}

%% file: figs/ex-limit-collapse.tex
\begin{tikzpicture}
  \matrix[anchor=center,row sep=0.3cm,column sep=0.325cm,every node/.style={draw,very thick,circle,minimum width=2.5pt,fill,inner sep=0pt,outer sep=2pt}] at (0,0) {
                 &                  &              & & & & &              & \node(abc){};
    \\
    \\
    \\
                 & \node(acd){};   
    \\
    \\
                 &                  &              & & & & &              &               & \node(a2){};
    \\
                 &                  &              & & & & & \node(a1){};
    \\
    \node(c1){}; &                  & \node(c2){}; & & & & &              &               &               & & & & & & & \node(abcd-2){};
    \\
    \\
    \\
    \\
    \\
                 & \node(abcd-1){}; & 
    \\
                 &                  &              & & & & &              &               &               & & & & & & & \node(e){};
    \\
    \\
    \\
    \\
    \\
                 & \node(f){}; 
    \\
  };
  \pgfdeclarelayer{background}
  \pgfdeclarelayer{foreground}
  \pgfsetlayers{background,main,foreground}
  
  \begin{pgfonlayer}{background}
    \draw[draw opacity=0,fill opacity=0.4,fill=royalblue!45] 
      (abc.center)--(a1.center)--(abcd-2.center)--(e.center) to[->,relative=false,out=12,in=0,distance=5.2cm] (abc.center);
    \draw[draw opacity=0,fill opacity=0.4,fill=forestgreen!30] 
      (acd.center)--(c2.center)--(abcd-1.center)--(f.center) to[->,relative=false,out=181,in=179,distance=3.7cm] (acd.center);
  \end{pgfonlayer}

  \begin{pgfonlayer}{foreground}
  \draw[<-,very thick,color=chocolate,>=latex](abc) -- ++ (90:0.5cm););
  \path (abc) ++ (0.1cm,0.35cm) node[right]{$\tightfbox{$abc$} = \avert$};
  \draw[->,very thick,royalblue] (abc) to node[right]{$\hspace*{-0.2em}\black{a}$} (a1);
  \draw[->,very thick,royalblue] (abc) to node[right]{$\hspace*{-0.2em}\black{a}$} (a2);
  \draw[->,very thick,royalblue,shorten >=6pt] (abc) to node[left,pos=0.3]{$\hspace*{0.2em}\black{c}$} (c1); 
  \draw[->,very thick,royalblue,shorten >=3pt] (abc) to node[right,pos=0.4,xshift=-0.15em]{$\black{c}$} (c2);
  \draw[->,very thick,royalblue,shorten >=4pt] (abc) to node[right,pos=0.85]{$\black{b}$} (f);
  \path (abc) ++ (-1.65cm,0.4cm) node{\LARGE$\aonechart\:/\:\aonecharthat$};
  %
%   % Mark the angle XAY
%   \begin{scope}
%     \path[clip] (A) -- (X) -- (Y);
%     \fill[red, opacity=0.5, draw=black] (A) circle (5mm);
%     \node at ($(A)+(30:7mm)$) {$\theta$};
%   \end{scope}
  %\pic[draw=magenta, text=blue, "$\loopnsteplab{1}$" ,thick]{angle=a1--abc--c2};
  %\pic [draw=red, text=blue, <->, "$\theta$", angle eccentricity=1.5] {angle = mary--origo--bob};
  %
  \path (acd) ++ (0cm,0.4cm) node{$\tightfbox{$acd$} = \bvertbari{1}$};
  \draw[->] (acd) to node[above,pos=0.225]{$a$} (a2);
  \draw[->] (acd) to node[below,%near start,yshift=0.6mm,xshift=0.3mm
                                pos=0.68]{$a$} (a1);
%   \begin{scope}
%     \path[draw,clip] (acd) -- (c1) -- (c2) -- cycle;
%     \fill[red, opacity=0.5, draw=black] (acd) circle (5mm);
%     %\node at ($(A)+(30:7mm)$) {$\theta$};
%   \end{scope}

  \draw[->,forestgreen,very thick] (acd) to node[left]{$\black{c}\hspace*{-0.23em}$} (c1);
  \draw[->,forestgreen,very thick] (acd) to node[left]{$\black{c}\hspace*{-0.15em}$} (c2);
  \draw[->,bend right,relative=false,out=-45,in=167.5,looseness=1] (acd) to node[above,pos=0.8]{$d$} (e);
  \path (a1) ++ (0.1cm,-0.375cm) node{\tightfbox{$a_1$}};
  \draw[->,shorten >=3pt] (a1) to node[below,pos=0.3]{$a_1$} (abcd-2);
  %\draw[->,bend left,distance=0.75cm,out=110] (a1) to (abc);
  \draw[->,bend left,distance=2.9cm,relative=false,out=0,in=-20,shorten >=4pt] (a1) to node[right]{$a_1\hspace*{-0.18em}$} (abc);
  \path (a2) ++ (0.32cm,0.4cm) node{\tightfbox{$a_2$}};
  \draw[->] (a2) to node[above,near end]{$a_2$} (abcd-2);
  \draw[->,bend left,distance=1.65cm,relative=false,out=-5,in=-20,shorten >=4pt] (a2) to node[left,pos=0.65]{$a_2$} (abc);
  %\draw[->,bend left,distance=2cm,out=80,in=120] (a2) to (abc);
  %
  \path (c1) ++ (-0.35cm,0.25cm) node{\tightfbox{$c_1$}};
  \draw[->,bend left,relative=false,out=200,in=180,looseness=1.9] (c1) to node[pos=0.55,right]{$c_1\hspace*{1.5em}$} (acd);
  \draw[->] (c1) to node[right]{$\hspace*{-0.18em}c_1$} (abcd-1);
  %\draw[->,bend right,relative=false,out=-30,in=0,looseness=2] (c1) ..controls ($(c2) + (1cm,-1cm)$) .. (acd);
  %
  \path (c2) ++ (0.35cm,-0.25cm) node{\tightfbox{$c_2$}};
  \draw[->] (c2) to node[right]{$\hspace*{-0.21em}c_2$} (abcd-1);
  \draw[->,out=220,in=180,looseness=3.75] (c2) to node[pos=0.275,below]{$c_2$\hspace*{-0.4em}} (acd);
  %\draw[->] (c2) .. controls ($(c2) + (,)$) and ($(acd) + ()$) .. (acd);
  %
  \path (abcd-2) ++ (0.05cm,0cm) node[right]{$\tightfbox{$(abcd)_2$} \parbox[t]{\widthof{${} = \bvertbari{2}$}}
                                                                             {${} = \bvertbari{2}$
                                                                              \\
                                                                              ${} = \bverti{2}$}$};
  \draw[->] (abcd-2) to node[right]{$d$} (e);
  \draw[->,densely dotted,thick,bend right,distance=1.25cm,relative=false,out=90,in=0] (abcd-2) to node[right,pos=0.45,xshift=0.05cm]{$\sone$} (abc);
  \path (abcd-1) ++ (0cm,-0.3cm) node[left]{$\bverti{1} = \tightfbox{$(abcd)_1$}$};
  \draw[->] (abcd-1) to node[left]{$b\hspace*{-0.1em}$} (f);
  \draw[->,densely dotted,thick,bend left,relative=false,out=180,in=180,looseness=2] (abcd-1) to node[left,pos=0.45,xshift=0.05cm]{$\sone$} (acd);
  %\draw[->,densely dotted] (abcd-2) .. controls ($(abcd-2) + (1.75cm,0.5cm)$) and ($(acd) + (1.75cm,-0.5cm)$) .. (acd);
  %
  \path (e) ++ (0cm,-0.4cm) node{\tightfbox{$e$}};
  % \draw[->] (e) .. controls ($(abcd-2) + (-2cm,0.5cm)$) and ($(abc) + (-1.75cm,-0.25cm)$) .. (abc);
  \draw[->,relative=false,out=12.5,in=0,distance=5cm] (e) to node[right]{$e$} (abc);
  \path (f) ++ (0cm,-0.45cm) node{\tightfbox{$f$}};
  %\path (f) ++ (0.35cm,0.4cm) node{\tightfbox{$f$}};
  \draw[->,relative=false,out=180,in=180,distance=3.5cm] (f) to node[right]{$f$} (acd);

  % bisimulation
  % \draw[densely dashed,very thick,magenta,bend right,distance=1.5cm] (abcd-1) to (abcd-2);
  %
  \draw[|->,very thick,out=-10,in=225,shorten <=1pt,shorten >=3pt,densely dashed,magenta,magenta,distance=1.5cm] 
    (abcd-1) to node[below,yshift=0.025cm]{${\scpfun}$} (abcd-2);
  \draw[|->,very thick,out=180,in=35,shorten <=8pt,shorten >=4.5pt,densely dashed,magenta,magenta,distance=1.5cm] 
    (abcd-2) to node[below,yshift=0.05cm,pos=0.375]{${\scpfun}$} (abcd-1);
  \end{pgfonlayer}
\end{tikzpicture}

%% file: figs/ex-limit-collapse-1-small.tex
\begin{tikzpicture}
  \matrix[anchor=center,row sep=0.3cm,column sep=0.325cm,every node/.style={draw,very thick,circle,minimum width=2.5pt,fill,inner sep=0pt,outer sep=2pt}] at (0,0) {
                 &                  &              & & & & &              & \node(abc){};
    \\
    \\
    \\
                 & \node(acd){};   
    \\
    \\
                 &                  &              & & & & &              &               & \node(a2){};
    \\
                 &                  &              & & & & & \node(a1){};
    \\
    \node(c1){}; &                  & \node(c2){}; & & & & &              &               &               & & & & & & & \node(abcd-2){};
    \\
    \\
    \\
    \\
    \\
                 & \node[draw=none,fill=none](abcd-1){}; & 
    \\
                 &                  &              & & & & &              &               &               & & & & & & & \node(e){};
    \\
    \\
    \\
    \\
    \\
                 & \node(f){}; 
    \\
  };
  \pgfdeclarelayer{background}
  \pgfdeclarelayer{foreground}
  \pgfsetlayers{background,main,foreground}
  
  \begin{pgfonlayer}{background}
    \draw[draw opacity=0,fill opacity=0.4,fill=royalblue!30] 
      (abc.center)--(a1.center)--(abcd-2.center)--(e.center) to[->,relative=false,out=12,in=0,distance=5.2cm] (abc.center);
  \end{pgfonlayer}
  
  \begin{pgfonlayer}{foreground}
    \draw[<-,very thick,color=chocolate,>=latex](abc) -- ++ (90:0.5cm););
    \path (abc) ++ (0.1cm,0.35cm) node[right]{$\tightfbox{$abc$} = \avert$};
    \draw[->,very thick,royalblue] (abc) to node[right]{$\hspace*{-0.2em}\black{a}$} (a1);
    \draw[->,very thick,royalblue] (abc) to node[right]{$\hspace*{-0.2em}\black{a}$} (a2);
    \draw[->] (abc) to node[left,pos=0.3]{$\hspace*{0.2em}\black{c}$} (c1); 
    \draw[->] (abc) to node[right,pos=0.4,xshift=-0.15em]{$\black{c}$} (c2);
    \draw[->] (abc) to node[right,pos=0.85]{$\black{b}$} (f);
    % 
    %\path (abc) ++ (-1.5cm,0.3cm) node{\LARGE $\aonecharti{1} \, / \, \aonecharthat_{1} $};
    %\path (abc) ++ (-1.75cm,0.5cm) node{\scalebox{2}{$\aonecharti{1} \, / \, \aonecharthati{1} $}};
    \path (abc) ++ (0cm,-5.5cm) node{\scalebox{2.5}{$\aonecharti{1} \, / \, \aonecharthati{1} $}};
    %\path (abc) ++ (-1.75cm,0.5cm) node{\LARGE $\connthroughin{\aonechart}{\bverti{1}}{\bverti{2}} \, / \, \connthroughin{\aonecharthat}{\bverti{1}}{\bverti{2}} $};
    %
  %   % Mark the angle XAY
  %   \begin{scope}
  %     \path[clip] (A) -- (X) -- (Y);
  %     \fill[red, opacity=0.5, draw=black] (A) circle (5mm);
  %     \node at ($(A)+(30:7mm)$) {$\theta$};
  %   \end{scope}
    %\pic[draw=firebrick, text=blue, "$\loopnsteplab{1}$" ,thick]{angle=a1--abc--c2};
    %\pic [draw=red, text=blue, <->, "$\theta$", angle eccentricity=1.5] {angle = mary--origo--bob};
    %
    \path (acd) ++ (0cm,0.4cm) node{$\tightfbox{$acd$}$};
    \draw[->] (acd) to node[above,pos=0.225]{$a$} (a2);
    \draw[->] (acd) to node[below,%near start,yshift=0.6mm,xshift=0.3mm
                                  pos=0.68]{$a$} (a1);
  %   \begin{scope}
  %     \path[draw,clip] (acd) -- (c1) -- (c2) -- cycle;
  %     \fill[red, opacity=0.5, draw=black] (acd) circle (5mm);
  %     %\node at ($(A)+(30:7mm)$) {$\theta$};
  %   \end{scope}

    \draw[->] (acd) to node[left]{$\black{c}\hspace*{-0.23em}$} (c1);
    \draw[->] (acd) to node[left]{$\black{c}\hspace*{-0.15em}$} (c2);
    \draw[->,bend right,relative=false,out=-45,in=167.5,looseness=1] (acd) to node[above,pos=0.8]{$d$} (e);
    \path (a1) ++ (0.1cm,-0.375cm) node{\tightfbox{$a_1$}};
    \draw[->,shorten >=3pt] (a1) to node[below]{$a_1$} (abcd-2);
    %\draw[->,bend left,distance=0.75cm,out=110] (a1) to (abc);
    \draw[->,bend left,distance=2.9cm,relative=false,out=0,in=-20,shorten >=4pt] (a1) to node[right]{$a_1\hspace*{-0.18em}$} (abc);
    \path (a2) ++ (0.32cm,0.4cm) node{\tightfbox{$a_2$}};
    \draw[->] (a2) to node[above,near end]{$a_2$} (abcd-2);
    \draw[->,bend left,distance=1.65cm,relative=false,out=-5,in=-20,shorten >=4pt] (a2) to node[left,pos=0.65]{$a_2$} (abc);
    %\draw[->,bend left,distance=2cm,out=80,in=120] (a2) to (abc);
    %
    \path (c1) ++ (-0.35cm,0.25cm) node{\tightfbox{$c_1$}};
    \draw[->,bend left,relative=false,out=200,in=180,looseness=1.9] (c1) to node[pos=0.55,right]{$c_1\hspace*{1.5em}$} (acd);
    \draw[->,out=280,in=210,distance=2.5cm] (c1) to node[below]{$c_1$} (abcd-2);
    %\draw[->,bend right,relative=false,out=-30,in=0,looseness=2] (c1) ..controls ($(c2) + (1cm,-1cm)$) .. (acd);
    %
    \path (c2) ++ (0.35cm,-0.25cm) node{\tightfbox{$c_2$}};
    \draw[->,out=260,in=200,distance=1.5cm,shorten >=3pt] (c2) to node[below,pos=0.425]{$c_2$} (abcd-2);
    \draw[->,out=220,in=180,looseness=3.75] (c2) to node[pos=0.275,below]{$c_2$\hspace*{-0.4em}} (acd);
    %\draw[->] (c2) .. controls ($(c2) + (,)$) and ($(acd) + ()$) .. (acd);
    %
    \path (abcd-2) ++ (0.05cm,0cm) node[right]{$\tightfbox{$(abcd)_2$} = \bverti{2} $};
    \draw[->] (abcd-2) to node[right]{$d$} (e);
    \draw[->,densely dotted,thick,bend right,distance=1.25cm,relative=false,out=90,in=0] (abcd-2) to node[right,pos=0.45,xshift=0.05cm]{$\sone$} (abc);
    \path (e) ++ (0cm,-0.4cm) node{\tightfbox{$e$}};
    % \draw[->] (e) .. controls ($(abcd-2) + (-2cm,0.5cm)$) and ($(abc) + (-1.75cm,-0.25cm)$) .. (abc);
    \draw[->,relative=false,out=12.5,in=0,distance=5cm] (e) to node[right]{$e$} (abc);
    \path (f) ++ (0cm,-0.45cm) node{\tightfbox{$f$}};
    %\path (f) ++ (0.35cm,0.4cm) node{\tightfbox{$f$}};
    \draw[->,relative=false,out=180,in=180,distance=3.5cm] (f) to node[right]{$f$} (acd);
  \end{pgfonlayer}
\end{tikzpicture}

%% file: figs/fig-ex-local-transfer-function.tex
\begin{figure}[t!]
\begin{flushleft}  
  \hspace*{-3ex}
  \scalebox{0.7}{%
\begin{tikzpicture}
  \matrix[anchor=center,row sep=0.5cm,column sep=1.65cm,ampersand replacement=\&,
          every node/.style={draw,very thick,circle,minimum width=2.5pt,fill,inner sep=0pt,outer sep=2pt}] at (0,0) {
  \\
                    \&  \node(abc--1){};
  \\
  \\
  \node(acd--1){};   
  \\
                    \& \node(a--1){};
  \\
  \node(c--1){};       \&                  \&  \node(abcd-2--1){};
  \\
  \\
  \node(abcd-1--1){};  \&                  \&  \node(e--1){};
  \\
  \\
  \node(f--1){};
  \\
  };
  \pgfdeclarelayer{background}
  \pgfdeclarelayer{foreground}
  \pgfsetlayers{background,main,foreground}
  
  \begin{pgfonlayer}{background}
    \draw[draw opacity=0,fill opacity=0.4,fill=mediumblue!40] (abc--1.center)--(a--1.center)--(abcd-2--1.center)--(e--1.center) to[->,out=-3,in=2,distance=2.95cm] (abc--1.center);
    \draw[draw opacity=0,fill opacity=0.4,fill=forestgreen!30] (acd--1.center)--(f--1.center) to[out=181,in=179,distance=2.67cm] (acd--1.center);
  \end{pgfonlayer}{background}
  \draw[<-,very thick,color=chocolate,>=latex](abc--1) -- ++ (180:0.5cm););
  \path (abc--1) ++ (-1.55cm,0.1cm) node{\LARGE $\aonecharti{s} \, / \, \aonecharthati{s}$};
  \path (abc--1) ++ (0cm,0.375cm) node{$\tightfbox{$abc$}$};
  % \draw[->] (abc--1) to node[right,xshift=-0.05cm]{\black{$a$}} (a--1);
  % \draw[->,shorten >=3pt]  (abc--1) to node[left,xshift=0.05cm,pos=0.3]{\black{$c$}} (c--1);
  % \draw[->,shorten >= 5pt] (abc--1) to node[right,pos=0.65,xshift=-0.05cm]{\black{$b$}} (f--1);
  \draw[->,thick,royalblue] (abc--1) to node[right,xshift=-0.05cm]{\black{$a$}} (a--1);
  \draw[->,thick,royalblue,shorten >=3pt]  (abc--1) to node[left,xshift=0.05cm,pos=0.3]{\black{$c$}} (c--1);
  \draw[->,thick,royalblue,shorten >= 5pt] (abc--1) to node[right,pos=0.65,xshift=-0.05cm]{\black{$b$}} (f--1);
  \path (acd--1) ++ (0cm,0.375cm) node{$\tightfbox{$acd$}$};
  \draw[->] (acd--1) to node[above,pos=0.25]{$a$} (a--1);
  % \draw[->] (acd--1) to node[left,xshift=0.085cm,pos=0.5]{\black{$c$}} (c--1);
  \draw[->,thick,forestgreen] (acd--1) to node[left,xshift=0.085cm,pos=0.5]{\black{$c$}} (c--1);
  \draw[->,bend right,distance=0.75cm] (acd--1) to node[below,pos=0.475]{$d$} (e--1);
  \path (c--1) ++ (-0.25cm,-0.3cm) node{$\tightfbox{$c$}$};
  \draw[->] (c--1) to node[right,xshift=-0.05cm]{$c$} (abcd-1--1);
  \draw[->,out=180,in=225] (c--1) to node[left,xshift=0.05cm]{$c$} (acd--1);
  \path (abcd-1--1) ++ (-0.55cm,-0.35cm) node{$\tightfbox{$abcd_1$}$};
  \draw[->,densely dotted,thick,out=180,in=209,shorten >=4.5pt,distance=1.5cm] (abcd-1--1) to node[left,xshift=0.05cm]{$\sone$} (acd--1);
  \draw[->] (abcd-1--1) to node[right,xshift=-0.075cm,pos=0.25]{$b$} (f--1);
  \path (f--1) ++ (0cm,-0.375cm) node{$\tightfbox{$f$}$};
  \draw[->,out=180,in=180,distance=2.5cm] (f--1) to node[left,xshift=0.05cm]{$f$} (acd--1);
  \path (a--1) ++ (-0.25cm,-0.25cm) node{$\tightfbox{$a$}$};
  \draw[->,out=0,in=-45,distance=0.75cm] (a--1) to node[right,xshift=-0.05cm]{$a$} (abc--1);
  \draw[->] (a--1) to node[above,pos=0.6]{$a$} (abcd-2--1);
  \path (abcd-2--1) ++ (0.55cm,-0.35cm) node{$\tightfbox{$abcd_2$}$};
  \draw[->,densely dotted,thick,out=0,in=-25,distance=1.5cm,shorten >=4.5pt] (abcd-2--1) to node[right,xshift=0.05cm]{$\sone$} (abc--1);
  \draw[->] (abcd-2--1) to node[left,xshift=0.05cm]{$d$} (e--1);
  \path (e--1) ++ (0cm,-0.375cm) node{$\tightfbox{$e$}$};
  \draw[->,out=0,in=0,distance=2.75cm] (e--1) to node[right]{$e$} (abc--1);

  %\draw[-,very thick,densely dashed,magenta,bend right,distance=0.75cm] (abcd-1--1) to (abcd-2--1);

  %
  \matrix[anchor=center,row sep=0.5cm,column sep=1.65cm,ampersand replacement=\&,
          every node/.style={draw,very thick,circle,minimum width=2.5pt,fill,inner sep=0pt,outer sep=2pt}] at (4.25,-3.85) {
  \\
                    \&  \node(abc--2){};
  \\
  \\
  \node(acd--2){};   
  \\
                    \& \node(a--2){};
  \\
  \node(c--2){};       \&                  \&  \node(abcd-2--2){};
  \\
  \\
  \node(abcd-1--2){};  \&                  \&  \node(e--2){};
  \\
  \\
  \node(f--2){};
  \\
  };
  \pgfdeclarelayer{background}
  \pgfdeclarelayer{foreground}
  \pgfsetlayers{background,main,foreground}
  
  \begin{pgfonlayer}{background}
    \draw[draw opacity=0,fill opacity=0.4,fill=mediumblue!40] (abc--2.center)--(a--2.center)--(abcd-2--2.center)--(e--2.center) to[->,out=-3,in=2,distance=2.95cm] (abc--2.center);
    \draw[draw opacity=0,fill opacity=0.4,fill=forestgreen!30] (acd--2.center)--(f--2.center) to[out=181,in=179,distance=2.67cm] (acd--2.center);
  \end{pgfonlayer}{background}
  \draw[<-,very thick,color=chocolate,>=latex](abc--2) -- ++ (180:0.5cm););
  \path (abc--2) ++ (2.2cm,0.3cm) node{\LARGE $\aonecharti{s} \, / \, \aonecharthati{s}$};
  \path (abc--2) ++ (0cm,0.375cm) node{$\tightfbox{$abc$}$};
    % \draw[->] (abc--2) to node[right,xshift=-0.05cm]{\black{$a$}} (a--2);
    % \draw[->,shorten >=3pt]  (abc--2) to node[left,xshift=0.05cm,pos=0.3]{\black{$c$}} (c--2);
    % \draw[->,shorten >= 5pt] (abc--2) to node[right,pos=0.65,xshift=-0.05cm]{\black{$b$}} (f--2);
  \draw[->,thick,royalblue] (abc--2) to node[right,xshift=-0.05cm]{\black{$a$}} (a--2);
  \draw[->,thick,royalblue,shorten >=3pt]  (abc--2) to node[left,xshift=0.05cm,pos=0.3]{\black{$c$}} (c--2);
  \draw[->,thick,royalblue,shorten >= 5pt] (abc--2) to node[right,pos=0.65,xshift=-0.05cm]{\black{$b$}} (f--2);
  \path (acd--2) ++ (0cm,0.375cm) node{$\tightfbox{$acd$}$};
  \draw[->] (acd--2) to node[above,pos=0.25]{$a$} (a--2);
    % \draw[->] (acd--2) to node[left,xshift=0.085cm,pos=0.5]{\black{$c$}} (c--2);
  \draw[->,thick,forestgreen] (acd--2) to node[left,xshift=0.085cm,pos=0.5]{\black{$c$}} (c--2);
  \draw[->,bend right,distance=0.75cm] (acd--2) to node[below,pos=0.55]{$d$} (e--2);
  \path (c--2) ++ (-0.25cm,-0.3cm) node{$\tightfbox{$c$}$};
  \draw[->] (c--2) to node[right,xshift=-0.05cm]{$c$} (abcd-1--2);
  \draw[->,out=180,in=225] (c--2) to node[left,xshift=0.05cm]{$c$} (acd--2);
  \path (abcd-1--2) ++ (-0.55cm,-0.35cm) node{$\tightfbox{$abcd_1$}$};
  \draw[->,densely dotted,thick,out=180,in=209,shorten >=4.5pt,distance=1.5cm] (abcd-1--2) to node[left,xshift=0.05cm]{$\sone$} (acd--2);
  \draw[->] (abcd-1--2) to node[right,xshift=-0.075cm,pos=0.25]{$b$} (f--2);
  \path (f--2) ++ (0cm,-0.375cm) node{$\tightfbox{$f$}$};
  \draw[->,out=180,in=180,distance=2.5cm] (f--2) to node[left,xshift=0.05cm]{$f$} (acd--2);
  \path (a--2) ++ (-0.25cm,-0.25cm) node{$\tightfbox{$a$}$};
  \draw[->,out=0,in=-45,distance=0.75cm] (a--2) to node[right,xshift=-0.05cm]{$a$} (abc--2);
  \draw[->] (a--2) to node[above,pos=0.6]{$a$} (abcd-2--2);
  \path (abcd-2--2) ++ (0.55cm,-0.35cm) node{$\tightfbox{$abcd_2$}$};
  \draw[->,densely dotted,thick,out=0,in=-25,distance=1.5cm,shorten >=4.5pt] (abcd-2--2) to node[right,xshift=0.05cm]{$\sone$} (abc--2);
  \draw[->] (abcd-2--2) to node[left,xshift=0.05cm]{$d$} (e--2);
  \path (e--2) ++ (0cm,-0.375cm) node{$\tightfbox{$e$}$};
  \draw[->,out=0,in=0,distance=2.75cm] (e--2) to node[right]{$e$} (abc--2);

  \draw[|->,very thick,densely dashed,magenta,distance=0.75cm,out=-32,in=205,distance=2cm] (abcd-1--1) to node[below,pos=0.25]{$\scpfunon{\black{\csetverts}}$} (abcd-2--2);
  \draw[|->,very thick,densely dashed,magenta,distance=0.75cm,out=235,in=140,distance=1.25cm] (abcd-2--1) to node[left,pos=0.325,xshift=0.05cm]{$\scpfunon{\black{\csetverts}}$} (abcd-1--2);
  
  \draw[-,thick,densely dashed,chocolate,out=200,in=190,distance=2.5cm] (f--1) to (f--2);
  \draw[-,thick,densely dashed,chocolate,out=40,in=105,distance=2.35cm] (a--1) to (a--2);
  \draw[-,thick,densely dashed,chocolate,out=-50,in=155,distance=1.5cm] (c--1) to (c--2);
  \draw[-,thick,densely dashed,chocolate,out=-30,in=202.5,distance=3.5cm] (e--1) to (e--2);
  \draw[-,thick,densely dashed,chocolate,out=20,in=10,distance=2.75cm,shorten <=3.5pt,shorten >=3.5pt] (abc--1) to (abc--2);
  \draw[-,thick,densely dashed,chocolate,out=-77.5,in=170,distance=1cm,shorten <=0pt,shorten >=3.5pt] (acd--1) to (acd--2);
  \draw[-,thick,densely dashed,chocolate,out=-45,in=165,distance=1.75cm,shorten <=6pt,shorten >=6pt] (abcd-1--1) to (abcd-1--2);
  \draw[-,thick,densely dashed,chocolate,out=35,in=90,distance=2cm] (abcd-2--1) to (abcd-2--2);
    
\end{tikzpicture}}\vspace*{-1ex}
\end{flushleft}
  \vspace*{-1.5ex}
  \caption{\label{fig:ex:local:transfer:function}%
    Two copies of a
      simplified version $\aonecharti{s}$ (with \LLEEwitness~$\aonecharthati{s}$) of the \onechart~$\aonechart$ in Fig.~\ref{fig:countex:collapse},
    linked by  a local transfer function $\magenta{\scpfunon{\black{\twinc}}}$ (\magenta{$\mapsto$ links}) 
      whose graph (a grounded \onebisimulation\ slice) 
      is extended to a \onebisimulation\ (\chocolate{added links}).
    Like $\aonechart$ in Fig.~\ref{fig:countex:collapse}, also $\aonecharti{s}$ is a \twincrystal\ shaped \LLEEonechart.  
    % Simplified version $\aonecharti{s}$ (with \LLEEwitness~$\aonecharthati{s}$) of the \LLEEonechart~$\aonechart$ in Fig.~\ref{fig:countex:collapse},
    %   and a local transfer function $\magenta{\scpfunon{\black{\twinc}}}$ on $\aonecharti{s}$ (\magenta{$\mapsto$} links) 
    %   whose graph (a grounded \onebisimulation\ slice) is extended to a \onebisimulation\ (added \chocolate{brown} links).
    }
\end{figure}

%% file: figs/fig-twincrystal-schema.tex
\begin{figure}[!t]
\begin{center}
\begin{tikzpicture}
  
  \node (pivot) at (0,0) {};
    \path (pivot) ++ (0cm,0.4cm) node{$\pivottwinc$};
    \path (pivot) ++ (-1.35cm,-2.75cm) node (P1-left) {};
    \path (pivot) ++ (1.35cm,-2.75cm) node (P1-right) {};
    
  % label pivot part
  \path (pivot) ++ (0cm,-3.05cm) node{$\twincpivot$};

  \node (top) at (3.5,0) {};
    \path (top) ++ (0cm,0.4cm) node{$\toptwinc$};
    \path (top) ++ (-1.35cm,-2.75cm) node (P2-left) {};
    \path (top) ++ (1.35cm,-2.75cm) node (P2-right) {};
    
  % label top part
  \path (top) ++ (0cm,-3.05cm) node{$\twinctop$};

  % label carrier twin crystal
  \path ($(pivot)!0.5!(top)$) ++ (0cm,0.65cm) node{$\twinc$};

  \draw [-,fill=forestgreen!20] (pivot.center) -- (P1-left.center) -- (P1-right.center) -- (pivot.center);

  \draw [-,fill=royalblue!35] (top.center) -- (P2-left.center) -- (P2-right.center) -- (top.center);
  
  % label E_2
  \path (top) ++ (0cm,-0.4cm) node {\large\bf\royalblue{$E_2$}};

  % (1-pivot)
    \draw[draw=none] (pivot) to node[pos=0.325](1-pivot){} (P1-left);
  \draw[->,thick,forestgreen] (pivot.center) to (1-pivot.center);
  % (2-pivot)
    \draw[draw=none] (pivot) to node[pos=0.425](2-pivot){} (P1-right);
  \draw[->,thick,forestgreen] (pivot.center) to (2-pivot.center);   
  %  
  %  
  % (1-top)
  \draw[draw=none] (top) to node[pos=0.425](1-top){} (P2-left);
    \draw[->,very thick,royalblue] (top.center) to (1-top.center);
  % (2-top)   
  \draw[draw=none] (top) to node[pos=0.325](2-top){} (P2-right);
    \draw[->,very thick,royalblue] (top.center) to (2-top.center);   
  \draw[->,shorten >=2pt] (pivot.center) to (1-top.center);
  \draw[->,shorten >=5.5pt] (pivot.center) to (2-top.center);  
  \draw[->,thick,royalblue,shorten >=5.5pt] (top.center) to (1-pivot.center);
  \draw[->,thick,royalblue,shorten >=2pt] (top.center) to (2-pivot.center);

  % (boundary-left-P1)
  \draw[draw=none] (pivot) to node[pos=0.675](boundary-left-P1){} (P1-left);
  % (boundary-right-P1)   
  \draw[draw=none] (pivot) to node[pos=0.75](boundary-right-P1){} (P1-right);
  %
  % boundary Am_1
  \draw[out=20,in=210,distance=1cm] (boundary-left-P1.center) to node[pos=0.4](boundary-P1-mid){} (boundary-right-P1.center);
    % boundary Am_1 mid, plus arrows from
    \path (boundary-P1-mid.center) ++ (0.5cm,0.1cm) node(v1){}; 
    \draw[->] (boundary-P1-mid.center) to (v1.center);
    \draw[->,thick,royalblue] (top.center) to (v1.center);
    % label
    \path (boundary-left-P1) ++ (-0.5cm,-0.1cm) node{$\boundaryof{\twincambigpivot}$};
    
  % 1-transitions back to pivot  
  \draw[->>,thick,densely dotted,out=170,in=180,distance=1.5cm,shorten >=-2pt] (boundary-left-P1.center) to (pivot);  
   
  % label Unam_1
  \path (boundary-P1-mid) ++ (-0.2cm,0.4cm) node{$\twincunambigpivot$}; 
  % label Am_1
  \path (boundary-P1-mid) ++ (-0.35cm,-0.4cm) node{$\twincambigpivot$};
  
  % loop entry transition from top to Am_1
  \path (boundary-P1-mid.center) ++ (0.35cm,-0.5cm) node(w1){}; 
  \draw[->,thick,forestgreen] (pivot.center) to (w1.center);
  
  % prohibited loop entry transitions Am_1 to Un_1
  \draw[->,red] ($(boundary-P1-mid.center) + (-0.6cm,-0.2cm)$) to node[sloped,pos=0.625]{/} ($(boundary-P1-mid.center) + (0.6cm,0.525cm)$);

  % (boundary-left-P2)
  \draw[draw=none] (top) to node[pos=0.7](boundary-left-P2){} (P2-left);
  % (boundary-right-P2)   
  \draw[draw=none] (top) to node[pos=0.55](boundary-right-P2){} (P2-right);
  % 
  % boundary Am_2
  \draw[out=-30,in=190,distance=1cm] (boundary-left-P2.center) to node[pos=0.5](helper-P2){} node[pos=0.9](boundary-P2-mid){} (boundary-right-P2.center); 
    % label 
    \path (boundary-right-P2) ++ (0.55cm,-0.1cm) node{$\boundaryof{\twincambigtop}$}; 
    % boundary Am_2 mid, plus arrows from
    \path (boundary-P2-mid.center) ++ (0cm,0.35cm) node(v2){}; 
    \draw[->] (boundary-P2-mid.center) to (v2.center);
    \draw[->] (pivot.center) to (v2.center);
  % 1-transitions back to pivot  
  \draw[->>,thick,densely dotted,out=30,in=0,distance=1.5cm,shorten >=-2pt] (boundary-right-P2.center) to (top);

  % label Unam_2
  \path (boundary-P2-mid) ++ (-0.35cm,0.15cm) node{$\twincunambigtop$}; 
  % label Am_2
  \path (boundary-P2-mid) ++ (0cm,-0.65cm) node{$\twincambigtop$};  
    
  % loop entry transition from top to Am_2
  \path (boundary-P2-mid.center) ++ (-0.9cm,-0.7cm) node(w2){}; 
  \draw[->,very thick,royalblue] (top.center) to (w2.center);  
  
  % prohibited loop entry transitions Am_1 to Un_1
  \draw[->,red] ($(helper-P2) + (0cm,-0.6cm)$) to node[pos=0.45,sloped]{/} ($(helper-P2) + (0cm,0.25cm)$);

  % collapse function
  \draw[|->,thick,out=-20,in=200,shorten <=0pt,shorten >=2pt,densely dashed,magenta,distance=0.75cm] 
    (w1) to node[below,yshift=0.05cm,pos=0.53]{$\scpfunon{\black{\twinc}}$} (w2);
  \draw[|->,thick,out=160,in=20,shorten <=0pt,shorten >=2pt,densely dashed,magenta,distance=0.75cm] 
    (w2) to node[above,yshift=-0.05cm]{$\scpfunon{\black{\twinc}}$} (w1);

%   \draw[|->,thick,out=20,in=160,shorten <=0pt,shorten >=2pt,densely dashed,magenta,distance=0.75cm] 
%     (w1) to node[above,yshift=-0.05cm]{$\scpfunon{\black{\twinc}}$} (w2);
%   \draw[|->,thick,out=200,in=-20,shorten <=0pt,shorten >=2pt,densely dashed,magenta,distance=0.75cm] 
%     (w2) to node[below,yshift=0.05cm,pos=0.47]{$\scpfunon{\black{\twinc}}$} (w1);   
    
\end{tikzpicture}
\end{center}
  \vspace*{-2ex}
  \caption{\label{fig:twincrystal:schema}%
    Structure schema of
    a \protect\twincrystal\ with carrier $\twinc$, 
      with part $\twincpivot$ of pivot vertex $\pivottwinc$,
      and part $\twinctop$ of top vertex $\toptwinc$,
      where $\twinctop$ is generated by the top entry transitions in \royalblue{$E_2$}.
    % Schematic illustration of a \protect\twincrystal\ 
    %   with carrier $\twinc$ with top vertex $\toptwinc$, and pivot vertex $\pivot$.
    % The pivot part $\twincpivot$ partitions into the sets $\twincunambigpivot$ and $\twincambigpivot$ 
    %   of unambiguous, and ambiguous vertices, respectively.
    % The top part $\twincpivot$ partitions into the sets $\twincunambigtop$ and $\twincambigtop$ 
    %   of unambiguous, and ambiguous vertices, respectively.
    % The boundary vertices $\boundaryof{\twincambigpivot}$ of $\twincpivot$ have \onetransition\ paths to $\pivottwinc$,
    %   and the boundary vertices $\boundaryof{\twincambigtop}$ of $\twinctop$ have \onetransition\ paths to $\toptwinc$.
    % There are no transitions directly from $\twincambigpivot\setminus\boundaryof{\twincambigpivot}$ to $\twincunambigpivot$,
    %   and no transitions directly from $\twincambigtop\setminus\boundaryof{\twincambigtop}$ to $\twincunambigtop$.
    % \oneBisimilar\ vertices in $\twincambigpivot$ and $\twincambigtop$ are related by the collapse (partial) function 
    %   $\scpfunon{\twinc}$.  
           }
\end{figure}

%% file: figs/fig-ill-pos-R123.tex
\begin{figure}[!tb]
\begin{center}
  \scalebox{0.85}{%
\begin{tikzpicture}[scale=1,every node/.style={transform shape}]
\tikzset{Rightarrow/.style={darkcyan,double distance=2pt,>={Implies},->},
         DashedRightarrow/.style={darkcyan,densely dashed,double distance=2pt,>={Implies},->},
         strike through/.append style={decoration={markings, mark=at position 0.5 with {\draw[-] ++ (-2pt,-2pt) -- (2pt,2pt);}},postaction={decorate}},
  loopentryred/.style={->,thick},        
  onebodyred/.style={->,very thick,densely dotted},
  onebodyredcol/.style={->,very thick,densely dotted,darkcyan},
  onebodyredalert/.style={->,thick,densely dotted,red},
  onebodyredrtc/.style={->>,very thick,densely dotted},
  onebodyredrtccol/.style={->>,very thick,densely dotted,darkcyan},
  onebodyredrtcalert/.style={->>,thick,densely dotted,red},
  % notoneactbodyredrtc/.style={->>,thick,densely dashed,red},
  % notonebodyredcol/.style={->,
  %                          decoration={markings,mark=at position 0.5 with \node{X};}
  %                          % decoration={markings,mark=at position 0.5 with 
  %                          %              {\draw [red,thick,-] 
  %                          %                 ++ ( 0,-\StrikeThruDistance) 
  %                          %                 -- ( 0, \StrikeThruDistance);}},
  %                          postaction={decorate},
  %                          thick,densely dotted,red},
  % notonebodyredrtccol/.style={->>,
  %                              decoration={markings,mark=at position 0.5 with 
  %                                           {\draw [red,thick,-] 
  %                                              ++ ( 0, -\StrikeThruDistance) 
  %                                              -- ( 0, \StrikeThruDistance);}},
  %                              postaction={decorate},
  %                              thick,densely dotted,red},
  convonebodyred/.style={<-,very thick,densely dotted},
  convonebodyredcol/.style={<-,very thick,densely dotted,darkcyan},
  oneactbodyred/.style={->,densely dashed},
  oneactbodyredcol/.style={->,thick,densely dashed,darkcyan},
  oneactbodyredtc/.style={-,densely dashed,thick,preaction={draw,DashedRightarrow}},
  oneactbodyredtccol/.style={-,densely dashed,thick,darkcyan,preaction={draw,DashedRightarrow}},
  oneactbodyredrtc/.style={->>,densely dashed},
  oneactbodyredrtccol/.style={->>,thick,densely dashed,darkcyan},
  oneactbodyredrtcalert/.style={->>,thick,densely dashed,red},
  convoneactbodyredrtc/.style={<<-,densely dashed},
  convoneactbodyredrtccol/.style={<<-,thick,densely dashed,darkcyan},
  noonebodyred/.style={-,densely dashed,thick,preaction={draw,Rightarrow}},
  noonebodyredcol/.style={-,thick,red,preaction={draw,Rightarrow}},%{-,densely dashed,thick,darkcyan,preaction={draw,Rightarrow}},
    % triple/.style={-,preaction={draw,Rightarrow}},
    % quadruple/.style={preaction={draw,Rightarrow,shorten >=0pt},shorten >=1pt,-,double,doubledistance=0.2pt}
  strike thru arrow/.style={
    decoration={markings, mark=at position 0.5 with {
        \draw [blue, thick,-] 
            ++ (-\StrikeThruDistance,-\StrikeThruDistance) 
            -- ( \StrikeThruDistance, \StrikeThruDistance);}
    },
    postaction={decorate},
}    
}

% 
%---------
% R1.1 
%--------- 
\matrix[anchor=center,ampersand replacement=\&,%
        row sep=0.8cm,column sep=1.15cm,every node/.style={draw,thick,circle,minimum width=2.5pt,fill,inner sep=0pt,outer sep=2pt}] at (0,0) {
  \node[draw=none,fill=none](w-1-bar){};  \&  \node[draw=none,fill=none](h0){};           
  \\[0.25cm]
  \node[draw=none,fill=none](w-1n){};     \&  
  \\[-0.275cm]   
                     \&  \node[draw=none,fill=none](label-R11){};
  \\[-0.275cm]
  \node[draw=none,fill=none](w-11){};     \& 
  \\[0.275cm]
  \node(w-1){};      \&                                       \&  \node(w-2){};
  \\
  \node[draw=none,fill=none](dummy-1){};  
                     \&                                       \&  \node[draw=none,fill=none](dummy-2){};
  \\
};
\calcLength(w-1-bar,h0){mylen};

\path (label-R11) ++ (0pt,-15pt) node{\Large \crtcrossreflabel{\mbox{\nf (R1.1)}}[R1.1]};

% w-1 outgoing  
%  
\path (w-1) ++ ({-0.2*\mylen pt},{-0.25*\mylen pt}) node{$\bverti{1}$}; 
  
% w-2 outgoing  
% 
\path (w-2) ++ ({0.2*\mylen pt},{-0.25*\mylen pt}) node{$\bverti{2}$}; 
% \path (w-2) ++ ({-0.3*\mylen pt},{0.35*\mylen pt}) node{\text{\nf \mediumblue{\large $\bverti{2}$ not normed}}}; 
\path (w-2) ++ ({-0.45*\mylen pt},{0.65*\mylen pt}) node{\text{\nf \mediumblue{\large $\bverti{1}$, $\bverti{2}$ not}}}; 
\path (w-2) ++ ({-0.3*\mylen pt},{0.35*\mylen pt}) node{\text{\nf \mediumblue{\large normed}}};    
  
\draw[oneactbodyredrtcalert,thick,red,shorten <=3pt,shorten >=3pt] (w-2) to node[pos=0.5,sloped]{$\pmb{/}$} (w-1);

\draw[-,very thick,out=-40,in=220,magenta,densely dashed,distance={1*\mylen pt}] 
  (w-1) to (w-2);

% 
%---------
% R1.2 
%--------- 
\matrix[anchor=center,ampersand replacement=\&,%
        row sep=0.8cm,column sep=1.15cm,every node/.style={draw,thick,circle,minimum width=2.5pt,fill,inner sep=0pt,outer sep=2pt}] at (3.75,0) {
  \node[chocolate](w-1-bar){};  \&  \node[draw=none,fill=none](h0){};           
  \\[0.25cm]
  \node(w-1n){};     \&  
  \\[-0.275cm]   
                     \&  \node[draw=none,fill=none](label-R12){};
  \\[-0.275cm]
  \node(w-11){};     \& 
  \\[0.275cm]
  \node(w-1){};      \&                                       \&  \node(w-2){};
  \\
  \node[draw=none,fill=none](dummy-1){};  
                     \&                                       \&  \node[draw=none,fill=none](dummy-2){};
  \\
};
\calcLength(w-1-bar,h0){mylen};

\draw[thick,chocolate] (w-1-bar) circle (0.12cm); 

\path (label-R12) ++ (0pt,-15pt) node{\Large \crtcrossreflabel{\mbox{\nf (R1.2)}}[R1.2]};  
\path ($(label-R11)!0.43!(label-R12)$) ++ (0pt,-5pt) node{\Large \crtcrossreflabel{\mbox{\nf (R1)}}[R1]};

% w-1-bar outgoing
% 
\path (w-1-bar) ++ ({0.075*\mylen pt},{0.3*\mylen pt}) node{$\bvertbari{1}$};

\path (w-1-bar) ++ ({-0.25*\mylen pt},{-0.4*\mylen pt}) node[draw,thick,circle,minimum width=2.5pt,fill,inner sep=0pt,outer sep=2pt](w-1-bar-w-1n){};
  \draw[loopentryred,shorten <= 2pt] (w-1-bar) to (w-1-bar-w-1n);
  \draw[oneactbodyredrtc,distance={0.25*\mylen pt},out=210,in=160] (w-1-bar-w-1n) to (w-1n);

% w-1n outgoing
% 
\draw[onebodyredcol,out=10,in=-40,distance={0.5*\mylen pt},shorten >=2pt] (w-1n) to (w-1-bar);
\draw[-,dotted,thick,shorten <={0.3*\mylen pt},shorten >={0.125*\mylen pt}] (w-1n) to (w-11);

\draw[loopentryred] (w-1n) to ($(w-1n) + ({-0.15*\mylen pt},{-0.3*\mylen pt})$);
\draw[loopentryred] (w-1n) to ($(w-1n) + ({0.15*\mylen pt},{-0.3*\mylen pt})$);

\draw[convonebodyredcol,out=-30,in=90,distance={0.15*\mylen pt}] (w-1n) to ($(w-1n) + ({0.4*\mylen pt},{-0.3*\mylen pt})$);

% w-11 outgoing

\path (w-11) ++ ({-0.25*\mylen pt},{-0.4*\mylen pt}) node[draw,thick,circle,minimum width=2.5pt,fill,inner sep=0pt,outer sep=2pt](w-11-w-1){};
  \draw[loopentryred] (w-11) to (w-11-w-1);
  \draw[oneactbodyredrtc,distance={0.25*\mylen pt},out=210,in=160] (w-11-w-1) to (w-1);
\draw[onebodyredcol,out=20,in=270,distance={0.25*\mylen pt}] (w-11) to ($(w-11) + ({0.35*\mylen pt},{0.4*\mylen pt})$);

% w-1 outgoing  
%  
\draw[onebodyredcol,out=30,in=-40,distance={0.5*\mylen pt}] (w-1) to (w-11);
 
\path (w-1) ++ ({-0.2*\mylen pt},{-0.25*\mylen pt}) node{$\bverti{1}$};

% w-2 outgoing  
% 
\draw[thick,densely dashed] (w-2) circle (0.12cm); 
\path (w-2) ++ ({0.2*\mylen pt},{-0.25*\mylen pt}) node{$\bverti{2}$}; 
\path (w-2) ++ ({-0.5*\mylen pt},{0.35*\mylen pt}) node{\text{\nf\mediumblue{\large $\bverti{1}$, $\bverti{2}$ normed}}};    
      
\draw[oneactbodyredrtcalert,thick,red,shorten <=3pt,shorten >=3pt] (w-2) to node[pos=0.5,sloped]{$\pmb{/}$} (w-1);  
   
\draw[-,very thick,bend right,magenta,densely dashed,distance={1*\mylen pt},out=-40,in=220,shorten >=3pt] 
  (w-1) to (w-2);

%---------
% R2
%--------- 
\matrix[anchor=center,ampersand replacement=\&,%
        row sep=0.8cm,column sep=1.15cm,every node/.style={draw,thick,circle,minimum width=2.5pt,fill,inner sep=0pt,outer sep=2pt}] at (6.3,0.25) {
                    \&  \node[draw=none,fill=none](v){};       \& \node(w-1){};
  \\[0.25cm]
  \node[draw=none,fill=none](w-1-bar){}; \&  \node[draw=none,fill=none](h0){};
                                         \&  \node(w-2-bar){};  
  \\
  \node[draw=none,fill=none](w-1n){};    \&                    \&  \node(w-2n){};
  \\[-0.275cm]
                    \& \node[draw=none,fill=none](label-R2){};
  \\[-0.275cm]
  \node[draw=none,fill=none](w-11){};    \&                    \&  \node(w-21){};
  \\
  \node[draw=none,fill=none](w-1-previous){};   \&  \node[draw=none,fill=none](h){}; 
                                                              \&  \node(w-2){};
  \\
};
\calcLength(v,h0){mylen};

\path (label-R2) ++ (0pt,18pt) node{\Large \crtcrossreflabel{\mbox{\nf (R2)}}[R2]};

\draw[thick,densely dashed] (w-1) circle (0.12cm);
\path (w-1) ++ ({0.075*\mylen pt},{0.25*\mylen pt}) node{$\bverti{1}$};

\path (w-1) ++ ({-0.25*\mylen pt},{-0.4*\mylen pt}) node[draw,thick,circle,minimum width=2.5pt,fill,inner sep=0pt,outer sep=2pt](w-1-w-2-bar){};
  \draw[loopentryred,shorten <=2pt] (w-1) to (w-1-w-2-bar);
  \draw[oneactbodyredrtc,distance={0.25*\mylen pt},out=210,in=160] (w-1-w-2-bar) to (w-2-bar);

% w-2-bar outgoing
% 
\path (w-2-bar) ++ ({0.5*\mylen pt},{-0.075*\mylen pt}) node{$\bvertbari{2}$};
%\draw[loopentryred,noonebodyredcol,thick,out=45,in=0,distance={1*\mylen pt}] (w-2-bar) to (v);
  
\draw[oneactbodyredrtccol,out=10,in=-40,distance={0.5*\mylen pt},shorten >=2pt] (w-2-bar) to (w-1);  

\path (w-2-bar) ++ ({-0.25*\mylen pt},{-0.4*\mylen pt}) node[draw,thick,circle,minimum width=2.5pt,fill,inner sep=0pt,outer sep=2pt](v-2-v-12){};
  \draw[loopentryred] (w-2-bar) to (v-2-v-12);
  \draw[oneactbodyredrtc,distance={0.25*\mylen pt},out=210,in=160] (v-2-v-12) to (w-2n);

% w-2n outgoing
% 
\draw[oneactbodyredrtccol,thick,out=10,in=-40,distance={0.5*\mylen pt}] (w-2n) to (w-2-bar);
\draw[-,dotted,thick,shorten <={0.3*\mylen pt},shorten >={0.125*\mylen pt}] (w-2n) to (w-21);
\draw[loopentryred] (w-2n) to ($(w-2n) + ({-0.15*\mylen pt},{-0.3*\mylen pt})$);
\draw[loopentryred] (w-2n) to ($(w-2n) + ({0.15*\mylen pt},{-0.3*\mylen pt})$);

\draw[convoneactbodyredrtccol,thick,out=-30,in=90,distance={0.15*\mylen pt}] (w-2n) to ($(w-2n) + ({0.4*\mylen pt},{-0.3*\mylen pt})$);

% (w-21) ingoing/outgoing    
%
\draw[oneactbodyredrtccol,thick,out=20,in=270,distance={0.25*\mylen pt}] (w-21) to ($(w-21) + ({0.35*\mylen pt},{0.4*\mylen pt})$);    

\path (w-21) ++ ({-0.25*\mylen pt},{-0.4*\mylen pt}) node[draw,thick,circle,minimum width=2.5pt,fill,inner sep=0pt,outer sep=2pt](v-n2-u-2){};   
  \draw[loopentryred] (w-21) to (v-n2-u-2);
  \draw[oneactbodyred,distance={0.25*\mylen pt},out=210,in=160] (v-n2-u-2) to (w-2);

% (w-2) outgoing    
%
\path (w-2) ++ ({0.075*\mylen pt},{-0.25*\mylen pt}) node{$\bverti{2}$};
\draw[oneactbodyredrtccol,out=10,in=-40,distance={0.5*\mylen pt}] (w-2) to (w-21);

\draw[-,very thick,magenta,densely dashed,out=180,in=180,distance={2*\mylen pt},shorten <=2pt] 
  (w-1) to (w-2);

% 
%---------
% R3.1 
%--------- 
\matrix[anchor=center,ampersand replacement=\&,%
        row sep=0.8cm,column sep=1.15cm,every node/.style={draw,thick,circle,minimum width=2.5pt,fill,inner sep=0pt,outer sep=2pt}] at (0.75,-4.75) {
                    \&  \node(v){};
  \\[0.25cm]
  \node(w-1-bar){}; \&  \node[draw=none,fill=none](h0){};
                                    \&  \node(w-2-bar){};  
  \\
  \node(w-1n){};    \&              \&  \node(w-2n){};
  \\[-0.275cm]
                    \&  \node[draw=none,fill=none](label-R31){};
  \\[-0.275cm]
  \node(w-11){};    \&              \&  \node(w-21){};
  \\
  \node(w-1){};     \&  \node[draw=none,fill=none](h){}; 
                                    \&  \node(w-2){};
  \\
};
\calcLength(v,h0){mylen};

\path (label-R31) ++ (0pt,-5pt) node{\Large \crtcrossreflabel{\mbox{\nf (R3.1)}}[R3.1]}; %{\large \ref{R3},\ref{R3.1}};

% (v) outgoing
%
\draw[thick,densely dashed] (v) circle (0.12cm); 
\path (v) ++ (0pt,{0.25*\mylen pt}) node{$\avert$};
  \path (v) ++ ({-0.5*\mylen pt},{-0.4*\mylen pt}) node[draw,thick,circle,minimum width=2.5pt,fill,inner sep=0pt,outer sep=2pt] (v-01){};
    \draw[loopentryred,shorten <=2pt] (v) to (v-01);
    \draw[oneactbodyredrtc,bend right,distance={0.35*\mylen pt}] (v-01) to (w-1-bar);
  \path (v) ++ ({0.5*\mylen pt},{-0.4*\mylen pt}) node[draw,thick,circle,minimum width=2.5pt,fill,inner sep=0pt,outer sep=2pt] (v-03){}; 
    \draw[loopentryred,shorten <=2pt] (v) to (v-03);
    \draw[oneactbodyredrtc,bend left,distance={0.35*\mylen pt}] (v-03) to (w-2-bar);

% (w-1-bar) outgoing   
%     
\draw[onebodyredcol,out=135,in=180,distance={1*\mylen pt},shorten >=2pt] (w-1-bar) to (v);         
\path (w-1-bar) ++ ({-0.4*\mylen pt},{0.1*\mylen pt}) node{$\bvertbari{1}$};

\path (w-1-bar) ++ ({0.25*\mylen pt},{-0.4*\mylen pt}) node[draw,thick,circle,minimum width=2.5pt,fill,inner sep=0pt,outer sep=2pt](v-1-v-11){};
  \draw[loopentryred] (w-1-bar) to (v-1-v-11);
  \draw[oneactbodyredrtc,distance={0.25*\mylen pt},out=-30,in=20] (v-1-v-11) to (w-1n);

% (w-1n) outgoing
% 
\draw[-,dotted,thick,shorten <={0.3*\mylen pt},shorten >={0.125*\mylen pt}] (w-1n) to (w-11);

\draw[onebodyredcol,out=170,in=220,distance={0.5*\mylen pt}] (w-1n) to (w-1-bar);

% (w-1n) incoming
%
\draw[convonebodyredcol,out=210,in=90,distance={0.15*\mylen pt}] (w-1n) to ($(w-1n) + ({-0.4*\mylen pt},{-0.3*\mylen pt})$);

% (w-11) 
%
\draw[onebodyredcol,out=160,in=270,distance={0.25*\mylen pt}] (w-11) to ($(w-11) + ({-0.35*\mylen pt},{0.4*\mylen pt})$);

% (w-1) 
%
\path (w-1) ++ ({0*\mylen pt},{-0.25*\mylen pt}) node{$\bverti{1}$};
\draw[onebodyredcol,out=170,in=220,distance={0.5*\mylen pt}] (w-1) to (w-11);

\path (w-11) ++ ({0.25*\mylen pt},{-0.4*\mylen pt}) node[draw,thick,circle,minimum width=2.5pt,fill,inner sep=0pt,outer sep=2pt](v-n1-u-1){};
  \draw[loopentryred] (w-11) to (v-n1-u-1);
  \draw[oneactbodyredrtc,distance={0.25*\mylen pt},out=-30,in=20] (v-n1-u-1) to (w-1);
\draw[loopentryred] (w-1n) to ($(w-1n) + ({-0.15*\mylen pt},{-0.3*\mylen pt})$);
\draw[loopentryred] (w-1n) to ($(w-1n) + ({0.15*\mylen pt},{-0.3*\mylen pt})$);

% (w-2-bar) outgoing
% 
\path (w-2-bar) ++ ({0.4*\mylen pt},{0.1*\mylen pt}) node{$\bvertbari{2}$};
\draw[onebodyredcol,out=45,in=0,distance={1*\mylen pt},shorten >=2pt] (w-2-bar) to (v);
\draw[onebodyredcol,out=10,in=-40,distance={0.5*\mylen pt}] (w-2n) to (w-2-bar);

\path (w-2-bar) ++ ({-0.25*\mylen pt},{-0.4*\mylen pt}) node[draw,thick,circle,minimum width=2.5pt,fill,inner sep=0pt,outer sep=2pt](v-2-v-12){};
  \draw[loopentryred] (w-2-bar) to (v-2-v-12);
  \draw[oneactbodyredrtc,distance={0.25*\mylen pt},out=210,in=160] (v-2-v-12) to (w-2n);

\draw[oneactbodyredrtcalert,shorten <= {0.1 *\mylen pt},shorten >={0.1 *\mylen pt}] 
  (w-2-bar) to node[pos=0.5,sloped]{$\pmb{/}$} (w-1-bar);

% w-2n outgoing
% 
\draw[-,dotted,thick,shorten <={0.3*\mylen pt},shorten >={0.125*\mylen pt}] (w-2n) to (w-21);
\draw[loopentryred] (w-2n) to ($(w-2n) + ({-0.15*\mylen pt},{-0.3*\mylen pt})$);
\draw[loopentryred] (w-2n) to ($(w-2n) + ({0.15*\mylen pt},{-0.3*\mylen pt})$);

\draw[convonebodyredcol,out=-30,in=90,distance={0.15*\mylen pt}] (w-2n) to ($(w-2n) + ({0.4*\mylen pt},{-0.3*\mylen pt})$);

% (w-21) ingoing/outgoing    
%
\draw[onebodyredcol,out=20,in=270,distance={0.25*\mylen pt}] (w-21) to ($(w-21) + ({0.35*\mylen pt},{0.4*\mylen pt})$);    

\path (w-21) ++ ({-0.25*\mylen pt},{-0.4*\mylen pt}) node[draw,thick,circle,minimum width=2.5pt,fill,inner sep=0pt,outer sep=2pt](v-n2-u-2){};
  \draw[oneactbodyredrtc,distance={0.25*\mylen pt},out=210,in=160] (v-n2-u-2) to (w-2);     
  \draw[loopentryred] (w-21) to (v-n2-u-2);

% (w-2) outgoing    
%
\path (w-2) ++ ({0.075*\mylen pt},{-0.25*\mylen pt}) node{$\bverti{2}$};
\draw[onebodyredcol,out=10,in=-40,distance={0.5*\mylen pt}] (w-2) to (w-21);

\draw[-,very thick,bend right,magenta,densely dashed,distance={0.75*\mylen pt}] 
  (w-1) to node[pos=0.5](mid){} 
           node[pos=0.65](left){} (w-2);

%---------
% R3.2
%--------- 
\matrix[anchor=center,ampersand replacement=\&,%
        row sep=0.8cm,column sep=1.15cm,every node/.style={draw,thick,circle,minimum width=2.5pt,fill,inner sep=0pt,outer sep=2pt}] at (5.25,-4.75) {
                    \&  \node(v){};
  \\[0.25cm]
  \node(w-1-bar){}; \&  \node[draw=none,fill=none](h0){};
                                     \&  \node(w-2-bar){};  
  \\
  \node(w-1n){};    \&               \&  \node(w-2n){};
  \\[-0.275cm]
                    \& \node[draw=none,fill=none](label-R32){};
  \\[-0.275cm]
  \node(w-11){};    \&               \&  \node(w-21){};
  \\
  \node(w-1){};     \&  \node[draw=none,fill=none](h){}; 
                                     \&  \node(w-2){};
  \\
};
\calcLength(v,h0){mylen};

\path (label-R32) ++ (0pt,-5pt) node{\Large \crtcrossreflabel{\mbox{\nf (R3.2)}}[R3.2]}; %node{\large \ref{R3},\ref{R3.2}};

% (v) outgoing
%
\draw[thick,densely dashed] (v) circle (0.12cm); 
\path (v) ++ (0pt,{0.25*\mylen pt}) node{$\avert$};
  \path (v) ++ ({-0.5*\mylen pt},{-0.4*\mylen pt}) node[draw,thick,circle,minimum width=2.5pt,fill,inner sep=0pt,outer sep=2pt] (v-01){};
    \draw[loopentryred,shorten <=2pt] (v) to (v-01);
    \draw[oneactbodyredrtc,bend right,distance={0.35*\mylen pt}] (v-01) to (w-1-bar);
  \path (v) ++ ({0.5*\mylen pt},{-0.4*\mylen pt}) node[draw,thick,circle,minimum width=2.5pt,fill,inner sep=0pt,outer sep=2pt] (v-03){}; 
    \draw[loopentryred,shorten <=2pt] (v) to (v-03);
    \draw[oneactbodyredrtc,bend left,distance={0.35*\mylen pt}] (v-03) to (w-2-bar);

% (w-1-bar) outgoing   
%       
\draw[oneactbodyredrtccol,out=135,in=180,distance={1*\mylen pt},shorten >=2pt] (w-1-bar) to (v);         
%\draw[onebodyredalert,out=145,in=150,distance={1.2*\mylen pt},shorten >= 5pt] (w-1-bar) to node[pos=0.5,sloped]{$\pmb{/}$} (v);

\path (w-1-bar) ++ ({-0.4*\mylen pt},{0.1*\mylen pt}) node{$\bvertbari{1}$};

\path (w-1-bar) ++ ({0.25*\mylen pt},{-0.4*\mylen pt}) node[draw,thick,circle,minimum width=2.5pt,fill,inner sep=0pt,outer sep=2pt](v-1-v-11){};
  \draw[loopentryred] (w-1-bar) to (v-1-v-11);
  \draw[oneactbodyredrtc,distance={0.25*\mylen pt},out=-30,in=20] (v-1-v-11) to (w-1n);

% (w-1n) outgoing
% 
\draw[-,dotted,thick,shorten <={0.3*\mylen pt},shorten >={0.125*\mylen pt}] (w-1n) to (w-11);

\draw[onebodyredcol,out=170,in=220,distance={0.5*\mylen pt}] (w-1n) to (w-1-bar);

% (w-1n) incoming
%
\draw[convonebodyredcol,out=210,in=90,distance={0.15*\mylen pt}] (w-1n) to ($(w-1n) + ({-0.4*\mylen pt},{-0.3*\mylen pt})$);

% (w-11) 
%
\draw[onebodyredcol,out=160,in=270,distance={0.25*\mylen pt}] (w-11) to ($(w-11) + ({-0.35*\mylen pt},{0.4*\mylen pt})$);

% (w-1) 
%
\path (w-1) ++ ({0*\mylen pt},{-0.25*\mylen pt}) node{$\bverti{1}$};
\draw[onebodyredcol,out=170,in=220,distance={0.5*\mylen pt}] (w-1) to (w-11);

\path (w-11) ++ ({0.25*\mylen pt},{-0.4*\mylen pt}) node[draw,thick,circle,minimum width=2.5pt,fill,inner sep=0pt,outer sep=2pt](v-n1-u-1){};
  \draw[loopentryred] (w-11) to (v-n1-u-1);
  \draw[oneactbodyredrtc,distance={0.25*\mylen pt},out=-30,in=20] (v-n1-u-1) to (w-1);
\draw[loopentryred] (w-1n) to ($(w-1n) + ({-0.15*\mylen pt},{-0.3*\mylen pt})$);
\draw[loopentryred] (w-1n) to ($(w-1n) + ({0.15*\mylen pt},{-0.3*\mylen pt})$);

% (w-2-bar) outgoing
% 
\path (w-2-bar) ++ ({0.4*\mylen pt},{0.1*\mylen pt}) node{$\bvertbari{2}$};
\draw[oneactbodyredrtccol,out=45,in=0,distance={1*\mylen pt},shorten >=2pt] (w-2-bar) to (v);
\draw[oneactbodyredrtccol,out=10,in=-40,distance={0.5*\mylen pt}] (w-2n) to (w-2-bar);

\path (w-2-bar) ++ ({-0.25*\mylen pt},{-0.4*\mylen pt}) node[draw,thick,circle,minimum width=2.5pt,fill,inner sep=0pt,outer sep=2pt](v-2-v-12){};
  \draw[loopentryred] (w-2-bar) to (v-2-v-12);
  \draw[oneactbodyredrtc,distance={0.25*\mylen pt},out=210,in=160] (v-2-v-12) to (w-2n);

\draw[oneactbodyredrtcalert,shorten <= {0.1 *\mylen pt},shorten >={0.1 *\mylen pt}] 
  (w-2-bar) to node[pos=0.5,sloped]{$\pmb{/}$} (w-1-bar);

% w-2n outgoing
% 
\draw[-,dotted,thick,shorten <={0.3*\mylen pt},shorten >={0.125*\mylen pt}] (w-2n) to (w-21);
\draw[loopentryred] (w-2n) to ($(w-2n) + ({-0.15*\mylen pt},{-0.3*\mylen pt})$);
\draw[loopentryred] (w-2n) to ($(w-2n) + ({0.15*\mylen pt},{-0.3*\mylen pt})$);

\draw[convoneactbodyredrtccol,out=-30,in=90,distance={0.15*\mylen pt}] (w-2n) to ($(w-2n) + ({0.4*\mylen pt},{-0.3*\mylen pt})$);

% (w-21) ingoing/outgoing    
%
\draw[oneactbodyredrtccol,out=20,in=270,distance={0.25*\mylen pt}] (w-21) to ($(w-21) + ({0.35*\mylen pt},{0.4*\mylen pt})$);    

\path (w-21) ++ ({-0.25*\mylen pt},{-0.4*\mylen pt}) node[draw,thick,circle,minimum width=2.5pt,fill,inner sep=0pt,outer sep=2pt](v-n2-u-2){};
  \draw[oneactbodyredrtc,distance={0.25*\mylen pt},out=210,in=160] (v-n2-u-2) to (w-2);     
  \draw[loopentryred] (w-21) to (v-n2-u-2);

% (w-2) outgoing    
%
\path (w-2) ++ ({0.075*\mylen pt},{-0.25*\mylen pt}) node{$\bverti{2}$};
\draw[oneactbodyredrtccol,out=10,in=-40,distance={0.5*\mylen pt}] (w-2) to (w-21);     
     
\draw[onebodyredrtcalert,out=0,in=30,distance={2.25*\mylen pt},shorten >= 6pt] (w-2) to node[pos=0.5,sloped,xshift={0.3*\mylen pt}]{$\pmb{/}$} (v);

\draw[-,very thick,bend right,magenta,densely dashed,distance={0.75*\mylen pt}] 
  (w-1) to node[pos=0.5](mid){} 
           node[pos=0.65](left){} (w-2);

% 
%---------
% R3.3 
%--------- 
\matrix[anchor=center,ampersand replacement=\&,%
        row sep=0.8cm,column sep=1.15cm,every node/.style={draw,thick,circle,minimum width=2.5pt,fill,inner sep=0pt,outer sep=2pt}] at (0.75,-10.25) {
                    \&  \node(v){};
  \\[0.25cm]
  \node(w-1-bar){}; \&  \node[draw=none,fill=none](h0){};
                                      \&  \node(w-2-bar){};  
  \\
  \node(w-1n){};    \&                \&  \node(w-2n){};
  \\[-0.275cm]
                    \&  \node[draw=none,fill=none](label-R33){};
  \\[-0.275cm]
  \node(w-11){};    \&                \&  \node(w-21){};
  \\
  \node(w-1){};     \&  \node[draw=none,fill=none](h){}; 
                                      \&  \node(w-2){};
  \\
};
\calcLength(v,h0){mylen};

\path (label-R33) ++ (0pt,-5pt) node{\Large \crtcrossreflabel{\mbox{\nf (R3.3)}}[R3.3]};

% (v) outgoing
%
\path (v) ++ (0pt,{0.25*\mylen pt}) node{$\avert$};
  \path (v) ++ ({-0.5*\mylen pt},{-0.4*\mylen pt}) node[draw,thick,circle,minimum width=2.5pt,fill,inner sep=0pt,outer sep=2pt] (v-01){};
    \draw[loopentryred] (v) to (v-01);
    \draw[oneactbodyredrtc,bend right,distance={0.35*\mylen pt}] (v-01) to (w-1-bar);
  \path (v) ++ ({0.5*\mylen pt},{-0.4*\mylen pt}) node[draw,thick,circle,minimum width=2.5pt,fill,inner sep=0pt,outer sep=2pt] (v-03){}; 
    \draw[loopentryred] (v) to (v-03);
    \draw[oneactbodyredrtc,bend left,distance={0.35*\mylen pt}] (v-03) to (w-2-bar);

% (w-1-bar) outgoing   
%     
\draw[oneactbodyredrtccol,out=135,in=180,distance={1*\mylen pt}] (w-1-bar) to (v);         
\path (w-1-bar) ++ ({-0.4*\mylen pt},{0.1*\mylen pt}) node{$\bvertbari{1}$};

\path (w-1-bar) ++ ({0.25*\mylen pt},{-0.4*\mylen pt}) node[draw,thick,circle,minimum width=2.5pt,fill,inner sep=0pt,outer sep=2pt](v-1-v-11){};
  \draw[loopentryred] (w-1-bar) to (v-1-v-11);
  \draw[oneactbodyredrtc,distance={0.25*\mylen pt},out=-30,in=20] (v-1-v-11) to (w-1n);
  
\draw[oneactbodyredrtcalert,shorten <= {0.1 *\mylen pt},shorten >={0.1 *\mylen pt},out=155,in=25] 
  (w-2-bar) to node[pos=0.5,sloped]{$\pmb{/}$} (w-1-bar);

% (w-1n) outgoing
% 
\draw[-,dotted,thick,shorten <={0.3*\mylen pt},shorten >={0.125*\mylen pt}] (w-1n) to (w-11);

\draw[oneactbodyredrtccol,out=170,in=220,distance={0.5*\mylen pt}] (w-1n) to (w-1-bar);

\draw[oneactbodyredrtccol,shorten <= {0.1 *\mylen pt},shorten >={0.1 *\mylen pt}] 
  (w-1-bar) to (w-2-bar);
  
% \draw[oneactbodyredrtcalert,shorten <= {0.1 *\mylen pt},shorten >={0.1 *\mylen pt}] 
%   (w-2-bar) to node[pos=0.5,sloped]{$\pmb{/}$} (w-1-bar);

% (w-1n) incoming
%
\draw[convoneactbodyredrtccol,out=210,in=90,distance={0.15*\mylen pt}] (w-1n) to ($(w-1n) + ({-0.4*\mylen pt},{-0.3*\mylen pt})$);

% (w-11) 
%
\draw[oneactbodyredrtccol,thick,out=160,in=270,distance={0.25*\mylen pt}] (w-11) to ($(w-11) + ({-0.35*\mylen pt},{0.4*\mylen pt})$);

% (w-1) 
%
\path (w-1) ++ ({0*\mylen pt},{-0.25*\mylen pt}) node{$\bverti{1}$};
\draw[oneactbodyredrtccol,out=170,in=220,distance={0.5*\mylen pt}] (w-1) to (w-11);

\path (w-11) ++ ({0.25*\mylen pt},{-0.4*\mylen pt}) node[draw,thick,circle,minimum width=2.5pt,fill,inner sep=0pt,outer sep=2pt](v-n1-u-1){};
  \draw[loopentryred] (w-11) to (v-n1-u-1);
  \draw[oneactbodyredrtc,distance={0.25*\mylen pt},out=-30,in=20] (v-n1-u-1) to (w-1);
\draw[loopentryred] (w-1n) to ($(w-1n) + ({-0.15*\mylen pt},{-0.3*\mylen pt})$);
\draw[loopentryred] (w-1n) to ($(w-1n) + ({0.15*\mylen pt},{-0.3*\mylen pt})$);

\draw[onebodyredrtcalert,out=180,in=195,distance={1.1*\mylen pt},shorten >= 6pt] (w-1) to node[pos=0.5,sloped]{$\pmb{/}$} (w-1-bar);

% (w-2-bar) outgoing
% 
\path (w-2-bar) ++ ({0.4*\mylen pt},{0.1*\mylen pt}) node{$\bvertbari{2}$};
\draw[onebodyredcol,out=45,in=0,distance={1*\mylen pt}] (w-2-bar) to (v);
\draw[onebodyredcol,out=10,in=-40,distance={0.5*\mylen pt}] (w-2n) to (w-2-bar);

\path (w-2-bar) ++ ({-0.25*\mylen pt},{-0.4*\mylen pt}) node[draw,thick,circle,minimum width=2.5pt,fill,inner sep=0pt,outer sep=2pt](v-2-v-12){};
  \draw[loopentryred] (w-2-bar) to (v-2-v-12);
  \draw[oneactbodyredrtc,distance={0.25*\mylen pt},out=210,in=160] (v-2-v-12) to (w-2n);

% w-2n outgoing
% 
\draw[-,dotted,thick,shorten <={0.3*\mylen pt},shorten >={0.125*\mylen pt}] (w-2n) to (w-21);
\draw[loopentryred] (w-2n) to ($(w-2n) + ({-0.15*\mylen pt},{-0.3*\mylen pt})$);
\draw[loopentryred] (w-2n) to ($(w-2n) + ({0.15*\mylen pt},{-0.3*\mylen pt})$);

\draw[convonebodyredcol,out=-30,in=90,distance={0.15*\mylen pt}] (w-2n) to ($(w-2n) + ({0.4*\mylen pt},{-0.3*\mylen pt})$);

% (w-21) ingoing/outgoing    
%
\draw[onebodyredcol,out=20,in=270,distance={0.25*\mylen pt}] (w-21) to ($(w-21) + ({0.35*\mylen pt},{0.4*\mylen pt})$);    

\path (w-21) ++ ({-0.25*\mylen pt},{-0.4*\mylen pt}) node[draw,thick,circle,minimum width=2.5pt,fill,inner sep=0pt,outer sep=2pt](v-n2-u-2){};
  \draw[oneactbodyredrtc,distance={0.25*\mylen pt},out=210,in=160] (v-n2-u-2) to (w-2);     
  \draw[loopentryred] (w-21) to (v-n2-u-2);

% (w-2) outgoing    
%
\path (w-2) ++ ({0.075*\mylen pt},{-0.25*\mylen pt}) node{$\bverti{2}$};
\draw[onebodyredcol,out=10,in=-40,distance={0.5*\mylen pt}] (w-2) to (w-21);

\draw[-,very thick,bend right,magenta,densely dashed,distance={0.75*\mylen pt}] 
  (w-1) to node[pos=0.5](mid){} 
           node[pos=0.65](left){} (w-2);

%---------
% R3.4
%--------- 
\matrix[anchor=center,ampersand replacement=\&,%
        row sep=0.8cm,column sep=1.15cm,every node/.style={draw,thick,circle,minimum width=2.5pt,fill,inner sep=0pt,outer sep=2pt}] at (5.25,-10.25) {
                    \&  \node(v){};
  \\[0.25cm]
  \node(w-1-bar){}; \&  \node[draw=none,fill=none](h0){};
                                     \&  \node(w-2-bar){};  
  \\
  \node(w-1n){};    \&               \&  \node(w-2n){};
  \\[-0.275cm]
                    \& \node[draw=none,fill=none](label-R34){};
  \\[-0.275cm]
  \node(w-11){};    \&               \&  \node(w-21){};
  \\
  \node(w-1){};     \&  \node[draw=none,fill=none](h){}; 
                                     \&  \node(w-2){};
  \\
};
\calcLength(v,h0){mylen};

\path (label-R34) ++ (0pt,-5pt) node{\Large \crtcrossreflabel{\mbox{\nf (R3.4)}}[R3.4]};
\path ($(label-R31)!0.495!(label-R34)$) ++ (0pt,10pt) node{\Large \crtcrossreflabel{\mbox{\nf (R3)}}[R3]};

% (v) outgoing
%
\path (v) ++ (0pt,{0.25*\mylen pt}) node{$\avert$};
  \path (v) ++ ({-0.5*\mylen pt},{-0.4*\mylen pt}) node[draw,thick,circle,minimum width=2.5pt,fill,inner sep=0pt,outer sep=2pt] (v-01){};
    \draw[loopentryred] (v) to (v-01);
    \draw[oneactbodyredrtc,bend right,distance={0.35*\mylen pt}] (v-01) to (w-1-bar);
  \path (v) ++ ({0.5*\mylen pt},{-0.4*\mylen pt}) node[draw,thick,circle,minimum width=2.5pt,fill,inner sep=0pt,outer sep=2pt] (v-03){}; 
    \draw[loopentryred] (v) to (v-03);
    \draw[oneactbodyredrtc,bend left,distance={0.35*\mylen pt}] (v-03) to (w-2-bar);
   
\draw[oneactbodyredrtcalert,shorten <= {0.1 *\mylen pt},shorten >={0.1 *\mylen pt},out=155,in=25] 
  (w-2-bar) to node[pos=0.5,sloped]{$\pmb{/}$} (w-1-bar); 
   
% (w-1-bar) outgoing   
% 
\draw[oneactbodyredrtccol,out=135,in=180,distance={1*\mylen pt}] (w-1-bar) to (v);  
\draw[onebodyredalert,out=145,in=150,distance={1.2*\mylen pt},shorten >= 5pt] (w-1-bar) to node[pos=0.5,sloped]{$\pmb{/}$} (v);     
%\draw[noonebodyredcol,out=135,in=180,distance={1*\mylen pt}] (w-1-bar) to (v);  

\draw[oneactbodyredrtccol,shorten <= {0.1 *\mylen pt},shorten >={0.1 *\mylen pt}] 
  (w-1-bar) to (w-2-bar);

\path (w-1-bar) ++ ({-0.4*\mylen pt},{0.1*\mylen pt}) node{$\bvertbari{1}$};

\path (w-1-bar) ++ ({0.25*\mylen pt},{-0.4*\mylen pt}) node[draw,thick,circle,minimum width=2.5pt,fill,inner sep=0pt,outer sep=2pt](v-1-v-11){};
  \draw[loopentryred] (w-1-bar) to (v-1-v-11);
  \draw[oneactbodyredrtc,distance={0.25*\mylen pt},out=-30,in=20] (v-1-v-11) to (w-1n);

% (w-1n) outgoing
% 
\draw[-,dotted,thick,shorten <={0.3*\mylen pt},shorten >={0.125*\mylen pt}] (w-1n) to (w-11);

\draw[onebodyredcol,out=170,in=220,distance={0.5*\mylen pt}] (w-1n) to (w-1-bar);

% (w-1n) incoming
%
\draw[convonebodyredcol,out=210,in=90,distance={0.15*\mylen pt}] (w-1n) to ($(w-1n) + ({-0.4*\mylen pt},{-0.3*\mylen pt})$);

% (w-11) 
%
\draw[onebodyredcol,out=160,in=270,distance={0.25*\mylen pt}] (w-11) to ($(w-11) + ({-0.35*\mylen pt},{0.4*\mylen pt})$);

% (w-1) 
%
\path (w-1) ++ ({0*\mylen pt},{-0.25*\mylen pt}) node{$\bverti{1}$};
\draw[onebodyredcol,out=170,in=220,distance={0.5*\mylen pt}] (w-1) to (w-11);

\path (w-11) ++ ({0.25*\mylen pt},{-0.4*\mylen pt}) node[draw,thick,circle,minimum width=2.5pt,fill,inner sep=0pt,outer sep=2pt](v-n1-u-1){};
  \draw[loopentryred] (w-11) to (v-n1-u-1);
  \draw[oneactbodyredrtc,distance={0.25*\mylen pt},out=-30,in=20] (v-n1-u-1) to (w-1);
\draw[loopentryred] (w-1n) to ($(w-1n) + ({-0.15*\mylen pt},{-0.3*\mylen pt})$);
\draw[loopentryred] (w-1n) to ($(w-1n) + ({0.15*\mylen pt},{-0.3*\mylen pt})$);

% (w-2-bar) outgoing
% 
\path (w-2-bar) ++ ({0.4*\mylen pt},{0.1*\mylen pt}) node{$\bvertbari{2}$};
\draw[onebodyredcol,out=45,in=0,distance={1*\mylen pt}] (w-2-bar) to (v);
\draw[onebodyredcol,out=10,in=-40,distance={0.5*\mylen pt}] (w-2n) to (w-2-bar);

\path (w-2-bar) ++ ({-0.25*\mylen pt},{-0.4*\mylen pt}) node[draw,thick,circle,minimum width=2.5pt,fill,inner sep=0pt,outer sep=2pt](v-2-v-12){};
  \draw[loopentryred] (w-2-bar) to (v-2-v-12);
  \draw[oneactbodyredrtc,distance={0.25*\mylen pt},out=210,in=160] (v-2-v-12) to (w-2n);

% w-2n outgoing
% 
\draw[-,dotted,thick,shorten <={0.3*\mylen pt},shorten >={0.125*\mylen pt}] (w-2n) to (w-21);
\draw[loopentryred] (w-2n) to ($(w-2n) + ({-0.15*\mylen pt},{-0.3*\mylen pt})$);
\draw[loopentryred] (w-2n) to ($(w-2n) + ({0.15*\mylen pt},{-0.3*\mylen pt})$);

\draw[convonebodyredcol,out=-30,in=90,distance={0.15*\mylen pt}] (w-2n) to ($(w-2n) + ({0.4*\mylen pt},{-0.3*\mylen pt})$);

% (w-21) ingoing/outgoing    
%
\draw[onebodyredcol,out=20,in=270,distance={0.25*\mylen pt}] (w-21) to ($(w-21) + ({0.35*\mylen pt},{0.4*\mylen pt})$);    

\path (w-21) ++ ({-0.25*\mylen pt},{-0.4*\mylen pt}) node[draw,thick,circle,minimum width=2.5pt,fill,inner sep=0pt,outer sep=2pt](v-n2-u-2){};
  \draw[oneactbodyredrtc,distance={0.25*\mylen pt},out=210,in=160] (v-n2-u-2) to (w-2);     
  \draw[loopentryred] (w-21) to (v-n2-u-2);

% (w-2) outgoing    
%
\path (w-2) ++ ({0.075*\mylen pt},{-0.25*\mylen pt}) node{$\bverti{2}$};
\draw[onebodyredcol,out=10,in=-40,distance={0.5*\mylen pt}] (w-2) to (w-21);

\draw[-,very thick,bend right,magenta,densely dashed,distance={0.75*\mylen pt}] 
  (w-1) to node[pos=0.5](mid){} 
           node[pos=0.65](left){} (w-2);

\end{tikzpicture}
  }
\end{center}
  \vspace*{-1.5ex}
  \caption{\label{fig:ill:pos:R123}%
           Reduced \protect\onebisimilarity\ redundancies, see Lem.~\ref{lem:reducing:1brs}. % of kind \ref{R1}, \ref{R2}, and \ref{R3}. 
           \protect\darkcyan{Upward dotted arrows}: body \protect\onetransition\ backlinks.
           \protect\darkcyan{Upward dashed double arrows}: \protect\bodytransition\ paths of (direct) \protect\txtloopsbackto\ links.
           \protect\darkcyan{Dashed double arrows}: paths of \protect\bodytransition{s}. 
           \protect\alert{Struck out red arrows}: prohibited body-transitions and body-tr.-paths.
           \protect\magenta{Dashed links, bottom}:~assumed~\protect\onebisimilarity. % of $\protect\bverti{1}$ and $\protect\bverti{2}$. 
           }
\end{figure}
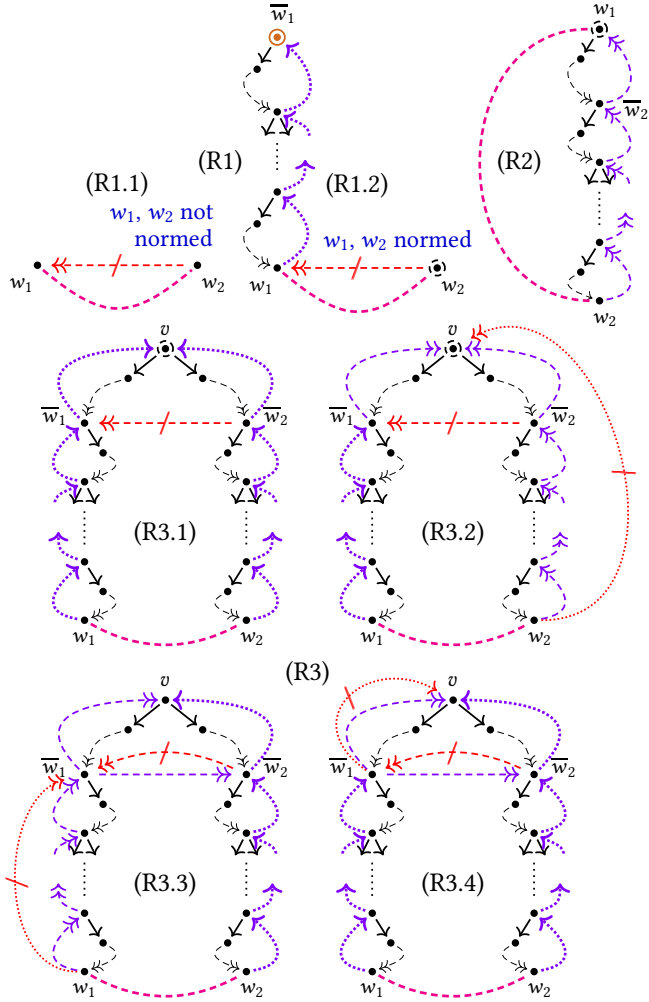  

%% file: figs/fig-ill-pos-R3-41-42.tex
\begin{figure}[!t]
\begin{center}
  \scalebox{0.85}{%
\begin{tikzpicture}[scale=1,every node/.style={transform shape}]
\tikzset{Rightarrow/.style={darkcyan,double distance=2pt,>={Implies},->},
         DashedRightarrow/.style={darkcyan,densely dashed,double distance=2pt,>={Implies},->},
         strike through/.append style={decoration={markings, mark=at position 0.5 with {\draw[-] ++ (-2pt,-2pt) -- (2pt,2pt);}},postaction={decorate}},
  loopentryred/.style={->,thick},        
  onebodyred/.style={->,very thick,densely dotted},
  onebodyredcol/.style={->,very thick,densely dotted,darkcyan},
  onebodyredalert/.style={->,thick,densely dotted,red},
  onebodyredrtc/.style={->>,very thick,densely dotted},
  onebodyredrtccol/.style={->>,very thick,densely dotted,darkcyan},
  onebodyredrtcalert/.style={->>,thick,densely dotted,red},
  % notoneactbodyredrtc/.style={->>,thick,densely dashed,red},
  % notonebodyredcol/.style={->,
  %                          decoration={markings,mark=at position 0.5 with \node{X};}
  %                          % decoration={markings,mark=at position 0.5 with 
  %                          %              {\draw [red,thick,-] 
  %                          %                 ++ ( 0,-\StrikeThruDistance) 
  %                          %                 -- ( 0, \StrikeThruDistance);}},
  %                          postaction={decorate},
  %                          thick,densely dotted,red},
  % notonebodyredrtccol/.style={->>,
  %                              decoration={markings,mark=at position 0.5 with 
  %                                           {\draw [red,thick,-] 
  %                                              ++ ( 0, -\StrikeThruDistance) 
  %                                              -- ( 0, \StrikeThruDistance);}},
  %                              postaction={decorate},
  %                              thick,densely dotted,red},
  convonebodyred/.style={<-,very thick,densely dotted},
  convonebodyredcol/.style={<-,very thick,densely dotted,darkcyan},
  oneactbodyred/.style={->,densely dashed},
  oneactbodyredcol/.style={->,thick,densely dashed,darkcyan},
  oneactbodyredtc/.style={-,densely dashed,thick,preaction={draw,DashedRightarrow}},
  oneactbodyredtccol/.style={-,densely dashed,thick,darkcyan,preaction={draw,DashedRightarrow}},
  oneactbodyredrtc/.style={->>,densely dashed},
  oneactbodyredrtccol/.style={->>,thick,densely dashed,darkcyan},
  oneactbodyredrtcalert/.style={->>,thick,densely dashed,red},
  convoneactbodyredrtc/.style={<<-,densely dashed},
  convoneactbodyredrtccol/.style={<<-,thick,densely dashed,darkcyan},
  noonebodyred/.style={-,densely dashed,thick,preaction={draw,Rightarrow}},
  noonebodyredcol/.style={-,thick,red,preaction={draw,Rightarrow}},%{-,densely dashed,thick,darkcyan,preaction={draw,Rightarrow}},
    % triple/.style={-,preaction={draw,Rightarrow}},
    % quadruple/.style={preaction={draw,Rightarrow,shorten >=0pt},shorten >=1pt,-,double,doubledistance=0.2pt}
  strike thru arrow/.style={
    decoration={markings, mark=at position 0.5 with {
        \draw [blue, thick,-] 
            ++ (-\StrikeThruDistance,-\StrikeThruDistance) 
            -- ( \StrikeThruDistance, \StrikeThruDistance);}
    },
    postaction={decorate},
}    
}

%---------
% R3.4.1
%--------- 
\matrix[anchor=center,ampersand replacement=\&,%
        row sep=0.8cm,column sep=0.8cm,every node/.style={draw,thick,circle,minimum width=2.5pt,fill,inner sep=0pt,outer sep=2pt}] at (0,0) {
                    \&  \node(v){};
  \\[0.25cm]
  \node(w-1-bar){}; \&  \node[draw=none,fill=none](h0){};
                                     \&  \node(w-2-bar){};  
  \\
  \node(w-1n){};    \&               \&  \node(w-2n){};
  \\[-0.275cm]
                    \& \node[draw=none,fill=none](label){};
  \\[-0.275cm]
  \node(w-11){};    \&               \&  \node(w-21){};
  \\
  \node(w-1){};     \&  \node[draw=none,fill=none](h){}; 
                                     \&  \node(w-2){};
  \\
};
\calcLength(v,h0){mylen};

\path (label) ++ (0pt,+9pt) node{\small \mbox{\nf (R3.4)}};
\path (label) ++ (0pt,-3pt)  node{\large \crtcrossreflabel{\mbox{\nf (R3.4.1)}}[R3.4.1]};

% (v) outgoing
%
\path (v) ++ (0pt,{0.25*\mylen pt}) node{$\avert$};
  \path (w-1-bar) ++ ({-0.25*\mylen pt},{-0.5*\mylen pt}) node[draw,thick,circle,minimum width=2.5pt,fill,inner sep=0pt,outer sep=2pt] (v-01){};
    \draw[loopentryred,out=240,in=30,distance={0.35*\mylen pt}] (v) to (v-01);
    \draw[oneactbodyredrtc,out=160,in=190,distance={0.5*\mylen pt}] (v-01) to (w-1-bar);
  \path (w-1) ++ ({-0.55*\mylen pt},{-0.4*\mylen pt}) node[draw,thick,circle,minimum width=2.5pt,fill,inner sep=0pt,outer sep=2pt] (v-03){}; 
    \draw[loopentryred,out=260,in=45,distance={1*\mylen pt}] (v) to (v-03);
    \draw[oneactbodyredrtc,out=135,in=190,distance={0.3*\mylen pt}] (v-03) to (w-1);
  \draw[->,thick,dashed,out=-80,in=-40,distance={0.7*\mylen pt},shorten >=2pt] (v) to (v);      
   
\draw[oneactbodyredrtcalert,shorten <= {0.1 *\mylen pt},shorten >={0.1 *\mylen pt},out=155,in=25] 
  (w-2-bar) to node[pos=0.5,sloped]{$\pmb{/}$} (w-1-bar); 
   
% (w-1-bar) outgoing   
% 
\draw[oneactbodyredrtccol,out=135,in=180,distance={1*\mylen pt}] (w-1-bar) to (v);  
\draw[onebodyredalert,out=145,in=150,distance={1.2*\mylen pt},shorten >= 5pt] (w-1-bar) to node[pos=0.5,sloped]{$\pmb{/}$} (v);     
%\draw[noonebodyredcol,out=135,in=180,distance={1*\mylen pt}] (w-1-bar) to (v);  

\draw[oneactbodyredrtccol,shorten <= {0.1 *\mylen pt},shorten >={0.1 *\mylen pt}] 
  (w-1-bar) to (w-2-bar);

\path (w-1-bar) ++ ({-0.4*\mylen pt},{0.1*\mylen pt}) node{$\bvertbari{1}$};

\path (w-1-bar) ++ ({0.25*\mylen pt},{-0.4*\mylen pt}) node[draw,thick,circle,minimum width=2.5pt,fill,inner sep=0pt,outer sep=2pt](v-1-v-11){};
  \draw[loopentryred] (w-1-bar) to (v-1-v-11);
  \draw[oneactbodyredrtc,distance={0.25*\mylen pt},out=-30,in=20] (v-1-v-11) to (w-1n);

% (w-1n) outgoing
% 
\draw[-,dotted,thick,shorten <={0.3*\mylen pt},shorten >={0.125*\mylen pt}] (w-1n) to (w-11);

\draw[onebodyredcol,out=170,in=220,distance={0.5*\mylen pt}] (w-1n) to (w-1-bar);

% (w-1n) incoming
%
\draw[convonebodyredcol,out=210,in=90,distance={0.15*\mylen pt}] (w-1n) to ($(w-1n) + ({-0.4*\mylen pt},{-0.3*\mylen pt})$);

% (w-11) 
%
\draw[onebodyredcol,out=160,in=270,distance={0.25*\mylen pt}] (w-11) to ($(w-11) + ({-0.35*\mylen pt},{0.4*\mylen pt})$);

% (w-1) 
%
\path (w-1) ++ ({0*\mylen pt},{-0.25*\mylen pt}) node{$\bverti{1}$};
\draw[onebodyredcol,out=170,in=220,distance={0.5*\mylen pt}] (w-1) to (w-11);

\path (w-11) ++ ({0.25*\mylen pt},{-0.4*\mylen pt}) node[draw,thick,circle,minimum width=2.5pt,fill,inner sep=0pt,outer sep=2pt](v-n1-u-1){};
  \draw[loopentryred] (w-11) to (v-n1-u-1);
  \draw[oneactbodyredrtc,distance={0.25*\mylen pt},out=-30,in=20] (v-n1-u-1) to (w-1);
\draw[loopentryred] (w-1n) to ($(w-1n) + ({-0.15*\mylen pt},{-0.3*\mylen pt})$);
\draw[loopentryred] (w-1n) to ($(w-1n) + ({0.15*\mylen pt},{-0.3*\mylen pt})$);

% (w-2-bar) outgoing
% 
\path (w-2-bar) ++ ({0.4*\mylen pt},{0.1*\mylen pt}) node{$\bvertbari{2}$};
\draw[onebodyredcol,out=45,in=0,distance={1*\mylen pt}] (w-2-bar) to (v);
\draw[onebodyredcol,out=10,in=-40,distance={0.5*\mylen pt}] (w-2n) to (w-2-bar);

\path (w-2-bar) ++ ({-0.25*\mylen pt},{-0.4*\mylen pt}) node[draw,thick,circle,minimum width=2.5pt,fill,inner sep=0pt,outer sep=2pt](v-2-v-12){};
  \draw[loopentryred] (w-2-bar) to (v-2-v-12);
  \draw[oneactbodyredrtc,distance={0.25*\mylen pt},out=210,in=160] (v-2-v-12) to (w-2n);

% w-2n outgoing
% 
\draw[-,dotted,thick,shorten <={0.3*\mylen pt},shorten >={0.125*\mylen pt}] (w-2n) to (w-21);
\draw[loopentryred] (w-2n) to ($(w-2n) + ({-0.15*\mylen pt},{-0.3*\mylen pt})$);
\draw[loopentryred] (w-2n) to ($(w-2n) + ({0.15*\mylen pt},{-0.3*\mylen pt})$);

\draw[convonebodyredcol,out=-30,in=90,distance={0.15*\mylen pt}] (w-2n) to ($(w-2n) + ({0.4*\mylen pt},{-0.3*\mylen pt})$);

% (w-21) ingoing/outgoing    
%
\draw[onebodyredcol,out=20,in=270,distance={0.25*\mylen pt}] (w-21) to ($(w-21) + ({0.35*\mylen pt},{0.4*\mylen pt})$);    

\path (w-21) ++ ({-0.25*\mylen pt},{-0.4*\mylen pt}) node[draw,thick,circle,minimum width=2.5pt,fill,inner sep=0pt,outer sep=2pt](v-n2-u-2){};
  \draw[oneactbodyredrtc,distance={0.25*\mylen pt},out=210,in=160] (v-n2-u-2) to (w-2);     
  \draw[loopentryred] (w-21) to (v-n2-u-2);

% (w-2) outgoing    
%
\path (w-2) ++ ({0.075*\mylen pt},{-0.25*\mylen pt}) node{$\bverti{2}$};
\draw[onebodyredcol,out=10,in=-40,distance={0.5*\mylen pt}] (w-2) to (w-21);

\draw[-,very thick,bend right,magenta,densely dashed,distance={0.65*\mylen pt}] 
  (w-1) to node[pos=0.5](mid){} 
           node[pos=0.65](left){} (w-2);

%---------
% R3.4.2
%--------- 
\matrix[anchor=center,ampersand replacement=\&,%
        row sep=0.8cm,column sep=0.6cm,every node/.style={draw,thick,circle,minimum width=2.5pt,fill,inner sep=0pt,outer sep=2pt}] at (4.85,0) {
                    \&  \node[draw=none,fill=none](h0){};           
                                     \&                     \& \node(v){};
  \\[0.65cm]
  \node(w-1-bar--1){}; \&  \node[draw=none,fill=none](h1){};
                                     \&  \node(w-2-bar--1){};  \&              \& \node(u-1-bar--2){}; \&                 \& \node(u-2-bar--2){};
  \\
  \node(w-1n--1){};    \&               \&  \node(w-2n--1){};  \&              \& \node(u-1n--2){};    \&                 \&  \node(u-2n--2){}; 
  \\[-0.275cm]
                    \& \node[draw=none,fill=none](label--1){}; \& \&  \node[draw=none,fill=none](label){};
                                                                               \&                      \& \node[draw=none,fill=none](label--2){};
  \\[-0.275cm]
  \node(w-11--1){};    \&               \&  \node(w-21--1){};  \&              \& \node(u-11--2){};    \&                 \&  \node(u-21--2){};
  \\
  \node(w-1--1){};     \&  \node[draw=none,fill=none](h){}; 
                                     \&  \node(w-2--1){};      \&              \& \node(u-1--2){};     \&                 \&  \node(u-2--2){};
  \\
};
%\calcLength(h0,h1){mylen};

\path (label--1) ++ (0pt,0.5pt) node{\mbox{\small\nf (R3.4)}};
\path (label) ++ (0pt,0.5pt) node{\large \crtcrossreflabel{\mbox{\nf (R3.4.2)}}[R3.4.2]};
\path (label--2) ++ (0pt,0.5pt) node{\mbox{\small\nf (R3.4)}};

%----------------
% left part R3.4
%----------------    

% (v) outgoing
%
\path (v) ++ (0pt,{0.25*\mylen pt}) node{$\avert$};
  \path (v) ++ ({-1.65*\mylen pt},{-0.6*\mylen pt}) node[draw,thick,circle,minimum width=2.5pt,fill,inner sep=0pt,outer sep=2pt] (v-01){};
    \draw[loopentryred] (v) to (v-01);
    \draw[oneactbodyredrtc,bend right,distance={0.35*\mylen pt}] (v-01) to (w-1-bar--1);
  \path (v) ++ ({-0.85*\mylen pt},{-0.6*\mylen pt}) node[draw,thick,circle,minimum width=2.5pt,fill,inner sep=0pt,outer sep=2pt] (v-03){}; 
    \draw[loopentryred] (v) to (v-03);
    \draw[oneactbodyredrtc,bend left,distance={0.35*\mylen pt}] (v-03) to (w-2-bar--1);
   
\draw[oneactbodyredrtcalert,shorten <= {0.1 *\mylen pt},shorten >={0.1 *\mylen pt},out=155,in=25] 
  (w-2-bar--1) to node[pos=0.5,sloped]{$\pmb{/}$} (w-1-bar--1);

% (w-1-bar--1) outgoing   
% 
\draw[oneactbodyredrtccol,out=135,in=180,distance={1*\mylen pt}] (w-1-bar--1) to (v);  
\draw[onebodyredalert,out=145,in=150,distance={1.2*\mylen pt},shorten >= 5pt] (w-1-bar--1) to node[pos=0.5,sloped]{$\pmb{/}$} (v);     
%\draw[noonebodyredcol,out=135,in=180,distance={1*\mylen pt}] (w-1-bar) to (v);  

\draw[oneactbodyredrtccol,shorten <= {0.1 *\mylen pt},shorten >={0.1 *\mylen pt}] 
  (w-1-bar--1) to (w-2-bar--1);

\path (w-1-bar--1) ++ ({-0.4*\mylen pt},{0.1*\mylen pt}) node{$\bvertbari{1}$};

\path (w-1-bar--1) ++ ({0.25*\mylen pt},{-0.4*\mylen pt}) node[draw,thick,circle,minimum width=2.5pt,fill,inner sep=0pt,outer sep=2pt](v-1-v-11--1){};
  \draw[loopentryred] (w-1-bar--1) to (v-1-v-11--1);
  \draw[oneactbodyredrtc,distance={0.25*\mylen pt},out=-30,in=20] (v-1-v-11--1) to (w-1n--1);

% (w-1n--1) outgoing
% 
\draw[-,dotted,thick,shorten <={0.3*\mylen pt},shorten >={0.125*\mylen pt}] (w-1n--1) to (w-11--1);

\draw[onebodyredcol,out=170,in=220,distance={0.5*\mylen pt}] (w-1n--1) to (w-1-bar--1);

% (w-1n--1) incoming
%
\draw[convonebodyredcol,out=210,in=90,distance={0.15*\mylen pt}] (w-1n--1) to ($(w-1n--1) + ({-0.4*\mylen pt},{-0.3*\mylen pt})$);

% (w-11--1) 
%
\draw[onebodyredcol,out=160,in=270,distance={0.25*\mylen pt}] (w-11--1) to ($(w-11--1) + ({-0.35*\mylen pt},{0.4*\mylen pt})$);

% (w-1--1) 
%
\path (w-1--1) ++ ({0*\mylen pt},{-0.25*\mylen pt}) node{$\bverti{1}$};
\draw[onebodyredcol,out=170,in=220,distance={0.5*\mylen pt}] (w-1--1) to (w-11--1);

\path (w-11--1) ++ ({0.25*\mylen pt},{-0.4*\mylen pt}) node[draw,thick,circle,minimum width=2.5pt,fill,inner sep=0pt,outer sep=2pt](v-n1-u-1--1){};
  \draw[loopentryred] (w-11--1) to (v-n1-u-1--1);
  \draw[oneactbodyredrtc,distance={0.25*\mylen pt},out=-30,in=20] (v-n1-u-1--1) to (w-1--1);
\draw[loopentryred] (w-1n--1) to ($(w-1n--1) + ({-0.15*\mylen pt},{-0.3*\mylen pt})$);
\draw[loopentryred] (w-1n--1) to ($(w-1n--1) + ({0.15*\mylen pt},{-0.3*\mylen pt})$);

% (w-2-bar--1) outgoing
% 
\path (w-2-bar--1) ++ ({0.4*\mylen pt},{0.1*\mylen pt}) node{$\bvertbari{2}$};
\draw[onebodyredcol,out=45,in=-110,distance={0.5*\mylen pt}] (w-2-bar--1) to (v);
\draw[onebodyredcol,out=10,in=-40,distance={0.5*\mylen pt}] (w-2n--1) to (w-2-bar--1);

\path (w-2-bar--1) ++ ({-0.25*\mylen pt},{-0.4*\mylen pt}) node[draw,thick,circle,minimum width=2.5pt,fill,inner sep=0pt,outer sep=2pt](v-2-v-12--1){};
  \draw[loopentryred] (w-2-bar--1) to (v-2-v-12--1);
  \draw[oneactbodyredrtc,distance={0.25*\mylen pt},out=210,in=160] (v-2-v-12--1) to (w-2n--1);

% (w-2n--1) outgoing
% 
\draw[-,dotted,thick,shorten <={0.3*\mylen pt},shorten >={0.125*\mylen pt}] (w-2n--1) to (w-21--1);
\draw[loopentryred] (w-2n--1) to ($(w-2n--1) + ({-0.15*\mylen pt},{-0.3*\mylen pt})$);
\draw[loopentryred] (w-2n--1) to ($(w-2n--1) + ({0.15*\mylen pt},{-0.3*\mylen pt})$);

\draw[convonebodyredcol,out=-30,in=90,distance={0.15*\mylen pt}] (w-2n--1) to ($(w-2n--1) + ({0.4*\mylen pt},{-0.3*\mylen pt})$);

% (w-21--1) ingoing/outgoing    
%
\draw[onebodyredcol,out=20,in=270,distance={0.25*\mylen pt}] (w-21--1) to ($(w-21--1) + ({0.35*\mylen pt},{0.4*\mylen pt})$);    

\path (w-21--1) ++ ({-0.25*\mylen pt},{-0.4*\mylen pt}) node[draw,thick,circle,minimum width=2.5pt,fill,inner sep=0pt,outer sep=2pt](v-n2-u-2--1){};
  \draw[oneactbodyredrtc,distance={0.25*\mylen pt},out=210,in=160] (v-n2-u-2--1) to (w-2--1);     
  \draw[loopentryred] (w-21--1) to (v-n2-u-2--1);

% (w-2--1) outgoing    
%
\path (w-2--1) ++ ({0.075*\mylen pt},{-0.25*\mylen pt}) node{$\bverti{2}$};
\draw[onebodyredcol,out=10,in=-40,distance={0.5*\mylen pt}] (w-2--1) to (w-21--1);

\draw[-,very thick,bend right,magenta,densely dashed,distance={0.4*\mylen pt}] 
  (w-1--1) to node[pos=0.5](mid){} 
           node[pos=0.65](left){} (w-2--1);

%----------------
% right part R3.4
%----------------    
          
% (v) outgoing
%
%\path (v) ++ (0pt,{0.25*\mylen pt}) node{$\avert$};
  %
  \path (v) ++ ({0.85*\mylen pt},{-0.6*\mylen pt}) node[draw,thick,circle,minimum width=2.5pt,fill,inner sep=0pt,outer sep=2pt] (v-01--2){};
    \draw[loopentryred] (v) to (v-01--2);
    \draw[oneactbodyredrtc,bend right,distance={0.35*\mylen pt}] (v-01--2) to (u-1-bar--2);
  \path (v) ++ ({1.65*\mylen pt},{-0.6*\mylen pt}) node[draw,thick,circle,minimum width=2.5pt,fill,inner sep=0pt,outer sep=2pt] (v-03--2){}; 
    \draw[loopentryred] (v) to (v-03--2);
    \draw[oneactbodyredrtc,bend left,distance={0.35*\mylen pt}] (v-03--2) to (u-2-bar--2);

% (u-1-bar--2) outgoing   
% 
\draw[oneactbodyredrtccol,out=135,in=-70,distance={0.5*\mylen pt}] (u-1-bar--2) to (v);  
\draw[onebodyredalert,out=45,in=30,distance={1.25*\mylen pt},shorten >= 5pt] (u-1-bar--2) to node[pos=0.5,sloped]{$\pmb{/}$} (v);     
%\draw[noonebodyredcol,out=135,in=180,distance={1*\mylen pt}] (u-1-bar) to (v);  
\draw[oneactbodyredrtccol,shorten <= {0.1 *\mylen pt},shorten >={0.1 *\mylen pt}] 
  (u-1-bar--2) to (u-2-bar--2);
  
\draw[oneactbodyredrtcalert,shorten <= {0.1 *\mylen pt},shorten >={0.1 *\mylen pt},out=140,in=60,distance={0.85*\mylen pt}] 
  (u-1-bar--2) to node[pos=0.45,sloped]{$\pmb{/}$} (w-1-bar--1);

\path (u-1-bar--2) ++ ({-0.325*\mylen pt},{0.1*\mylen pt}) node{$\cvertbari{1}$};

\path (u-1-bar--2) ++ ({0.25*\mylen pt},{-0.4*\mylen pt}) node[draw,thick,circle,minimum width=2.5pt,fill,inner sep=0pt,outer sep=2pt](v-1-v-11--2){};
  \draw[loopentryred] (u-1-bar--2) to (v-1-v-11--2);
  \draw[oneactbodyredrtc,distance={0.25*\mylen pt},out=-30,in=20] (v-1-v-11--2) to (u-1n--2);

% (u-1n--2) outgoing
% 
\draw[-,dotted,thick,shorten <={0.3*\mylen pt},shorten >={0.125*\mylen pt}] (u-1n--2) to (u-11--2);

\draw[onebodyredcol,out=170,in=220,distance={0.5*\mylen pt}] (u-1n--2) to (u-1-bar--2);

% (u-1n--2) incoming
%
\draw[convonebodyredcol,out=210,in=90,distance={0.15*\mylen pt}] (u-1n--2) to ($(u-1n--2) + ({-0.4*\mylen pt},{-0.3*\mylen pt})$);

% (u-11--2) 
%
\draw[onebodyredcol,out=160,in=270,distance={0.25*\mylen pt}] (u-11--2) to ($(u-11--2) + ({-0.35*\mylen pt},{0.4*\mylen pt})$);

% (u-1--2) 
%
\path (u-1--2) ++ ({0*\mylen pt},{-0.25*\mylen pt}) node{$\cverti{1}$};
\draw[onebodyredcol,out=170,in=220,distance={0.5*\mylen pt}] (u-1--2) to (u-11--2);

\path (u-11--2) ++ ({0.25*\mylen pt},{-0.4*\mylen pt}) node[draw,thick,circle,minimum width=2.5pt,fill,inner sep=0pt,outer sep=2pt](v-n1-u-1--2){};
  \draw[loopentryred] (u-11--2) to (v-n1-u-1--2);
  \draw[oneactbodyredrtc,distance={0.25*\mylen pt},out=-30,in=20] (v-n1-u-1--2) to (u-1--2);
\draw[loopentryred] (u-1n--2) to ($(u-1n--2) + ({-0.15*\mylen pt},{-0.3*\mylen pt})$);
\draw[loopentryred] (u-1n--2) to ($(u-1n--2) + ({0.15*\mylen pt},{-0.3*\mylen pt})$);

% (u-2-bar--2) outgoing
% 
\path (u-2-bar--2) ++ ({0.4*\mylen pt},{0.1*\mylen pt}) node{$\cvertbari{2}$};
\draw[onebodyredcol,out=45,in=0,distance={1*\mylen pt}] (u-2-bar--2) to (v);
\draw[onebodyredcol,out=10,in=-40,distance={0.5*\mylen pt}] (u-2n--2) to (u-2-bar--2);

\path (u-2-bar--2) ++ ({-0.25*\mylen pt},{-0.4*\mylen pt}) node[draw,thick,circle,minimum width=2.5pt,fill,inner sep=0pt,outer sep=2pt](v-2-v-12--2){};
  \draw[loopentryred] (u-2-bar--2) to (v-2-v-12--2);
  \draw[oneactbodyredrtc,distance={0.25*\mylen pt},out=210,in=160] (v-2-v-12--2) to (u-2n--2);

\draw[oneactbodyredrtcalert,shorten <= {0.1 *\mylen pt},shorten >={0.1 *\mylen pt},out=155,in=25] 
  (u-2-bar--2) to node[pos=0.5,sloped]{$\pmb{/}$} (u-1-bar--2);

% (u-2n--2) outgoing
% 
\draw[-,dotted,thick,shorten <={0.3*\mylen pt},shorten >={0.125*\mylen pt}] (u-2n--2) to (u-21--2);
\draw[loopentryred] (u-2n--2) to ($(u-2n--2) + ({-0.15*\mylen pt},{-0.3*\mylen pt})$);
\draw[loopentryred] (u-2n--2) to ($(u-2n--2) + ({0.15*\mylen pt},{-0.3*\mylen pt})$);

\draw[convonebodyredcol,out=-30,in=90,distance={0.15*\mylen pt}] (u-2n--2) to ($(u-2n--2) + ({0.4*\mylen pt},{-0.3*\mylen pt})$);

% (u-21--2) ingoing/outgoing    
%
\draw[onebodyredcol,out=20,in=270,distance={0.25*\mylen pt}] (u-21--2) to ($(u-21--2) + ({0.35*\mylen pt},{0.4*\mylen pt})$);    

\path (u-21--2) ++ ({-0.25*\mylen pt},{-0.4*\mylen pt}) node[draw,thick,circle,minimum width=2.5pt,fill,inner sep=0pt,outer sep=2pt](v-n2-u-2--2){};
  \draw[oneactbodyredrtc,distance={0.25*\mylen pt},out=210,in=160] (v-n2-u-2--2) to (u-2--2);     
  \draw[loopentryred] (u-21--2) to (v-n2-u-2--2);

% (u-2--2) outgoing    
%
\path (u-2--2) ++ ({0.075*\mylen pt},{-0.25*\mylen pt}) node{$\cverti{2}$};
\draw[onebodyredcol,out=10,in=-40,distance={0.5*\mylen pt}] (u-2--2) to (u-21--2);

\draw[-,very thick,bend right,magenta,densely dashed,distance={0.4*\mylen pt}] 
  (u-1--2) to node[pos=0.5](mid){} 
           node[pos=0.65](left){} (u-2--2);

\end{tikzpicture}
  }
\end{center}
  \vspace*{-2ex}
  \caption{\label{fig:ill:pos:R3:41:42}%
           A \protect\precrystalline\ \protect$\sone$-bisimil.\ redundancy $\protect\pair{\bverti{1}}{\bverti{2}}$ is crystalline
           if it is neither of form \protect\ref{R3.4.1} nor~of~form~\protect\ref{R3.4.2}.
           }
\end{figure}

%% file: figs/fig-ex-elim-R-3-2.tex
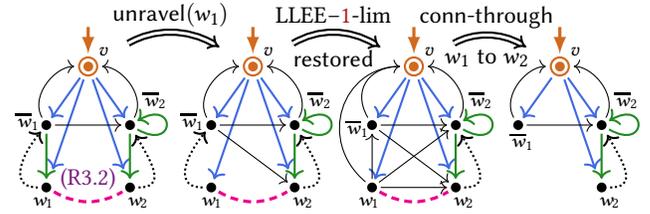
\begin{figure}[t!]
\begin{center}
\begin{tikzpicture}
  
  \matrix[anchor=center,row sep=0.55cm,column sep=0.425cm,every node/.style={draw,very thick,circle,minimum width=2.5pt,fill,inner sep=0pt,outer sep=2pt}] at (0,0) {
       & \node[chocolate](v--1){}; & &[0.1cm] & & \node[chocolate](v--2){}; & &[0.05cm] & & \node[chocolate](v--3){}; & &[-0.15cm] & & \node[chocolate](v--4){};  
    \\[0.1cm]
    \node(w1-bar--1){}; & & \node(w2-bar--1){}; & & \node(w1-bar--2){}; & & \node(w2-bar--2){}; & &  
    \node(w1-bar--3){}; & & \node(w2-bar--3){}; & & \node(w1-bar--4){}; & & \node(w2-bar--4){};  
    \\[0.15cm]  
    \node(w1--1){}; & & \node(w2--1){}; & & \node(w1--2){}; & & \node(w2--2){}; & &  
    \node(w1--3){}; & & \node(w2--3){}; & &                 & & \node(w2--4){}; 
    \\
    };
  \draw[thick,chocolate] (v--1) circle (0.12cm);
  \draw[thick,chocolate] (v--2) circle (0.12cm);
  \draw[thick,chocolate] (v--3) circle (0.12cm);
  \draw[thick,chocolate] (v--4) circle (0.12cm);
  \draw[<-,very thick,>=latex,chocolate,shorten <=1.5pt](v--1) -- ++ (90:0.525cm);
  \draw[<-,very thick,>=latex,chocolate,shorten <=1.5pt](v--2) -- ++ (90:0.525cm);
  \draw[<-,very thick,>=latex,chocolate,shorten <=1.5pt](v--3) -- ++ (90:0.525cm);
  \draw[<-,very thick,>=latex,chocolate,shorten <=1.5pt](v--4) -- ++ (90:0.525cm);
  %
  %--------
  % chart 1
  %--------
  %
  % \path (v--1) ++ (-0.7cm,0.275cm) node{\Large $\bonechart$};
  %
  \path (v--1) ++ (0.2cm,0.2cm) node{$\scriptstyle\avert$};
  \path (w1-bar--1) ++ (-0.25cm,0.05cm) node{$\scriptstyle\bvertbari{1}$};
  \path (w1--1) ++ (-0.05cm,-0.2cm) node{$\scriptstyle\bverti{1}$};
  \path (w2-bar--1) ++ (0.33cm,0.33cm) node{$\scriptstyle\bvertbari{2}$};
  \path (w2--1) ++ (0.1cm,-0.2cm) node{$\scriptstyle\bverti{2}$};
  \draw[->,thick,royalblue,shorten <=2pt] (v--1) to (w1-bar--1);
  \draw[->,thick,royalblue,shorten <=2pt,shorten >= 4pt] (v--1) to (w1--1);
  \draw[->,thick,royalblue,shorten <=2pt] (v--1) to (w2-bar--1);
  \draw[->,thick,royalblue,shorten <=2pt,shorten >= 4pt] (v--1) to (w2--1);
  \draw[->,thick,forestgreen] (w1-bar--1) to (w1--1);
  \draw[->] (w1-bar--1) to (w2-bar--1);
  \draw[->,out=120,in=180,shorten >= 2pt] (w1-bar--1) to (v--1);
  \draw[->,thick,densely dotted,out=180,in=230,distance=0.4cm] (w1--1) to (w1-bar--1);
  \draw[->,thick,forestgreen] (w2-bar--1) to (w2--1);
  \draw[->,out=60,in=0,shorten >= 2pt] (w2-bar--1) to (v--1);
  \draw[->,thick  ,forestgreen,out=-30,in=30,distance=0.6cm] (w2-bar--1) to (w2-bar--1);
  \draw[->,thick,densely dotted,out=0,in=-60,shorten >=1pt,distance=0.4cm] (w2--1) to (w2-bar--1);
  \draw[-,very thick,bend right,magenta,densely dashed,distance=] (w1--1) to (w2--1);

  % label R3.2
  %
  \path ($(w1-bar--1)!0.5!(w2--1)$) ++ (0cm,-0.3cm) node{\small \ref{R3.2}};

  %--------
  % chart 2
  %--------
  %
  \path (v--2) ++ (0.2cm,0.2cm) node{$\scriptstyle\avert$};
  \path (w1-bar--2) ++ (-0.25cm,0.05cm) node{$\scriptstyle\bvertbari{1}$};
  \path (w1--2) ++ (0cm,-0.2cm) node{$\scriptstyle\bverti{1}$};
  \path (w2-bar--2) ++ (0.33cm,0.33cm) node{$\scriptstyle\bvertbari{2}$};
  \path (w2--2) ++ (0.1cm,-0.2cm) node{$\scriptstyle\bverti{2}$};
  \draw[->,thick,royalblue,shorten <=2pt] (v--2) to (w1-bar--2);
  \draw[->,thick,royalblue,shorten <=2pt] (v--2) to (w1--2);
  \draw[->,thick,royalblue,shorten <=2pt] (v--2) to (w2-bar--2);
  \draw[->,thick,royalblue,shorten <=2pt,shorten >= 4pt] (v--2) to (w2--2);
  \draw[->] (w1-bar--2) to (w2--2);
  \draw[->] (w1-bar--2) to (w2-bar--2);
  \draw[->,out=120,in=180,shorten >= 2pt] (w1-bar--2) to (v--2);
  \draw[->,thick,densely dotted,out=180,in=230,distance=0.4cm] (w1--2) to (w1-bar--2);
  \draw[->,thick,forestgreen] (w2-bar--2) to (w2--2);
  \draw[->,out=60,in=0,shorten >= 2pt] (w2-bar--2) to (v--2);
  \draw[->,thick,forestgreen,out=-30,in=30,distance=0.6cm] (w2-bar--2) to (w2-bar--2);
  \draw[->,thick,densely dotted,out=0,in=-60,shorten >= 1pt,distance=0.4cm] (w2--2) to (w2-bar--2);
  \draw[-,very thick,bend right,magenta,densely dashed,distance=] (w1--2) to (w2--2);

\draw[-implies,thick,double equal sign distance, bend left,distance=0.9cm,
               shorten <= 0.5cm,shorten >= 0.4cm,yshift=0.05cm
               ] (v--1) to node[above]{\small $\fap{\textsf{unravel}}{\bverti{1}}$} (v--2);

  %--------
  % chart 3
  %--------
  %
  \path (v--3) ++ (0.2cm,0.2cm) node{$\scriptstyle\avert$};
  \path (w1-bar--3) ++ (-0.2cm,-0.05cm) node{$\scriptstyle\bvertbari{1}$};
  \path (w1--3) ++ (-0.05cm,-0.2cm) node{$\scriptstyle\bverti{1}$};
  \path (w2-bar--3) ++ (0.33cm,0.33cm) node{$\scriptstyle\bvertbari{2}$};
  \path (w2--3) ++ (0.1cm,-0.2cm) node{$\scriptstyle\bverti{2}$};
  \draw[->,thick,royalblue,shorten <=2pt] (v--3) to (w1-bar--3);
  \draw[->,thick,royalblue,shorten <=2pt] (v--3) to (w1--3);
  \draw[->,thick,royalblue,shorten <=2pt] (v--3) to (w2-bar--3);
  \draw[->,thick,royalblue,shorten <=2pt,shorten >= 4pt] (v--3) to (w2--3);
  \draw[->] (w1-bar--3) to (w2--3);
  \draw[->] (w1-bar--3) to (w2-bar--3);
  \draw[->,out=120,in=180,shorten >= 2pt,distance=0.45cm] (w1-bar--3) to (v--3);
  \draw[->] (w1--3) to (w1-bar--3);
  \draw[->] (w1--3) to (w2-bar--3);
  \draw[->] (w1--3) to (w2--3);
  \draw[-,out=140,in=180,shorten >= 2pt,distance=0.85cm] (w1--3) to (v--3);
  \draw[->,thick,forestgreen] (w2-bar--3) to (w2--3);
  \draw[->,out=60,in=0,shorten >= 2pt] (w2-bar--3) to (v--3);
  \draw[->,thick,forestgreen,out=-30,in=30,distance=0.6cm] (w2-bar--3) to (w2-bar--3);
  \draw[->,thick,densely dotted,out=0,in=-60,shorten >= 1pt,distance=0.4cm] (w2--3) to (w2-bar--3);
  \draw[-,very thick,bend right,magenta,densely dashed,distance=] (w1--3) to (w2--3);

\draw[-implies,thick,double equal sign distance, bend left,distance=0.9cm,
               shorten <= 0.5cm,shorten >= 0.4cm
               ] (v--2) to node[below,yshift=-0.05cm]{\small $\textsf{restored}$}
                           node[above,yshift=0.05cm]{\small \LLEEonelim} (v--3);

  %--------
  % chart 4
  %--------
  %
  % \path (v--4) ++ (0.75cm,0.275cm) node{\Large $\bonechart'''$};
  % \path (v--4) ++ (-0.1cm,-1.45cm) node{\Large $\bonechart'''$};
  %
  \path (v--4) ++ (0.2cm,0.2cm) node{$\scriptstyle\avert$};
  \path (w1-bar--4) ++ (0.05cm,-0.225cm) node{$\scriptstyle\bvertbari{1}$};
  \path (w2-bar--4) ++ (0.33cm,0.33cm) node{$\scriptstyle\bvertbari{2}$};
  \path (w2--4) ++ (0.1cm,-0.2cm) node{$\scriptstyle\bverti{2}$};
  \draw[->,thick,royalblue,shorten <=2pt] (v--4) to (w1-bar--4);
  \draw[->,thick,royalblue,shorten <=2pt] (v--4) to (w2-bar--4);
  \draw[->,thick,royalblue,shorten <=2pt,shorten >= 4pt] (v--4) to (w2--4);
  \draw[->] (w1-bar--4) to (w2-bar--4);
  \draw[->,out=120,in=180,shorten >= 2pt] (w1-bar--4) to (v--4);
  \draw[->,thick,forestgreen] (w2-bar--4) to (w2--4);
  \draw[->,out=60,in=0,shorten >= 2pt] (w2-bar--4) to (v--4);
  \draw[->,thick,forestgreen,out=-30,in=30,distance=0.6cm] (w2-bar--4) to (w2-bar--4);
  \draw[->,thick,densely dotted,out=0,in=-60,shorten >= 1pt,distance=0.4cm] (w2--4) to (w2-bar--4);

\draw[-implies,thick,double equal sign distance, bend left,distance=0.8cm,
               shorten <= 0.5cm,shorten >= 0.4cm
               ] (v--3) to node[above,yshift=0.025cm]{\small $\textsf{conn-through}$}
                           node[below]{\small $\bverti{1} \textsf{ to } \bverti{2}$} (v--4);

\end{tikzpicture}
\end{center}  
  \vspace*{-2ex}
  \caption{\label{fig:ex:elim:R3.2}%
           Example for the \protect\LLEEpreserving\ elimination of a red.\ $\sone$\nb-bisim.\ red.\ of kind \protect\ref{R3.2}
             by redirecting transitions to \protect\onebisimilar\ targets. Here all proper action labels~are~the~same.
           }

\end{figure}

%% file: figs/fig-ex-pars-insulated-exp.tex
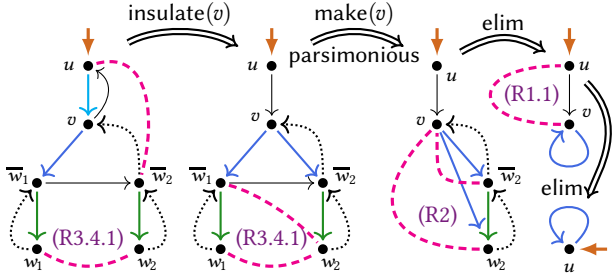
\begin{figure}[t!]
\begin{center}
\begin{tikzpicture}
   \matrix[anchor=center,row sep=0.55cm,column sep=0.55cm,every node/.style={draw,very thick,circle,minimum width=2.5pt,fill,inner sep=0pt,outer sep=2pt}] at (0,0) {
       & \node(u--1){};  & &[-0.15cm]   & & \node(u--2){}; &  &[-0.35cm] & & \node(u--3){}; &       &[-0.15cm] & \node(u--4){}; 
    \\[0.1cm]
       & \node(v--1){};  & &            & & \node(v--2){}; &         &   & & \node(v--3){}; &       & & \node(v--4){}; 
    \\[0.1cm]
    \node(w1bar--1){};  & & \node(w2bar--1){}; & & \node(w1bar--2){};   & & \node(w2bar--2){}; & &  \node[draw=none,fill=none](w1bar--3){}; & & \node(w2bar--3){}; 
    \\[0.2cm]  
    \node(w1--1){};      & & \node(w2--1){};     & & \node(w1--2){};       & & \node(w2--2){};     & &   \node[draw=none,fill=none](w1--3){};   & & \node(w2--3){}; & & \node(u--5){};
    \\
    };
    \draw[<-,very thick,>=latex,chocolate,shorten <=0.5pt](u--1) -- ++ (90:0.5cm);
    \draw[<-,very thick,>=latex,chocolate,shorten <=0.5pt](u--2) -- ++ (90:0.5cm);
    \draw[<-,very thick,>=latex,chocolate,shorten <=0.5pt](u--3) -- ++ (90:0.5cm);
    \draw[<-,very thick,>=latex,chocolate,shorten <=0.5pt](u--4) -- ++ (90:0.5cm);
    \draw[<-,very thick,>=latex,chocolate,shorten <=0.5pt](u--5) -- ++ (0:0.5cm);

%--------
% chart 1
%--------
%
%\path (v--1) ++ (-0.7cm,0.275cm) node{\Large $\bonechart$};
%
\path (u--1) ++ (-0.25cm,0cm) node{$\scriptstyle\cvert$};
\path (v--1) ++ (-0.225cm,0.05cm) node{$\scriptstyle\avert$};
\path (w1bar--1) ++ (-0.25cm,0.1cm) node{$\scriptstyle\bvertbari{1}$};
\path (w1--1) ++ (-0.05cm,-0.2cm) node{$\scriptstyle\bverti{1}$};
\path (w2bar--1) ++ (0.3cm,0.1cm) node{$\scriptstyle\bvertbari{2}$};
\path (w2--1) ++ (0.1cm,-0.2cm) node{$\scriptstyle\bverti{2}$};

\draw[->,thick,cyan] (u--1) to (v--1);
\draw[->,thick,royalblue] (v--1) to (w1bar--1);
\draw[->,out=45,in=-30,distance=0.3cm] (v--1) to (u--1);
\draw[->,thick,forestgreen] (w1bar--1) to (w1--1);
\draw[->] (w1bar--1) to (w2bar--1);
\draw[->,thick,densely dotted,out=180,in=230,distance=0.4cm] (w1--1) to (w1bar--1);
\draw[->,thick,forestgreen] (w2bar--1) to (w2--1);
\draw[->,thick,densely dotted,out=90,in=0,distance=0.5cm] (w2bar--1) to (v--1);
\draw[->,thick,densely dotted,out=0,in=-50,distance=0.4cm] (w2--1) to (w2bar--1);
\draw[-,very thick,bend right,magenta,densely dashed,distance=] (w1--1) to (w2--1); 
\draw[-,very thick,out=15,in=75,magenta,densely dashed,distance=0.75cm] (u--1) to (w2bar--1);

\path ($(w1bar--1)!0.5!(w2--1)$) ++ (0cm,-0.3cm) node{\small \ref{R3.4.1}};

%--------
% chart 2
%--------
%
%\path (v--2) ++ (-0.7cm,0.275cm) node{\Large $\bonechart$};
%
\path (u--2) ++ (-0.25cm,0cm) node{$\scriptstyle\cvert$};
\path (v--2) ++ (-0.225cm,0.1cm) node{$\scriptstyle\avert$};
\path (w1bar--2) ++ (-0.25cm,0.1cm) node{$\scriptstyle\bvertbari{1}$};
\path (w1--2) ++ (-0.05cm,-0.2cm) node{$\scriptstyle\bverti{1}$};
\path (w2bar--2) ++ (0.3cm,0.1cm) node{$\scriptstyle\bvertbari{2}$};
\path (w2--2) ++ (0.1cm,-0.2cm) node{$\scriptstyle\bverti{2}$};

\draw[->] (u--2) to (v--2);
\draw[->,thick,royalblue] (v--2) to (w1bar--2);
\draw[->,thick,royalblue] (v--2) to (w2bar--2);
\draw[->,thick,forestgreen] (w1bar--2) to (w1--2);
\draw[->] (w1bar--2) to (w2bar--2);
\draw[->,thick,densely dotted,out=180,in=230,distance=0.4cm] (w1--2) to (w1bar--2);
\draw[->,thick,forestgreen] (w2bar--2) to (w2--2);
\draw[->,thick,densely dotted,out=90,in=0,distance=0.5cm] (w2bar--2) to (v--2);
\draw[->,thick,densely dotted,out=0,in=-50,distance=0.4cm] (w2--2) to (w2bar--2);
\draw[-,very thick,bend right,magenta,densely dashed,distance=] (w1--2) to (w2--2); 
\draw[-,very thick,out=-15,in=130,magenta,densely dashed,distance=0.4cm] (w1bar--2) to (w2--2);

\path ($(w1bar--2)!0.5!(w2--2)$) ++ (0cm,-0.3cm) node{\small \ref{R3.4.1}};

\draw[-implies,thick,double equal sign distance, bend left,distance=0.9cm,
               shorten <= 0.5cm,shorten >= 0.4cm
               ] (u--1) to node[above]{\small $\fap{\textsf{insulate}}{\avert}$} (u--2);

%--------
% chart 3
%--------
%
%\path (v--3) ++ (-0.7cm,0.275cm) node{\Large $\bonechart$};
%
\path (u--3) ++ (0.2cm,-0.15cm) node{$\scriptstyle\cvert$};
\path (v--3) ++ (-0.225cm,0.05cm) node{$\scriptstyle\avert$};
\path (w2bar--3) ++ (0.3cm,0.1cm) node{$\scriptstyle\bvertbari{2}$};
\path (w2--3) ++ (0.1cm,-0.2cm) node{$\scriptstyle\bverti{2}$};

\draw[->] (u--3) to (v--3);
\draw[->,thick,royalblue] (v--3) to (w2bar--3);
\draw[->,thick,royalblue,shorten >=5pt] (v--3) to (w2--3);
\draw[->,thick,forestgreen] (w2bar--3) to (w2--3);
\draw[->,thick,densely dotted,out=90,in=0,distance=0.5cm] (w2bar--3) to (v--3);
\draw[->,thick,densely dotted,out=0,in=-50,distance=0.4cm] (w2--3) to (w2bar--3);

\path ($(w1bar--3)!0.5!(w2--3)$) ++ (0cm,0cm) node{\small \ref{R2}};    

\draw[-,very thick,densely dashed,magenta,out=180,in=270,distance=0.6cm]
     (w2bar--3) to (v--3);
\draw[-,very thick,densely dashed,magenta,out=180,in=225,distance=1.25cm]
     (w2--3) to  (v--3);

\draw[-implies,thick,double equal sign distance, bend left,distance=0.9cm,
               shorten <= 0.5cm,shorten >= 0.4cm
               ] (u--2) to node[above]{\small $\fap{\textsf{make}}{\avert}$} node[below]{\small $\textsf{parsimonious}$} (u--3);     
  
%--------
% chart 4
%--------
%
%\path (v--4) ++ (-0.7cm,0.275cm) node{\Large $\bonechart$};
%
\path (u--4) ++ (0.2cm,0cm) node{$\scriptstyle\cvert$};
\path (v--4) ++ (0.225cm,0.15cm) node{$\scriptstyle\avert$};
  
\draw[->] (u--4) to (v--4);
\draw[->,thick,royalblue,out=-45,in=225,distance=0.9cm] (v--4) to (v--4);

\path ($(u--4)!0.5!(v--4)$) ++ (-0.5cm,0cm) node{\small \ref{R1.1}}; 
  
\draw[-,very thick,densely dashed,magenta,out=180,in=180,distance=1.25cm]
     (v--4) to  (u--4);      

\draw[-implies,thick,double equal sign distance, bend left,distance=0.75cm,
               shorten <= 0.25cm,shorten >= 0.25cm
               ] (u--3) to node[above]{\small ${\textsf{elim}}$} (u--4);

%--------
% chart 5
%--------
%
%\path (v--4) ++ (-0.7cm,0.275cm) node{\Large $\bonechart$};
%
\path (u--5) ++ (0cm,-0.25cm) node{$\scriptstyle\cvert$};
\draw[->,thick,royalblue,out=45,in=135,distance=0.9cm] (u--5) to (u--5);

\draw[-implies,thick,double equal sign distance,out=-45,in=70,distance=1cm,
               shorten <= 0.2cm,shorten >= 0.6cm
               ] (u--4) to node[left,pos=0.72]{\small ${\textsf{elim}}$} (u--5);  
  
\end{tikzpicture}
\end{center}  
  \vspace*{-2ex}
  \caption{\label{fig:ex:pars:insulated:exp}%
           \protect\LLEEPreserving\ parsimonious insulation from above of a red.\ $\protect\sone$-bisimilarity redundancy \protect\ref{R3.4.1}
             that here leads to its elimination, and permits further minimization.}% 

\end{figure}  

%% file: figs/fig-counterex-transf-fun-elevation.tex
\begin{figure}[t!]
\begin{flushleft}
  \hspace*{-1ex}%
\scalebox{0.5}{%
\begin{tikzpicture}
  \pgfdeclarelayer{background}
  \pgfdeclarelayer{foreground}
  \pgfsetlayers{background,main,foreground}
  %
  %--------------
  % ground floor
  %--------------
  %
  \matrix[anchor=center,row sep=0.5cm,column sep=1.65cm,ampersand replacement=\&,
          every node/.style={draw,very thick,circle,minimum width=2.5pt,fill,inner sep=0pt,outer sep=2pt}] at (0,-7) {
  \\
                    \&  \node(abc--1-gf){};
  \\
  \\
  \node(acd--1-gf){};   
  \\
                    \& \node(a--1-gf){};
  \\
  \node(c--1-gf){};       \&                  \&  \node(abcd-2--1-gf){};
  \\
  \\
  \node(abcd-1--1-gf){};  \&                  \&  \node(e--1-gf){};
  \\
  \\
  \node(f--1-gf){};
  \\
  };
  \path (abc--1-gf) ++ (0cm,0.375cm) node{$\tightfbox{$abc$}$};
  \draw[->,thick,royalblue] (abc--1-gf) to node[right,xshift=-0.05cm]{\black{$a$}} (a--1-gf);
  \draw[->,thick,royalblue,shorten >=3pt]  (abc--1-gf) to node[left,xshift=0.05cm,pos=0.3]{\black{$c$}} (c--1-gf);
  \draw[->,thick,royalblue,shorten >= 5pt] (abc--1-gf) to node[right,pos=0.65,xshift=-0.05cm]{\black{$b$}} (f--1-gf);
  % \draw[->,thick,royalblue] (abc--1-gf) to node[right,xshift=-0.05cm]{\black{$a$}} (a--1-gf);
  % \draw[->,thick,royalblue,shorten >=3pt]  (abc--1-gf) to node[left,xshift=0.05cm,pos=0.3]{\black{$c$}} (c--1-gf);
  % \draw[->,thick,royalblue,shorten >= 5pt] (abc--1-gf) to node[right,pos=0.65,xshift=-0.05cm]{\black{$b$}} (f--1-gf);
  %
  \path (acd--1-gf) ++ (0cm,0.375cm) node{$\tightfbox{$acd$}$};
  \draw[->] (acd--1-gf) to node[above,pos=0.25]{$a$} (a--1-gf);
  \draw[->,thick,forestgreen] (acd--1-gf) to node[left,xshift=0.085cm,pos=0.5]{\black{$c$}} (c--1-gf);
    % \draw[->,thick,forestgreen] (acd--1-gf) to node[left,xshift=0.085cm,pos=0.5]{\black{$c$}} (c--1-gf);
  \draw[->,bend right,distance=0.75cm] (acd--1-gf) to node[below,pos=0.475]{$d$} (e--1-gf);
  \path (c--1-gf) ++ (-0.25cm,-0.3cm) node{$\tightfbox{$c$}$};
  \draw[->] (c--1-gf) to node[right,xshift=-0.05cm]{$c$} (abcd-1--1-gf);
  \draw[->,out=180,in=225] (c--1-gf) to node[left,xshift=0.05cm]{$c$} (acd--1-gf);
  \path (abcd-1--1-gf) ++ (-0.55cm,-0.35cm) node{$\tightfbox{$abcd_1$}$};
  \draw[->,densely dotted,thick,out=180,in=209,shorten >=4.5pt,distance=1.5cm] (abcd-1--1-gf) to node[left,xshift=0.05cm,pos=0.5]{$\sone$} (acd--1-gf);
  \draw[->] (abcd-1--1-gf) to node[right,xshift=-0.075cm,pos=0.25]{$b$} (f--1-gf);
  \path (f--1-gf) ++ (0cm,-0.375cm) node{$\tightfbox{$f$}$};
  \draw[->,out=180,in=180,distance=2.5cm] (f--1-gf) to node[left,xshift=0.05cm]{$f$} (acd--1-gf);
  \path (a--1-gf) ++ (-0.25cm,-0.25cm) node{$\tightfbox{$a$}$};
  \draw[->,out=0,in=-45,distance=0.75cm] (a--1-gf) to node[right,xshift=-0.05cm]{$a$} (abc--1-gf);
  \draw[->] (a--1-gf) to node[above,pos=0.4]{$a$} (abcd-2--1-gf);
  \path (abcd-2--1-gf) ++ (0.55cm,-0.35cm) node{$\tightfbox{$abcd_2$}$};
  \draw[->,densely dotted,thick,out=0,in=-25,distance=1.5cm,shorten >=4.5pt] (abcd-2--1-gf) to node[right,xshift=0.05cm,pos=0.65]{$\sone$} (abc--1-gf);
  \draw[->] (abcd-2--1-gf) to node[left,xshift=0.075cm,pos=0.35]{$d$} (e--1-gf);
  \path (e--1-gf) ++ (0cm,-0.375cm) node{$\tightfbox{$e$}$};
  \draw[->,out=0,in=0,distance=2.75cm] (e--1-gf) to node[right]{$e$} (abc--1-gf);

  \draw[rounded corners,densely dashdotted] ($(abc--1-gf) + (3.5cm,0.75cm)$) rectangle ($(f--1-gf) + (-2.4cm,-0.75cm)$);
    \path (abc--1-gf) ++ (2.6cm,0.95cm) node{ground floor};

  \matrix[anchor=center,row sep=0.5cm,column sep=1.65cm,ampersand replacement=\&,
          every node/.style={draw,very thick,circle,minimum width=2.5pt,fill,inner sep=0pt,outer sep=2pt}] at (8.75,-7) {
  \\
                    \&  \node(abc--2-gf){};
  \\
  \\
  \node(acd--2-gf){};   
  \\
                    \& \node(a--2-gf){};
  \\
  \node(c--2-gf){};       \&                  \&  \node(abcd-2--2-gf){};
  \\
  \\
  \node(abcd-1--2-gf){};  \&                  \&  \node(e--2-gf){};
  \\
  \\
  \node(f--2-gf){};
  \\
  };
  \path (abc--2-gf) ++ (0cm,0.375cm) node{$\tightfbox{$abc$}$};
  \draw[->,thick,royalblue] (abc--2-gf) to node[right,xshift=-0.05cm]{\black{$a$}} (a--2-gf);
  \draw[->,thick,royalblue,shorten >=3pt]  (abc--2-gf) to node[left,xshift=0.05cm,pos=0.3]{\black{$c$}} (c--2-gf);
  \draw[->,thick,royalblue,shorten >= 5pt] (abc--2-gf) to node[right,pos=0.65,xshift=-0.05cm]{\black{$b$}} (f--2-gf);
    % \draw[->,thick,royalblue] (abc--2-gf) to node[right,xshift=-0.05cm]{\black{$a$}} (a--2-gf);
    % \draw[->,thick,royalblue,shorten >=3pt]  (abc--2-gf) to node[left,xshift=0.05cm,pos=0.3]{\black{$c$}} (c--2-gf);
    % \draw[->,thick,royalblue,shorten >= 5pt] (abc--2-gf) to node[right,pos=0.65,xshift=-0.05cm]{\black{$b$}} (f--2-gf);
  %
  \path (acd--2-gf) ++ (0cm,0.375cm) node{$\tightfbox{$acd$}$};
  \draw[->] (acd--2-gf) to node[above,pos=0.25]{$a$} (a--2-gf);
  \draw[->,thick,forestgreen] (acd--2-gf) to node[left,xshift=0.085cm,pos=0.5]{\black{$c$}} (c--2-gf);
    %\draw[->,thick,forestgreen] (acd--2-gf) to node[left,xshift=0.085cm,pos=0.5]{\black{$c$}} (c--2-gf);
  \draw[->,bend right,distance=0.75cm] (acd--2-gf) to node[below,pos=0.475]{$d$} (e--2-gf);
  \path (c--2-gf) ++ (-0.25cm,-0.3cm) node{$\tightfbox{$c$}$};
  \draw[->] (c--2-gf) to node[right,xshift=-0.05cm]{$c$} (abcd-1--2-gf);
  \draw[->,out=180,in=225] (c--2-gf) to node[left,xshift=0.05cm]{$c$} (acd--2-gf);
  \path (abcd-1--2-gf) ++ (-0.55cm,-0.35cm) node{$\tightfbox{$abcd_1$}$};
  \draw[->,densely dotted,thick,out=180,in=209,shorten >=4.5pt,distance=1.5cm] (abcd-1--2-gf) to node[left,xshift=0.05cm,pos=0.5]{$\sone$} (acd--2-gf);
  \draw[->] (abcd-1--2-gf) to node[right,xshift=-0.075cm,pos=0.25]{$b$} (f--2-gf);
  \path (f--2-gf) ++ (0cm,-0.375cm) node{$\tightfbox{$f$}$};
  \draw[->,out=180,in=180,distance=2.5cm] (f--2-gf) to node[left,xshift=0.05cm]{$f$} (acd--2-gf);
  \path (a--2-gf) ++ (-0.25cm,-0.25cm) node{$\tightfbox{$a$}$};
  \draw[->,out=0,in=-45,distance=0.75cm] (a--2-gf) to node[right,xshift=-0.05cm]{$a$} (abc--2-gf);
  \draw[->] (a--2-gf) to node[above,pos=0.4]{$a$} (abcd-2--2-gf);
  \path (abcd-2--2-gf) ++ (0.55cm,-0.35cm) node{$\tightfbox{$abcd_2$}$};
  \draw[->,densely dotted,thick,out=0,in=-25,distance=1.5cm,shorten >=4.5pt] (abcd-2--2-gf) to node[right,xshift=0.05cm,pos=0.65]{$\sone$} (abc--2-gf);
  \draw[->] (abcd-2--2-gf) to node[left,xshift=0.05cm]{$d$} (e--2-gf);
  \path (e--2-gf) ++ (0cm,-0.375cm) node{$\tightfbox{$e$}$};
  \draw[->,out=0,in=0,distance=2.75cm] (e--2-gf) to node[right]{$e$} (abc--2-gf);

  \draw[rounded corners,densely dashdotted] ($(abc--2-gf) + (3.5cm,0.75cm)$) rectangle ($(f--2-gf) + (-2.4cm,-0.75cm)$);
    \path (abc--2-gf) ++ (2.6cm,0.95cm) node{ground floor};

  \draw[|->,very thick,densely dashed,magenta,shorten <=4pt,out=15,in=170,distance=1.75cm] (abc--1-gf) to (abc--2-gf);
  \draw[|->,very thick,densely dashed,magenta,shorten >=4.5pt,out=45,in=170,distance=1.75cm] (acd--1-gf) to (acd--2-gf);  
  \draw[|->,very thick,densely dashed,magenta,shorten >=4.5pt,out=25,in=170,distance=2cm] (a--1-gf) to (a--2-gf);
  \draw[|->,very thick,densely dashed,magenta,out=10,in=170,distance=1.75cm,shorten >=4.5pt] (c--1-gf) to (c--2-gf);
  \draw[|->,very thick,densely dashed,magenta,out=35,in=160,distance=1.5cm,shorten >=4.5pt] (abcd-2--1-gf) to (abcd-2--2-gf);
  \draw[|->,very thick,densely dashed,magenta,out=7.5,in=150,distance=1.3cm] (abcd-1--1-gf) to (abcd-1--2-gf);
  \draw[|->,very thick,densely dashed,magenta,out=-25,in=205,distance=2cm] (e--1-gf) to (e--2-gf);
  \draw[|->,very thick,densely dashed,magenta,out=-15,in=205,distance=1.5cm] (f--1-gf) to (f--2-gf);

  %
  %--------------
  % first floor
  %--------------
  %
  \matrix[anchor=center,row sep=0.5cm,column sep=1.65cm,ampersand replacement=\&,
          every node/.style={draw,very thick,circle,minimum width=2.5pt,fill,inner sep=0pt,outer sep=2pt}] at (0,0) {
  \\
                    \&  \node(abc--1-ff){};
  \\
  \\
  \node(acd--1-ff){};   
  \\
                    \& \node(a--1-ff){};
  \\
  \node(c--1-ff){};       \&                  \&  \node(abcd-2--1-ff){};
  \\
  \\
  \node(abcd-1--1-ff){};  \&                  \&  \node(e--1-ff){};
  \\
  \\
  \node(f--1-ff){};
  \\
  };
  \path (abc--1-ff) ++ (0cm,0.375cm) node{$\tightfbox{$abc$}$};
  \draw[->,out=-60,in=60,distance=2.3cm] (abc--1-ff) to node[right,xshift=-0.05cm,pos=0.1]{\black{$a$}} (a--1-gf);
  \draw[->,shorten >=3pt]  (abc--1-ff) to node[left,xshift=0.05cm,pos=0.275]{\black{$c$}} (c--1-gf);
  \draw[->,shorten >= 5pt] (abc--1-ff) to node[right,pos=0.3,xshift=-0.05cm]{\black{$b$}} (f--1-gf);
  % \draw[->,thick,royalblue] (abc--1-ff) to node[right,xshift=-0.05cm]{\black{$a$}} (a--1-ff);
  % \draw[->,thick,royalblue,shorten >=3pt]  (abc--1-ff) to node[left,xshift=0.05cm,pos=0.3]{\black{$c$}} (c--1-ff);
  % \draw[->,thick,royalblue,shorten >= 5pt] (abc--1-ff) to node[right,pos=0.65,xshift=-0.05cm]{\black{$b$}} (f--1-ff);
  %
  \path (acd--1-ff) ++ (0cm,0.375cm) node{$\tightfbox{$acd$}$};
  \draw[->,shorten >=4pt] (acd--1-ff) to node[right,pos=0.1,xshift=-0.05cm]{$a$} (a--1-gf);
  \draw[->,out=180,in=180,distance=2.85cm] (acd--1-ff) to node[left,xshift=0.025cm,pos=0.15]{\black{$c$}} (c--1-gf);
    % \draw[->,thick,forestgreen] (acd--1-ff) to node[left,xshift=0.085cm,pos=0.5]{\black{$c$}} (c--1-ff);
  \draw[->,distance=3cm,out=-20,in=125] (acd--1-ff) to node[right,pos=0.475,xshift=-0.05cm]{$d$} (e--1-gf);
  \path (c--1-ff) ++ (-0.25cm,-0.3cm) node{$\tightfbox{$c$}$};
  \draw[->] (c--1-ff) to node[right,xshift=-0.05cm]{$c$} (abcd-1--1-ff);
  \draw[->,out=180,in=180,distance=2.65cm] (c--1-ff) to node[left,xshift=0.05cm,pos=0.3]{$c$} (acd--1-gf);
  \path (abcd-1--1-ff) ++ (-0.55cm,-0.35cm) node{$\tightfbox{$abcd_1$}$};
  \draw[->,densely dotted,thick,out=180,in=209,distance=1.5cm] (abcd-1--1-ff) to node[left,xshift=0.05cm,pos=0.65]{$\sone$} (acd--1-ff);
  \draw[-,out=-77.5,in=82,shorten >=4.5pt] (abcd-1--1-ff) to node[right,xshift=-0.075cm,pos=0.1]{$b$} (f--1-gf);
  \path (f--1-ff) ++ (0cm,-0.375cm) node{$\tightfbox{$f$}$};
  \draw[->,out=180,in=180,distance=2cm] (f--1-ff) to node[left,xshift=0cm,pos=0.175]{$f$} (acd--1-gf);
  \path (a--1-ff) ++ (0.15cm,0.35cm) node{$\tightfbox{$a$}$};
  \draw[->,shorten >=15pt] (a--1-ff) to node[left,xshift=0.05cm,pos=0.6]{$a$} (abc--1-gf);
  \draw[->] (a--1-ff) to node[above,pos=0.65,xshift=0cm]{$a$} (abcd-2--1-ff);
  \path (abcd-2--1-ff) ++ (0.55cm,-0.35cm) node{$\tightfbox{$abcd_2$}$};
  \draw[->,densely dotted,thick,out=0,in=-25,distance=1.5cm] (abcd-2--1-ff) to node[right,xshift=0.05cm,pos=0.65]{$\sone$} (abc--1-ff);
  \draw[->,out=240,in=125,distance=1.5cm] (abcd-2--1-ff) to node[left,xshift=0.05cm,pos=0.175]{$d$} (e--1-gf);
  \path (e--1-ff) ++ (0cm,-0.375cm) node{$\tightfbox{$e$}$};
  \draw[->,out=0,in=0,distance=1.5cm] (e--1-ff) to node[right,pos=0.425,xshift=-0.05cm]{$e$} (abc--1-gf);

  \draw[rounded corners,densely dashdotted] ($(abc--1-ff) + (3.5cm,0.75cm)$) rectangle ($(f--1-ff) + (-2.4cm,-0.75cm)$);
    \path (abc--1-ff) ++ (2.6cm,0.95cm) node{first floor};
  
  \draw[rounded corners,thick] ($(abc--1-ff) + (3.7cm,1.5cm)$) rectangle ($(f--1-gf) + (-2.6cm,-0.95cm)$);
  %\draw[rounded corners,thick] ($(abc--1-ff) + (3.7cm,1.95cm)$) rectangle ($(f--1-gf) + (-2.6cm,-0.95cm)$);
  \path (abc--1-ff) ++ (-0.35cm,2.1cm) node(label-left){\Huge $\elevationofabove{\alert{\asetverts}}{\aonecharti{s}}$};

  % \draw[ellipse through=abcd-1--1-ff and abcd-2--1-ff and e--1-ff];
  \ellipsebyfoci{draw,red}{$(abcd-1--1-ff) + (-1.4cm,-0.5cm)$}{$(abcd-2--1-ff) + (1.3cm,0.5cm)$}{1.035}
    \path (abcd-2--1-ff) ++ (1.4cm,1.05cm) node{\Large $\alert{\asetverts}$};

  \matrix[anchor=center,row sep=0.5cm,column sep=1.65cm,ampersand replacement=\&,
          every node/.style={draw,very thick,circle,minimum width=2.5pt,fill,inner sep=0pt,outer sep=2pt}] at (8.75,0) {
  \\
                    \&  \node(abc--2-ff){};
  \\
  \\
  \node(acd--2-ff){};   
  \\
                    \& \node(a--2-ff){};
  \\
  \node(c--2-ff){};       \&                  \&  \node(abcd-2--2-ff){};
  \\
  \\
  \node(abcd-1--2-ff){};  \&                  \&  \node(e--2-ff){};
  \\
  \\
  \node(f--2-ff){};
  \\
  };
  \path (abc--2-ff) ++ (0cm,0.375cm) node{$\tightfbox{$abc$}$};
  %\draw[->] (abc--2-ff) to node[right,xshift=-0.05cm]{\black{$a$}} (a--2-ff);
  \draw[->,out=-60,in=60,distance=2.3cm] (abc--2-ff) to node[right,xshift=-0.05cm,pos=0.1]{\black{$a$}} (a--2-gf);
  \draw[->,shorten >=3pt] (abc--2-ff) to node[left,xshift=0.05cm,pos=0.25]{\black{$c$}} (c--2-gf);
  \draw[->,shorten >=5pt] (abc--2-ff) to node[right,pos=0.3,xshift=-0.05cm]{\black{$b$}} (f--2-gf);
    % \draw[->,thick,royalblue] (abc--2-ff) to node[right,xshift=-0.05cm]{\black{$a$}} (a--2-ff);
    % \draw[->,thick,royalblue,shorten >=3pt]  (abc--2-ff) to node[left,xshift=0.05cm,pos=0.3]{\black{$c$}} (c--2-ff);
    % \draw[->,thick,royalblue,shorten >= 5pt] (abc--2-ff) to node[right,pos=0.65,xshift=-0.05cm]{\black{$b$}} (f--2-ff);
    %
  \path (acd--2-ff) ++ (0cm,0.375cm) node{$\tightfbox{$acd$}$};
  \draw[->,shorten >=4pt] (acd--2-ff) to node[right,pos=0.1,xshift=-0.05cm]{$a$} (a--2-gf);
  \draw[->,out=180,in=180,distance=2.85cm] (acd--2-ff) to node[left,xshift=0.025cm,pos=0.15]{\black{$c$}} (c--2-gf);
    %\draw[->,thick,forestgreen] (acd--2-ff) to node[left,xshift=0.085cm,pos=0.5]{\black{$c$}} (c--2-ff);
  %\draw[->,bend right,distance=0.75cm] (acd--2-ff) to node[below,pos=0.475]{$d$} (e--2-ff);
  \draw[->,distance=3cm,out=-20,in=125] (acd--2-ff) to node[right,pos=0.475,xshift=-0.05cm]{$d$} (e--2-gf);
  \path (c--2-ff) ++ (-0.25cm,-0.3cm) node{$\tightfbox{$c$}$};
  \draw[->] (c--2-ff) to node[right,xshift=-0.05cm]{$c$} (abcd-1--2-ff);
  \draw[->,out=180,in=180,distance=2.65cm] (c--2-ff) to node[left,xshift=0.05cm,pos=0.3]{$c$} (acd--2-gf);
  \path (abcd-1--2-ff) ++ (-0.55cm,-0.35cm) node{$\tightfbox{$abcd_1$}$};
  \draw[->,densely dotted,thick,out=180,in=209,distance=1.5cm] (abcd-1--2-ff) to node[left,xshift=0.05cm,pos=0.65]{$\sone$} (acd--2-ff);
  \draw[-,out=-77.5,in=82,shorten >=4.5pt] (abcd-1--2-ff) to node[right,xshift=-0.075cm,pos=0.1]{$b$} (f--2-gf);
  \path (f--2-ff) ++ (0cm,-0.375cm) node{$\tightfbox{$f$}$};
  \draw[->,out=180,in=180,distance=1.5cm] (f--2-ff) to node[left,xshift=0cm,pos=0.175]{$f$} (acd--2-gf);
  \path (a--2-ff) ++ (0.15cm,0.35cm) node{$\tightfbox{$a$}$};
  \draw[->,shorten >=15pt] (a--2-ff) to node[left,xshift=0.05cm,pos=0.6]{$a$} (abc--2-gf);
  \draw[->] (a--2-ff) to node[above,pos=0.65,xshift=-0.05cm]{$a$} (abcd-2--2-ff);
  \path (abcd-2--2-ff) ++ (0.55cm,-0.35cm) node{$\tightfbox{$abcd_2$}$};
  \draw[->,densely dotted,thick,out=0,in=-25,distance=1.5cm] (abcd-2--2-ff) to node[right,xshift=0.05cm,pos=0.65]{$\sone$} (abc--2-ff);
  %\draw[->] (abcd-2--2-ff) to node[left,xshift=0.05cm]{$d$} (e--2-ff);
  \draw[->,out=240,in=125,distance=1.5cm] (abcd-2--2-ff) to node[left,xshift=0.05cm,pos=0.175]{$d$} (e--2-gf);
  \path (e--2-ff) ++ (0cm,-0.375cm) node{$\tightfbox{$e$}$}; 
  \draw[->,out=0,in=0,distance=1.5cm] (e--2-ff) to node[right,pos=0.425,xshift=-0.05cm]{$e$} (abc--2-gf);

  %\draw[|->,very thick,densely dashed,magenta,bend right,distance=0.75cm,out=-35,in=265,distance=4.5cm] (abcd-1--1-ff) to (abcd-2--2-ff);
  \draw[|->,very thick,densely dashed,magenta,distance=0.75cm,out=-20,in=225,distance=3.6cm,shorten >=3pt] (abcd-1--1-ff) to (abcd-2--2-ff);
  \draw[|->,very thick,densely dashed,magenta,distance=0.75cm,out=35,in=135,distance=1.5cm] (abcd-2--1-ff) to (abcd-1--2-ff);

  \draw[rounded corners,densely dashdotted] ($(abc--2-ff) + (3.5cm,0.75cm)$) rectangle ($(f--2-ff) + (-2.4cm,-0.75cm)$);
    \path (abc--2-ff) ++ (2.6cm,0.95cm) node{first floor};

  \draw[rounded corners,thick] ($(abc--2-ff) + (3.7cm,1.5cm)$) rectangle ($(f--2-gf) + (-2.6cm,-0.95cm)$);
  % \draw[rounded corners,thick] ($(abc--2-ff) + (3.7cm,1.95cm)$) rectangle ($(f--2-gf) + (-2.6cm,-0.95cm)$);
  \path (abc--2-ff) ++ (-0.35cm,2.1cm) node(label-right){\Huge $\elevationofabove{\alert{\asetverts}}{\aonecharti{s}}$}; 
    
  \draw[->,%-left to,
        very thick,densely dashed,magenta,shorten <=0.3cm,shorten >=0.3cm] (label-left) to node[above]{\Huge $\sliftltfunn{\magenta{\scpfunon{\black{\twinc}}}}$} (label-right);

  \ellipsebyfoci{draw,red}{$(abcd-1--2-ff) + (-1.4cm,-0.5cm)$}{$(abcd-2--2-ff) + (1.3cm,0.5cm)$}{1.035}  
    \path (abcd-2--2-ff) ++ (1.4cm,1cm) node{\Large $\alert{\asetverts}$};
  
  \begin{pgfonlayer}{background}
    
    \draw[draw opacity=0,fill opacity=0.4,fill=royalblue!30] 
      (abc--1-ff.center)--(a--1-ff.center)--(abcd-2--1-ff.center)--(e--1-ff.center) to[->,relative=false,out=-1,in=1,distance=2.92cm] (abc--1-ff.center);
    \draw[draw opacity=0,fill opacity=0.4,fill=forestgreen!15] 
      (acd--1-ff.center)--(c--1-ff.center)--(abcd-1--1-ff.center)--(f--1-ff.center) to[->,relative=false,out=181,in=179,distance=2.67cm] (acd--1-ff.center);
    \draw[draw opacity=0,fill opacity=0.4,fill=royalblue!30] 
      (abc--2-ff.center)--(a--2-ff.center)--(abcd-2--2-ff.center)--(e--2-ff.center) to[->,relative=false,out=-1,in=1,distance=2.92cm] (abc--2-ff.center);
    \draw[draw opacity=0,fill opacity=0.4,fill=forestgreen!15] 
      (acd--2-ff.center)--(c--2-ff.center)--(abcd-1--2-ff.center)--(f--2-ff.center) to[->,relative=false,out=181,in=179,distance=2.67cm] (acd--2-ff.center);
    \draw[draw opacity=0,fill opacity=0.4,fill=royalblue!45] 
      (abc--1-gf.center)--(a--1-gf.center)--(abcd-2--1-gf.center)--(e--1-gf.center) to[->,relative=false,out=-1,in=1,distance=2.92cm] (abc--1-gf.center);
    \draw[draw opacity=0,fill opacity=0.4,fill=forestgreen!30] 
      (acd--1-gf.center)--(c--1-gf.center)--(abcd-1--1-gf.center)--(f--1-gf.center) to[->,relative=false,out=181,in=179,distance=2.67cm] (acd--1-gf.center);
    \draw[draw opacity=0,fill opacity=0.4,fill=royalblue!45] 
      (abc--2-gf.center)--(a--2-gf.center)--(abcd-2--2-gf.center)--(e--2-gf.center) to[->,relative=false,out=-1,in=1,distance=2.92cm] (abc--2-gf.center);
    \draw[draw opacity=0,fill opacity=0.4,fill=forestgreen!30] 
      (acd--2-gf.center)--(c--2-gf.center)--(abcd-1--2-gf.center)--(f--2-gf.center) to[->,relative=false,out=181,in=179,distance=2.67cm] (acd--2-gf.center);
  \end{pgfonlayer}{background}
    
\end{tikzpicture}%
  }
\end{flushleft}  
  \caption{\label{fig:ex:lem:transfer:lift:local:transfer}%
    Lifting of the local transfer function $\protect\magenta{\scpfunon{\black{\twinc}}}$ on $\aonecharti{s}$ from Fig.~\protect\ref{fig:ex:local:transfer:function}
      with domain and range $\protect\alert{\asetverts}$ %$\protect\alert{\asetverts} \defdby \domof{\protect\magenta{\scpfunon{\black{\twinc}}}} = \ranof{\protect\magenta{\scpfunon{\black{\twinc}}}}$
    to a transfer function $\protect\sliftltfunn{\protect\magenta{\scpfunon{\black{\twinc}}}}$
    on the elevation $\protect\elevationofabove{\protect\alert{\protect\asetverts}}{\protect\aonecharti{s}}$ of $\protect\alert{\asetverts}$,
    which is a \protect\LLEEoneLTS. 
    % The elevation $\protect\elevationofabove{\protect\alert{\protect\asetverts}}{\protect\aonecharti{s}}$ is a \LLEEoneLTS\ as well.
      % as stated by Lem.~\protect\ref{lem:transfer:lift:local:transfer}.%
           }
\end{figure}

%% file: figs/diag-lifting-transfer-fun.tex
\begin{tikzpicture}
  \matrix[anchor=center,row sep=0.5cm,column sep=1.4cm] {
      % \node{local transfer function $\sphifun$}; 
        % &[-1cm]  
          \node(topleft){$\aoneLTS$};
            & \node(topright){$\aoneLTS$};
    \\[-0.25ex]
      %\node{projections $\sproji{1}$ are transfer functions};
        &  
    \\[-0.5ex]
      % \node[yshift=0.075cm]{transfer function $\sliftltfun$}; 
        %&[-1cm]       
          \node(bottomleft){$\elevatechart{\asetverts}{\aoneLTS}%^{(\firstfloor)}
                                                                $};
            & \node(bottomright){$\elevatechart{\asetverts}{\aoneLTS}%^{(\firstfloor)}
                                                                     $};
    \\  
                                                       };
  \draw[-{Straight Barb[left]}] (topleft)    to node[above]{$\sphifun$} (topright);
  \draw[-{Straight Barb[left]}] (bottomleft) to node[above]{$\sliftltfun$} (bottomright);
  \draw[->] (bottomleft)  to node[left,xshift=0.05cm]{$\sproji{1}$}   (topleft);
  \draw[->] (bottomright) to node[right,xshift=-0.05cm]{$\sproji{1}$} (topright); 
\end{tikzpicture}

%% file: cryst-arxiv-app-v1v2.tex
% \appendix

\newpage\onecolumn%\mbox{}

%% Appendix
%=================== 
\appendix%
\section{Appendix} %{: supplements, more proof details, and omitted proofs}%
%\section*{S \mbox{} Supplements (Appendix)}%
  \label{appendix}%
%===================
%\addtocounter{section}{19}  
% \footnotetext{\hspace*{0.5pt}%
%   {\nf\bf Guidelines appendix.}
%       `If necessary, detailed proofs of technical results may be included in a clearly-labeled appendix, 
%       to be consulted at the discretion of program committee members.'}    

% In this appendix we provide additional material for the sections of the article, for more scrutiny of our argumentation.
% In particular, we give details for most of the formal statements of our proof. 
% However, 
% the material in this appendix is still under development for a monograph that we are writing on the completeness proof for $\milnersys$.
% Here we collected the parts that are most crucial for our argumentation.
%   We do not repeat proofs of results that have been published already;
%     in those cases we place references, and explain the formal connections. 
%     %
% -- Please find below a table of contents of the article, and of this appendix.   

% %----------
% % \subsection*{Table of Contents}
% %----------
% \enlargethispage{5ex}
% %---------------
% %\begin{KeepFromToc}%
%   \setcounter{tocdepth}{3}%
% \addtocontents{toc}{\protect\setcounter{tocdepth}{-1}}
% \tableofcontents
% \addtocontents{toc}{\protect\setcounter{tocdepth}{3}}
% %\end{KeepFromToc}
% %---------------

%\newpage
%-------------------------------------------------------------------
\subsection{Supplements for Section~\ref{motivation:proof:strategy}}
%-------------------------------------------------------------------

%-------------
\subsubsection{Obstacle for the bisimulation-chart proof strategy (more detail)} \mbox{}
  \label{intro::app}%
%-------------
\newcommand{\EP}{\textit{EP}}%
\newcommand{\EPacc}{{\EP\hspace*{1pt}^{\prime}}}%

\medskip
%\vspace*{-1.5ex}
\noindent
In Sect.~\ref{motivation:proof:strategy}, starting on page~\ref{bisim:chart:strategy:start}, %--\ref{bisim:chart:strategy:end}
  we argued that a bisimulation-chart proof strategy that would operate 
    in analogy with Salomaa's proof in \cite{salo:1966} of completeness of $\Fone$ for language equivalence of regular expressions 
    does not work for showing completeness of Milner's system with respect to the process interpretation. 
In particular, we argued that an extraction procedure $\EP$ of a solution from any given guarded linear system of recursion equations
  as in Salomaa's proof is not possible, not even from systems that are solvable. 
The reason is as follows. 
An extraction procedure $\EP$ that were completely analogous to the one used by Salomaa would have the following two properties:  
\begin{enumerate}[label={($\EP$-\arabic{*})},align=right,leftmargin=*,itemsep=0.5ex]
  \item{}\label{EP:1}
    $\EP$ extracts a solution from any guarded linear system $\aspec$ of recursion equations.
  \item{}\label{EP:2}
    $\EP$ proceeds `data-obliviously' by mechanically combining
    the actions in the equations of a recursion system $\aspec$ according to a fixed way of traversing it that,
    while copying and transporting specific letters from the equations over to the solution, 
    never compares different letters and never takes decisions on the basis of the specific letter at hand. 
    In particular $\EP$ also never takes a decision that the system would not be solvable,
      but always produces a result.    
\end{enumerate} 
Now property \ref{EP:1} cannot be obtained, because as Milner showed, see Example~\ref{exa:charts:not:expressible} below, 
  there are guarded linear systems of recursion equations that are unsolvable by star expressions in the process semantics.
While this shows that an extraction procedure $\EP$ with \ref{EP:1} and \ref{EP:2} is impossible,
  it leaves open the possibility of an extraction procedure $\EPacc$ that satisfies \ref{EP:2} except for that it may sometimes not terminate (see \ref{EPacc:2} below),
    and the restriction \ref{EPacc:1} of \ref{EP:1} to solvable recursion equations (see \ref{EPacc:1}, \ref{EPacc:1:a}) 
    together with a soundness condition (see \ref{EPacc:1}, \ref{EPacc:1:b}):
\begin{enumerate}[label={($\EPacc$-\arabic{*})},align=right,leftmargin=*,itemsep=0.5ex]
  \item{}\label{EPacc:1}
    \begin{enumerate}[label={(\alph{*})},align=right,leftmargin=*,itemsep=0.25ex]
      \item{}\label{EPacc:1:a}
        $\EPacc$ extracts a solution from any guarded linear system $\aspec$ of recursion equations \underline{\smash{that is solvable}}.
      \item{}\label{EPacc:1:b}
        Whenever $\EPacc$ obtains a result for a guarded linear system $\aspec$ of recursion equations, then that is a solution of $\aspec$.  
    \end{enumerate}
  \item{}\label{EPacc:2}
    $\EPacc$ proceeds data-obliviously in the same way as \ref{EP:2} requires for $\EP$. However, the extraction process is not required to be terminating. 
    %
    % \item{}\label{EPacc:3}
    %   The extraction process from a guarded linear recursion system $\aspec$ is sound, that is, whenever $\EPacc$ obtains a result on $\aspec$, 
    %     then that is a solution of $\aspec$.
\end{enumerate}    
It turns out that such a restricted extraction procedure $\EPacc$ is not possible, either.
The reason is that there are unsolvable specifications $\aspeci{\textit{uso}}$ that are `data-obliviously' the same 
  (that is, the underlying process graph has the same structure when action names are ignored)
  with  a (respectively) corresponding solvable specification $\aspeci{\textit{so}}$.
  Then, in order to safeguard solvability for $\aspeci{\textit{so}}$, see \ref{EPacc:1}, \ref{EPacc:1:a},
    $\EPacc$ would need to produce a result for both $\aspeci{\textit{uso}}$ and $\aspeci{\textit{so}}$;
  but this then leads to a contradiction with \ref{EPacc:1}, \ref{EPacc:1:b}, because the result extracted from $\aspeci{\textit{uso}}$ cannot be a solution,
    as $\aspeci{\textit{uso}}$ is unsolvable.
We give two examples of such pairs of unsolvable and solvable specifications that are data-obliviously the same 
  below in Example~\ref{exa:charts:not:expressible} and Example~\ref{exa:charts:not:expressible:a},
where the specifications in Example~\ref{exa:charts:not:expressible:a}
  are solvable counterparts of unsolvable specifications in Example~\ref{exa:charts:not:expressible}.
%   are the specific
% The reason is that, as shown in Example~\ref{exa:charts:not:expressible:a} below,
%   recursion systems $\fap{\aspec}{G_1^{(a)}}$ and $\fap{\aspec}{G_2^{(a)}}$
% that are `data-obliviously' the same as the unsolvable specifications $\fap{\aspec}{G_1}$ and $\fap{\aspec}{G_2}$ in Example~\ref{exa:charts:not:expressible} 
%   are solvable under the process semantics. 
            
Such pairs of unsolvable and solvable process graphs that are `data-obliviously' the same are not only artificial counterexamples,
  but can also arise from bisimulation charts that link bisimilar expressible graphs, see Example~4.1 in \cite{grab:fokk:2020:lics,grab:fokk:2020:lics:arxiv}.
    (The such a solvable bisimulation chart is paired with a chart that arises from it by relabeling each transition by a different action.)

\begin{exa}[not expressible process graphs, unsolvable recursive specifications]\label{exa:charts:not:expressible}
  As mentioned in Sect.~\ref{motivation:proof:strategy} on page~\pageref{quotation:unsolvable:specs}, 
    Milner noticed that guarded systems of recursion equations cannot always be solved by star expressions
    under the process semantics:
    \begin{quotation}
      ``[In] contrast with the case for languages---an arbitrary system of guarded equations 
        in [star]-behaviours cannot in general be solved in star expressions'' \cite{miln:1984}.
    \end{quotation}   
  In fact, Milner showed in \cite{miln:1984}
    that the linear specification $\fap{\aspec}{G_1}$ defined by the process graph $G_1$  below 
    does not have a star expression solution modulo bisimilarity. 
  He conjectured that that holds also for the easier specification $\fap{\aspec}{G_2}$ defined by the process graph $G_2$ below.
    That that is indeed the case was confirmed and proved later by Bosscher \cite{boss:1997}.    
  \begin{center}\label{fig:milner-bosscher-expressible}
    \input{figs/cexs-milner-bosscher.tex}
  \end{center}
  Here the start vertex of a process graph is again highlighted by a brown arrow~\picarrowstart,
  and a vertex $\avert$ with immediate termination
  is emphasized in brown as \pictermvert\ including a boldface ring.
  
  $G_1$ and $G_2$ are finite process graphs that are not bisimilar to the process interpretation of any star expression.
  In this sense, $G_1$ and $G_2$ are \emph{not} expressible by a regular expression (under the process semantics). 
  This sets the process semantics of regular expressions apart from the standard language semantics,
  with respect to which every language that is accepted by a finite-state automaton is the interpretation of some regular expression.
  
  Furthermore it is easy to see that both of $G_1$ and $G_2$ do not satisfy \LEE.
    Namely, both of these process graphs do not contain loop subcharts, but each of them represents an infinite behavior.
    Therefore the loop elimination procedure stops immediately on either of them, unsuccessfully.
\end{exa}

\begin{exa}\label{exa:charts:not:expressible:a}
  While the specifications $\fap{\aspec}{G_1}$ and $\fap{\aspec}{G_2}$  in Example~\ref{exa:charts:not:expressible}
  are not solvable by star expressions in the process semantics,
  this situation changes drastically if all actions in $\fap{\aspec}{G_1}$ and $\fap{\aspec}{G_2}$
    are replaced by a single action $\aact$.
  Then we obtain the following process graphs $G_1^{(a)}$ and $G_2^{(a)}$ with appertaining specifications   
    $\fap{\aspec}{G_1^{(a)}}$ and $\fap{\aspec}{G_2^{(a)}}$:
  \begin{center}\label{fig:milner-bosscher-expressible-a}
    \input{figs/cexs-milner-bosscher-a.tex}
  \end{center}
  These specifications are solvable
    by setting $X_1 \defdby X_2 \defdby X_3 \defdby a^* \prod 0$ in $\fap{\aspec}{G_1^{(a)}}$,
    and by setting $Y_1 \defdby Y_2 \defdby a^*$ in $\fap{\aspec}{G_2^{(a)}}$,
  because the following identities are provable in Milner's system $\milnersys$ (which in any case is sound for the process semantics):
  \begin{align*}
    a^* \prod 0 
      & \begin{aligned}[t]
            & \milnersyseq
          (1 + a \prod a^*) \prod 0 
              \milnersyseq
           1 \prod 0 + (a \prod a^*) \prod 0
              \milnersyseq
           0 + a \prod (a^* \prod 0)  \punc{,}    
          \\  
            & \milnersyseq
          a \prod (a^* \prod 0) 
              \milnersyseq
          a \prod (a^* \prod 0) + a \prod (a^* \prod 0) \punc{.}    
        \end{aligned}
    &
    a^* 
      & {} \milnersyseq
        1 + a \prod a^*
  \end{align*}  
  From these \provablein{\milnersys} identities the \provablein{\milnersys} correctness conditions for these settings
    follow directly, for example for $X_1$ in $\fap{\aspec}{G_1^{(a)}}$,
    and for $Y_1$ in $\fap{\aspec}{G_2^{(a)}}$:
  \begin{align*}
    X_1 & {} 
      \synteq
        a^* \prod 0
      \milnersyseq
        a \prod (a^* \prod 0) + a \prod (a^* \prod 0) 
      \synteq
        a \prod X_2 + a \prod X_3  \punc{,} 
    &
    Y_1 & {}
      \synteq
        a^*
      \milnersyseq
        1 + a \prod a^*
      \synteq
        1 + a \prod Y_2 \punc{.}
  \end{align*}
  These \provablein{\milnersys} identities show, together with the analogous ones for $X_2$, $X_3$, and $Y_2$
    that the settings
      $X_1 \defdby X_2 \defdby X_3 \defdby a^* \prod 0$ 
        and 
      $Y_1 \defdby Y_2 \defdby a^*$ 
        define solutions of $\fap{\aspec}{G_1^{(a)}}$, and $\fap{\aspec}{G_2^{(a)}}$, respectively,
  because \provablein{\milnersys} identities are also identities with respect to the process semantics (as Milner's system is sound for the process semantics).
    
  The crucial reason why we have obtained solvable specifications $\fap{\aspec}{G_1^{(a)}}$ and $\fap{\aspec}{G_2^{(a)}}$
    from unsolvable specifications $\fap{\aspec}{G_1}$ and $\fap{\aspec}{G_2}$, respectively,
  is that the underlying process graphs $\fap{\aspec}{G_1^{(a)}}$ and $\fap{\aspec}{G_2^{(a)}}$ 
    are not bisimulation collapses, in contrast with the process graphs $G_1$ and $G_2$ from which they arose by renaming all actions to a single one. 
  Indeed the bisimulation collapses $G_{10}$ of $G_1$, and $G_{20}$ of $G_2$ 
    are of particularly easy form that are obviously solvable by the star expressions $a^* \prod 0$, and $a^*$, respectively.
  These solutions can be transferred backwards over the functional bisimulations
    $G_1^{(a)} \funbisim G_{10}^{(a)}$, and $G_2^{(a)} \funbisim G_{20}^{(a)}$
  (indicated by magenta links in the pictures below)
  according to Lem.~\ref{lem:preservation:sols}, \ref{it:conv:funbisims:lem:preservation:sols}.\vspace{-1.5ex}
  \begin{center}\label{fig:milner-bosscher-expressible-a-collapse}
    \input{figs/cexs-milner-bosscher-a-collapse.tex}
  \end{center}
  in order to obtain the \provablein{\milnersys} solutions described above.
  
  This example witnesses the result from \cite{grab:fokk:2020:lics,grab:fokk:2020:lics:arxiv} 
    that a process graph $G$ without \onetransitions\ is expressible by a \onefree\ star expression in the process semantics
      if the bisimulation collapse of $G$ has the property \LEE. 
  Here the bisimulation collapses $G_{10}^{(a)}$ and $G_{20}^{(a)}$ satisfy \LEE,
    although the process graphs from which they arise by collapse, $G_{1}^{(a)}$ and $G_{2}^{(a)}$, do not satisfy \LEE.
\end{exa}

%% file: figs/cexs-milner-bosscher.tex
\begin{tikzpicture}
  
%==========================================
% not expressible chart 1 (no-loop example)
%========================================== 
%
\matrix[anchor=center,row sep=0.924cm,column sep=0.75cm,
        every node/.style={draw,very thick,circle,minimum width=2.5pt,fill,inner sep=0pt,outer sep=2pt}] at (0,0) {
                   &                  &  \node(C-1-2){};
  \\
  \node(C-1-1){};  &                  &  \node[draw=none,fill=none](helper){};               
  \\
                   &                  &  \node(C-1-3){};  
  \\
};
\draw[<-,very thick,>=latex,color=chocolate](C-1-1) -- ++ (180:0.5cm);  % ($(C-1-1)+(-0.125cm,0cm)$)

\path(C-1-1) ++ (-0.325cm,0.5cm) node{\Large $\iap{G}{1}$};
\path(C-1-1) ++ (-0.2cm,-0.4cm) node{$X_1$};
\draw[->,bend right,distance=0.65cm] (C-1-1) to node[above]{$\aacti{2}$} (C-1-2); 
\draw[->,bend right,distance=0.65cm] (C-1-1) to node[left]{$\aacti{3}$}  (C-1-3);

\path(C-1-2) ++ (0cm,0.4cm) node{$X_2$};
\draw[->,bend right,distance=0.65cm]  (C-1-2) to node[above]{$\aacti{1}$} (C-1-1); 
\draw[->,bend left,distance=0.65cm]  (C-1-2) to node[right,xshift=-1pt]{$\aacti{3}$} (C-1-3);

\path(C-1-3) ++ (0cm,-0.4cm) node{$X_3$};
\draw[->,bend right,distance=0.65cm] (C-1-3) to node[left]{$\aacti{1}$}  ($(C-1-1)+(+0.15cm,-0.05cm)$);
\draw[->,bend left,distance=0.65cm]  (C-1-3) to node[right,xshift=-1.5pt]{$\aacti{2}$} (C-1-2);

\path (helper) ++ (1.25cm,0cm) node[right]{$
  \fap{\aspec}{G_1} = 
  \left\{\,
  \begin{aligned}
    X_1 & {} = a_2 \prod X_2  +  a_3 \prod X_3
    \\ 
    X_2 & {} = a_1 \prod X_1  +  a_3 \prod X_3
    \\
    X_3 & {} = a_1 \prod X_1  +  a_2 \prod X_2
  \end{aligned}
  \,\right.
  $};

%================================================
% not expressible chart 2 (double-exit iteration)
%================================================ 
%
\matrix[anchor=center,row sep=0cm,column sep=1.5cm,
        every node/.style={draw,very thick,circle,minimum width=2.5pt,fill,inner sep=0pt,outer sep=2pt}] at (9.5,0) {
  \node[color=chocolate](C-2-0){};  &   \node[color=chocolate](C-2-1){};
  \\
};
\draw[<-,very thick,>=latex,chocolate,shorten <=2pt](C-2-0) -- ++ (180:0.58cm);
\path(C-2-0) ++ (0cm,-0.4cm) node{$Y_1$};
\path(C-2-0) ++ (-0.325cm,0.5cm) node{\Large $\iap{G}{2}$};

% \draw[thick] (C-2-2) circle (0.12cm);
\path (C-2-1) ++ (0cm,-0.45cm) node{$Y_2$};

\draw[thick,chocolate] (C-2-1) circle (0.12cm);
\draw[thick,chocolate] (C-2-0) circle (0.12cm);
\draw[->,bend left,distance=0.65cm,shorten <=2pt,shorten >=2pt] (C-2-0) to node[above]{$\aact$} (C-2-1); 
\draw[->,bend left,distance=0.65cm,shorten <=2pt,shorten >=2pt] (C-2-1) to node[below]{$\bact$} (C-2-0); 

% \draw[->,bend right,distance=0.6cm,shorten >=2pt] (C-2-0) to node[left]{$b$} (C-2-2);
% \draw[->,bend left,distance=0.6cm,shorten >=2pt] (C-2-1) to node[right]{$c$} (C-2-2);

\path (C-2-1) ++ (0.85cm,0cm) node[right]{$
  \fap{\aspec}{G_2} = 
  \left\{\,
  \begin{aligned}
    Y_1 & {} = 1 + a \prod Y_2
    \\ 
    Y_2 & {} = 1 + b \prod Y_1
  \end{aligned}
  \,\right.
  $};

\end{tikzpicture}

%% file: figs/cexs-milner-bosscher-a.tex
\begin{tikzpicture}
  
%==========================================
% not expressible chart 1 (no-loop example)
%========================================== 
%
\matrix[anchor=center,row sep=0.924cm,column sep=0.75cm,
        every node/.style={draw,very thick,circle,minimum width=2.5pt,fill,inner sep=0pt,outer sep=2pt}] at (0,0) {
                   &                  &  \node(C-1-2){};
  \\
  \node(C-1-1){};  &                  &  \node[draw=none,fill=none](helper){};               
  \\
                   &                  &  \node(C-1-3){};  
  \\
};
\draw[<-,very thick,>=latex,color=chocolate](C-1-1) -- ++ (180:0.5cm);  % ($(C-1-1)+(-0.125cm,0cm)$)

\path(C-1-1) ++ (-0.325cm,0.5cm) node{\Large $\bpap{G}{1}{(a)}$};
\path(C-1-1) ++ (-0.2cm,-0.4cm) node{$X_1$};
\draw[->,bend right,distance=0.65cm] (C-1-1) to node[above]{$\aact$} (C-1-2); 
\draw[->,bend right,distance=0.65cm] (C-1-1) to node[left]{$\aact$}  (C-1-3);

\path(C-1-2) ++ (0cm,0.4cm) node{$X_2$};
\draw[->,bend right,distance=0.65cm]  (C-1-2) to node[above]{$\aact$} (C-1-1); 
\draw[->,bend left,distance=0.65cm]  (C-1-2) to node[right,xshift=-1pt]{$\aact$} (C-1-3);

\path(C-1-3) ++ (0cm,-0.4cm) node{$X_3$};
\draw[->,bend right,distance=0.65cm] (C-1-3) to node[left]{$\aact$}  ($(C-1-1)+(+0.15cm,-0.05cm)$);
\draw[->,bend left,distance=0.65cm]  (C-1-3) to node[right,xshift=-1.5pt]{$\aact$} (C-1-2);

\path (helper) ++ (1cm,0cm) node[right]{$
  \fap{\aspec}{G_1^{(a)}} = 
  \left\{\,
  \begin{aligned}
    X_1 & {} = a \prod X_2  +  a \prod X_3
    \\ 
    X_2 & {} = a \prod X_1  +  a \prod X_3
    \\
    X_3 & {} = a \prod X_1  +  a \prod X_2
  \end{aligned}
  \,\right.
  $};

%================================================
% not expressible chart 2 (double-exit iteration)
%================================================ 
%
\matrix[anchor=center,row sep=0cm,column sep=1.5cm,
        every node/.style={draw,very thick,circle,minimum width=2.5pt,fill,inner sep=0pt,outer sep=2pt}] at (9.5,0) {
  \node[color=chocolate](C-2-0){};  &   \node[color=chocolate](C-2-1){};
  \\
};
\draw[<-,very thick,>=latex,chocolate,shorten <=2pt](C-2-0) -- ++ (180:0.58cm);
\path(C-2-0) ++ (0cm,-0.4cm) node{$Y_1$};
\path(C-2-0) ++ (-0.325cm,0.5cm) node{\Large $\bpap{G}{2}{(a)}$};

% \draw[thick] (C-2-2) circle (0.12cm);
\path (C-2-1) ++ (0cm,-0.45cm) node{$Y_2$};

\draw[thick,chocolate] (C-2-1) circle (0.12cm);
\draw[thick,chocolate] (C-2-0) circle (0.12cm);
\draw[->,bend left,distance=0.65cm,shorten <=2pt,shorten >=2pt] (C-2-0) to node[above]{$a$} (C-2-1); 
\draw[->,bend left,distance=0.65cm,shorten <=2pt,shorten >=2pt] (C-2-1) to node[below]{$a$} (C-2-0); 

% \draw[->,bend right,distance=0.6cm,shorten >=2pt] (C-2-0) to node[left]{$b$} (C-2-2);
% \draw[->,bend left,distance=0.6cm,shorten >=2pt] (C-2-1) to node[right]{$c$} (C-2-2);

\path (C-2-1) ++ (0.65cm,0cm) node[right]{$
  \fap{\aspec}{G_2^{(a)}} = 
  \left\{\,
  \begin{aligned}
    Y_1 & {} = 1 + a \prod Y_2
    \\ 
    Y_2 & {} = 1 + a \prod Y_1
  \end{aligned}
  \,\right.
  $};

\end{tikzpicture}

%% file: figs/cexs-milner-bosscher-a-collapse.tex
\begin{tikzpicture}
  
%==========================================
% not expressible chart 1 (no-loop example)
%========================================== 
%
\matrix[anchor=center,row sep=0.924cm,column sep=0.75cm,
        every node/.style={draw,very thick,circle,minimum width=2.5pt,fill,inner sep=0pt,outer sep=2pt}] at (0,0) {
                   &                  &  \node(C-1-2){};
  \\
  \node(C-1-1){};  &                  &  \node[draw=none,fill=none](helper){};  &  &  & \node(C-10){};           
  \\
                   &                  &  \node(C-1-3){};  
  \\
};
\draw[<-,very thick,>=latex,color=chocolate](C-1-1) -- ++ (180:0.5cm);  % ($(C-1-1)+(-0.125cm,0cm)$)
\draw[<-,very thick,>=latex,color=chocolate](C-10) -- ++ (90:0.5cm);  % ($(C-1-1)+(-0.125cm,0cm)$)

\path(C-1-1) ++ (-0.325cm,0.5cm) node{\Large $\bpap{G}{1}{(a)}$};
\path(C-1-1) ++ (-0.2cm,-0.4cm) node{$X_1$};
\draw[->,bend right,distance=0.65cm] (C-1-1) to node[above]{$\aact$} (C-1-2); 
\draw[->,bend right,distance=0.65cm] (C-1-1) to node[left]{$\aact$}  (C-1-3);

\path(C-1-2) ++ (0cm,0.4cm) node{$X_2$};
\draw[->,bend right,distance=0.65cm]  (C-1-2) to node[above]{$\aact$} (C-1-1); 
\draw[->,bend left,distance=0.65cm]  (C-1-2) to node[right,xshift=-1pt]{$\aact$} (C-1-3);

\path(C-1-3) ++ (0cm,-0.4cm) node{$X_3$};
\draw[->,bend right,distance=0.65cm] (C-1-3) to node[left]{$\aact$}  ($(C-1-1)+(+0.15cm,-0.05cm)$);
\draw[->,bend left,distance=0.65cm]  (C-1-3) to node[right,xshift=-1.5pt]{$\aact$} (C-1-2);

\path(C-10) ++ (0cm,-0.4cm) node{$X$};
\path(C-10) ++ (0.55cm,0.5cm) node{\Large $\bpap{G}{10}{(a)}$};

\draw[->,out=-45,in=225,distance=1.75cm,shorten <=2pt,shorten >=2pt] (C-10) to node[below]{$\aact$} (C-10);

\path (C-10) ++ (2cm,-0.3cm) node{\parbox{\widthof{provable~}}
                                      {\hspace*{\fill}$a^* \prod 0$ is\hspace*{\fill}\mbox{}%
                                         \\[-0.25ex]
                                       \hspace*{\fill}provable\hspace*{\fill}\mbox{}%
                                         \\[-0.25ex]
                                       \hspace*{\fill}\mbox{}solution\hspace*{\fill}\mbox{}}};

\path (C-10) ++ (0cm,-2cm) node{$
  \fap{\aspec}{G_{10}^{(a)}} = 
  \left\{\,
    X = a \prod X
  \,\right\} 
  $};

\draw[magenta,densely dashed,thick,shorten <=3pt,shorten >=3pt,out=80,in=120,distance=2.5cm] (C-1-1) to (C-10);
\draw[magenta,densely dashed,thick,shorten <=3pt,shorten >=3pt,out=-10,in=170,distance=1cm] (C-1-2) to (C-10);
\draw[magenta,densely dashed,thick,shorten <=3pt,shorten >=3pt,out=10,in=190,distance=1cm] (C-1-3) to (C-10);

%================================================
% not expressible chart 2 (double-exit iteration)
%================================================ 
%
\matrix[anchor=center,row sep=0cm,column sep=1.5cm,
        every node/.style={draw,very thick,circle,minimum width=2.5pt,fill,inner sep=0pt,outer sep=2pt}] at (8.75,0) {
  \node[color=chocolate](C-2-0){};  &   \node[color=chocolate](C-2-1){};  &[0.5cm] \node[chocolate](C-20){}; 
  \\
};
\draw[<-,very thick,>=latex,chocolate,shorten <=2pt](C-2-0) -- ++ (180:0.58cm);
\draw[<-,very thick,>=latex,chocolate,shorten <=2pt](C-20) -- ++ (90:0.58cm);
\path(C-2-0) ++ (0cm,-0.4cm) node{$Y_1$};
\path(C-2-0) ++ (-0.325cm,0.5cm) node{\Large $\bpap{G}{2}{(a)}$};

% \draw[thick] (C-2-2) circle (0.12cm);
\path (C-2-1) ++ (0cm,-0.45cm) node{$Y_2$};

\path(C-20) ++ (0cm,-0.4cm) node{$Y$};
\path(C-20) ++ (0.55cm,0.5cm) node{\Large $\bpap{G}{20}{(a)}$};

\draw[thick,chocolate] (C-2-1) circle (0.12cm);
\draw[thick,chocolate] (C-2-0) circle (0.12cm);
\draw[thick,chocolate] (C-20) circle (0.12cm);

\draw[->,bend left,distance=0.65cm,shorten <=2pt,shorten >=2pt] (C-2-0) to node[above]{$\aact$} (C-2-1); 
\draw[->,bend left,distance=0.65cm,shorten <=2pt,shorten >=2pt] (C-2-1) to node[below]{$\aact$} (C-2-0); 

\draw[->,out=-45,in=225,distance=1.75cm,shorten <=2pt,shorten >=2pt] (C-20) to node[below]{$\aact$} (C-20);

\draw[magenta,densely dashed,thick,out=65,in=150,distance=1.5cm,shorten <=3pt,shorten >=3pt] (C-2-0) to (C-20);
\draw[magenta,densely dashed,thick,out=20,in=160,distance=0.5cm,shorten <=3pt,shorten >=3pt] (C-2-1) to (C-20);

\path (C-20) ++ (2cm,-0.3cm) node{\parbox{\widthof{provable solution}}
                                      {\hspace*{\fill}$a^*$ is\hspace*{\fill}\mbox{}%
                                       \\[-0.25ex]
                                       \hspace*{\fill}provable\hspace*{\fill}\mbox{}%
                                       \\[-0.25ex]
                                       \hspace*{\fill}solution\hspace*{\fill}\mbox{}}};

\path (C-20) ++ (0cm,-2cm) node{$
  \fap{\aspec}{G_{20}^{(a)}} = 
  \left\{\,
    Y = 1 + a \prod Y
  \,\right\} 
  $};

\end{tikzpicture}